\documentclass[reqno]{amsart}
\usepackage{amsthm,amsmath,amsfonts,amssymb,amscd,mathrsfs,graphics}
\usepackage{txfonts}
\usepackage{hyperref,supertabular}

\newcommand{\eb}{{\bar{1}}}
\newcommand{\rb}{{\bar{2}}}

\newcommand{\B}{{\mathscr B}}
\newcommand{\Bb}{{\bar{B}}}
\newcommand{\Bci}{{B^\circ_i}}

\newcommand{\C}{{\mathbb{C}}}

\newcommand{\D}{{\mathscr D}}
\newcommand{\Dom}{\operatorname{Dom}}

\newcommand{\End}{{\operatorname{End}}}

\newcommand{\Ib}{{\bar{I}}}

\newcommand{\Jb}{{\bar{J}}}

\newcommand{\LL}{{\mathscr L}} 

\newcommand{\PP}{{\mathscr P}}

\newcommand{\IP}{{\mathbb{P}}}
\newcommand{\Q}{{\mathbb{Q}}}

\newcommand{\R}{{\mathbb{R}}}
\newcommand{\rank}{{\operatorname{rank}}}
\newcommand{\Sp}{\operatorname{Sp}}

\newcommand{\U}{{\mathscr U}}
\newcommand{\Uct}{{{\mathscr U}_\tau^\circ}}
\newcommand{\Ub}{{\bar{\mathscr U}}}

\newcommand{\Z}{{\mathbb{Z}}}
\newcommand{\ab}{{\bar{a}}}
\newcommand{\ad}{{a^\dag}}
\newcommand{\adj}{\operatorname{ad}}
\newcommand{\abd}{{\bar{a}^\dag}}
\newcommand{\aut}{\operatorname{Aut}}

\newcommand{\bb}{{\bar{b}}}

\newcommand{\bnabla}{\bar{\nabla}}

\newcommand{\cb}{{\bar{c}}}
\newcommand{\cC}{\mathscr C} 

\newcommand{\chat}{\hat{c}}
\newcommand{\ch}{{\mathscr H}}

\newcommand{\db}{{\bar{d}}}
\newcommand{\diag}{\operatorname{diag}}

\newcommand{\fb}{{\bar{f}}}

\newcommand{\hstar}{{\hat{*}}}
\newcommand{\heta}{{\hat{\eta}}}
\newcommand{\hg}{{\hat{g}}}
\newcommand{\ib}{{\bar{i}}}

\newcommand{\im}{\operatorname{im}}

\newcommand{\jb}{{\bar{j}}}

\newcommand{\ka}{{\kappa}}

\newcommand{\kb}{{\bar{k}}}
\newcommand{\lb}{{\bar{l}}}

\newcommand{\mb}{{\bar{m}}}
\newcommand{\mub}{{\bar{\mu}}}

\newcommand{\nb}{{\bar{n}}}
\newcommand{\nub}{{\bar{\nu}}}

\newcommand{\om}{{\omega}}
\newcommand{\pb}{{\bar{p}}}

\newcommand{\qb}{{\bar{q}}}
\newcommand{\re}{\operatorname{Re}}
\newcommand{\res}{\operatorname{res}}

\newcommand{\taub}{{\bar{\tau}}}
\newcommand{\tb}{{\bar{t}}}

\newcommand{\wt}{\operatorname{wt}}

\renewcommand{\O}{{\mathscr{O}}}
\renewcommand{\hom}{\operatorname{Hom}}

\newcommand{\pat}{{\partial}}
\newcommand{\bpat}{{\bar{\partial}}}

\newtheorem{thm}{Theorem}[section]
\newtheorem{lm}[thm]{Lemma}
\newtheorem{prop}[thm]{Proposition}
\newtheorem{crl}[thm]{Corollary}
\newtheorem{prob}[thm]{Problem}

\theoremstyle{definition}
\newtheorem{rem}[thm]{Remark}

\newtheorem{df}[thm]{Definition}
\newtheorem{ex}[thm]{Example}
\theoremstyle{remark}
 
\title[Schr\"odinger equation and $tt^*$-geometry]
{Schr\"odinger equations, deformation theory and $tt^*$-geometry }
\author[Huijun Fan]{Huijun Fan\\ \\School of Mathematical Sciences, \\Peking University\\ and \\
{\\Beijing International Center \\for Mathematical Research,
Beijing, China}}
\address{School of Mathematical Sciences, Peking University, Beijing
100871, China}
\thanks{Partially Supported by NSFC 10631050\\} \email{fanhj@math.pku.edu.cn}
\date{\today}
\begin{document}
\maketitle
\begin{abstract} This is the first of a series of papers to construct the
deformation theory of the form Schr\"odinger equation, which is
related to a section-bundle system $(M,g,f)$, where $(M,g)$ is a
noncompact complete K\"ahler manifold with bounded geometry and $f$
is a holomorphic function defined on $M$. The deformation of $f$
will induce the deformation of the corresponding Sch\"odinger
equation. We will prove that the Schr\"oinder operator has purely
discrete spectrum if the section-bundle system $(M,g,f)$ is strongly
tame. As the conclusion, we can obtain the Hodge theory, including
the Hodge decomposition theorem and the Hard Lefschetz theorem. If
the manifold is Stein, then the $L^2$ cohomology of $\bpat_f$ can be
computed explicitly and only the middle dimensional homology
information is preserved.

If the deformation $\tau f_t$ of $f$ is strong, then we obtain a
Hodge bundle over the deformation space. Several operators on the
Hodge bundle are defined and they can be assembled into a total
connection $\D$. The integrability of $\D$ is obtained by proving
the Cecotti-Vafa's equation ($tt^*$ equation) and the "Fantastic
equation" found in this paper, which describes the behavior along
the coupling constant $\tau$ direction. We can also define a family
of flat connections and prove they satisfy some interesting
identities.

The strong deformation of a strongly tame system includes the following interesting examples:
\begin{itemize}
\item Marginal and relevant deformation of $(\C^n, \sqrt{-1}\sum_i dz^i\wedge dz^\ib, W)$, where $W$ is a non-degenerate
quasihomogeneous polynomial, e.g., $W$ being the quintic polynomials

\item Subdiagram deformation of $((\C^*)^n,\sqrt{-1}\sum_i\frac{dz^i}{z^i}\wedge \frac{dz^\ib}{z^\ib},f)$, where $f$ is a non-degenerate
 and convenient Laurent polynomials, e.g., $f=z_1+\cdots+z_n+\frac{e^{-t}}{z_1\cdots z_n}$.

\item The miniversal deformation of the simple singularities $A,D,E$ and the parabolic singularities $X_8, J_9, P_{10}$ in Arnold's singularity list.

\end{itemize}
As the simple applications, we can get the mixed Hodge structure,
the isomonodromic deformation of o.d.e. with respect to $\tau$ and
the Frobenius manifold structure on the deformation space via a
primitive vector which is defined by the oscillating integral along
the Lefschetz thimbles.

This work is also the first step attempting to understand the whole
Landau-Ginzburg B-model including the higher genus invariants. Our
work is mainly based on the pioneer work of Cecotti, Cecotti and
Vafa \cite{Ce1,Ce2,CV}.

\end{abstract}

\tableofcontents
\section{Introduction}
\subsection{LG model and $tt*$ geometry: the origin in physics}\

\

The $N=2$ superconformal field theories have important role in the
study of string theory since they can be taken as the exactly
soluble toy model and meanwhile as the reduction of the higher
dimensional quantum field theories. The corresponding lagrangian of
the quantum field theory is invariant under the action of the
$N=(2,2)$ superconformal group whose Super Lie algebra consists of
three parts: the even part generated by $H,P,M$, the Noether charges
for the time translation, the $1$-dimensional space translation and
the Lorentz rotation; the odd part generated by the supercharges:
$Q_+,Q_-,\bar{Q}_+,\bar{Q}_-$; the Noether charges $F_V, F_A$ of
vector $U(1)$ and axial $U(1)$ rotations. They satisfy the following
commutation relations (in the sense of super Lie algebra):
\begin{align*}
&Q_+^2=Q_-^2=\bar{Q}_+^2=\bar{Q}_-^2=0\\
&[Q_\pm,\bar{Q}_\pm]=H\pm P\\
&[\bar{Q}_+,\bar{Q}_-]=[Q_+,Q_-]=0\\
&[Q_-,\bar{Q}_+]=[Q_+,\bar{Q}_-]=0\\
&[iM,Q_\pm]=\mp Q_{\pm},[iM,\bar{Q}_\pm]=\mp \bar{Q}_\pm\\
&[iF_V, Q_\pm]=-iQ_\pm,[iF_V,\bar{Q}_\pm]=i\bar{Q}_\pm\\
&[iF_A,Q_{\pm}]=\mp iQ_\pm,[iF_A,\bar{Q}_\pm]=\pm i\bar{Q}_\pm.
\end{align*}
Here all the operators acting on the fock space (Hilbert space), and
except that $Q^\dag_\pm=\bar{Q}_\pm$ are conjugate to each other,
all other operators are Hermitian. The symmetries will be broken if
the Lagrangian does not preserve the corresponding invariance. The
Calabi-Yau $\sigma$-model (ref.\cite{Ka} from the view of
physicists) viewed as a geometrical realization of the $N=2$
superconformal algebra has been studied intensively by physicists
and mathematicians in the past twenty years. Because of the $\Z_2$
outer automorphism of this algebra, the geometrical realization of
this algebra induces the mirror symmetry phenomenon: The A-model
theory in one Calabi-Yau manifold $M$ should match the B-model
theory of another mirror Calabi-Yau manifold $\check{M}$. The $A$
model theory, the $B$-model theory and their mirror relations have
become a main subject in mathematics since $80$'s (for instance, see
the series of publications \cite{Ya}). In each model, there are
versions of closed string or open string theory. Here we only talk
about theory about closed string. The $A$ model theory has been
well-understood so far. It concerns the study of pseudo-holomorphic
curves in a symplectic manifold, in another word, studying the
moduli problem of the solutions of the following Cauchy-Riemann
equation:
\begin{equation}
\bpat_J u=0,
\end{equation}
where $u: (\Sigma,j)\to (M,\omega,J)$, $(\Sigma,j)$ is a Riemann
surface with complex structure $j$ and $(M,\om,J)$ is a symplectic
manifold with symplectic form $\om$ and compatible almost complex
structure $J$. The virtual fundamental cycle has been constructed
(ref. \cite{LiT, FO, Ru}) and one can define the correlation
functions (Gromov-Witten invariants) and the generating function and
get the A model theory. This model is a topological field theory (or
cohomological field theory ) in which the observables are
represented by the cohomological class of the target manifold $M$.
If $M$ is a K\"ahler manifold, the $A$ model is supposed to be
related the deformation of symplectic structures and the $B$ model
is supposed to be related to the deformation of the complex
structure. The Kodaira-Spencer theory \cite{Ko} characterized the
moduli germ of the complex structure of a compact complex manifold.
So far the geometrical construction of the correlation functions for
the higher genus ($\ge 2$) has not been finished. The genus $1$ part
and the holomorphic anomaly equations have been discussed by
\cite{BCOV2}. The genus $0$ part (or Frobenius manifold structure)
was given by Kontsevich-Barannikov construction \cite{BK}. The genus
$1$ invariants has been defined by H. Fang, Z. Lu and K. Yoshikawa
\cite{FLY}. A. Zinger \cite{Zi} and Fang-Lu-Yoshikawa have also
solved a conjecture of BCOV about the quintic threefold. A most
recent progress has been made by K. Castello and S. Li \cite{CL}.

However, the computation of the invariants in Calabi-Yau model is
rather difficult. Actually physicists have known anther effective
way to calculate the invariants much earlier. This is the $N=2$ SUSY
Landau-Ginzburg model. In the paper \cite{Ge}, Gepner found a
mysterious correspondence between some of the superconformal models
that he had obtained by taking the tensor product of exactly soluble
minimal models, and some specific Calabi-Yau manifolds. Further work
\cite{CGP1, CGP2, LVW,Mar,GVW} show that it is not an occasional
correspondence, and it is really a correspondence between
Landau-Ginzburg model and Calabi-Yau model.

The simplest supersymmetric Lagrangian of the Landau-Ginzburg model
has the following form
\begin{equation}
\int dz^2 d^4\theta K(\bar{\Phi_i,\Phi_i})+\left(\int d^2 z
d^2\theta W(\Phi_i)+c.c.\right),
\end{equation}
where the first term is called the $D$-term and the second term is
called the $F$-term and is given by a quasi-homogeneous holomorphic
function $W$ which is the function of the superfields $\Phi_i$. Like
CY $\sigma$ model, LG model has also A model theory and B model
theory and the mirror symmetry phenomena between them. A physical
description has been given by Witten \cite{Wi4}, Guffin and Sharpe
\cite{GS}. The A model is related to the moduli problem of the
following nonlinear Cauchy-Riemann equation:
$$
\bpat u+\overline{\frac{\pat W}{\pat u}}=0,
$$
where $u:\Sigma\to M$ is a map from the Riemann surface to $M$
(where the target K\"ahler manifold $M$ is assumed to have $U(1)$
symmetry which is the requirement of supersymmetry (ref.
\cite{Del})).

However, this equation is not well-defined on Riemann surface with
any genus. For $M=\C^n$, this equation was defined via some line
bundle structure (called $W$-structure) on Riemann surface by
Witten. This equation for $W=x^r$ case has been used by Witten
\cite{Wi1} to construct the virtual cycle on the moduli space of
$r$-spin curves, an attempt to generalize his conjecture \cite{Wi}
(later solved by Kontsevich and called Witten-Kontsevich theorem) to
$A,D,E$ cases. A more general definition has been obtained very
recently by Witten coupling with gauge fields \cite{Wi2}. We call
this globally defined equation as Witten equation. The moduli
problem of Witten equation has been studied by Fan-Jarvis-Ruan and
the corresponding cohomological field theory, quantum singularity
theory, has been constructed in \cite{FJR2,FJR3}. They also proved the
generalized Witten conjecture for $A, D, E$ cases (A case solved
earlier by Farber-Shadrin-Zvonkine \cite{FSZ}). The CY/LG
correspondence for A model has been checked in some cases by
Choido-Ruan \cite{ChR1,ChR2}.

The B side of the LG model has intimate relation with the
singularity theory (ref. \cite{AGV}). Singularity theory (or
catastrophe theory) studies the germ of the singularity of a
holomorphic functions. To get the information of the singularity,
one must use the deformation of the given holomorphic function.
Based on the deformation parameter space, the topological and
geometrical structure of a singularity appear. Similar things
happened in physics. Zamolodchikov \cite{Zam} use the deformed
massive Conformal field theory to study the conformal one. Like in
CY model, the observables, correlations functions and partition
function occupy the central position of LG model. In the papers,
\cite{CGP1,CGP2,Ce1,Ce2}, the Chiral ring structure has been
studied. In particular, in \cite{Ce1}, \cite{CV}, the $tt^*$
equations has been founded which induces an integrable connection on
the vacuum bundle. The discovery is remarkable, and we call this
equations as Cecotti-Vafa's equation. A realization of the CV
equations is to use the supersymmetric quantum mechanics. Given a
holomorphic function $f$, one can obtain the twisted operator
$\bpat_f=\bpat+\pat f\wedge$. Like the study of $\bpat$ operator in
compact K\"ahler manifold, the operator $\bpat_f$ induces a series
of operators $\pat_f, \bpat_f^\dag,
\Delta_f=\bpat_f\bpat_f^\dag+\bpat_f^\dag\bpat_f...$. This will
induce the Hodge theory and Hard Lefschetz theorem as in compact
K\"ahler manifolds. Those ingredients were conjectured  by these
physicists. In Cecotti and Vafa's paper, \cite{CV}, Page 41, they
speculate the existence of a generalized special geometry (compare
to the Variation of Hodge structure in compact K\"ahler manifolds).

However, there is a fundamental problem not solved hidden in their
formulation:
\begin{prob}
Is the zero spectrum of these Schr\"odinger operators a discrete
spectrum?
\end{prob}

Since $f$ is a holomorphic function defined on $\C^n$, $\C^n$ is a
complete noncompact manifold, it is not necessary that $0$ is the
discrete spectrum of the minimal self-adjoint extension of the
schr\"odinger operator $\Delta_f$. It is quite easy to find an
counterexample such that $zero$ is a continuous spectrum. In this
case, the Hodge decomposition theorem does not hold in the usual
form as in the compact K\"ahler manifold and then all the
conclusions based on Hodge decomposition theorem are ineffective.

Furthermore, to get the Cecotti-Vafa's equation, one has to
frequently use taking derivatives with respect to the deformation
parameters and sometimes using Stokes integration formulas. To
justify all of these operations, one must study carefully the
analytic properties of the Schr\"odinger equation and its solutions.

The Cecotti-Vafa's equations is very important. It provided and will
provide many integrable systems and meanwhile provide the
(topological) solutions. On the other hand, the authors in
\cite{CFIV} have founded a new supersymmetric index $Tr(-1)^F
Fe^{-\tau H}$ for a system with Hamiltonian $H$ compared to the
Witten index $Tr(-1)^F e^{-\tau H}$ \cite{Wi} and this index as the
function of $\beta$ satisfies an differential equation and has tight
connection with the Cecotti-Vafa's equation. Some examples in
\cite{CFIV} show that the index will introduce some integrable
systems. The author actually asked the following question:
\begin{prob}
Explain why the new index satisfy the Cecotti-Vafa's equation while
satisfies an integrable differential equation (original: coupled
integral equations from TBA, \emph{thermodynamic Bethe Ansatz}, ref.
\cite{Zam}) with $\tau$ as the variable?
\end{prob}

\subsection{$tt^*$ geometry: progress in mathematics}\

\

Mathematician began to realize the importance of $tt^*$ equations
only ten years later after Cecotti-Vafa's discovery. In \cite{Het1},
C. Hertling found that the $tt^*$ equations and the induced
integrable structure are naturally related to the Frobenius manifold
structure, which was introduced by B. Dubrovin in the study of
topological field theory in the early of $90$'s \cite{Du2,Du,Du3}.
 Such a structure is a geometrical description
of some integrable systems, among them, the WDVV equations, the
isomonodromic deformation of linear differential systems\cite{Man}.
However, such structure was found more earlier, in the earlier of
$80$'s, by K. Saito in the study of the universal deformation of
hypersurface singularities \cite{Sai1,Sai2,Sai3, ST}. The flat
structure comes from the existence of the "primitive forms" called
by K. Saito. K. Saito solved the existence for some examples, and
the general cases were solved by M. Saito via solving the
Riemann-Hilbert problem on $\IP^1$ \cite{Sm,Het1}.

The key concepts in \cite{Het1} defined by Hertling after he studied
Cecotti-Vafa's work are the TERP structures (or TLEP structures). He
compared the trTERP or trTLEP structure with other known structures,
including Variation of twistor structure by C. Simpson \cite{Si},
the Cecotti-Vafa's structure, Frobenius type structure and the
classical Variation of Hodge structure and mixed Hodge structure.
The relation between CV structure and Frobenius type structure has
been studied by Sabbah\cite{Sa}. The combination of CV structure and
Frobenius type structure defined on the holomorphic tangent bundle
of a complex manifold will introduce the so called CDV
structure\cite{Het1}. Essentially, CDV structure is the Frobenius
manifold structure plus an extra compatible real structure defined
on the real tangent bundle. A most recent work on these structures
has been done by J. Lin \cite{Lin}.

In \cite{Het1}, Hertling has also constructed a TERP structure for
the germ of any hypersurface singularities by using the oscillating
integration. Consequently, the primitive forms used to define the
Frobenius manifold are just horizontal sections of the connection
and meanwhile are the eigenfunctions of the index operator (new
supersymmetric index by CFIV) and consists of a one-dimensional
space. This avoid to use M. Saito's existence technique.

Globally, the oscillating integrals and the Fourier-Laplace
transformation technique has been used by Sabbah\cite{Sa2,Sa4},
Douai\cite{Do} to construct the Frobenius manifold structure for a
holomorphic function defined on an affine manifold (with more or
less standard metric), in particular, for the non-degenerate
convenient Laurent polynomials defined on the algebraic torus. The
simplest among them is $z+\frac{e^\tau}{z}$, which is known to be
the LG mirror of the Fano manifold $\C P^1$.

So to my understanding, the $tt^*$ geometry arises from
Cecotti-Vafa's equation and developed to a TERP structure by
Hertling. It is a generalization (or a complement ) of VHS, twistor
structure and Frobenius manifold structure. To obtain a TERP
structure from a geometrical structure is a more difficult work.

Because of the essence of the $N=2$ superconformal algebra, the TERP
structure should be a universal property and is possible to be
defined on the abstract algebra. This construction has been done by
Iritani \cite{Ir} and should also exist on Fan-Jarvis-Ruan's quantum
singularity theory \cite{FJR2} and even other higher dimensional
models.

Though some progress has been achieved, in particular, in the study
of singularity theory. There are many problems untouched by
mathematicians so far. One of the most important side of CV's
equations is its relation to integrable systems. It is unlikely a
natural way for algebraic geometry to generate a new integrable
system, in particular an integrable partial differential equations
and write down the explicit solution and study its asymptotic
approximation formula at the degenerate point of the deformation
space. The problems arising from physics and mathematics in the
study of $tt^*$ geometry are the start point of our present
research.

\subsection{Our deformation theory of Schr\"odinger operators}\

\

This paper is the first of a series of papers attempting to
understand, explore and complement the geometrical structure
concerning the $tt^*$ equations that physicists have discovered. We
will construct the mathematical foundation of the deformation theory
of the supersymmetric quantum mechanics associated to a tuple
$(M,g,f)$. The tuple $(M,g,f)$ consists of a complete noncompact
K\"ahler manifold with metric $g$, which has bounded geometry (see
Section for an explanation), and a holomorphic function $f$ defined
on $M$. We give a name to it: section-bundle system. The reason is
that $(M,g,f)$ is equivalent to the tuple $(M\times \C, g\oplus g_E,
s)$, where $s=\tau f$ is the section of the trivial bundle $M\times
\C$, where $\tau$ is the coordinate of the fiber $\C$ and $g_E$ is
the standard K\"ahler metric on $\C$.  This viewpoint would be very
important for the later generalization \cite{FT}. For the section
bundle system $(M,g,f)$, one obtain a twist operator
$\bpat_f=\bpat+\pat f\wedge$ acting on the space $\Omega^k(M)$ of
smooth $k$-forms and its formal conjugate operator $\bpat^\dag_f$
with respect to the K\"ahler metric $g$. So we can define the
twisted Laplace operator $\Delta_f=\bpat^\dag_f \bpat_f+\bpat_f
\bpat^\dag_f$. These operators have closed extension to the space
$L^2(\Lambda^k(M))$ of $L^2$ $k$-forms. Notice that these operators
only keep the real grading of forms and do not keep the Hodge
grading. We require the restriction of the section-bundle system
$(M,g,f)$ to be strongly tame, which will be given in Definition
\ref{df:tameness}:

\begin{df} The section-bundle system is said to be strongly tame, if for any
constant $C>0$, there is
\begin{equation}
|\nabla f|^2-C|\nabla^2 f|\to \infty, \;\text{as}\;d(x, x_0)\to
\infty.
\end{equation}
Here $d(x,x_0)$ is the distance between the point $x$ and the base
point $x_0$.
\end{df}

For strongly tame section-bundle system $(M,g,f)$, we have the
fundamental theorem in this paper:

\begin{thm}[Theorem \ref{thm:main-1}]\label{thm:intro-1} Suppose that $(M,g)$ is a K\"ahler manifold with bounded
geometry. If $\{(M.g),f\}$ is a strongly tame section-bundle system,
then the form Laplacian $\Delta_f$ has purely discrete spectrum and
all the eigenforms form a complete basis of the Hilbert space
$L^2(\Lambda^\bullet(M))$.
\end{thm}

This theorem can deduce the following conclusions:

\begin{crl}[Theorem \ref{thm:main-hypersurface}\label{crl:intro-2} and Theorem \ref{thm:spec-Laurent}]
The section-bundle systems $(\C^n, i\sum_j dz^j\wedge dz^\jb, W)$
and $((\C^*)^n, i\sum_j \frac{dz^j}{z^j}\wedge
\frac{dz^\jb}{z^\jb},f)$ are strongly tame, if
\begin{itemize}
\item $W$ is a non-degenerate quasi-homogeneous polynomial with
homogeneous weight $1$ and of type $(q_1,\cdots,q_n)$ with all
$q_i\le 1/2$.
\item $f$ is a convenient and non-degenerate Laurent polynomial
defined on the algebraic torus $(\C^*)^n$.
\end{itemize}
Therefore, the corresponding form Laplacian has purely discrete
spectrum and all the eigenforms form a complete basis of the Hilbert
space $L^2(\Lambda^\bullet(M))$.
\end{crl}

Note that the conditions for $W$ and $f$ in the above corollary are
equivalent to the facts that the hypersurface defined by $\{W=0\}$
in the weighted projective space and the hypersurface defined by
$\{f=0\}$ in the toric variety are smooth. Therefore, the above
theorem says that for any smooth hypersurface in the weighted
projective spaces or in the toric varieties there is a corresponding
strongly tame section-bundle system such that the form Laplacian has
purely discrete spectrum.

Theorem \ref{thm:intro-1} is based on an important spectrum theorem,
Theorem \ref{thm:Kon-shubin} of Kondrat'ev-Shubin \cite{KS} for
scalar Schr\"odinger equation. We can prove that the form Laplacian
for strongly tame section-bundle system is the compact perturbation
of the scalar Schr\"odinger operator. So the scalar and the form
cases have the same spectrum properties.

Once we have the fundamental spectrum theorem, we can proceed as in
the case of compact K\"ahler manifold. We can obtain the Hodge
decomposition theorem, Theorem \ref{thm:Hodge}, and the Hard
Lefschetz theorem, Theorem
\ref{thm:hard-lefs-1}-\ref{thm:hard-lefs-2}. The proof of the Hard
Lefschetz theorem is due to the K\"ahler-Hodge identities and the
$N=2$ superconformal algebra structure for those twisted operators.

However, unlike in the compact case, we need compute the
$L^2$-cohomology, at best in terms of the topological data of the
section-bundle system. We have
\begin{thm}[Theorem \ref{crl:stein}] Let $(M,g)$ be a K\"ahler stein manifold with bounded
geometry and $(M,g,f)$ be strongly tame. If $f$ is a Morse function,
then
\begin{equation}
\dim \ch^k=
\begin{cases}
0,\;k<n\\
\mu,\;k=n.
\end{cases}
\end{equation}
and there is an explicit isomorphisms:
\begin{equation}
i_{0h}:\ch^n \to \Omega^n(M)/df\wedge\Omega^{n-1}(M).
\end{equation}
Here $\mu$ is the number of critical points of $f$.
\end{thm}

Then we go to the deformation theory in Section 3. We study the
strong deformation (Definition \ref{df:strong-deform}) of a strongly
tame section bundle system $(M,g,f)$. The strong deformation
contains the following cases:
\begin{itemize}
\item The marginal and relevant deformation of $(\C^n, i\sum_j dz^j\wedge
dz^\jb,W)$, where $W$ is the polynomial as in Corollary
\ref{crl:intro-2}. In particular, including the universal unfolding
of the simple singularities $A_n, D_n, E_6,E_7,E_8$, and the singularities $P_8, X_9, J_{10}$
in Arnold's list \cite{AGV}.
\item The deformation $((\C^*)^n, i\sum_j \frac{dz^j}{z^j}\wedge
\frac{dz^\jb}{z^\jb}, \tau f_t)$, where $f_t=f(z,t)$ is the
polynomial as in Corollary \ref{crl:intro-2}. In particular, this
including the deformation: $z_1+\cdots+z_n+\frac{e^{-t}}{z_1\cdots
z_n}$ which is the mirror of $\C P^n$ and its subdiagram
deformation.
\end{itemize}
Note that the above list contains many important cases:
\begin{itemize}
\item corresponds to the deformation theory of smooth hypersurfaces
in (weighted) projective space.
\item corresponds to the minimal model $A,D,E$ cases studied in
physics and in Givental's formal Gromov-Witten theory \cite{Gi}.
\item $P_8, X_9, J_{10}$ has central charge $1$ and corresponds to
the elliptic curves and modular forms.
\item corresponds to the LG model mirror to the toric Fano varieties
in toric geometry (ref. \cite{Bar}, \cite{HV}).
\end{itemize}
All of our conclusions in this paper can be applied to the strong
deformation of a strongly tame section-bundle system.

Section 3 is the analytic kernel of the deformation theory, the
similar role as in the compact deformation theory in projective
space. There are some difference. In our case, we have a fixed
background manifold and the deformation takes places when the
superpotential function changes and this makes the life easily than
in the compact case. On the other hand, since the manifold is
non-compact, one must treat all the related quantities in infinite
far place, i.e, treating the infinite far boundary. Usually this is
very difficult if the potential function has bad behavior at the
infinity. In our case, we have maximum principle which controls the
asymptotic behavior of the eigenforms, and the Green functions at
the infinity. The interior estimate can be obtained via the routine
technique from elliptic differential equation of second type. We
have the following conclusions:
\begin{itemize}
\item Any eigenform of the Schr\"odinger equation is exponential
decaying and the weighted $C^k$ norms which involving the potential
of its higher derivatives are exponential decay. This is given by
Proposition \ref{crl:eign-decay}. The similar conclusion hold for
Green function which is given by Proposition \ref{prop:
Green-func-decay}.
\item The asymptotic behavior of the Green function $G(z,w)$ as
$z\to  w$ is given by Proposition \ref{prop:Green-func-appr}.
\item The existence and the regularity of the solutions of the non-homogeneous
equation are given by Theorem \ref{thm:exist-non-homo}.
\end{itemize}

Based on those estimates, we obtain Theorem
\ref{thm:defor-conti-spectrum}, the continuity theorem of the
eigenvalues. The stability theorems are given by the
differentiability theorems, Theorem \ref{thm:stabi-resolv},
\ref{thm:stabi-proj} and \ref{thm:stabi-Green} which shows the
differentiability of the resolvent operators, projective operators
and the Green functions with respect to the deformation parameters.
In particular, this shows that the eigenforms are $C^\infty$
differentiable with respect to the deformation parameter if the
deformation is strong deformation.

Therefore Section 2 and 3 has answered the Problem 1.1 in the most
interesting cases and meanwhile proved the differentiability and
decaying estimate problems after Problem 1.1. This gives a rigorous
mathematical foundation for Hodge theory related to the differential
geometrical structure of the twisted operators. The decaying
estimate of eigenforms also allow us applying the $L^1$-stokes
theorem to eigenforms.

Now we go to Section 4. Let $f_{\tau,t}(z):=f_\tau(z):=\tau f(t,z)$
be a strong deformation defined on $\C^*\times S\times M$, where
$S\subset \C^m$ is a domain and $\C^*=\C-\{0\}$. Then there is a
family of deformation operators
$\Delta_{f_\tau},\bpat_{f_\tau},\pat_{f_\tau},\bpat_{f_\tau}^\dag,\pat_{f_\tau
}^\dag$ which depend on the parameter $(\tau,t)\in \C^*\times S$ and
the corresponding space of harmonic forms $\ch^\bullet_{\tau,t}$
depending on the parameter $(\tau,t)$. Therefore we obtain the Hodge
bundle $\ch^\bullet\to \C^*\times S$ which is the subbundle of the
trivial complex Hilbert bundle
$\Lambda^\bullet_\C:=L^2\Lambda^\bullet \times \C^*\times S\to
\C\times S$. We have the ordinary derivatives $\pat_\tau \alpha_a,
\pat_i \alpha_a, \pat_\taub \alpha_a, \pat_\ib \alpha_a$. The
covariant derivatives $D_\tau,...$ are defined as the projections of
the ordinary derivatives to the perpendicular subspace of the space
of harmonic forms. The multiplication operators (marginal operators)
$f,\pat_i f,\pat_\ib \fb$ are defined as the unbounded operators
acting on the Hilbert bundle and their restrictions to the Hodge
bundle are defined as the operator $\U, B_i, \Ub,\Bb_i$. The metric
on $\ch$ is naturally obtained by the background metric $g$. The
connection $D+\bar{D}$ is the Hermitian connection with respect to
$g$. The ordinary real structure $\tau_\R$ is compatible with $D$.
The commutation check of those operators induces the first family of
equation, \emph{\textbf{Cecotti-Vafa's equations}} (or $tt^*$
equation) given by Theorem \ref{thm:CV-equa},
\begin{align*}
&D_{\bar{i}}B_j=D_iB_\jb=0,\;D_iB_j=D_jB_i,\;D_{\bar{i}}B_\jb=D_{\bar{j}}B_\ib\\
&[D_i,D_j]=[D_{\bar{i}},D_{\bar{j}}]=[B_i,B_j]=[B_\ib,B_\jb]=0\\
&[D_i,D_{\bar{j}}]=-[B_i,B_\jb].
\end{align*}
However, we found that this is not the total story and some
information has been missed in the previous study of $tt^*$
structure (even in physics literatures), the information coming from
the coupling parameter $\tau$. Actually the related operators are
$f,\tau D_\tau, \U_\tau$ and their conjugates. After a carefully
check, we find the second family of equations given by Theorem
\ref{thm:Fan-eq}, which is called by us as \emph{\textbf{Fantastic
equations}},
\begin{align*}
&(1)\;[D_i,{\U_\tau}]+[B_i,\tau
D_\tau]=0,\;[D_\ib,\Ub_\tau]+[B_\ib,\bar{\tau}
D_{\bar{\tau}}]=0\\
&(2)\;[D_i,\Ub_\tau]=0,\;[D_\ib,{\U_\tau}]=0\\
&(3)\;[\taub D_{\bar{\tau}},B_i]
=[\tau D_{\tau},B_\ib]=0,\;\\
&(4)\;[B_i,{\U_\tau}]=0,\;[B_\ib,\Ub_\tau]=0\\
&(5)\;[\tau D_\tau,D_\ib]=-[{\U_\tau},B_\ib],\;[{\bar{\tau}}D_{\bar{\tau}},D_i]=-[\Ub_\tau,B_i]\\
&(6)\;[\tau D_\tau,D_i]=[\taub D_{\bar{\tau}},D_\ib]=0\\
&(7)\;[\tau D_\tau, \Ub_\tau]=[\bar{\tau}D_{\bar{\tau}},{\U_\tau}]=0\\
&(8)\;[\tau D_\tau,\bar{\tau}D_{\bar{\tau}}]=-[{\U_\tau},\Ub_\taub]\\
\end{align*}

In the same way as done in \cite{Ce1,CV}, a connection $\D$ combined
with all the operators $D_i,\tau D_\tau, B_i, \U_\tau$ and their
conjugates is obtained, which is a nearly flat connection by Theorem
\ref{thm:Tota-GM}, i.e., $\D^2=0 \mod \ch^\perp$. We call it as the
Gauss-Manin connection of the bundle $\ch\to \C^*\times S$. To
obtain a flat connection, we will transfer the structures on $\ch$
to another related bundle $\ch_{\ominus,top}$ in Section
\ref{subsec:4.3}.

From Section \ref{subsec:4.3}, one assumes that $M$ is a K\"ahler
stein manifold so that by Theorem \ref{crl:stein}, only the middle
dimensional (co)homology information is preserved.

Let $H^\ominus\to \C^*\times S$ be the relative homologic bundle
with fiber $H^\ominus_{(\tau,t)}=H_n(M,f^{-\infty}_{(\tau,t)},\C)$.
Then $\ch_{\ominus,top}\to \C^*\times S$ is the dual bundle of
$H^\ominus\to \C^*\times S$ with fiber
$\ch^{(\tau,t)}_{\ominus,top}=H^n(M,f^{-\infty}_{(\tau,t)},\C)$. The
intersection pairing, Poincare duality, Witten index, periodic
matrix and much concepts have been discussed in this section. We
define the bundle homomorphism $\psi_{\ominus}:\ch_\ominus\to
\ch_{\ominus,top}$ as
\begin{equation}
[\psi_{\ominus}(\alpha)](\tau,t)=e^{f_{(\tau,t)}+\fb_{(\tau,t)}}\alpha(\tau,t)=:S^-,\;(\tau,t)\in
\C^*\times S.
\end{equation}
Similarly, we have the homomorphism $\psi_{\oplus}:\ch_\oplus\to
\ch_{\oplus,top}$ as
\begin{equation}
[\psi_{\oplus}(\alpha)](\tau,t)=e^{-f_{(\tau,t)}-\fb_{(\tau,t)}}\alpha(\tau,t):=S^+,\;(\tau,t)\in
\C^*\times S.
\end{equation}
We can prove the following conclusion.
\begin{thm}[Theorem \ref{thm:bund-isom}] The bundle map $\psi_\ominus$ provides an isomorphism
between two real Hermitian holomorphic bundles:
$(\ch_\ominus,g,\tau_R)$ and $(\ch_{\ominus,top}, \hg,\tau_R)$ and
the same for $\psi_\oplus$.
\end{thm}
Here the metric $\hg$ is defined in Section \ref{subsec:4.3}.
Furthermore, by integrating $S^-_a$ along the corresponding
Lefschetz thimble, we can obtain the horizontal section $\Pi^-_a$.
Furthermore, we can prove
\begin{thm}[Theorem \ref{thm:correspondence}] The correspondence between
$(\ch_{\ominus,top},S^-_a,D+\bar{D},\heta,\tau_R,\hstar)$ and
$(\ch_{\ominus,top},\Pi^-_a,\D,\heta,\tau_R,\hstar)$ is an
isomorphism between two real Hermitian bundles with the associated
Hermitian connections.
\end{thm}

In fact, we obtain the important conclusion:
\begin{thm}[Theorem \ref{thm:true-flat}] The Cecotti-Vafa equations and the Fantastic equations
hold on $\ch_{\ominus,top}\to \C^*\times S_m$ and the connection
$\D$ is flat.
\end{thm}
Here $S_m\subset S$ is the set of $(\tau,t)$ such that
$f_{(\tau,t)}$ is a Morse function.

A very important fact is that the equation $\D \Pi^-=0$ and its
solutions $\Pi^-$ are given explicitly in terms of the wave
functions $\alpha_a$ of the Schr\"odinger equation
$\Delta_{f_\tau}=0$, where the Schr\"odinger equation can be defined
in a more complicated manifold $M$ except the standard $\C^n$ or
$(\C^*)^n$.

In Section \ref{subsec:4.4}, we discuss a family of flat connections
$\nabla^{G,s}$ by gauge transformation. The gauge transformation is
used to change the $(0,1)$-part of $\D$ into $\bpat$. The
derivatives of $\nabla^{G,s}$ provides another family of flat
connections satisfying some interesting identities. The gauge
transformation transforms the horizontal sections $\Pi^-_a$ into the
horizontal and holomorphic sections $\Pi^{-,h}_a$.

So far, by considering the deformation theory of the Schr\"odinger
equations, we have obtained the $tt^*$ geometry that physicists hope
to get. $tt^*$ structure can deduce many interesting structures. It
induces the mixed Hodge structure in Section \ref{subsec:4.5} while
inducing the isomonodromic deformation in Section
\ref{subsec:4.5.3}. In Section \ref{subsec:4.7}, we can obtain the
harmonic Higgs bundle structure. In Section \ref{subsec:4.6}, we
will compute the periodic matrix and the primitive vector based on
$S_m^1\subset S_m$. In some situations, e.g., $(\C^n,W)$ or
$((\C^*)^n,f)$, the computation can be simplified. In particular, in
$(\C^n,W)$ case, the primitive vector is given by the following
oscillating integral:
\begin{align}
&\Pi^{-,\circ}_1=(\Pi^{-,\circ}_{11},\cdots,\Pi^{-,\circ}_{\mu 1})^T\nonumber\\
&\Pi^{-,\circ}_{a1}=\tau^{n/2-1} e^{-A}\cdot\int_{\cC_a}e^{\tau
f+\taub\fb}dz^1\wedge\cdots\wedge dz^n,
\end{align}
which is defined in $S_m^1$.

Via this primitive vector, we can construct the Frobenius manifold
structure in Section \ref{subsec:4.7}.

\subsection{Relation to other mathematical branches}\

\

Naturally, the deformation theory of Schr\"odinger operators have
very tight connections with many mathematical branches. We will
describe simply the connections what we knew.

\subsubsection{\underline{Singularity theory}}\

\

$tt^*$ geometry was firstly considered by experts studying the
singularity theory. The first Frobenius manifold structure for the
universal unfolding of a hypersurface singularity was founded by K.
Saito and later developed by M. Saito. Though our theory can only
treat the universal unfolding of singularity with central charge
$c\le 1$, it deserves to build the detail correspondence. One big
difference is about the existence of primitive form (primitive
vector). In our case, this is obtained naturally (even given by
explicit form) by the period matrix. In M Saito's theory, this was
obtained by solving Riemann-Hilbert problem on $\IP^1$.

For the other hypersurface singularities, the best treatment is to
consider the boundary value problem of Schr\"odinger operators,
which will help build the local theory. Hence in some sense, the
singularity theory can be described in the framework of differential
geometry.

\subsubsection{\underline{Deformation theory of hypersurfaces}}\

\

Since the deformation theory of Schr\"odinger operators can be
applied to the defining functions of smooth hypersurface in a
(weighted) projective space. It is natural to compare the VHS and
period mappings in two sides. The VHS in projective space was
proposed by Grifitths, and developed by Deligne, Schmidt,Shrenk,...
and many mathematicians. The same comparison can be done for toric
hypersurfaces. One should compare with Batyrev's work \cite{Ba}

\subsubsection{\underline{Topological field theory: CY/LG correspondence and mirror
symmetry}}\

\

The deformation theory of Schr\"odinger operators can be seen as the
mathematical theory of Landau-Ginzburg B model. One can compare the
Frobenius structure with LG A model, Fan-Jarvis-Ruan's quantum
singularity theory to check the mirror symmetry. On the other hand,
one can check the LG/CY correspondence: compare with
Kontsevich-Barannikov's Calabi-Yau B model. This correspondence will
generalize many results obtained by physicists and generalize the
correspondence from the chiral ring structure to Frobenius manifold
structure.

In Manin's book \cite{Ma}, some different Frobenius manifold
structures were proposed, they are: the isomonodromic deformation of
linear differential system, Saito's Frobenius manifold structure
from singularity theory, Kontsevich-Barannikov's construction on CY
manifold and formal Frobenius manifold from quantum cohomology
theory. Recently, the Frobenius manifold structure from
Fan-Jarvis-Ruan's quantum singularity theory also appear. Some
attempts has been done (ref. \cite{HM},\cite{ChR1,ChR2}) to identify
these structures by comparing the initial data. The deformation
theory of Schr\"odinger operators build the rest one corner of the
whole theory and provide the hope to build the direct geometrical
correspondence between those structures.

\subsection{Further problems in the deformation theory}\

\

There are many subsequent problems. We will study the relation of
the connection $\D$ with various structures such as CV-structure,
TERP-structure, and Saito-Frobenius Structure in our following
papers. Also we will discuss the DGBV algebraic structure induced by
our deformation theory. After make clear all the relations, we will
turn to the computation of the Frobenius manifold structures.

There are also two challenge directions. One is to study the
integrable systems and their transformations coming from such a
deformation theory. The other is to use the quantization method to
capture the information from the deformation theory and try to
construct the higher genus invariants for LG B-model.

Finally the deformation theory of Schr\"odinger type equations
should exist not only for the trivial section-bundle system
$(M\times \C, M,g,\tau f)$, but also for the other more complicated
section-bundle systems. For instance, the model relating the
complete intersection CY model.

\subsection{Acknowledgments} The idea of this paper appeared very
earlier, and I reported several times in different conferences
beginning from my talk in Peking university in January, 2010. I have
chance to talk with many mathematicians, T. Milanov, Y. Ruan, C.
Hertling, K. Saito, G. Tian, P. Griffiths, E. Witten and etc. I
would like to thank them for their patience and the helpful
suggestions. In particular, I would like to thank K. Saito, who
spent a lot of time to discuss with me and pointed out many problems
when I begin to write this paper. Most of this paper has been
written when I visited Princeton University in the fall of 2010. I
appreciate Prof. G. Tian for his interests to my work and the kind
invitation to visit Princeton. I also want to thank Prof. Alice
Chang for her arrangement and thank the math department there for
hospitality.

I would like to thank my landlord Mrs. Mei-Yu Tsai for providing me
a comfortable living place and for the generous help when I was in
Princeton. Finally, this paper is also dedicated to Mrs. Grace G. in
the memory of all she gave me in my hard time.

\section{Differential geometry of Schr\"odinger operators}

\subsection{Preliminary of functional analysis}

\subsubsection{\underline{Min-max principle}} We write down the Theorem
XIII.1 of \cite{RS}:

\begin{thm}\label{thm:prel-minmax} Let $H$ be a self adjoint operator of a Hilbert space that is bounded
below, i.e., $H\ge CI$ for some $C$. Define
\begin{equation}
\mu_n(H)=\sup_{\varphi_1,\cdots,\varphi_{n-1}}U_H(\varphi_1,\cdots,\varphi_{n-1}),
\end{equation}
where
\begin{equation}
U_H(\varphi_1,\cdots,\varphi_m)=\inf_{\stackrel{\psi\in
D(H);||\psi||=1}{\psi\in [\varphi_1,\cdots,\varphi_m]^\perp}}(\psi,
H\psi).
\end{equation}

where $D(H)$ represent the domain of $H$ and
$[\varphi_1,\cdots,\varphi_m]^\perp$ is the perpendicular complement
space of the subspace $[\varphi_1,\cdots,\varphi_m]$ generating by
$\varphi_1,\cdots,\varphi_m$.

Then for each $n$, one of the following two conditions must hold:

\begin{enumerate}
\item there are $n$ eigenvalues (counting multiplicity) below the
bottom of the essential spectrum and $\mu_n(H)$ is the $n$-th
eigenvalue (counting multiplicity);

\item $\mu_n$ is the bottom of the essential spectrum, i.e.,
$\mu_n=\inf\{\lambda|\lambda \in \sigma_{ess}(H)\}$ and in that case
$\mu_n=\mu_{n+1}=\mu_{n+2}=\cdots$ and there are at most $n-1$
eigenvalues (counting multiplicity) below $\mu_n$.
\end{enumerate}
\end{thm}

\subsubsection{\underline{Discreteness of spectrum}} The following
theorem which was the Theorem XIII.64 of \cite{RS} gives the
equivalent conditions that the self-adjoint operator $H$ has only
discrete spectrum.

\begin{thm}\label{thm:prel-disc-spectrum} Let $H$ be a self-adjoint operator that is bounded
below. Then the following conditions are equivalent:
\begin{enumerate}
\item $(H-\mu)^{-1}$ is compact for some $\mu\in \rho(H)$.

\item $(H-\mu)^{-1}$ is compact for all $\mu\in \rho(H)$.

\item $\{\psi\in D(H) | ||\psi||\le 1, ||H\psi||\le b\}$ is compact
for all $b$.

\item $\{\psi\in D(H) | ||\psi||\le 1, (H\psi,\psi)\le b\}$ is compact
for all $b$.

\item There exists a complete orthonormal basis
$\{\varphi_n\}_{n=1}^{\infty}$ in $D(H)$ so that $H\varphi_n=\mu_n
\varphi_n$ with $\mu_1\le \mu_2\le \cdots $ and $\mu_n\to \infty$.

\item $\mu_n(H)\to \infty$ where $\mu_n(H)$ is given by the
min-max principle.
\end{enumerate}

\end{thm}

\subsubsection{\underline{Compact perturbation}}

\begin{df} Let $H_0$ be a self-adjoint operator with domain
$D(H_0)$ and $H_1$ be a symmetric operator with domain
$D(H_0)\subset D(H_1)$. $H_1$ is said to be $H_0$-bounded in form
with bound $a$, if there exists a constant $b$ such that for any
$\psi\in D(H_0)$ the following holds
$$
|(H_1\psi, \psi)|\le a (H_0\psi, \psi)+b(\psi,\psi).
$$
Moreover, if for any $a>0$, there exists a $b$ such that the above
inequality holds, then $H_1$ is said to be $H_0$-compact.
\end{df}

It is known that if $a<1$, then $H_0+H_1$ is also self-adjoint
operator.

\begin{thm}\label{thm-fund-compact} Let $H_0$ be a semibounded self-adjoint operator with
only discrete spectrum. Let $H_1$ be a symmetric operator which is
$H_0$-bounded with bound $a<1$. Then the self-adjoint operator
$H=H_0+H_1$ has only discrete spectrum. In particular, if $H_1$ is
$H_0$-compact, then $H$ has only discrete spectrum.
\end{thm}

\begin{proof} Since
$$
|(H_1\psi, \psi)|\le a (H_0\psi, \psi)+b(\psi,\psi),
$$
We have for any $\psi \in D(H)=D(H_0)$,
$$
(H\psi, \psi)\ge (1-a)(H_0\psi, \psi)-b(\psi,\psi).
$$
By Min-max principle, Theorem \ref{thm:prel-minmax}, we have
$$
\mu_n(H)\ge (1-a)\mu_n(H_0)-b.
$$
By Theorem \ref{thm:prel-disc-spectrum}, we know that $\mu_n(H_0)\to
\infty$ as $n\to \infty$. Hence
$$
\mu_n(H)\to\infty,
$$
as $n\to \infty$.

So by Theorem \ref{thm:prel-disc-spectrum} again, we know that $H$
has only discrete spectrum.
\end{proof}

In 1934, K. Friedrichs \cite{Fr} has proved the following theorem
(see Theorem XIII.67 of \cite{RS}):

\begin{thm} Let $V\in L^1_{loc}(\R^n)$ be bounded from below and
suppose that $V(x)\to \infty$, as $|x|\to \infty$. Then
$H_0=-\Delta+V$ defined as a sum of quadratic forms is an operator
having compact resolvent. In particular, $H_0$ has purely discrete
spectrum and a complete set of eigenfunctions.
\end{thm}

This theorem was generalized in many cases. Here we will use the
Theorem of Kondrat'ev-Shubin about the Schr\"odinger operators on
manifolds with bounded geometry.

\subsubsection{\underline{Bounded geometry}} Let $(M,g)$ be a
$n$-dimensional connected complete Riemannian manifold with metric
$g$. $(M,g)$ is said to have a bounded geometry, if the following
conditions hold:
\begin{enumerate}
\item the injectivity radius $r_0$ of $M$ is positive.
\item $|\nabla^m R|\le C_m$, where $\nabla^m R$ is the $m$-th
covariant derivative of the curvature tensor and $C_m$ is a constant
only depending on $m$.
\end{enumerate}

\begin{ex} Let $\bar{M}$ be a compact K\"ahler manifold, and
$D=\sum_{i=1}^p D_i$ a normal crossing divisor in $\bar{M}$. Let
$M=\bar{M}-D$ and $s_i$ be the defining section of $D_i$, and
$|\cdot|$ be a Hermitian metric on the associated line bundle
$[D_i]$ with $|s_i|<1$. Then we have the Poincar\'e metric $\omega$
associated to $D$ given by
$$
\om=C\bar{\om}-2\sum_{i=1}^p dd^c\log(-\log|s_i|^2),,
$$
for sufficiently large $C$. Here $\bar{\om}$ is the K\"ahler form of
the manifold $\bar{M}$. Locally, the neighborhood near the divisor
is given by polydisc of the form $(\Delta^*)^k\times \Delta^{n-k}$
and the classical poincar\'e metric on $\Delta^*$ is given by
$$
\om_{\Delta^*}=\frac{i}{2\pi}\frac{dz d\bar{z}}{|z|^2(\log|z|^2)^2}.
$$
However, this metric has zero injective radius. To get a bounded
geometrical structure, Cheng-Yau \cite{CY} introduced the local
quasi-coordinate system.

For $(\Delta^*,\om_{\Delta^*})$, the quasi-local coordinate system
is defined as follows: for any $0<\eta<3/4$, define the coordinate
map
\begin{equation}
\phi_\eta(v)=\exp(\frac{(1+\eta)(v+1)}{(1-\eta)(v-1)}), \forall v\in
\Delta_{3/4}.
\end{equation}
Then $\cup_{\eta\in (0,1)}\phi_\eta(\Delta_{3/4})$ covers $\Delta^*$
and the pull-back of the classical Poincar\'e metric
$$
\phi^*_\eta(\om_{\Delta^*})=\frac{i}{2\pi}\frac{dv\wedge
d\bar{v}}{1-|v|^2}
$$
is independent of $\eta$. One can check this metric has bounded
geometry. For a polydisc $(\Delta^*)^k$, one can take the coordinate
system
$$
V_\eta=(\Delta_{3/4})^k\times \Delta^{n-k}, \forall
\eta=(\eta_1,\cdots,\eta_k)\in (0,1)^k.
$$
with coordinate map:
$$
\Phi_\eta(v)=(\phi_{\eta_1}(v),\cdots,\phi_{\eta_k},v_{k+1},\cdots,v_n),
\forall v\in V_\eta.
$$
where $\phi_{\eta_i}\equiv \phi_\eta$, which is defined before. So
$\cup_{\eta\in (0,1)^k}\Phi_\eta(V_\eta)$ gives a quasi-coordinate
system such that the pull-back metric $\Phi_\eta^*(\om)$ has bounded
geometry. Hence via the quasi-coordinate system, $(M,\om)$ has
bounded geometry. The reader can refer \cite{CY} and \cite{Wd1,Wd2}
for the proof.
\end{ex}

The Laplace-Beltrami operator on scaler functions on $M$ is defined
as
$$
\Delta=\frac{1}{\sqrt{g}}\frac{\pat}{\pat
x^i}(\sqrt{g}g^{ij}\frac{\pat u}{\pat x^j}),
$$
where $(x^1,\cdots,x^n)$ are local coordinates, $g=\det(g_{ij})$ and
$(g^{ij})$ is the inverse matrix of the metric matrix $(g_{ij})$.
Let $V(x)$ be a smooth potential function defined on $M$. The
Schr\"odinger operator is
\begin{equation}
H_0=-\Delta+V(x).
\end{equation}

\begin{thm}\label{thm:Kon-shubin} Assume that $(M,g)$ is a Riemannian manifold of bounded
geometry, $H_0=-\Delta+V$ is the Schr\"odinger operator with a
locally $L^2$-integrable potential $V(x)$ which is semi-bounded:
$$
V(x)\ge -C,\;x\in E.
$$
There exists $c>0$, depending only on $(M,g)$ such that the spectrum
$\sigma(H)$ consisting of only discrete spectrum if and only if the
following condition is satisfied:

\begin{itemize}
\item[(D)] For any sequence $\{x_k|,k=1,2,\cdots\}\subset M$ such
that $x_k\to \infty$ as $k\to \infty$, for any $r<r_0/2$ and any
compact subsets $F_k\subset \bar{B}(x_k,r)$ such that $Cap(F_k)\le
cr^{n-2}$ in case $n\ge 3$ and $cap(F_k)\le c(\ln\frac{1}{r})^{-1}$
in case $n=2$,
\begin{equation}
\int_{B(x_k,r)/F_k}V(x)dvol_E\to \infty,\;\text{as}\;k\to\infty.
\end{equation}
\end{itemize}
Here $Cap(F_k)$ for $n\ge 3$ means that the harmonic capacity of the
set $F_k$ in the normal (geodesic) coordinates centered at $x_k$,
and for $n=2$ it means the same capacity with respect to a ball
$B(0,R)$ of fixed radius $R<r_0$.
\end{thm}

The following is the definition of the harmonic capacity

\begin{df} For any compact set $F\subset \Omega$, the harmonic
capacity of $F$ with respect to $\Omega$ is
$$
Cap_{\Omega}(F)=\left\{\int_\Omega |\nabla u|^2 |u\in
C_0^\infty(\Omega), u=1 \text{near}\;F,0\le u\le 1
\;\text{in}\;\Omega\right\}
$$
\end{df}

An easy corollary similar to Friedrichs' theorem is:

\begin{crl}\label{crl:spect-1} If the potential $V(x)$ is locally
$L^2$-integrable and is semi-bounded:
$$
V(x)\ge -C,\;x\in E.
$$
Then if $V(x)\to \infty$ as $d(x,x_0)\to \infty$, the spectrum of
the Schr\"odinger operator $H=-\Delta+V$ has only discrete spectrum.
Here $x_0$ is fixed base point on $E$.
\end{crl}

\subsection{Example: Complex harmonic oscillators}

\subsubsection{\underline{Witten deformation}}\

\

Let $f$ be a $C^\infty$ function on a $n$-dimensional Riemannian
manifold $M$. Then this forms a simple model of section bundle
system $(M\times \R,M, f_\tau:=\tau f)$.

We can choose a local orthonormal basis $\{e^1,\cdots,e^n\}$ of the
cotangent bundle around a point $p$ and denote the dual basis as
$\{e_1,\cdots,e_n\}.$ Define two operators
\begin{equation}
\ad_j =e^j\wedge,\;a_k=\iota(e_k),
\end{equation}
where $\iota(e_k)$ is the contraction with the vector field $e_k$.
The two operators satisfy the commutative relations:
\begin{equation}
[\ad_j,\ad_k]=0,\;[a_j,\ad_k]=\delta_{jk},\;[a_j,a_k]=0.
\end{equation}

Define the operators acting on the space of $p$-forms
$\Lambda^p(M)$:
$$
d_{f_\tau}=d+\tau df\wedge=e^{-\tau df}\circ d \circ e^{\tau
df},\;d_{f_\tau}^*=e^{\tau df}\circ d^* \circ e^{-\tau
df}=d^*+\iota(\tau df)
$$
and the Laplace operator
$$
\Box_{f_\tau}=d_{f_\tau}d_{f_\tau}^*+d_{f_\tau}^*d_{f_\tau}.
$$
The Laplace operator has the following form:
\begin{equation}\label{eq:Schrodinger}
\Box_{f_\tau}= -\Delta+\tau \sum_{k,l}\nabla_l\nabla_k f(\ad_k
a_l+\ad_l a_k)+\tau^2 |\nabla f|^2,
\end{equation}
where $\Delta$ is the Laplace-Beltrami operator of $M$.

This is just the Witten deformation Laplacian in Riemannian
geometry. In particular, by studying the spectrum behavior of
$\Box_{f_\tau}$, one can obtain the Morse inequality (See
\cite{Wi}).

\subsubsection{\underline{real $1$-dimensional harmonic oscillator}}\

\

The simplest section-bundle system is given by $(\R\times\R, \R,
f=\tau t x^2)$ for $\tau>0$. Now the operator (\ref{eq:Schrodinger})
becomes
$$
-\Delta+4\tau t(\ad a)+4t^2\tau^2x^2:\Lambda^p\to\Lambda^p.
$$
For $p=0$, we have
$$
-\frac{d^2}{dx^2}+4t^2\tau^2x^2=2t\tau\left(-\frac{d^2}{d(\sqrt{2t\tau}x)^2}+(\sqrt{2t\tau
}x)^2\right)
$$
which is precisely the harmonic oscillator and has eigenvalues
$$
\lambda_k=2t\tau(1+2k),k=0,1,\cdots,
$$
and the corresponding eigenfunctions:
$$
u_k(\sqrt{2t\tau}x)=H_k(\sqrt{2t\tau}x)e^{-\frac{(\sqrt{2t
\tau}x)^2}{2}},
$$
where $u_k(x)=H_k(x)e^{-\frac{x^2}{2}}$ are the eigenfunctions of
the standard Harmonic oscillator
$$
-\frac{d^2}{dx^2}+x^2
$$
and $H_k(x)$ are the Hermitian polynomials.

For $p=1$, if we set $\varphi=udx$, we have
$$
\Box_{f_\tau}\varphi=\left(-\frac{d^2}{dx^2} u+4t^2\tau^2x^2
u+4t\tau u\right)dx,
$$
which has the eigenvalues
$\lambda_k+4t\tau=2t\tau(3+2k),k=0,1,\cdots$, and the corresponding
eigenforms are
\begin{equation}
\varphi_m=u_m dx.
\end{equation}
The smallest eigenform is
\begin{equation}
\varphi_0=e^{-\tau tx^2}dx.
\end{equation}

In general, one can consider the SBS $(\R^n\times \R, \R^n,
f=\tau\sum_j t_j x_j^2)$. The operator (\ref{eq:Schrodinger})
becomes
$$
-\sum_j \frac{d^2}{dx_j^2}+4\tau\sum_j(t_j \ad_j a_j)+4\tau^2\sum_j
t^2_j x^2_j:\Lambda^p\to\Lambda^p
$$
Let $u_m^j,\lambda_m^j$ be the eigenfunctions and eigenvalues of the
$j$-th harmonic oscillator $\frac{d^2}{dx_j^2}+4\tau^2 t^2_j x^2_j
$, then the total eigenforms are
$$
u_{m_1}^1\cdots u_{m_n}^n dx^1\wedge\cdots dx^p,
$$
and the eigenvalues are
$\lambda_{m_1}^1+\cdots+\lambda_{m_n}^n+4\tau\sum_j t_j$. These
eigenforms form a basis in $L^2\Lambda^p(\R^n)$. The smallest
eigenvalue is $6\tau\sum_j t_j$ and the corresponding eigenform is
\begin{equation}
\varphi_0=e^{-\tau (t_1x_1^2+\cdots+t_px_p^2)}dx^1\wedge
\cdots\wedge dx^p.
\end{equation}

\subsubsection{\underline{complex $1$-dimensional harmonic oscillator}}\

\

Now we consider the section-bundle system $(\C\times
\C,\C,f_\tau=\tau z^2)$. Then we have the twisted operator
$\bpat_{f_{\tau}}=\bpat+\pat f_\tau \wedge $, the conjugate
$\bpat_{f_\tau}^\dag$, and the complex Laplace operator (see next
section for the computation)
\begin{align*}
\Delta_{f_\tau}&=\bpat_{f_\tau}\bpat_{f_\tau}^\dag+\bpat_{f_\tau}^\dag\bpat_{f_\tau}\\
&=-\frac{\pat^2}{\pat z\pat \bar{z}}+2(\tau \ad \ab+\taub\abd
a)+4|\tau|^2|z|^2:\Lambda^p\to \Lambda^p
\end{align*}

If $p=0$, then $\Delta_{f_\tau} u=0$ has only zero solution. If
$p=1$ and set $\varphi=udz+vd\bar{z}$, then we have
$$
\Delta_{f_\tau} \varphi=\left(-\frac{\pat^2 u}{\pat z\pat
\bar{z}}+2\tau v +4|\tau|^2|z|^2 u\right)dz+\left(-\frac{\pat^2
v}{\pat z\pat \bar{z}}+2\bar{\tau}u +4|\tau|^2|z|^2 v\right)d\bar{z}
$$
Let $H:=-\frac{\pat^2 }{\pat z\pat \bar{z}}+4|\tau|^2|z|^2$, and
define a $\C$ linear operator $\hat{\tau}_t:\Lambda^1\to \Lambda^1$
given by
$$
\hat{\tau}(udz+vd\bar{z})=(\tau v dz+\taub u d\bar{z}).
$$
Then $\hat{\tau}^2(\varphi)=|\tau|^2\varphi$, i.e.,
$\tau_\R:=1/|\tau|\hat{\tau}$ defines an involution (real
structure). Now we can show that
\begin{equation}
\Delta_{f_\tau} \tau_\R\varphi=\tau_\R\Delta_{f_\tau}\varphi.
\end{equation}
This shows that $\Delta_{f_\tau}$ can be reduced to two operators
acting on the two eigenspaces $E_{\pm}$ of $\tau_\R$ with respect to
the two eigenvalues $\pm 1$. $\varphi=udz+vd\bar{z}\in E_{\pm}$ iff
$$
\tau v=\pm |\tau| u,\;\text{and}\;\bar{\tau}u=\pm |\tau|v.
$$
Therefore the eigenvalue problem of $\Delta_{f_\tau}$ is reduced to
the eigenvalue problem of the following operator
$$
H\pm 2|\tau|=-\frac{\pat^2 }{\pat z\pat \bar{z}}+4|\tau|^2|z|^2\pm
2|\tau|.
$$
which are the sum of two real 1-dimensional harmonic oscillators. In
particular, $\Delta_{f_\tau}\varphi=0$ for $\varphi\in L^2\Lambda^1$
iff there exists $v\in L^2$ such that
\begin{equation}\label{eq:harm-1}
\varphi=-\frac{\tau}{|\tau|}vdz+vd\bar{z}
\end{equation}
and $v$ satisfies
\begin{equation}\label{eq:comp-osci}
-\frac{\pat^2 v}{\pat z\pat \bar{z}}+4|\tau|^2|z|^2v- 2|\tau|v=0.
\end{equation}
The solution space is $1$-dimensional and is generated by
$$
\varphi(z)=\{u_0(2\sqrt{\tau}x)u_0(2\sqrt{\tau}y)\}(-\frac{\tau}{|\tau|}dz+d\bar{z})=e^{-\tau(|z|^2)}(-\frac{\tau}{|\tau|}dz+d\bar{z}),
$$
where $u_0$ is the smallest eigenfunction of the real 1-dimensional
harmonic oscillator.

We have the asymptotic relation as $|\tau|\to \infty$,
\begin{equation}
(\varphi,\varphi)_{L^2}=\int_{\C}\varphi\wedge
*\overline{\varphi}\sim \frac{C(\frac{\tau}{|\tau|})}{|\tau|},
\end{equation}
where $C(\frac{\tau}{|\tau|})$ is a constant depending only on the
direction $\frac{\tau}{|\tau|}$.

\begin{rem} Now the Witten index $Tr(-1)^F e^{-\tau \Delta_{f_\tau}}$ is $1$ and
one can also calculate the partition function $Tr e^{-\tau
\Delta_{f_\tau}}=\sum_{\lambda_i(\tau)} e^{-\tau \lambda_i(\tau)}$.
\end{rem}

\subsubsection{\underline{complex $n$-dimensional harmonic oscillator}}\

\

Consider the section bundle system $(\C^n\times \C,\C^n,
f_\tau=\tau(z_1^2+\cdots+z_n^2) )$. Then the Laplacian is
$$
\Delta_{f_\tau}=\Delta_1+\cdots+\Delta_n,
$$
where $\Delta_i$ is the complex $1$-dimensional harmonic oscillator
with respect to $z_i$. It is easy to see that the $L^2$ harmonic
$k$-form for $k<n$ is trivial, and the space of $L^2$ harmonic
$n$-form is $1$-dimensional and is generated by the following
element:
\begin{equation}
\Psi(z_1,\cdots,z_n,\tau)=\varphi(z_1)\wedge \cdots \varphi(z_n),
\end{equation}
where the $1$-form $\varphi(\cdot)$ is given by (\ref{eq:harm-1}).

Similarly, as $|\tau|\to \infty,$ we have the asymptotic relation:

\begin{equation}
(\Psi,\Psi)_{L^2}=\int_{\C^n}\Psi\wedge
*\overline{\Psi}\sim \frac{C(\frac{\tau}{|\tau|})}{|\tau|^n}.
\end{equation}

\subsection{Complete non-compact K\"ahler manifold with potential
function}\

\

Let $(M, g)$ be a complex manifold with a Hermitian metric $g$.
Denote by $T=TM, T^*=T^*M, \bar{T}=\overline{TM}$, and
$\overline{T^*}=\overline{T^* M}$ the holomorphic tangent bundle,
the holomorphic cotangent bundle, the anti-holomorphic tangent
bundle and anti-holomorphic cotangent bundles respectively. Let $U$
be a coordinate chart of $M$ having local coordinates
$z=(z^1,\cdots,z^n)$. Here we assume that $z_j=x_{2j-1}+i x_{2j},
j=1,\cdots,n$. So locally, $g=g_{i,\jb}dz^i dz^{\jb}$ (by using
Einstein summation convention) is a section of the bundle
$T^*\otimes \overline{T^*}$. The associated form $\om=\frac{i}{2}
g_{i,\jb}dz^i\wedge dz^{\jb}$ and
\begin{equation}
\frac{\om^n}{n!}=g dx^1\wedge dx^2 \cdots \wedge dx^{2n}=dv_M,
\end{equation}
where $g=\det(g_{i\jb})$ and $dv_M$ is the real volume element with
respect to the Hermitian metric defined on $T_\R M=T\oplus \bar{T}$.

A $(p,q)$-form is a differential section $\varphi: z\to [z,
\varphi_{i_1,\cdots,i_p,\jb_1,\cdots,\jb_q}]$ of the tensor bundle
$\Lambda^{p,q}=(\otimes T^*)^p\otimes (\otimes \bar{T}^*)^q$ such
that the components $\varphi_{i_1,\cdots,i_p,\jb_1,\cdots,\jb_q}$
are skew-symmetric with respect to
$i_1,\cdots,i_p,\jb_1,\cdots,\jb_q$. Locally it has the form
$$
\varphi=\frac{1}{p!q!}\sum_{\stackrel{i_1,\cdots,i_p}{\jb_1,\cdots,\jb_q}}
\varphi_{i_1,\cdots,i_p,\jb_1,\cdots,\jb_q}dz^{i_1}\wedge
\cdots\wedge dz^{i_p}\wedge dz^{\jb_1}\wedge \cdots\wedge
dz^{\jb_q}.
$$
or written as the form without fractional factors
$$
\varphi=\sum_{\stackrel{i_1\le \cdots\le i_p}{\jb_1\le \cdots\le
\jb_q}} \varphi_{i_1,\cdots,i_p,\jb_1,\cdots,\jb_q}dz^{i_1}\wedge
\cdots\wedge dz^{i_p}\wedge dz^{\jb_1}\wedge \cdots\wedge
dz^{\jb_q}.
$$

Let $(g^{\ib j})$ be the inverse matrix of the metric matrix
$(g_{i,\jb})$ such that $\sum_{k}g^{\jb,k}g_{k \ib}=\delta^j_i$ and
$\sum_{\kb}g_{i\kb}g^{\kb j}=\delta^j_i$.

We choose the convention of the form action as follows: if
$\xi=\xi^p\frac{\pat}{\pat z^p}+\xi^{\qb}\frac{\pat}{\pat z^{\qb}}$
and $\eta=\eta^p\frac{\pat}{\pat z^p}+\eta^{\qb}\frac{\pat}{\pat
z^{\qb}}$, then
$$
\om(\xi,\eta)=ig_{i \jb}dz^i\wedge dz^{\jb}(\xi^p\frac{\pat}{\pat
z^p}+\xi^{\qb}\frac{\pat}{\pat z^{\qb}},\eta^p\frac{\pat}{\pat
z^p}+\eta^{\qb}\frac{\pat}{\pat z^{\qb}})=ig_{i
\jb}(\xi^i\eta^{\jb}-\eta^i\xi^{\jb}).
$$
If $\xi$ is a holomorphic tangent vector, then the length of $\xi$
is $|\xi|^2=g_{i\jb}\xi^i\overline{\xi^j}$. and the inner product at
point $z$ of two holomorphic tangent vector $\xi,\eta$ is defined as
$$
(\xi,\eta)(z)=g_{i\jb}\xi^i\overline{\eta^j}.
$$
Let
\begin{align*}
&\varphi(z)=\frac{1}{p!q!}\sum_{\stackrel{i_1,\cdots,i_p}{\jb_1,\cdots,\jb_q}}
\varphi_{i_1,\cdots,i_p,\jb_1,\cdots,\jb_q}dz^{i_1}\wedge
\cdots\wedge dz^{i_p}\wedge dz^{\jb_1}\wedge \cdots\wedge
dz^{\jb_q},\;\\
&\psi(z)=\frac{1}{p!q!}\sum_{\stackrel{k_1,\cdots,k_p}{\lb_1,\cdots,\lb_q}}
\psi_{k_1,\cdots,k_p,\lb_1,\cdots,\lb_q}dz^{k_1}\wedge \cdots\wedge
dz^{k_p}\wedge dz^{\lb_1}\wedge \cdots\wedge dz^{\lb_q},
\end{align*}
then the inner product of the two $(p,q)$-forms is defined as
\begin{equation}
(\varphi,\psi)(z)=\frac{1}{p!q!}\sum_{i,j,k,l}g^{\kb_1 i_1}\cdots
g^{\jb_2 l_2}\cdots g^{\jb_q l_q
}\varphi_{i_1,\cdots,i_p,\jb_1,\cdots,\jb_q}\overline{\psi_{{k_1,\cdots,k_p,\lb_1,\cdots,\lb_q}}}.
\end{equation}

Then the $L^2$ inner product on $M$ is defined as
\begin{equation}
(\varphi,\psi)=\int_M (\varphi,\psi)(z)dv_M.
\end{equation}

Denote by $I_p=(i_1,\cdots,i_p)$ the multiple index with length $p$.
Then a $(p,q)$ form $\varphi$ can be simply written as
\begin{equation}
\varphi(z)=\frac{1}{p!q!}\sum_{I_p,J_q}\varphi_{I_p
\Jb_q}dz^{I_p}\wedge dz^{\Jb_q}.
\end{equation}

\begin{df} Let $\psi(z)=\frac{1}{p!q!}\sum_{I_p,J_q}\psi_{I_p
\Jb_q}dz^{I_p}\wedge dz^{\Jb_q}$. We can define a $\C$-linear
operator $*: \Omega^{p,q}\to \Omega^{n-q,n-p}$ (where $\Omega^{p,q}$
represent the section set of the bundle $\Lambda^{p,q}$) as follows:
\begin{equation}
*\psi=(i)^n (-1)^{\frac{n(n-1)}{2}+pn} \sum_{I_q,J_p} g_{I_q
I_{n-q}\Jb_p\Jb_{n-p}}\psi^{\Jb_p I_q}dz^{I_{n-q}}\wedge
dz^{\Jb_{n-p}}.
\end{equation}
Here $g_{I_q I_{n-q}\Jb_p\Jb_{n-p}}$ is the tensor products of the
metric tensor $g_{i\jb}$:
$$
g_{I_q I_{n-q}\Jb_p\Jb_{n-p}}=g_{i_1\cdots
i_n\jb_1\cdots\jb_n}=\det(g_{i_a\jb_b}).
$$
and
\begin{equation}
\psi^{\Jb_p I_q}:=\sum_{k l}g^{\jb_1 k_1}\cdots g^{\jb_p
k_p}g^{\lb_1 i_1}\cdots g^{\lb_q i_q}\psi_{k_1\cdots k_p\lb_1\cdots
\lb_q}.
\end{equation}
\end{df}

It is routine to prove the following lemma (See \cite{MK} or
\cite{Ko}).

\begin{prop} The $*$ operator satisfies the properties:
\begin{enumerate}
\item $(\varphi,\psi)(z)dv_M=\varphi\wedge *\overline{\psi}$
\item $\overline{*\psi}=*\overline{\psi}$
\item $*^2=(-1)^{p+q}$.
\item $\bar{\varphi}\wedge *\psi=\psi\wedge *\bar{\varphi}$.
\end{enumerate}
\end{prop}

\begin{df} Denote by $\bpat^\dag,\pat^\dag, d^\dag$ the adjoint
operators of $\bpat,\pat,d$ with respect to the inner product
$(\cdot,\cdot)$, namely we define those adjoint operators by the
following identities:
\begin{equation}
(\bpat \varphi,\psi)=(\varphi, \bpat^\dag \psi),\;(\pat
\varphi,\psi)=(\varphi, \pat^\dag \psi),\;(d\varphi,\psi)=(\varphi,
d^\dag \psi),
\end{equation}
for compactly supported form $\varphi$ or $\psi$.
\end{df}

Sometimes it is convenient distinguish the $*$ operator by their
domain. The operator $*:\Omega^{p,q}\to\Omega^{n-q,n-p}$ can be
written as $*^{p,q}$. In this way, the relation $*^2=(-1)^{p+q}$ is
actually
\begin{equation}
*^{n-q,n-p}*^{p,q}=(-1)^{p+q}.
\end{equation}

If $A$ is a differential operator, we denote by $A^\dag$ the adjoint
operator of $A$ with respect to the metric $(\cdot,\cdot)$. Now the
adjoint operator $(*^{p,q})^\dag$ has the identity:
\begin{equation}
(*^{p,q})^\dag=(-1)^{p+q}*^{n-q,n-p},\;(*^{p,q})^\dag*^{p,q}=I
\end{equation}

The following result is well-known.
\begin{prop}
\begin{equation}
\bpat^\dag\psi=-*\pat*\psi,\;\pat^\dag\psi=-*\bpat*\psi,\;d^\dag=-*d*\psi.
\end{equation}
\end{prop}

\begin{prop} For $\psi\in \Gamma(\Lambda^{p,q+1})$,
\begin{equation}
(\bpat^\dag \psi)^{\Ib_p
j_1,\cdots,j_q}=(-1)^{p+1}\sum_{k}(\frac{\pat}{\pat z^k}+\frac{\pat
\log g}{\pat z^k})\psi^{\Ib_p k j_1\cdots j_q}.
\end{equation}
\end{prop}

One can find this proposition in \cite{MK}.

Now the complex Laplacian operator is
$\Delta=\bpat\bpat^\dag+\bpat^\dag \bpat$.

Let $\iota_v$ be the contraction operator contracting with the
vector $v$.
\begin{lm} We have the properties:
\begin{align}
&(dz^\mu\wedge)^\dag=g^{\mub \nu}\iota_{\pat_\nu},\;(dz^{\mub}\wedge)^\dag=g^{\nub \mu}\iota_{\pat_\nub}\\
&(\iota_{\pat_\nu})^\dag=g_{\mu
\nub}dz^\mu\wedge,\;(\iota_{\pat_\nub})^\dag=g_{\nu
\mub}dz^{\mub}\wedge.
\end{align}
\end{lm}

\begin{proof} It suffices to prove the first identity, since the
other ones can be induced easily.

Let
\begin{equation*}
\varphi=\frac{1}{p!q!}\sum \varphi_{i_1\cdots i_p\jb_1\cdots
\jb_q}dz^{i_1}\wedge\cdots\wedge
dz^{\jb_q},\;\psi=\frac{1}{(p+1)!q!}\psi_{l_0\cdots
l_p\kb_1\cdots\kb_q}dz^{l_0}\wedge \cdots dz^{\kb_q}.
\end{equation*}
Then
\begin{align*}
dz^\mu\wedge \varphi=&\frac{1}{p!q!}\sum \varphi_{i_1\cdots
i_p\jb_1\cdots \jb_q}dz^\mu\wedge\cdots\wedge dz^{i_1}\cdots\wedge dz^{\jb_q}\\
=&\frac{1}{(p+1)!q!}\sum (dz^\mu\wedge \varphi)_{i_0\cdots
\jb_q}dz^{i_0}\wedge dz^{i_1}\wedge \cdots dz^{\jb^q},
\end{align*}
where
$$
(dz^\mu\wedge \varphi)_{i_0\cdots \jb_q}=\sum_{s=0}^p
(-1)^s\delta^\mu_{i_s}\varphi_{i_0\cdots \widehat{i_s}\cdots \jb_q}.
$$
So
\begin{align*}
(dz^\mu\wedge \varphi, \psi)=&\frac{1}{(p+1)!q!}(dz^\mu\wedge
\varphi)_{i_0\cdots \jb_q}\overline{\psi_{l_0\cdots \kb_q}}g^{\lb_0
i_0}\cdots g^{\jb_q
k_q}\\
=&\frac{1}{(p+1)!q!}\sum_{s=0}^p (-1)^s g^{\lb_s
\mu}\varphi_{i_0\cdots \widehat{i_s}\cdots
\jb_q}\overline{\psi_{l_0\cdots \kb_q}}g^{\lb_0
i_0}\cdots \widehat{g^{\lb_s \mu}}\cdots g^{\jb_q k_q}\\
=&\frac{1}{p!q!}\varphi_{i_1\cdots i_p \jb_1\cdots
\jb_q}\overline{\psi_{\nu l_1\cdots \kb_q}}g^{\nub \mu}g^{\lb_1
i_1}\cdots g^{\jb_q k_q}.
\end{align*}
On the other hand, we have
\begin{align*}
g^{\mub \nu}\iota_{\nu}\psi=& g^{\mub
\nu}\iota_\nu\left(\frac{1}{(p+1)!q!}\sum \psi_{l_0\cdots
\kb_q}dz^{l_0}\wedge\cdots \wedge
dz^{\kb_q}\right)\\
=&\frac{1}{(p+1)!q!}\sum g^{\mub \nu}\psi_{l_0\cdots l_p\kb_1\cdots
\kb_q}\sum_{s=0}^p \delta^{l_s}_\nu (-1)^s dz^{l_0}\wedge
\cdots\widehat{dz^{l_s}}\cdots \wedge dz^{\kb_q}\\
=&\frac{1}{(p+1)!q!}\sum^p_{s=0}g^{\mub l_s}\psi_{l_sl_0\cdots
\widehat{l_s}\cdots\wedge \kb_q}dz^{l_0}\wedge
\cdots\widehat{dz^{l_s}}\cdots\wedge dz^{\kb_q}\\
=&\frac{1}{p!q!}g^{\mub \nu}\psi_{\nu l_1\cdots \wedge dz^{\kb_q}}
\end{align*}
Therefore, we have
\begin{align*}
(\varphi, g^{\mub \nu}\iota_\nu
\psi)=&\frac{1}{p!q!}\varphi_{i_1\cdots \jb_q}\overline{\psi_{\nu
l_1\cdots \wedge \kb_q}}g^{\nub \mu}g^{\lb_1 i_1}\cdots g^{\jb_q
k_q}.
\end{align*}
In summary, we obtain
$$
(dz^\mu\wedge \varphi, \psi)=(\varphi, g^{\mub \nu}\iota_\nu \psi)
$$
\end{proof}

The following lemma shows the commutation relations of the operators
$\iota_{\pat_\nu}, dz^\mu,\dots$ with the $*$ operator.

\begin{lm}
\begin{align*}
&
*^{p,q-1}(dz^l\wedge)^\dag=(-1)^{p+q-1}(dz^\lb\wedge)*^{p,q},\;*^{p-1,q}(\iota_{\pat_m})=(-1)^{p+q+1}(g_{m\lb}dz^\lb\wedge)*^{p,q}
\end{align*}
and the other commutation relations can be obtained by taking
complex conjugate or adjoint.
\end{lm}

\subsubsection{\underline{Clifford algebra and $sl_2(\R)$ Lie algebra
}}\

These operators $\iota_{\pat_\nu}, dz^\mu,\dots$ acting on
$\Omega^{p,q}$ and increase or decrease the degree by $1$, we think
them as degree $1$ operators. We always define our Lie bracket as
super Lie bracket, i.e., if $A $ and $B$ are two operators,
\begin{equation}
[A,B]=AB-(-1)^{|A||B|}BA,
\end{equation}
where $|A|$ denotes the degree of $A$.

It is easy to prove the following formulas which implies the
Clifford algebra structure of operators.
\begin{prop}\label{prop-Cliff}
\begin{align}
&[(\iota_{\pat_\nu})^\dag, \iota_{\pat_\mu}]=g_{\mu
\nub},\;[(dz^\mu\wedge)^\dag, dz^{\nu}\wedge ]=g^{\mub \nu},\\
&[(\iota_{\pat_\nu})^\dag, \iota_{\pat_\mub}]=0,\;[(dz^\mu\wedge)^\dag, dz^{\nub}\wedge ]=0,\\
&[\iota_{\pat_\nu},\iota_{\pat_\mu}]=[\iota_{\pat_\nu},\iota_{\pat_\mub}]=[dz^\mu\wedge
,dz^\nu\wedge]=[dz^\mu\wedge ,dz^\nub\wedge]=0
\end{align}
\end{prop}

Define the following operators:
\begin{align}
&L:= i g_{\mu\nub}dz^\mu\wedge
dz^{\nub}\wedge ,\\
&\Lambda:=i g^{\mub
\nu}\iota_{\pat_\nu}\iota_{\pat_\mub}. \\
&F\varphi^{p,q}=(p+q-n)\varphi^{p,q}.
\end{align}

Since
$$
\Lambda^\dag=-i g^{\nub
\mu}(\iota_{\pat_\mub})^\dag(\iota_{\pat_\nu})^\dag=-i g^{\nub
\mu}g_{\kb \mu}dz^\kb\wedge g_{l\nub}dz^l=L,
$$
$\Lambda$ is the adjoint operator of $L$ with respect to the inner
product $(\cdot,\cdot)$.

It is obvious that $L,\Lambda,F$ are real operators and furthermore
We have the well-known $sl_2(\R)$ Lie algebra relation, which can be
easily proved using Proposition \ref{prop-Cliff}.
\begin{prop}
\begin{equation}
[\Lambda, L]=F,\;[F,L]=-2L,\;[F,\Lambda]=2\Lambda.
\end{equation}
\end{prop}

\subsubsection{\underline{Connection and curvature}}

Now We assume that $(M,g)$ is a K\"ahler manifold with K\"ahler
metric $g=g_{i\jb}dz^i dz^{\jb}$. Let $T_\R M$ be the real tangent
bundle of $M$ and $T_\C M$ be the complex extension of $T_\R M$. Let
$g$ be the Hermitian metric on $T_\C M$ such that
$g(\frac{\pat}{\pat z^i},\frac{\bpat}{\bpat z^j})=g_{i\jb}$ and
$\nabla$ is the Chern connection defined on $T_\C M$. When
restricted to the holomorphic tangent bundle $T$, we have the
holomorphic connection, still denoted by $\nabla$ such that if
$\xi=\xi^i \frac{\pat}{\pat z^i}$ is a holomorphic vector field,
then the covariant derivative is given by
\begin{equation}
\nabla_j \xi^i=\frac{\pat \xi^i}{\pat z^j}+\sum_k \Gamma^i_{jk}
\xi^k,
\end{equation}
where $\Gamma^i_{jk}=g^{\lb i}\frac{\pat g_{k \lb}}{\pat z^j}$. When
restricted to $\bar{T}$, we have $\nabla=\bpat$. So if
$\bar{\eta}=\eta^{\ib}\frac{\bpat}{\bpat z^i}\in \bar{T}$, we have
\begin{equation}
\nabla_i \eta^{\jb}=\bpat_i \eta^{\jb}.
\end{equation}
We also have the induced connection on the cotangent bundle such
that if $\xi=\xi_i dz^i$ and $\eta= \eta_{\ib} dz^{\ib}$, then we
have
\begin{align}
&\nabla_j \xi_i=\pat_j \xi_i-\Gamma^k_{ji}\xi_k,\\
&\nabla_j \eta_{\ib}=\bpat_j \eta_{\ib}.
\end{align}
We can also define the complex conjugate connection $\bnabla$ such
that
\begin{align}
&\bnabla_i \xi^j=\pat_i \xi^j\\
&\bnabla_i \xi_j=\pat_i \xi_j\\
&\bnabla_i \eta^{\jb}=\bpat_i
\eta^{\jb}+\overline{\Gamma^j_{ik}}\xi^{\kb}\\
&\bnabla_i \eta_{\jb}=\bpat_i
\eta_{\jb}-\overline{\Gamma^k_{ij}}\xi_{\kb}.
\end{align}

Define
$$
[\nabla_i, \bnabla_j]=\nabla_i\bnabla_j-\bnabla_j \nabla_i.
$$
Then it is easy to obtain

\begin{prop}
\begin{equation}
[\nabla_i, \bnabla_j]\xi^k=-\sum_l \pat_{\jb} \Gamma^k_{il}\xi^l.
\end{equation}
\end{prop}

\begin{df} Set
\begin{equation}
R^k_{il\jb}=-\pat_{\jb}\Gamma^k_{il}.
\end{equation}
The tensor field $R^k_{il\jb}$ is called the curvature operator of
the connection $\nabla$. Define the curvature tensor of $\nabla$ to
be
\begin{equation}
R_{i\kb l\jb}=g_{\kb p}R^p_{il\jb}.
\end{equation}
The Ricci curvature is defined as
\begin{equation}
R_{l\jb}=g^{\kb i}R_{i\kb l\jb},
\end{equation}
and the scalar curvature is
\begin{equation}
R=g^{\kb i}R_{i\kb}.
\end{equation}
\end{df}

Since $\Gamma^k_{il}=g^{\pb k}\pat_i g_{l \pb}$, it is easy to show
that
\begin{align}
&R_{i\kb l\jb}=-\pat_{\jb}\pat_i g_{l\kb}+g^{\qb p}\pat_{\jb}g_{\kb
p}\pat_i g_{l\qb},\\
&R_{l\jb}=-\pat_l\pat_{\jb}\log g.
\end{align}
Here we must use the fact that $\sum_i
\Gamma^i_{il}=\frac{1}{g}\pat_l g$.

It is well-known that the curvature tensor satisfies the symmetric
properties:
\begin{enumerate}
\item
$$
R_{i\kb l\jb}=R_{l\kb i\jb}=R_{l\jb i\kb}=R_{i\jb l\kb}.
$$
\item
$$
\overline{R_{i\kb l\jb}}=R_{k\ib j\lb}.
$$
\end{enumerate}
A straightforward computation shows that
\begin{prop}
\begin{align}
&[\nabla_i, \bnabla_j]\xi^k=R^k_{il\jb} \xi^l\\
&[\nabla_i, \bnabla_j]\xi_k=-R^l_{ik\jb}\xi_l \\
&[\nabla_i, \bnabla_j]\xi_{\kb}=R^{\lb}_{\jb i \kb}\xi_{\lb}
\end{align}
\end{prop}
Note that we have
\begin{equation}
\overline{R^l_{ik\jb}}=R^{\lb}_{\ib j \kb}.
\end{equation}
Similarly, we can prove that
$$
[\nabla_i,
\bnabla_j]\xi^l_{p\qb}=R^l_{im\jb}\xi^m_{p\qb}-R^m_{ip\jb}\xi^l_{m\qb}+R^{\mb}_{\jb
i \qb}\xi^l_{p\mb}.
$$

\begin{prop} In K\"ahler case, for $\varphi=\varphi_{i_1\cdots
\jb_q}dz^{i_1}\wedge \cdots dz^{\jb_q}$ we have
\begin{align}
&\pat \varphi=\sum_k \nabla_k \varphi_{i_1\cdots \jb_q} dz^k\wedge
dz^{i_1}\wedge dz^{\jb_q}\\
&\bpat \varphi=\sum_k \bnabla_k \varphi_{i_1\cdots \jb_q}
dz^{\kb}\wedge dz^{i_1}\wedge dz^{\jb_q}\\
&(\bpat^\dag \varphi)_{i_1\cdots \jb_{q-1}}=-(-1)^p\sum_{m,n}g^{\mb
n}\nabla_n \varphi_{i_1\cdots i_p\mb \jb_1\cdots \jb_{q-1}}\\
\end{align}
\end{prop}

One can find the proof of the above proposition from \cite{MK}.

\begin{thm} For any $(p,q)$-form $\varphi=\frac{1}{p!q!}\varphi_{i_1\cdots i_p \jb_1\cdots
\jb_q}dz^{i_1}\wedge \cdots dz^{i_p}\wedge dz^{\jb_1}\wedge\cdots
\wedge dz^{\jb_q}$, we have
\begin{align}
(\Delta \varphi)_{i_1\cdots \jb_q}=&-\sum_{\mu \nu}g^{\nub
\mu}\nabla_\mu \nabla_\nub \varphi_{I\jb_1\cdots
\jb_q}\nonumber\\
&+\sum_{\lambda=1}^q\sum^p_{i=1}(-1)^{\lambda}g^{\nub
\mu}\left(\sum_{l_i}R^{l_i}_{\mu \jb_\lambda
\tau_i}\varphi_{\tau_1\cdots l_i\cdots \tau_p\nu \jb_1\cdots
\widehat{\jb_\lambda}\cdots \jb_q}-R^\mb_{\jb_\lambda\mu
\nub}\varphi_{I\mb\jb_1\cdots\widehat{\jb_\lambda}\jb_q}
 \right)
\end{align}
\end{thm}

\begin{proof} Let $I$ represent the multiple index with length $|I|=p$. Note that
\begin{equation}
(\bpat\varphi)_{I\jb_0\cdots \jb_q}=(-1)^p \sum_{\mu=0}^q (-1)^\mu
\nabla_{\jb_\mu}\varphi_{I\jb_0\cdots \widehat{\jb_\mu}\cdots
\jb_q}.
\end{equation}
Therefore, we have
\begin{align}
(\bpat^\dag \bpat \varphi)_{I\jb_1\cdots \jb_q}=-&(-1)^p \sum_{\nu,
\mu} g^{\nub \mu}\nabla_\mu (\bpat\varphi)_{I\nub \mub_1\cdots
\mub_q}\nonumber\\
=&-\sum_{\nu \mu}\left(g^{\nub \mu}\nabla_\mu \nabla_\nub
\varphi_{I\jb_1\cdots \jb_q}+\sum_{\lambda=1}^q(-1)^\lambda g^{\nub
\mu}\nabla_\mu \nabla_{\jb_\lambda}\varphi_{I\nub \cdots
\widehat{\jb_\lambda}\cdots \jb_q}\right).
\end{align}
On the other hand, we have
\begin{align}
&(\bpat^\dag \varphi)_{I\jb_2\cdots \jb_q}=-(-1)^p\sum_{\nu
\mu}g^{\nub \mu}\nabla_\mu \varphi_{I \nub \jb_2\cdots \jb_q}\\
&(\bpat \bpat^\dag
\varphi)_{I\jb_1\cdots\jb_q}=-\sum^q_{\lambda=1}(-1)^{\lambda+1}\nabla_{\jb_\lambda}(g^{\nub
\mu}\nabla_\mu \varphi_{I\nub \jb_1\cdots
\widehat{\jb_\lambda}\cdots \jb_q}).
\end{align}
Hence, we obtain
\begin{align}
(\Delta\varphi)_{I\jb_1\cdots \jb_q}&=-\sum_{\mu \nu}g^{\nub
\mu}\nabla_\mu \nabla_\nub \varphi_{I\jb_1\cdots
\jb_q}+\sum^q_{\lambda=1}(-1)^{\lambda+1}g^{\nub \mu}[\nabla_\mu,
\nabla_{\jb_\lambda}]\varphi_{I\nub \cdots
\widehat{\jb_\lambda}\cdots \jb_q}\nonumber\\
&=-\sum_{\mu \nu}g^{\nub \mu}\nabla_\mu \nabla_\nub
\varphi_{I\jb_1\cdots
\jb_q}\\
&+\sum_{\lambda=1}^q\sum^p_{i=1}(-1)^{\lambda}g^{\nub
\mu}\left(\sum_{l_i}R^{l_i}_{\mu \jb_\lambda
\tau_i}\varphi_{\tau_1\cdots l_i\cdots \tau_p\nu \jb_1\cdots
\widehat{\jb_\lambda}\cdots \jb_q}-R^\mb_{\jb_\lambda\mu
\nub}\varphi_{I\mb\jb_1\cdots\widehat{\jb_\lambda}\jb_q}
 \right)\\
&-\sum_{\lambda=1}^q\sum^p_{i=1}(-1)^{\lambda+1}g^{\nub \mu}\left(
\sum_{m\neq \lambda}R^{\kb_m}_{\jb_\lambda \mu \jb_m}\varphi_{I\nub
\jb_1\cdots \kb_m\cdots \widehat{\jb_\lambda}\cdots \jb_q}\right).
\end{align}
Since $g^{\nub \mu}R^{\kb_m}_{\jb_\lambda \mu \jb_m}$ is symmetric
about $\nub, \kb_m$ indices, but $\varphi_{I\nub \jb_1\cdots
\kb_m\cdots \widehat{\jb_\lambda}\cdots \jb_q}$ is antisymmetric
about $\nub ,\kb_m$ indices, we know the last term vanishes. Hence
we obtain the result.
\end{proof}

\subsubsection{\underline{Twisted operators}}\

Now we assume that $f$ is a holomorphic function defined on the
K\"ahler manifold $(M,g)$.

Define the twisted $\bpat$ operators
\begin{align*}
&\bpat_f=\bpat+\pat f\wedge=\bpat+f_idz^i\wedge \\
&\pat_f=\pat+\bpat \bar{f}\wedge=\pat+\overline{f_i} dz^\ib\wedge,
\end{align*}
where $f_i:=\pat_i f$.  Denote by $\pat_f^\dag$ and $\bpat_f^\dag$
their adjoint operators with respect to the K\"ahler metric. The
Laplace operator is defined as
$$
\Delta_f=(\bpat_f+\bpat_f^\dag)^2.
$$

It is easy to prove the following result.
\begin{lm}
\begin{equation}\label{lm:inte-by-parts}
\bpat(\alpha\wedge *\bar{\beta})=\bpat_f \alpha\wedge
*\bar{\beta}+\alpha\wedge *\overline{*\pat_{-f}*\beta}
\end{equation}
and so there is
\begin{align*}
&\bpat_f^\dag=-*\pat_{-f}*\;\\
&\pat_f^\dag=-*\bpat_{-f}*.
\end{align*}
\end{lm}

On the other hand, we know that
\begin{align}
&\bpat_f^\dag=\bpat^\dag+\overline{f_\nu}g^{\nub
\mu}\iota_{\pat_{\mu}}\\
&\pat_f^\dag=\pat^\dag+f_\nu g^{\mub \nu}\iota_{\pat_\mub}.
\end{align}
Obviously, $\bpat_f^,\pat_f$ are of degree $1$ and
$\bpat_f^\dag,\pat_f^\dag$ are of degree $-1$, but they don't
preserve the Hodge grading.

\begin{lm}
\begin{equation}
\begin{array}{ll}
\pat_f^2=\bpat^2_f=0\;& \bpat_f\pat_f+\pat_f\bpat_f=0\\
(\bpat_f^\dag)^2=(\pat_f^\dag)^2=0\;&\bpat_f^\dag\pat_f^\dag+\pat_f^\dag\bpat_f^\dag=0.
\end{array}
\end{equation}
\end{lm}

\begin{proof} These identities are the consequence of the
commutation relations of operators $\pat,\bpat,\pat^\dag,\bpat^\dag$
and Lemma \ref{prop-Cliff}.
\end{proof}

\subsubsection{\underline{K\"ahler-Hodge
identities}}
\begin{prop}
\begin{align}
[\pat_f,\Lambda]=-i\bpat_f^\dag,\;&[\bpat_f, \Lambda]=i\pat_f^\dag\\
[\pat_f^\dag,L]=-i\pat_f, \;&[\bpat_f^\dag,L]=i\pat_f
\end{align}
\end{prop}

\begin{proof} It suffices to prove the first identity. The others
can be obtained by taking the complex conjugate or the adjoint
action.

Now
\begin{align*}
[\pat_f,\Lambda]=&[\pat+\overline{f_k}dz^\kb\wedge , ig^{\mub
\nu}\iota_{\pat_\nu}\iota_{\pat_\mub}]\\
=&i\left\{[\pat,ig^{\mub
\nu}\iota_{\pat_\nu}\iota_{\pat_\mub}]+[\overline{f_k}dz^\kb\wedge ,
g^{\mub \nu}\iota_{\pat_\nu}\iota_{\pat_\mub}]\right\}
\end{align*}
The first term
\begin{align*}
[\pat,g^{\mub \nu}\iota_{\pat_\nu}\iota_{\pat_\mub}]=&[\nabla_k
dz^k\wedge, g^{\mub \nu}\iota_{\pat_\nu}\iota_{\pat_\mub}]=g^{\mub
\nu}\nabla_k\left\{[dz^k,\iota_{\pat_\nu}]\iota_{\pat_\mub}-\iota_{\pat_\nu}[dz^k\wedge,
\iota_{\pat_{\mub}}] \right\}\\
=&g^{\mub \nu}\nabla_\nu\iota_{\pat_\mub}=-\bpat^\dag.\\
\end{align*}
This is the classical Hodge identity.

The second term
\begin{equation*}
[\overline{f_k}dz^\kb\wedge , g^{\mub
\nu}\iota_{\pat_\nu}\iota_{\pat_\mub}]=-\overline{f_k}g^{\mub
\nu}\iota_{\pat_\nu}\delta^k_\mu=-\overline{f_\mu}g^{\mub
\nu}\iota_{\pat_\nu}=-(f_\mu dz^\mu)^\dag.
\end{equation*}
Combining the two terms, we get
$$
[\pat_f,\Lambda]=-i\bpat_f^\dag.
$$

\end{proof}

By the K\"ahler-Hodge identities, one can easily prove the following
conclusions

\begin{crl}
\begin{align*}
[\pat_f, \bpat_f^\dag]&=0,\;[\bpat_f,\pat_f^\dag]=0\\
[\pat_f, \pat_f^\dag]&=\pat_f\pat_f^\dag+\pat_f^\dag\pat_f=\Delta_f\\
[\bpat_f,\bpat_f^\dag]&=\bpat_f\bpat_f^\dag+\bpat_f^\dag\bpat_f=\Delta_f.
\end{align*}
\end{crl}

\begin{prop} The (complex) form Laplacian operator $\Delta_f$ has the
following local expression:
\begin{equation}
\Delta_f=\Delta+(g^{\mub \nu}\nabla_\nu f_l
\iota_{\pat_{\mub}}dz^l\wedge+\overline{g^{\mub \nu}\nabla_\nu f_l
\iota_{\pat_{\mub}}dz^l\wedge} )+|\nabla f|^2.
\end{equation}
\end{prop}

\begin{proof} We have
\begin{align*}
[\bpat_f^\dag,\bpat_f]&=[(\bpat+f_k dz^k\wedge)^\dag,\bpat+f_l
dz^l]\\
&=[\bpat^\dag,
\bpat]+[(\nabla_{\nub}dz^\nub)^\dag,f_ldz^l]+[\overline{f_k}
(dz^\kb\wedge)^\dag, \nabla_\lb
dz^\lb]+[\overline{f_k}(dz^k)^\dag,f_l dz^l].
\end{align*}
The second term
$$
[(\nabla_{\nub}dz^\nub)^\dag,f_ldz^l\wedge]=[\nabla_\nu, f_l]g^{\mub
\nu}\iota_\mub
dz^l\wedge+[dz^\nub)^\dag,dz^l\wedge]f_l\nabla_\nu=\nabla_\nu
f_lg^{\mub \nu}\iota_\mub dz^l\wedge.
$$
Similarly, we can compute out the third term and the fourth term is
$|\nabla f|^2$.
\end{proof}

This shows that
$\Delta_f=(\bpat_f^\dag+\bpat_f)^2=[\bpat_f^\dag,\bpat_f]$ is a real
operator and $\Delta_f=[\pat_f^\dag, \pat_f]$.

\begin{rem} Notice that the operator $\Delta_f$ does not preserve
the Hodge grading, only preserve the grading of the real forms,
i.e., $\Delta_f$ is an operator from $\Omega^p$ to $\Omega^p$. Let
$$
\varphi=\sum_{k=0}^p \mathop{\sum_{0\le i_1\le i_k\le n }}_{0\le
\jb_1\le \jb_{n-k}\le n}\varphi_{i_1\cdots i_k \jb_1\cdots
\jb_{n-k}}.
$$
Then
\begin{align}
&(\Delta_f\varphi)_{i_1\cdots i_k\jb_1\cdots \jb_{n-k}}=-\sum_{\mu
\nu}g^{\nub \mu}\nabla_\mu \nabla_\nub \varphi_{i_1\cdots i_k
\jb_1\cdots \jb_{n-k}}+[R\circ (\varphi_{I_k\Jb_{n-k}})\nonumber]_{i_1\cdots i_k\jb_1\cdots \jb_{n-k}}\\
&+[L_f\circ (\varphi_{I_k\Jb_{n-k}})]_{i_1\cdots i_k\jb_1\cdots
\jb_{n-k}} +|\nabla f|^2\varphi_{i_1\cdots i_k\jb_1\cdots \jb_{n-k}}
\end{align}
Where at a point $z\in M$, $R, L:\Lambda^p_z\to \Lambda^p_z$ are
linear maps, and $[\psi]_{i_1\cdots i_k\jb_1\cdots \jb_{n-k}}$
represent the $(i_1\cdots i_k\jb_1\cdots \jb_{n-k})$ component of
the vector (or $(k,n-k)$-tensor) $\psi$.

Notice that $R$ only depends on the curvature tensor $R_{i\jb k\lb}$
and the metric tensor $g_{\mub \nu}$ and $L$ only depends on the
metric tensor $g_{\mub \nu}$ and the tensor $(\nabla_k f_l)$.
\end{rem}

It is also interesting to discuss the commutation relations of $*$
operators with those twisted operators.

\begin{prop} The following identities hold ($M$ is only required to be a complex manifold here):
\begin{align*}
&*\bpat_f^\dag=(-1)^{p+q}\pat_{-f}*,\;*\pat_f^\dag=(-1)^{p+q}\bpat_{-f}*.\\
&*(\pat_f)=(-1)^{p+q+1}\bpat_{-f}^\dag
*,\;*\bpat_f=(-1)^{p+q+1}\pat^\dag_{-f}*\\
&*\Delta_f=-\overline{\Delta}_{-f}*.
\end{align*}
In particular, when $M$ is K\"ahler, there is
$$
*\Delta_f=-{\Delta}_{-f}*
$$
\end{prop}

\begin{ex} Given the complex 2-dimensional section-bundle system
$(\C^2,f)$ with the standard K\"ahler metric, let us give the
explicit formula for the twisted Laplacian equation. Because of the
$*$-action, it suffices to consider the $0,1,2$-forms.
\begin{enumerate}
\item The case that $\psi$ is a function. then the equation
$\Delta_f\psi=0$ is a complex scalar Schr\"odinger equation:
\begin{equation}
\Delta \psi+|\nabla f|^2\psi=0.
\end{equation}
\item The case that $\psi$ is a 1-form. Assume that
$$
\psi=\psi_1dz^1+\psi_2dz^2+\psi_\eb dz^\eb+\psi_\rb dz^\rb.
$$
The vector Schr\"odinger equation has four component equations:
$$
\begin{cases}
\Delta\psi_1-(f_{11}\psi_\eb+f_{12}\psi_\rb)+|\nabla f|^2\psi_1=0\\
\Delta\psi_2-(f_{21}\psi_\eb+f_{22}\psi_\rb)+|\nabla f|^2\psi_2=0\\
\Delta\psi_\eb-(f_{\eb\eb}\psi_1+f_{\eb\rb}\psi_2)+|\nabla
f|^2\psi_\eb=0\\
\Delta\psi_\rb-(f_{\rb\eb}\psi_1+f_{\rb\rb}\psi_2)+|\nabla
f|^2\psi_\rb=0.
\end{cases}
$$
\end{enumerate}
If we set $\psi=(\psi_1,\psi_2,\psi_\eb,\psi_\rb)$, then the action
of the linear operator $L_f\circ$ can be written as the matrix form:
$$
L_f\circ(\psi)=
\begin{pmatrix}
0 &0&f_{11}&f_{12}\\
0&0&f_{21}&f_{22}\\
f_{\eb\eb}&f_{\eb\rb}&0&0\\
f_{\rb\eb}&f_{\rb\rb}&0&0.
\end{pmatrix}
$$
All the entries of $L_f$ consists of the second order derivatives of
$f$.
\item The case that $\psi$ is a 2-form. $\psi$ has 6 components and
has the form:
$$
\psi=\psi_{12}dz^1\wedge dz^2+\psi_{1\eb}dz^1\wedge
dz^\eb+\psi_{1\rb}dz^1\wedge dz^\rb+\psi_{2\eb}dz^2\wedge
dz^\eb+\psi_{2\rb}dz^2\wedge dz^\rb+\psi_{\eb\rb}dz^\eb\wedge
dz^\rb.
$$
The vector Schr\"odinger equation has 6 components:
$$
\begin{cases}
\Delta\psi_{12}-(f_{21}\psi_{1\eb}-f_{22}\psi_{1\rb}+f_{11}\psi_{2\eb}+f_{12}\psi_{2\rb})+|\nabla f|^2\psi_{12}=0\\

\Delta\psi_{1\eb}-(f_{\eb\rb}\psi_{12}-f_{12}\psi_{\eb\rb})+|\nabla f|^2\psi_{1\eb}=0\\

\Delta\psi_{1\rb}-(f_{\rb\rb}\psi_{12}+f_{11}\psi_{\eb\rb})+|\nabla f|^2\psi_{1\rb}=0\\

\Delta\psi_{2\eb}-(f_{\eb\eb}\psi_{12}+f_{22}\psi_{\eb\rb})+|\nabla f|^2\psi_{2\eb}=0\\

\Delta\psi_{2\eb}+(f_{\rb\eb}\psi_{12}-f_{21}\psi_{\eb\rb})+|\nabla f|^2\psi_{2\rb}=0\\

\Delta\psi_{\eb\rb}+(f_{\rb\eb}\psi_{1\eb}-f_{1\eb}\psi_{1\rb}+f_{\rb\rb}\psi_{2\eb}-f_{\eb\rb}\psi_{2\rb})+
|\nabla f|^2\psi_{\eb\rb}=0.\\
\end{cases}
$$
If we set
$\psi=(\psi_{12},\psi_{1\eb},\psi_{1\rb},\psi_{2\eb},\psi_{2\rb},\psi_{\eb\rb})$,
Then the matrix $L_f$ has the form
$$
L_f\circ(\psi)=
\begin{pmatrix}
0 &-f_{21}&-f_{22}&f_{11}&f_{12}&0\\
-f_{\eb\rb}&0&0&0&0&-f_{12}\\
-f_{\rb\rb}&0&0&0&0&-f_{11}\\
f_{\eb\eb}&0&0&0&0&f_{22}\\
f_{\rb\eb}&0&0&0&0&-f_{21}\\
0&f_{\rb\eb}&-f_{1\eb}&f_{\eb\eb}&-f_{\eb\rb}&0
\end{pmatrix}.
$$
\end{ex}

\subsubsection{\underline{$N=2$ supersymmetric algebra}}

Define operators
\begin{align*}
&d_f=\pat_f+\bpat_f,\;d^c_f=-i(\pat_f-\bpat_f)\\
&d^\dag_f=\pat_f^\dag+\bpat_f^\dag,\;(d^c_f)^\dag=i(\pat_f^\dag-\bpat_f^\dag)\\
&\Box_f=(d_f+d^\dag_f)^2=[d_f,d^\dag_f]=d_fd^\dag_f+d^\dag_fd_f\\
&\Box^c_f=(d^c_f+(d^c_f)^\dag)^2=[d_f^c,(d^c_f)^\dag]=d^c_f(d_f^c)^\dag+(d_f^c)^\dag
d_f^c
\end{align*}
Those operators satisfy
\begin{prop}
\begin{align*}
&\Box_f=\Box_f^c=2\Delta_f\\
&[d_f,d^c_f]=[d^\dag_f,(d^c_f)^\dag]=0,\;[d_f,(d^c_f)^\dag]=[d_f^\dag,d^c_f]=0.
\end{align*}
In particular, the Laplacian $\Box_f$ commutes with all above
operators.
\end{prop}

For the number operator $F$, we have

\begin{prop}
\begin{align*}
&[F,\pat_f]=-\pat_f,\;[F,\bpat_f^\dag]=\bpat_f^\dag\\
&[F,\bpat_f]=-\bpat_f,\;[F,\pat_f^\dag]=\pat_f^\dag.
\end{align*}
\end{prop}

Now the operators $L,\Lambda,F$ and $\Box_f$ generates the even part
of a Lie superalgebra, where $L,\Lambda, F$ generates $sl(2,\R)\cong
spin(3)$ and $\Box_f$ generates $u(1)\cong so(2)$. The odd operators
$\pat_f$ and $\bpat_f^\dag$ span a spinor representation $S$, since
there hold
\begin{align*}
&[L,\pat_f]=0,\;[\Lambda,\pat_f]=i\bpat_f^\dag,\;[F,\pat_f]=-\pat_f\\
&[L,\bpat_f^\dag]=-i\pat_f,\;[\Lambda,\bpat_f^\dag]=0,\;[F,\bpat_f^\dag]=\bpat_f^\dag\\
&[\pat_f,\pat_f]=0,[\pat_f,
\bpat_f^\dag]=0,\;[\bpat_f^\dag,\bpat_f^\dag]=0.
\end{align*}
The dual representation $S^*$ is given by the operators $\bpat_f$
and $\pat_f^\dag$. The pairing is giving by the relations
\begin{align*}
&[\pat_f,\pat_f^\dag]=\frac{1}{2}\Box_f,\;[\bpat_f,\bpat_f^\dag]=\frac{1}{2}\Box_f\\
\end{align*}

Hence we obtain a $spin_\C(3)$ supersymmetric algebra.

\subsection{Spectrum theory of Schr\"odinger operators for differential forms}

\emph{In this section, we always assume that $(M,g)$ is a complete
non-compact manifold with bounded geometry.}

\subsubsection{\underline{Sobolev Norms}} Since we have the decomposition
$$
\Lambda^p=\oplus_{\mu+\nu=p}\Lambda^{\mu,\nu},
$$
the $L^2$-inner product is defined as
$$
(\cdot,\cdot)=\oplus_{\mu+\nu=p}(\cdot,\cdot)_{\mu,\nu}.
$$
In local chart, $\Lambda^p$ has a basis
$$
\left\{dz^{i_1}\wedge\cdots dz^{i_k}\wedge dz^{\jb_1}\cdots\wedge
dz^{\jb_l}, 1\le i_1\le \cdots\le i_k\le n, 1\le \jb_1\cdots\le
\jb_l\le n,k+l=p\right\}.
$$
By Pascard's rule, $\sum_{k=0}^p
\binom{n}{k}\binom{n}{p-k}=\binom{2n}{p}$ is just the dimension of
$\Lambda^p$.

For any $\varphi\in \Omega^p$, we can define the Sobolev norm
\begin{equation}
||\varphi||_{k,2}:=\left(\sum_{|a|+|b|\le k}||\nabla^a \bnabla^b
\varphi||^2 \right)^{1/2}
\end{equation}
and the Sobolev space $W^{k,2}(\Lambda^p)$ consisting of all weak
differentiable $p$-forms having up to $k$-order $L^2$ integrable
derivatives. Let $W_0^{k,2}(\Lambda^p)$ be the closed subspace in
$W^{k,2}(\Lambda^p)$ which is the closure of the compactly supported
forms. In addition, we denote by $\Omega^k_{(2)}(M)=\Omega^k(M)\cap
L^2(\Lambda^k(M))$ and $\Omega_0^k(M)$ the set of the smoothly
compactly supported $k$-forms.

\subsubsection{\underline{Some analytic theorems on Riemnnian manifold with bounded
geometry}}\

\

The following $L^1$-Stokes theorem was proved by Gaffney \cite{Ga},
which is obvious in Euclidean space.

\begin{thm} Let $M$ be a $n$-dimensional orientable complete Riemannian manifold
whose Riemannian tensor is of $C^2$. Let $\gamma$ be a $n-1$-form of
class $C^1$ with the property that both $\gamma$ and $d\gamma$ are
in $L^1$. Then $\int_M d\gamma=0$.
\end{thm}

\begin{thm}[Density theorem]\label{thm:density} Let $(M,g)$ be a complete Riemannian
manifold with positive injectivity radius and let $k\ge 2$ be an
integer. Suppose that there exists a positive constant $C$ such that
for any $j=0,\cdots,k-2$, $|\nabla^j Ricc|\le C$. Then for any $p\ge
1$, $W^{k,p}(M)=W^{k,p}_0(M)$.
\end{thm}

\begin{thm}[Sobolev embedding theorem]\label{thm:sobo-embe} The Sobolev embeddings are
valid for complete manifolds with Ricci curvature bounded from below
and positive injectivity radius.
\end{thm}

One can find the two theorems and proofs in \cite{He}. The above two
theorems can be applied to our case:

\begin{crl} Let $(M,g)$ be a complete non-compact K\"ahler manifold
with bounded geometry. Then
$W^{k,2}(\Lambda^p)=W^{k,2}_0(\Lambda^p)$ for any $k$, and Sobolev
embedding theorem holds.
\end{crl}

\subsubsection{\underline{Close extension of operators}}\

The multiplication operator $\pat f:\Omega^{p}\to \Omega^{p+1}$ is a
closable operator. Its closure $\widetilde{\pat f}$ has domain
$$
\Dom(\widetilde{\pat f})=\{\psi\in \Omega_{(2)}^p| \int_M |\nabla
f|^2|\psi|^2<\infty\}.
$$
Similarly, the formal adjoint operator
$(df\wedge)^\dag=\overline{f_i} g^{\ib j}\iota_{\pat_j}$ is also a
closable operator and its closure $\widetilde{(df\wedge)^\dag}$ has
domain
$$
\Dom(\widetilde{(df\wedge)^\dag})=\{\psi\in \Omega_{(2)}^{p+1}|
\int_M |\overline{f_i} g^{\ib j}\iota_{\pat_j}\psi|^2<\infty\}.
$$
The operator $(\widetilde{\pat f})^\dag=\widetilde{(\pat f)^\dag}$.

If $\psi\in \Omega^{p+1}_{(2)}, \varphi\in \Omega^p_0(M)$, by Stokes
theorem, we have
$$
(\bpat \psi, \varphi)=(\psi, \bpat^\dag \varphi).
$$
Since $\Omega^k_0(M)$ is a dense set in $L^2$ space, $\bpat$ as the
formal adjoint operator has maximal closed extension
$\widetilde{\bpat}_M$ with domain
$$
\Dom(\widetilde{\bpat}_M)=\{\psi\in L^2(\Lambda^{p+1}(M))|\bpat f\in
L^2({\Lambda^p}(M))\}=W^{1,2}(\Lambda^p(M)).
$$
On the other hand, with the Dirichlet boundary condition, $d$ has
the closed extension, the minimal extension $\widetilde{\bpat}_m$
with domain $W^{1,2}_0(\Lambda^p(M))$. However because of the
density theorem, Theorem \ref{thm:density}, we have
$$
\widetilde{\bpat}_M=\widetilde{\bpat}_m.
$$
This means that $\bpat$ has a unique close extension. Similarly, we
can prove that the operator $\bpat^\dag$ has a unique close
extension. We drop off the symbol $\widetilde{}$ if no confusion
occurs.

Now $\bpat_f=\bpat+\pat f\wedge $ is a unbounded closable operator,
its closure has the domain
$$
\Dom(\bpat_f)=\{\psi\in L^2(\Lambda^p)| \int_M |\bpat_f
\psi|^2+|\psi|^2<\infty \}
$$

In the rest of the paper, we only consider the closure of $\bpat_f$
and $\bpat_f^\dag$ and don't think $\bpat_f$ as the sum of two
unbounded operators.

The Laplace operator $\Delta_f=\bpat_f\bpat_f^\dag+\bpat^\dag\bpat$
is lower bounded real symmetric operator. There is an associated
closable quadratic form
$$
Q_f(\psi,\varphi)=(\bpat_f \psi,\bpat_f
\psi)+(\bpat_f^\dag\psi,\bpat_f^\dag\psi).
$$
By the functional theorem (see \cite{RS0}, Theorem VIII.15), the
closure of $Q_f$ uniquely determines an self-adjoint extension
$\widetilde{\Delta_f}$ of $\Delta_f$. We will still use $\Delta_f$
to represent this self-adjoint extension.

Denote by
$$
\Delta_f=H_f^0+H^1_f,
$$
where
\begin{equation}
H^f_0:=-\sum_{\mu \nu}g^{\nub \mu}\nabla_\mu \nabla_\nub +|\nabla
f|^2,
\end{equation}
and
\begin{equation}
H^1_f:=R\circ (\varphi_{I_k\Jb_{n-k}})+L_f\circ
(\varphi_{I_k\Jb_{n-k}}).
\end{equation}
In one coordinate chart $U$, these operators can be viewed as
operators acting on the vector-valued $L^2$ space $L^2(U,
\C^{p(n)})$, where $p(n)=\binom{2n}{p}$ is the dimension of
$\Lambda^p$.

\subsubsection{Tameness of the section-bundle system}

\begin{df}\label{df:tameness} The section bundle system $\{(M, g),f\}$ is said to be fundamental
tame, if there exists a compact set $K$ and a constant $C_0>0$ such
that $|\nabla f|>C_0$ outside a compact subset $K\subset M$.

The section-bundle system is said to be strongly tame, if for any
constant $C>0$, there is
\begin{equation}
|\nabla f|^2-C|\nabla^2 f|\to \infty, \;\text{as}\;d(x, x_0)\to
\infty.
\end{equation}
Here $d(x,x_0)$ is the distance between the point $x$ and the base
point $x_0$.
\end{df}

Now we have the fundamental theorem of our theory.

\begin{thm}\label{thm:main-1} Suppose that $(M,g)$ is a K\"ahler manifold with bounded
geometry. If $\{(M.g),f\}$ is a strongly tame section-bundle system,
then the form Laplacian $\Delta_f$ has purely discrete spectrum and
all the eigenforms form a complete basis of the Hilbert space
$L^2(\Lambda^\bullet(M))$.
\end{thm}

\begin{proof} Since the section-bundle system is strongly tame, we
know that
$$
|\nabla f(x)|\to\infty,\;\text{as}\;d(x,x_0)\to \infty.
$$
Therefore by Corollary \ref{crl:spect-1}, we know that $H_f^0$ has
purely discrete spectrum. Now we want to prove that $H_f^1$ is a
compact perturbation of $H_f^0$.

At first, since $(M,g)$ has bounded geometry, there exists a
universal constant $C_R$ such that at point $x$,
\begin{equation}
|(R\circ \varphi,\varphi )(x)|\le C_R(\varphi,\varphi).
\end{equation}
Secondly, for any $\epsilon>0$, we can find a compact set
$K_\epsilon$ such that
\begin{equation}
|(L_f\circ \varphi,\varphi)(x)|\le \epsilon (|\nabla
f|^2\varphi,\varphi)(x)
\end{equation}
hold on $M-K_\epsilon$.

Since $f$ and $(M,g)$ are smooth, we can find a constant
$C_{\epsilon}>0$ such that for any $x\in M$,
\begin{equation}
|(L_f\circ \varphi,\varphi)|\le \epsilon (|\nabla
f|^2\varphi,\varphi)+C_\epsilon (\varphi,\varphi).
\end{equation}
Combining the inequalities of $R$ and $L_f$, we have
\begin{equation}
|(H_f^1\varphi,\varphi)|\le \epsilon (|\nabla
f|^2\varphi,\varphi)+(C_\epsilon+C_R)(\varphi,\varphi).
\end{equation}
Therefore, we obtain
\begin{equation}
|(H_f^1\varphi,\varphi)|\le
\epsilon(H_f^0\varphi,\varphi)+(C_\epsilon+C_R)(\varphi,\varphi).
\end{equation}
Now by Theorem \ref{thm-fund-compact}, we get the final conclusion.
\end{proof}

As the application of our fundamental theorem, Theorem
\ref{thm:main-1}, we will consider some important examples in the
following part.

\subsubsection{\underline{Hypersurface section-bundle system $(\C^{N+1},
W)$}}\

\

Let $W:\C^{N+1}\to \C$ be a quasi-homogeneous polynomial of type
$(q_0,\cdots,q_N)$, i.e., for all $\lambda\in \C$, there is
\begin{equation}
W(\lambda^{q_0}z_0,\cdots,\lambda^{q_N}z_N)=\lambda
W(z_0,\cdots,z_N).
\end{equation}

$W$ is called non-degenerate, if $0$ is its only isolated critical
point in $\C^{N+1}$. Equivalently, if we write
$$
W=\sum^s_{l=1}W_l=\sum^s_{l=1}c_j\mathop{\prod^{N}}_{j=0}z_j^{a_{lj}},
$$
then the matrix $A=(a_{lj})_{s\times (N+1)}$ has rank $N+1$ and
$$
A\cdot
\begin{pmatrix}
q_0\\
\vdots\\
q_N
\end{pmatrix}
=\begin{pmatrix}
1\\
\vdots\\ 1
\end{pmatrix}.
$$

The non-degeneracy of $W$ also implies that $\{W=0\}$ defines a
smooth hypersurface in the weighted projective space
$\IP^N_{(k_0,\cdots,k_N)}$, where $q_i=\frac{k_i}{d},(k_i,d)=1$.

\begin{df} Let $\C^N$ be the canonical K\"ahler manifold with
standard flat metric. Let $W:\C^{N+1}\to \C$ be a non-degenerate
quasi-homogeneous polynomial. The combination $(\C^{N+1},W)$ is
called a hypersurface section-bundle system.
\end{df}

\begin{thm}[\cite{FS}, Theorem 3.12]\label{thm:main-hypersurface} Let $(\C^{N+1},W)$ be a
hypersurface section-bundle system. Suppose that the weights $q_i\le
\frac{1}{2}$, then the system $(\C^{N+1},W)$ is strongly tame.
Therefore, the form Laplacian $\Delta_W$ defined on $\C^{N+1}$ has
purely discrete spectrum.
\end{thm}

Let $G_j$ be monomials of quasi-homogeneous weight less than $1$.
Then the following deformation
\begin{equation}
F(z,t)=W+\sum_j t_j G_j
\end{equation}
is called by Arnold, Gusein-Zade and Varchenko (\cite{AGV}, PP.416)
as the lower deformation of $W$. Physicists call it as the "relative
deformation".

\begin{thm}\label{thm:main-lowerdeform} Let $W$ be a quasi-homogeneous polynomial with homogeneous weight $1$ and of type $(q_0,\cdots,q_N)$.
Suppose that $q_i\le 1/2$. Then for any deformation parameter $t$,
the section-bundle system $(\C^{N+1}, F(z,t))$ is strongly tame.
Therefore, the form Laplacian $\Delta_{F(z,t)}$ defined on
$\C^{N+1}$ has purely discrete spectrum.
\end{thm}

The proof of Theorem \ref{thm:main-lowerdeform} is similar to the
proof of Theorem \ref{thm:main-hypersurface}. For the convenience of
the reader, we give the proof here.

\begin{proof} The proof is based on the following important
inequality:

\begin{lm}[\cite{FJR1}, Theorem 5.7]\label{thm-new}
Let $W \in \mathbb{C}[x_1, \dots, x_N]$ be a non-degenerate,
quasi-homogeneous polynomial with weights $q_i:=\wt(x_i)<1$ for each
variable $x_i,i=1, \dots,N$. Then for any $t$-tuple $(u_1, \dots,
u_N) \in \mathbb{C}^N$ we have

\[ |u_i| \leq C \left(\sum^N_{i=1}\left|\frac{\partial W}{\partial x_i}(u_1, \dots,
u_N)\right|+1 \right)^{\delta_i},\] where
$\delta_i=\frac{q_i}{\min_j(1-q_j)}$ and the constant $C$ depends
only on $W$. If $q_i\leq 1/2$ for all $ i \in \{1, \dots, N\}$, then
$\delta_i\le 1$ for all $ i \in \{1, \dots, N\}$. If $q_i<1/2$ for
all $ i \in \{1, \dots, N\}$, then $\delta_i<1$ for all $ i \in \{1,
\dots, N\}$.
\end{lm}

We set $F(z,t)=:W(z)+G(z,t)$. Above all, by Lemma \ref{thm-new} we
know that
\begin{equation}\label{eq:infty}
|\pat W|\to \infty, \;\text{as}\;|z|\to \infty.
\end{equation}

Now let $W_l=c_l\prod z_i^{b_{li}}$ be a quasi-homogeneous monomial
with weight $1$ (not necessary a monomial of $W$). By Lemma
\ref{thm-new}, we have for $p\neq q$ (the proof for $p=q$ case is
the same),
\begin{align}
&|\nabla_p\nabla_q W_l|\le
C|z_p|^{b_{lp}-1}|z_q|^{b_{lq}-1}\prod_{i\neq p,q}|z_i|^{b_{li}} \le
C \left(\sum^N_{i=1}\left|\frac{\partial W}{\partial z_i}(z_1,
\dots, z_N)\right|+1 \right)^{\sum
b_{li}\delta_i-\delta_p-\delta_q}\nonumber\\
&=C\left(\sum^N_{i=1}\left|\frac{\partial W}{\partial z_i}(z_1,
\dots, z_N)\right|+1
\right)^{\frac{1}{\min_j(1-q_j)}-\delta_p-\delta_q}\le
C\left(\sum^N_{i=1}\left|\frac{\partial W}{\partial z_i}(z_1, \dots,
z_N)\right|+1 \right)^{2-2\hat{\delta}_0}\nonumber.
\end{align}
Here $\hat{\delta}_0=\frac{\min q_i}{\min_j(1-q_j)}$.

Therefore, we have
\begin{equation}\label{ineq:lowerdeform-1}
|\nabla^2 W_l|\le C(|\nabla W|^2+1)^{-2\hat{\delta}_0}|\nabla
W|^2+C.
\end{equation}

In particular, we have
\begin{equation}
|\nabla^2 W|\le C(|\nabla W|^2+1)^{-2\hat{\delta}_0}|\nabla W|^2+C.
\end{equation}
where $C_0,C_1$ are constants.

We can assume that $G(z,t)$ is a monomial, since its weight is less
than $1$, there exists polynomial $\hat{G}$ with weight $1$ such
that the power of each variable of $z_i$ is no less than the
corresponding power of $G$. Hence as $|z_i|>>1$ we have
\begin{align}
|\nabla G|^2&\le |\nabla\hat{G}|^2\le
C\left(\sum^N_{i=1}\left|\frac{\partial W}{\partial z_i}(z_1, \dots,
z_N)\right|+1 \right)^{2-\hat{\delta}_0}\\
&\le C(|\nabla
W|^2+1)^{-\hat{\delta}_0}||\nabla W|^2+C\\
|\nabla^2 G|&\le |\nabla^2 \hat{G}|\le C(|\nabla
W|^2+1)^{-2\hat{\delta}_0}|\nabla W|^2+C.\\.
\end{align}

Therefore, we have for any $C_0>0$,
\begin{align*}
&|\nabla F|^2-C_0|\nabla^2 F|\ge |\nabla W|^2+|\nabla G|^2-2|\nabla W||\nabla G|-C|\nabla^2 W|-C|\nabla^2 G|\\
&\ge \frac{1}{2}|\nabla W|^2-C(|\nabla
W|^2+1)^{-\hat{\delta}_0}||\nabla W|^2-C|\nabla^2 W|-C(|\nabla
W|^2+1)^{-2\hat{\delta}_0}|\nabla W|^2-C\\
&\ge \frac{1}{100}|\nabla W|^2-C.
\end{align*}
This shows that $(\C^{N+1},F(z,t))$ is a strongly tame
section-bundle system, by Theorem \ref{thm:main-1} we get the
conclusion.
\end{proof}

\begin{rem} One type of quasi-homogeneous polynomial is extremely
interesting in the study of the Laudau-Ginzburg model in Topological
field theory. Let
\begin{equation}
W=\sum^N_{l=1}c_j\mathop{\prod^{N}}_{j=0}z_j^{a_{lj}}.
\end{equation}
Then the exponents can be written as a $N\times N$ matrix
$A=(a_{lj})$. In this case, $W$ is called an invertible singularity.
One can transpose the matrix to obtain an exponent matrix $A^T$ and
get a singularity $W^T$. $(W,W^T)$ forms a mirror pair in
Laudan-Ginzburg model which was observed by Beglund and H\"ubsch
\cite{BH}. Kreuzer and Skarke \cite{KS} proved that an invertible
singularity $W$ is non-degenerate (i.e.,$\rank A=N$,)if and only if
it can be written as a sum of (decoupled) invertible polynomials of
one of the following three basic types:
\begin{align*}
&W_{\text{Fermat}}=z^a.\\
&W_{\text{loop}}=z_1^{a_1}z_2+z_2^{a_2}z_3+\cdots+z^{a_{N-1}}_{N-1}z_N+z^{a_N}_N
z_1\\
&W_{\text{chain}}=z_1^{a_1}z_2+z_2^{a_2}z_3+\cdots+z_{N-1}^{a_{N-1}}z_N+z_N^{a_N}.
\end{align*}
If we assume that all $a_i\ge 2$, then the weight $q_i\le 1/2$ and
the three type polynomials are all strongly tame.

All $A,D,E$ polynomials, unimodal singularities and etc. are
strongly tame according to our definition. The reader can refer to
\cite{AGV} and some recent papers \cite{Kr},\cite{KP}.
\end{rem}

\subsubsection{\underline{Section bundle system with Laurent polynomial
potential}}\

\

Let $T$ be the torus
$Spec\C[z_1,\cdots,z_n,z_1^{-1},\cdots,z_n^{-1}]\cong (\C^*)^n$ and
let $f\in \C[z_1,\cdots,z_n,z_1^{-1},\cdots,z_n^{-1}]$ be a Laurent
polynomial having the form
\begin{equation}
f(z_1,\cdots,z_n)=\sum_{\alpha=(\alpha_1,\cdots, \alpha_n)\in \Z^n}
c_\alpha z^\alpha.
\end{equation}

The Newton polyhedron $\Delta=\Delta(f)$ of $f$ is the convex hull
of the integral points $\alpha=(\alpha_1,\cdots,\alpha_n)\in \Z^n$.
$f$ is said to convenient if $0$ is in the interior of the Newton
polyhedron.

\begin{df} Let $\Delta'\subset \Delta$ be an $l$-dimensional face of
$\Delta$. Define also the Laurent polynomial with the Newton
polyhedron $\Delta'$'
\begin{equation}
f^{\Delta'}(z)=\sum_{\alpha'\in \Delta'} c_{\alpha'}z^{\alpha'}.
\end{equation}
\end{df}

For any Laurent polynomial $g$, denote by $g_i,1\le i\le n$ the
logarithmic derivatives of $g$:
$$
g_i(z)=z_i\frac{\pat}{\pat z_i}g(z).
$$

\begin{df} A Laurent polynomial $f$ is called non-degenerate if for
every $l$-dimensional edge $\Delta'\subset \Delta (l>0)$ the
polynomial equations
$$
f^{\Delta'}(z)=f_1^{\Delta'}(z)=\cdots=f^{\Delta'}_n(z)=0
$$
has no common solutions in $T$.
\end{df}

\begin{prop}\label{prop:Laurent-strong-1} If $f$ is a convenient and
non-degenerate Laurent polynomial, then $f$ is strongly tame.
\end{prop}

\begin{proof} First we prove the sufficiency. $f$ is convenient
means that the origin is in the interior of the Newton polyhedron.
This is equivalent to say that for any ray
$\beta=(\beta_1,\cdots,\beta_n)\in \R^n$, there exists two integer
points $\alpha_\pm$ as the vertices of the Newton polyhedron such
that the standard inner products satisfy $\langle \alpha_-,\beta
\rangle<0$ and $\langle \alpha_+,\beta\rangle >0$.

Do coordinate transformation: let $z_i=e^{t_i}$. Then the metric
$g=i\sum_i dt^i\wedge \overline{dt^i}$ and the corresponding
K\"ahler connection is trivial. The function
$$
f(z_1,\cdots,z_n)=f(t_1,\cdots,t_n)=\sum_\alpha a_\alpha e^{\langle
\alpha,t\rangle}.
$$
Let $F(t)=|\nabla f|^2-C|\nabla^2 f|$ for any $C>0$. Then
\begin{equation}
F(t)\ge \sum_i \left|\sum_\alpha a_\alpha \alpha_i e^{\langle
\alpha,t\rangle}\right|^2-C\sum_{ij}\left|\sum_\alpha a_\alpha
\alpha_i\alpha_j e^{\langle \alpha,t\rangle}\right|,
\end{equation}
where $\alpha=(\alpha_1,\cdots,\alpha_n)$.

Now taking any ray $t_i=\beta_i |t|,\beta_i\in \C,|\beta|\le 1$.
Then real part $R(\beta)=(R(\beta_1),\cdots,R(\beta_n))$ defines a
line in $\R^n$. Since $f$ is convenient, there exists at least two
directions $\alpha_\pm$ such that $\langle R(\beta),\alpha_-\rangle
<0$ and $\langle R(\beta),\alpha_+\rangle >0$. We arrange all the
vertices $\alpha$ such that
\begin{equation}
\langle R(\beta),\alpha_0\rangle\le \langle
R(\beta),\alpha_1\rangle\cdots<\cdots\le 0\le \cdots<\cdots \langle
R(\beta),\alpha_s\rangle,
\end{equation}
where $s$ is the number of the vertices of the Newton polyhedron.

Denote by $M(\beta)=\langle R(\beta),\alpha_s\rangle$, $M$ the set
of all of $\alpha$ such that $\langle
R(\beta),\alpha\rangle=M(\beta))$ and $m^+$ the set of all $\alpha$
such that $0<\langle R(\beta),\alpha\rangle<M(\beta)$.

Hence we have along the line $\beta$ that
\begin{align}\label{inequ-laurent-1}
\sum_{ij}\left|\sum_\alpha a_\alpha \alpha_i\alpha_j e^{\langle
\alpha,t\rangle }\right|\le & C |e^{M(\beta)|t|}|\le \epsilon
|e^{M(\beta)|t|}|^2+C.
\end{align}

On the other hand, we have along the line $\beta$:
\begin{align}\label{inequ:laurent-2}
&\sum_i\left|\sum_\alpha a_\alpha \alpha_i e^{\langle
\alpha,t\rangle}\right|^2\ge \sum_i \left|\sum_{\alpha\in M}
a_\alpha \alpha_i e^{\langle \alpha,\beta\rangle|t|}\right|^2-
C\sum_{\alpha\in m^+}|e^{\langle \alpha,\beta\rangle|t||}|^2-C
\end{align}

Claim: for any $\theta \in [0,2\pi]$, there holds
\begin{equation}\label{inequ-laurent-2.5}
\sum_i\left|\sum_{\alpha\in M} a_\alpha \alpha_{i_0} e^{i Im(\langle
\alpha,\beta\rangle)\theta}\right|>0.
\end{equation}
If this is true, then we have
\begin{equation}\label{inequ-laurent-3}
\sum_i \left|\sum_{\alpha\in M} a_\alpha \alpha_i e^{\langle
\alpha,\beta\rangle|t|}\right|^2\ge C_0 |e^{M(\beta)|t|}|^2.
\end{equation}
Combining (\ref{inequ-laurent-1}), (\ref{inequ:laurent-2}) and
(\ref{inequ-laurent-3}) and using Young inequality, we get
$$
F(t)\to \infty,
$$
as $|t|\to \infty$.

To prove the Claim, we prove by contradiction. Suppose that
(\ref{inequ-laurent-2.5}) is not true, then there exists a
$\theta_0$ such that for any $i=1,\cdots,n$, there is
\begin{equation}
\sum_{\alpha\in M} a_\alpha \alpha_{i} e^{i Im(\langle
\alpha,\beta\rangle)\theta_0}=0,i=1,\cdots,n,
\end{equation}
Multiplying the above equality by $Re(\beta_i)$ and taking the sum,
noticing that for any $\alpha\in M$ $\sum_i \alpha_i
Re(\beta_i)=M(\beta)$, we obtain
\begin{equation}
\sum_{\alpha\in M} a_\alpha  e^{i Im(\langle
\alpha,\beta\rangle)\theta_0}=0.
\end{equation}
This contradicts with the fact that $f$ is non-degenerate.
Therefore, we proved the Claim.
\end{proof}

Now we have the following theorem.

\begin{thm}\label{thm:spec-Laurent} Suppose that $f$ is convenient and nondegenerate. Then the form
Laplacian $\Delta_f$ defined on $T$ has purely discrete spectrum.
\end{thm}

\begin{rem} Nondegenerate Laurent polynomials was initially studied
by Kouchnirenko \cite{Ko} in singularity theory. Later it becomes a
very popular objects in algebraic geometry because of their
connection to toric geometry \cite{Ba} and the other algebraic
theories (see \cite{CVj} and references there). We knew the concepts
of nondegeneracy and convenience from the paper by Sabbah and Douai
(see \cite{Sa}, \cite{Do}, where they constructed the Frobenius
structure for such polynomials by algebraic method.

On the other hand, a Laurent polynomial can define an affine
hypersurface on the algebraic torus $T$.
$$
Z_f=\{z\in T|f(z)=0\}.
$$

Let $\Delta$ be the Newton polyhedron of $f$. Then the corresponding
toric variety $\IP_\Delta$ is a compactification of the
$n$-dimensional torus $T_\Delta=T$ by algebraic tori $T_{\Delta'}$
for any faces $\Delta'\subset \Delta$. Let $\bar{Z}_f$ be the
closure of $Z_f$ in $\IP_\Delta$. For any face $\Delta'\subset
\Delta$ we have the hypersurface $Z_{f,\Delta'}=\bar{Z}_f\cap
T_{\Delta'}$ in $T_{\Delta'}$. Then Batyrev \cite{Ba} gave another
geometrical characterization to the nondegeneracy:

\emph{A Laurent polynomial $f$ is nondegenerate if and only if
$Z_{f,\Delta'}$ is a smooth affine subvariety in $T_\Delta'$ of
codimension $1$ for any face $\Delta'\subset \Delta$.}\

Non-degenerate and convenient Laurent polynomials become important
examples in the study of mirror symmetry. For example, the
polynomial
\begin{equation}
f(z_1,\cdots,z_n)=z_1+\cdots z_n+\frac{1}{z_1\cdots z_n},
\end{equation}
is taken as a superpotential in the Laudau-Ginzburg model of
topological field theory which is the mirror object of the
projective space $\IP^n$. It is well-known that the $B$-model of
this Laudau-Ginzburg model should coincide with the $A$-model, i.e,
the Gromov-Witten theory for $\IP^n$.

\end{rem}

\subsection{$L^2$-cohomology and Hodge decomposition theorem}\

\

\subsubsection{\underline{$L^2$-cohomology of $\bpat_f$ operator}}\

Since $\bpat_f^2=0$, we have the $\bpat_f$-complex for smooth
sections,
$$
\cdots\to \Omega^{k-1}(M)\xrightarrow{\bpat_f}
\Omega^{k}(M)\xrightarrow{\bpat_f}\Omega^{k+1}(M)\to\cdots.
$$
and for compactly supported forms:
$$
\cdots\to \Omega^{k-1}_0(M)\xrightarrow{\bpat_f}
\Omega^{k}_0(M)\xrightarrow{\bpat_f}\Omega^{k+1}_0(M)\to\cdots.
$$
We denote their cohomology groups as $H_{\bpat_f}(M)$ and
$H_{(0,\bpat_f)}^*(M)$.

We also have the smooth $L^2$-complex:
$$
\cdots\to \Omega^{k-1}_{(2)}(M)\xrightarrow{\bpat_f}
\Omega^{k}_{(2)}(M)\xrightarrow{\bpat_f}\Omega^{k+1}_{(2)}(M)\to\cdots.
$$
We denote its cohomology as $H^*_{((2),\bpat_f)}(M)$.

On the other hand, we can use the closure of $\bpat_f$ forms a
$L^2$-complex (the composition $\bpat_f\circ \bpat_f=0$ is in
distribution sense):
$$
\cdots \to L^2\Lambda^{k-1}(M)\xrightarrow{\bpat_f}
L^2\Lambda^{k}(M)\xrightarrow{\bpat_f}L^2\Lambda^{k+1}(M)\to\cdots.
$$
We denote its cohomology as $H^*_{((2),\bpat_f,\#)}(M)$.

\begin{lm} We have the isomorphism of two cohomologies:
\begin{equation}
H^*_{((2),\bpat_f)}(M)\simeq H^*_{((2),\bpat_f,\#)}(M).
\end{equation}
\end{lm}

\begin{proof} We first prove the following regularity
conclusion:\

\emph{Let $\psi\in \Dom(\bpat_f)$ such that $\bpat_f\psi=\varphi$ is
$L^2$ integrable, then $\psi\in \Omega_{(2)}^*$.}\

Since $\bpat \psi=-\pat f \wedge \psi+\varphi$, by interior estimate
of $\bpat$-operator (refer to \cite{DS}), there exists the estimate:
$$
||\psi||_{W^{1,2}(B_r(x_0))}\le C(||\pat f\wedge
\psi||_{L^2(B_R(x_0))}+||\varphi||_{L^2(B_R(x_0))}),
$$
where $r\le R$, $x_0\in M$ and $B_r(x_0)$ is the geodesic ball of
radius $r$ in $M$.

Now using the sobolev embedding theorem \ref{thm:sobo-embe}, we know
that $\psi\in L^\alpha(\Lambda^p(B_r(x_0)))$ for $2\le \alpha\le
\frac{2n}{n-1}$. Here $n$ is the complex dimension of $M$. By
standard bootstrap argument and H\"older estimate, we know that
$\psi$ is a smooth form in case that $\varphi$ is smooth.

Now if $\varphi\in \ker(\bpat_f)$, then $\varphi$ is a smooth forms.
If $\varphi=\bpat_f\psi$, then it lies in the kernel of $\bpat_f$ as
a $L^2$ solution. So $\varphi$ is automatically smooth by our
regularity conclusion. Finally we know that $\psi$ is a smooth $L^2$
form.

Hence we proved the isomorphism.
\end{proof}

By the above lemma, we denote the $L^2$-cohomology simply by
$H^*_{((2),\bpat_f)}(M)$ which corresponds to either of the two
complexes.

\subsubsection{\underline{Hodge decomposition}}\

\

\emph{In this part, we always assume that our section bundle system
$(M,g,f)$ is strongly tame}.\

If $(M, g, f)$ is strongly tame, then by Theorem \ref{thm:main-1},
$\Delta_f$ has only discrete spectrum in $L^2\Lambda^*(M)$ space.
Let $\ch\subset \Dom(\Delta_f)$ be the subspace of
$\Delta_f$-harmonic forms. Then we know that $\dim \ch<\infty$.

Let $E_\mu$ be the eigenspace with respect to the eigenvalue $\mu$
of $\Delta_f$, $P_\mu: L^2\Lambda^k\to E_\mu$ be the projection
operators, then we have the spectrum decomposition formulas:
\begin{align*}
&L^2\Lambda^k=\ch\oplus \oplus_{i=1}^\infty E_{\mu_i}\\
&\Delta_f=\sum_i \mu_i P_{\mu_i}.
\end{align*}

The Green operator $G_f$ of $\Delta_f$ satisfies
$$
G_f\Delta_f+P_0=\Delta_f G_f+P_0=I.
$$
This implies the following Hodge-De Rham theorem.
\begin{thm}\label{thm:Hodge} There are orthogonal decomposition
\begin{align}
& L^2\Lambda^k=\ch^k\oplus \im(\bpat_f)\oplus
\im (\bpat_f^\dag)\\
&\ker \bpat_f=\ch^k\oplus \im(\bpat_f).
\end{align}
In particular, we have the isomorphism
$$
H^*_{((2),\bpat_f)}(M)\cong \ch^*,
$$
where $\ch^k$ means the space of harmonic $k$-forms.
\end{thm}

\begin{proof}For any $\varphi\in L^2$, we have the decomposition:
\begin{equation}
\varphi=P_0\varphi\oplus \bpat_f\bpat_f^\dag G\circ \varphi\oplus
\bpat_f^\dag\bpat_f G\circ \varphi,
\end{equation}
where $\oplus$ is the orthogonal direct sum. This gives the
decomposition theorem and the decomposition of $\ker(\bpat_f)$.
Since $\dim \ch<\infty$, the image $\im(\bpat_f)$ is a closed
subspace, and we have the isomorphism
$$
\ch^*\simeq H^*_{((2),\bpat_f)}(M)
$$
between closed subspace.
\end{proof}

The same conclusions can be obtained for other operators.

\begin{thm} The following decomposition hold
\begin{align}
&L^2\Lambda^k=\ch^k\oplus \im (\pat_f)\oplus \im(\pat_f^\dag)\\
&L^2\Lambda^k=\ch^k\oplus \im d_f\oplus \im d^*_f\\
&L^2\Lambda^k=\ch^k\oplus \im d_f^c \oplus \im (d_f^c)^*.
\end{align}
Combined with the above decomposition, we have the five-fold
decomposition
\begin{align}
L^2\Lambda^k=&\ch^k\oplus \im d_fd_f^c\oplus d_f^* d_f^c\oplus \im
(d_f(d_f^c)^*)\im (d_f^*(d_f^c)^*)\\
=&\ch^k\oplus \im \pat_f\bpat_f\oplus \im \pat_f^\dag\bpat_f\im
\pat_f\bpat_f^\dag\oplus \im \pat_f^\dag\bpat_f^\dag\\
=&\ch^k\oplus \im \bpat_f\pat_f\oplus \im \bpat_f^\dag\pat_f\im
\bpat_f\pat_f^\dag\oplus \im \bpat_f^\dag\pat_f^\dag
\end{align}
\end{thm}

\subsubsection{\underline{Hard Lefschetz theorem}}

\begin{df} A homogeneous form $\alpha\in \Lambda^k (\Lambda^{p,q})$ is called
primitive if $\Lambda\alpha=0$. The space of $k$-primitive forms is
denoted by
$$
\PP^k:=\oplus_{p+q=k}\PP^{p,q}.
$$
\end{df}

If $\alpha\in \PP^k$, we can use the formula
$$
\Lambda^s L^r \alpha=\Lambda^{s-1}(\Lambda L^r-L^r
\Lambda)\alpha=r(N-k-r+1)\Lambda^{s-1}L^{r-1}\alpha.
$$
to deduce the following conclusions:

\begin{lm}\label{lm:prim}\begin{enumerate}
\item If $\alpha\in \PP^k$, then $L^s \alpha=0$ for $s\ge
(N+1-k)_+$, in particular for $\alpha\in \PP^N$, there is
$L\alpha=0$.

\item $\PP^k=0$ for $N+1\le k\le 2N$.
\end{enumerate}
\end{lm}

\begin{thm}\label{thm:prim}(Primitive decomposition formula) For every $\alpha\in
\Lambda^k$, there is a unique decomposition
$$
\alpha=\sum_{r\ge (k-N)_+}L^r \alpha_r,\;\alpha_r\in \PP^{k-2r}.
$$
Furthermore, $\alpha_r=\Phi_{k,r}(L,\Lambda)\alpha$ where
$\Phi_{k,r}$ is noncommutative polynomial in $L,\Lambda$ with
rational coefficients. As a consequence, we have the space
decompositions
\begin{align*}
&\Lambda^k=\oplus_{r\ge (k-N)_+}L^r\PP^{k-2r}\\
&\Lambda^{p,q}=\oplus_{r\ge (p+q-N)_+}L^r\PP^{p-r,q-r}
\end{align*}
\end{thm}

\begin{crl} The linear operators
\begin{align*}
&L^{N-k}:\Lambda^k\to \Lambda^{2N-k},\\
&L^{N-p-q}:\Lambda^{p,q}\to \Lambda^{N-p,N-q},
\end{align*}
are isomorphisms for $k\le N,p+q\le N$.
\end{crl}
The proof of the primitive decomposition formula can be found in
\cite{De,GH,CMP}.

Since the operators $\Lambda, L, F$ forms a $sl_2$ Lie algebra and
commute with the Laplacian operator $\Delta_f$. We can apply the
above primitive decomposition theorem to the space of $L^2$ harmonic
forms. Therefore in the same way to prove the Hard Lefschetz theorem
for compact K\"ahler manifolds, we can obtain the following result.

\begin{thm}[Hard Lefschetz theorem]\label{thm:hard-lefs-1} Let $(M,g,f)$ be a strongly tame
section-bundle system with complex dimension $n$. Then we have the
isomorphism
\begin{equation}
L^k: H^{n-k}_{((2),\bpat_f)}\simeq H^{n+k}_{((2),\bpat_f)},\;1\le
k\le n.
\end{equation}
\end{thm}

\begin{df} Let $n=\dim_\C M$. For $k\le n$, the primitive
$(n-k)$ cohomology, $PH^{n-k}_{((2),\bpat_f)}(M)$ is the kernel of
$L^{k+1}$ acting on $H^{n-k}_{((2),\bpat_f)}(M)$. Equivalently, it
is the kernel of the $\Lambda$-operator:
\begin{align*}
PH^{n-k}_{((2),\bpat_f)}(M)&=\ker\{L^{k+1}:
H^{n-k}_{((2),\bpat_f)}\to
H^{n+k+2}_{((2),\bpat_f)}\}\\
&=\ker\{\Lambda: H^{n-k}_{((2),\bpat_f)}\to
H^{n-k-2}_{((2),\bpat_f)}\}.
\end{align*}
In particular, $PH^k_{((2),\bpat_f)}=0$ for $k>n$.
\end{df}

Consequently, we have

\begin{thm}[Lefschetz's decomposition theorem]\label{thm:hard-lefs-2} Let $(M,g,f)$ be a
strongly tame section-bundle system. Then there is a direct sum
decomposition
\begin{equation}
H^k_{((2),\bpat_f)}=PH^k_{((2),\bpat_f)}\oplus L\cdot
PH^{k-2}_{((2),\bpat_f)}\oplus L^2\cdot
PH^{k-4}_{((2),\bpat_f)}\cdots.
\end{equation}
\end{thm}

\begin{rem} Though the $\bpat_f$ is defined in terms of a
holomorphic function $f$, the Lefschetz's decomposition theorem even
holds on real field. The reason is the Laplacian operator $\Delta_f$
and $L, \Lambda$ are real operators and the decomposition can be
restricted to the real subspace of the space of harmonic forms.
\end{rem}

\subsubsection{\underline{Computation of the $\bpat_f$ cohomology}}\

We have defined three cohomology groups:
$$
H^*_{(0,\bpat_f)}(M), H^*_{((2),\bpat_f)}, H^*_{\bpat_f}(M).
$$
We want to discuss their relations. Above all we have the (twisted)
poincare lemma:

\begin{lm}[$\bpat_f$ Poincare lemma]\label{lm:poincare} Assume that $U\subset M$ be a simply connected domain.
Let $\varphi $ be a $\bpat_f$-closed $k$-form on $U$. Then there
exist a $k-1$ form $\psi$ and a unique holomorphic $(k,0)$-form
$\phi$ module $df\wedge \Omega^{k-1}(U)$ such that
$$
\varphi=\phi+\bpat_f \psi.
$$
In particular, for the case $k<n$ and the case $k=n$ with the
condition $df\neq 0$ on $U$, $\phi \equiv 0$; if $df=0$ at some
points on $U$, then $\phi\neq 0$.
\end{lm}

\begin{proof} Assume $\varphi$ has the following decomposition
$$
\varphi=\sum_p \varphi^{p,k-p}.
$$
Then $\bpat_f\varphi=0$ decompose into
$$
\bpat\varphi^{p,k-p}+\pat f\wedge \varphi^{p-1,k-p+1}=0, 0\le p\le
k.
$$
By the Dolbeault lemma, there exists a $(0,k-1)$ form $\psi^{0,k-1}$
such that $\varphi^{0,k}=\bpat \psi^{0,k-1}$. Then we have
\begin{align*}
0=&\bpat \varphi^{1,k-1}+\pat f\wedge \bpat \psi^{0,k-1}\\
=&\bpat[\varphi^{1,k-1}-\pat f\wedge \psi^{0,k-1}]\\
\end{align*}

By Dolbeault lemma again, we can obtain $\psi^{1,k-2}$ such that
$$
\varphi^{1,k-1}=\bpat \psi^{1,k-2}+\pat f\wedge \psi^{0,k-1}.
$$
By induction, once we obtain $\psi^{p-1,k-p}$ we can find
$\psi^{p,k-p-1}$ such that
$$
\varphi^{p,k-p}=\bpat\psi^{p,k-p-1}+\pat f\wedge \psi^{p-1,k-p}.
$$
In particular, we arrive at the equation
$$
\bpat[\varphi^{k,0}-d f\wedge \psi^{k-1,0}]=0.
$$
Setting $\phi=\varphi^{k,0}-d f\wedge \psi^{k-1,0}$, which is
holomorphic, and defining $\psi=\sum_p \psi^{p,k-p-1}$, we obtain:
$$
\varphi=\phi+\bpat_f \psi.
$$

Consider the Koszul complex $K(df,U)$
$$
0\to \O_{U}\xrightarrow{df\wedge}\cdots\to \Lambda^{k-1,0}
(U)\xrightarrow{df\wedge}\Lambda^{k,0}(U)\to\cdots\Lambda^{n,0}(U)\to
0,
$$
the division lemma of De Rham (see \cite{Ku}, PP. 18) shows that the
cohomology groups
\begin{equation}
H^*(K(df,U))\cong
\begin{cases}
0,\;&k<n\\
\frac{\Omega^n(U)}{\Omega^{n-1}(U)\wedge df},\;&k=n
\end{cases}
\end{equation}
So if $k<n$, there exists a holomorphic form $\hat{\phi}$ such that
$\phi=df\wedge \hat{\phi}=\bpat_f(\hat{\phi})$ and
$\varphi=\bpat_f(\hat{\phi}+\psi)$. $k=n$ case can be considered
similarly. Hence we get the conclusion.
\end{proof}

Now we turn to the global computation of the complex
$(\Omega^{p}(M),\bpat_f=\bpat+df\wedge)$. This complex can be viewed
as the total complex of the double complex $(\Omega^{*,*}(M),
\bpat,df\wedge)$. Obviously we have the relations
$$
\bpat^2=(df\wedge)^2=[\bpat,df\wedge]=0.
$$
There are two filtrations on $(\Lambda^k,\bpat_f)$ given by
$$
'F^p\Lambda^k=\oplus_{p_1+q=k,p_1\ge
p}\Lambda^{p_1,q},\;\;''F^q\Lambda^k=\oplus_{p+q_2=k,q_2\ge
q}\Lambda^{p,q_2}
$$
So there are two spectral sequences, $\{'E_r\}$ and $\{''E_r\}$,
both abutting to $H^*((\Lambda^*,\bpat_f))=H^*_{\bpat_f}$. One can
obtain the result
$$
\left\{\begin{array}{l} 'E_2^{p,q}\cong H^p_{\bpat}(H^q_{\pat
f\wedge}(\Lambda^{*,*}))\\
''E_2^{p,q}\cong H^q_{\pat f\wedge}(H^p_{\bpat}(\Lambda^{*,*}))
\end{array}
\right.
$$

The second identity is the $q$-th cohomology group of the following
complex:
$$
\cdots\to H^p_\bpat(\Lambda^{q,*})\xrightarrow{df\wedge}
H^p_\bpat(\Lambda^{q+1,*})\xrightarrow{df\wedge}\cdots,
$$
which is equivalent to
\begin{equation}
\cdots\to H^{q,p}_\bpat(M)\xrightarrow{df\wedge}
H^{q+1,p}_\bpat(M)\xrightarrow{df\wedge}\cdots.
\end{equation}

We can obtain some explicit result under some requirement to the
manifold $M$. If $M$ is a stein manifold, then there is
\begin{equation}
H^{q,p}_\bpat(M)=0,\;\forall p>0.
\end{equation}

This is true for punctured polydiscs $(\C^*)^k\times \C^l$.

In this case, we get the Koszul complex:
\begin{equation}
\cdots\to \Omega^q(M)\xrightarrow{df\wedge}
\Omega^{q+1}(M)\xrightarrow{df\wedge}\cdots.
\end{equation}

\begin{thm} Suppose that $M$ is a complete noncompact stein
manifold, then we have
\begin{align*}
H^k_{\bpat_f}&=\oplus_{0\le p\le k}Gr^pH^k_{\bpat_f}=\oplus_{0\le
p\le k}H^{k-p}_{\pat f\wedge}(\Omega^p)\\
&=\left\{\begin{array}{ll} 0,\;&\text{if}\;k<n\\
\Omega^n/df\wedge\Omega^{n-1},\;&\text{if}\;k=n,
\end{array}\right.
\end{align*}
In particular, we have the two consequences:
\begin{enumerate}
\item if $(\C^{n+1}, W)$ is a section-bundle system with the potential
being a non-degenerate quasihomogeneous polynomial, there is
isomorphism
\begin{equation}
H^k_{\bpat_f}=\left\{\begin{array}{ll} 0,\;&\text{if}\;k<n\\
\C[z_0,\cdots,z_N]/J_W,\;&\text{if}\;k=n,
\end{array}\right.
\end{equation}
Here $J_W$ is the Jacobi ideal of $W$.
\item if $(T^n, f)$ is a section-bundle system with non degenerate and
convenient Laurent polynomial $f$, then
$$
H^k_{\bpat_f}=\left\{\begin{array}{ll} 0,\;&\text{if}\;k<n\\
\C[z_1,\cdots,z_n,z_1^{-1},\cdots,z_n^{-1}]/J_f,\;&\text{if}\;k=n.
\end{array}\right.
$$
\end{enumerate}
\end{thm}

There are two natural maps:
$$
i_1:H^*_{(0,\bpat_f)}(M)\to H^*_{(\bpat_f,(2))}(M),\;i_2:
H^*_{(\bpat_f,(2))}\to H^*_{\bpat_f}(M)
$$

\begin{prop} The maps $$ i_1:H^*_{(0,\bpat_f)}(M)\to
H^*_{(\bpat_f,(2))}(M),\;j:=i_2\circ i_1: H^*_{(0,\bpat_f)}\to
H^*_{\bpat_f}(M)
$$
are injective.
\end{prop}

\begin{proof} It suffices to prove that $i_2\circ i_1=j:H^*_{0,\bpat_f}(M)\to
H^*_{\bpat_f}(M)$ is injective. Let $\varphi$ be a class in
$H^k_{0,\bpat_f}(M)$. Suppose that the image $j([\varphi])$ is zero
in $H^*_{\bpat_f}(M)$, i.e, there exists a smooth $k-1$ form
$\hat{\psi}$ such that
$$
\varphi=\bpat_f \hat{\psi}.
$$
Notice that $\bpat_f \hat{\psi}\equiv 0$ near the infinity place.
Hence there exists a $(k-2)$ form $\eta$ such that
$\hat{\psi}-\bpat_f \eta$ vanishes near the infinity place. This
fact is true since the cohomology groups $H^*_{\bpat_f}$ are trivial
when restricted to the infinite far place (fundamental tame
condition). This shows that $i_1$ is injective.
\end{proof}

\begin{rem}\label{rem:holo-decomp-form} The map $i_2:
H^*_{(\bpat_f,(2))}\to H^*_{\bpat_f}(M)$ need not be injective or
surjective. Since $\ch^n\simeq H^n_{(\bpat_f,(2))}(M)$ and
$H^n_{(\bpat_f)}(M)\simeq \Omega^n/df\wedge\Omega^{n-1}$, we can
define the map $i_2$ by the following map $i_{0h}:\ch^n\to
\Omega^n/df\wedge\Omega^{n-1}$ if $M$ is simply-connected: let
$\alpha\in \ch^n$, then by poincare lemma \ref{lm:poincare}, there
exists a unique holomorphic $n$-form
$i_{0h}(\alpha)\in\Omega^n(M)/df\wedge\Omega^{n-1}(M)$ and a smooth
$(n-1)$ form $\psi$ such that
\begin{equation}
\alpha=i_{0h}(\alpha)+\bpat_f \psi.
\end{equation}
We can define another map $i_{h0}: \Omega^n/df\wedge\Omega^{n-1}\to
\ch^n$ as follows. Let $g\in \Omega^n(M)/df\wedge\Omega^{n-1}(M)$.
Since the Koszul complex $K(df,U)$ over a small neighborhood $U$ of
$\infty$ is exact, there exists a holomorphic $n-1$ form $\hat{g}$
such that $g=df\wedge \hat{g}=\bpat_f (\hat{g})$ on $U$. Take a
cut-off function $\chi$ which is $1$ on a smaller domain of $U$ and
with compact support in $U$, then $g-\bpat_f(\chi \hat{g})$ is a
compactly supported $n$-form on $M$, then we define
$i_{h0}(g)=P_0(g-\bpat_f(\chi \hat{g}))$. This definition is
independent of the choice of $\chi,\hat{g}$. If there is another
pair $\chi_0,\hat{g}_0$ such that $g-\bpat_f(\chi_0\hat{g}_0)$ is
compactly supported, then $\bpat_f(\chi
\hat{g})-\bpat_f(\chi_0\hat{g}_0)$ is also compactly supported, and
then $P_0[\bpat_f(\chi \hat{g})-\bpat_f(\chi_0\hat{g}_0)]=0$.
\end{rem}

\begin{thm}\label{crl:stein} Let $(M,g)$ be a K\"ahler stein manifold with bounded
geometry and $(M,g,f)$ be strongly tame. If $f$ is a Morse function,
then
\begin{equation}
\dim \ch^k=
\begin{cases}
0,\;k<n\\
\mu,\;k=n.
\end{cases}
\end{equation}
and there is an explicit isomorphisms:
\begin{equation}
i_{0h}:\ch^n \to \Omega^n(M)/df\wedge\Omega^{n-1}(M).
\end{equation}
\end{thm}

\begin{proof} Let $f_\tau=\tau f, \tau\in \C$ and $p_1,\cdots,p_\mu$ be the
critical points of $f$ or $f_\tau$. By Witten-Morse-Smale theory,
there is an isomorphism between the Witten complex
$(L^2(\Omega^*(M)),d_{f+\fb})$ and the Morse-Smale complex $(C_*(M,2
Re(f)),\pat)$ given by the non-degenerate critical points and the
negative gradient flow of $2 Re(f)$. Since all the non degenerate
critical points of $2 Re(f)$ has Morse index $n$, so only the middle
dimensional homology group is nontrivial and with dimension $\mu$.
By our Hodge theorem and the isomorphism $H^*_{((2),\bpat_f)}\cong
H^*(L^2(\Omega^*(M)),d_{f+\fb})\cong H_*(C_*(M,2 Re(f)),\pat)$, we
obtain the conclusion. We get a indirect isomorphism to the
$\C$-space $\Omega^n/df\wedge\Omega^{n-1}$. Given an order of those
critical points.  This order is the same for any $\tau$. As
$|\tau|\to\infty$, the mass of the harmonic $n$-forms will
concentrate at the critical points $p_i$. For large $\tau$, the map
$i_{0h}$ is isomorphic and provides the explicit isomorphism between
$\ch^n$ and $\Omega^n(M)/df\wedge\Omega^{n-1}(M)$.
\end{proof}

By Theorem \ref{crl:stein}, we know that the singular behavior of
$\ch$ can only happen when $f$ is a degenerate holomorphic function.

\subsection{$\Z_2$-symmetries, orbifoldizing and splitting}

\subsubsection{\underline{$\Z_2$-symmetries}}\

In compact K\"ahler manifold, the standard complex conjugate $\tau$
and the $*$ operator provides the $\Z_2$ symmetries to the Hodge
theory. In our case, if the section-bundle system $(M,g,f)$ is
K\"ahler, the twisted Laplace operator $\Delta_f$ is real. Therefore
the complex conjugate $\tau$ provides the $\Z_2$ symmetry to the
space $\ch^*$ of harmonic forms. However, the $*$ operator does not
communicate with $\Delta_f$ operator, instead we have
$$
*\Delta_{f}=-\Delta_{-f}*.
$$
Denote by $\ch^*_\ominus$ the space of $\Delta_f$-harmonic forms and
$\ch^*_\oplus$ the space of $\Delta_{-f}$-harmonic forms. Therefore
if we want our theory keep the $\Z_2$ symmetry of $*$-operator, we
should consider the total space:
$$
\ch^*_{\text{tot}}:=\ch^*_\ominus\oplus \ch^*_\oplus.
$$
The Riemannian-Hodge bilinear relation on $\ch^*_{\text{tot}}$ is
given by
\begin{equation}
\langle\phi,\psi\rangle=\int_M \phi\wedge \psi\wedge
\omega^{n-k},\;\forall \phi,\psi\in \ch^*_{\text{tot}}.
\end{equation}
it is $(-1)^k$ symmetric and satisfies the relation:
\begin{equation}
g(\phi,\psi)=\langle\phi,*\overline{\psi}\rangle.
\end{equation}
Therefore, the bilinear relation naturally satisfies the positivity
for any primitive harmonic forms in $\ch_{\text{tot}}$.

We have the analogous $\Z_2$ symmetry relation between cohomology
groups. In addition to the operator $\bpat^\dag_f$ and its complex,
we can also consider the complex $(\Omega^k, \pat_f)$ and the
cohomology $H^*_{\pat_f}$. They have the dual relation because of
the action of the $*$ operator. $*$ operator gives the following
commutation relation:
$$
*\bpat_f^\dag=(-1)^{p+q}\pat_{-f}*.
$$
Therefore we have the following conclusion.

\begin{prop} We have the isomorphisms
\begin{equation}
\ch_\ominus\simeq H^*_{((2),\bpat_f)}\simeq
H^*_{((2),\pat_{-f})}\simeq \ch_\oplus.
\end{equation}
\end{prop}

\subsubsection{\underline{orbifoldizing and splitting}}\

\

If $M$ is a Stein manifold and $f$ is fundamental tame, then only
the middle dimensional cohomology groups of $H^*_{((2),\bpat_f)}$
and $H^*_{\bpat_f}(M)$ are not trivial. Hence they have no ring
structures when compared to the cohomology ring defined on a closed
manifold. To compensate this default in LG/CY correspondence,
Intriligator and Vafa \cite{IV} used the orbifoldizing method. Such
idea has been developed by R. Kaufmann to the abstract algebraic
structures \cite{Kau}. Their constructions only provided the state
space structure and grading information. The orbifoldized
singularity theory for a non-degenerate quasi-homogeneous
singularities has been recently constructed by M. Krawitz \cite{Kr}.
The state space correspondence between LG A and B model, and between
CY model and LG A model for the invertible non-degenerate
quasi-homogeneous polynomials have been proved in \cite{Kr} and
\cite{ChR2}. Such idea can be used in our case at the conformal
point of the deformation space.

We take $(\C^n, W)$ as example, where $W(z)=W(z_1,\cdots,z_n)$ is a
non-degenerate quasi-homogeneous polynomial with weight $q_i=w_i/d$
for each variable $z_i$. Let $G_W\subset U(1)\times \cdots\times
U(1)$ be the diagonal symmetric group of $W$, i.e.,
$$
G_W=\{(\alpha_1,\cdots,\alpha_n)\in U(1)\times\cdots\times U(1)|
W(\alpha_1 z_1,\cdots,\alpha_n z_n)=W(z_1,\cdots,z_n),\forall
(z_1,\cdots,z_n)\}.
$$
There is a cyclic subgroup $\langle J_W\rangle\subset G_W$, where
$J_W=(e^{2\pi iq_1},\cdots,e^{2\pi iq_n})$. Let subgroup $G$ satisfy
$\langle J_W\rangle \subset G\subset G_W$ and take $\gamma\in G$, we
can define
\begin{equation}
\C^n_\gamma=\{z\in \C^n| \gamma\cdot
z=z\},\;n_\gamma=\dim_\C(\C^n_\gamma),\;W_\gamma=W|_{\C^n_\gamma},
\end{equation}
and call sub section-bundle system $(\C^n_\gamma, W_\gamma)$ as the
$\gamma$-twistor sector of $(\C^n,W)$. Note that by Lemma 3.2.1 of
\cite{FJR2}, $0$ is the only singularity of $W_\gamma$. On each
twister sector we can study the deformation theory of the
Schr\"odinger operator associated to $W_\gamma$. What extra
information can be extracted and the comparison with other
orbifolding process is an interesting problem.

Except the orbifoldizing operation, there is another simple
phenomenon, i.e., the splitting of the variables of the
Schr\"odinger operators due to the splitting of the section-bundle
system $(\C^n, W=W_1+W_2)=(\C^{n_1}, W_1)\times (\C^{n_2},W_2)$. In
this case, the harmonic form of the total space is the product of
two harmonic forms of lower dimensional Schr\"odinger systems. Note
that such decomposition can't be done for projective hypersurface.
An example is given by the Fermat polynomials
$z_1^{w_1}+\cdots+z_n^{w_2}$. Therefore such splitting at the
conformal point of the deformation parameter space will be very
helpful to the computation of the topological quantities of the
corresponding projective hypersurfaces.

\section{Deformation theory}

\subsection{Deformation of superpotential}\

\

A superpotential of a schr\"odinger system is a holomorphic function
$f:M\to \C$ defined on a (non-compact)complete complex manifold $M$
with dimension $n$. The deformation of $f$ will induce the
deformation of the Schr\"odinger equation. In this section, we will
consider the deformation of the potential and define the so called
strong deformation which is required in this paper.

\subsubsection{\underline{Milnor numbers}}\

\

Let $z\in M$ be an isolated critical point of $f$, and $B(z)$ be a
ball centered at $z$ such that $z$ is the only critical point of $f$
at $B(z)$. Then it is well-known (or see \cite{Mi}) that the
topological degree of the map:
$$
\pat B\to S^{2n-1}: z\to \frac{\pat f}{|\pat f|}
$$
equals to the Milnor number of $f$ at $z$,
$$
\mu_z(f):=\dim_\C \O_z/J_f.
$$

We can also define the global Milnor number of $f$ on a open domain
$U$.

\begin{df} Let $\pat U$ be the boundary of $U$ which is assumed to
be a smooth compact $2n-1$ dimensional manifold and $\pat f|_{\pat
U}\neq 0$. Then the Milnor number of $f$ on $U$ is defined to be the
topological degree of the map
$$
\pat U\to S^{2n-1}: z\to \frac{\pat f}{|\pat f|}.
$$
We denote the Milnor number of $f$ in $U$ by $\mu_U(f)$.

In particular, if $z$ is not a critical point of $f$, we define the
Milnor number at $z$ to be zero. If $\pat f$ has positive
dimensional zero locus in $U$, then we define $\mu_U(f)=\infty$.
\end{df}

The proof of the following conclusions about the Milnor number can
be found in Appendix B of \cite{Mi}.

\begin{thm} Let $U\subset M$ be an open set and $f: U\to M$ be a
holomorphic function. Let $V\subset U$ be an open subset with
compact closure in $U$ with the boundary $\pat V$ being a smooth
compact manifold, and such that $\pat f|_{\pat V}\neq 0$, then $f$
has only finitely many isolated critical points $p_1,\cdots,p_l$ in
$V$ and has the identities:
\begin{equation}
\mu_V(f)=\sum_{j=1}^l \mu_{p_j}(f).
\end{equation}
Furthermore, if $f(z,t)$ is a holomorphic function defined on
$U\times [0,1]$ such that $\pat_z f(z,t)|\neq 0$ on $\pat V\times
[0,1]$, then
\begin{equation}
\mu_V(f(z,0))=\mu_U(f(z,1)).
\end{equation}
\end{thm}

We follow Broughton's definition (see \cite{Br}) of a tame
holomorphic function.

\begin{df} A holmorphic function $f$ is called tame if there is a
compact neighborhood $U$ of the critical points of $f$ such that
$||\pat f||$ is bounded away from $0$ on $M-U$.
\end{df}

For tame holomorphic function $f$ we can define the global Milnor
number $\mu(f)$ on $\C^m$ and certainly we have
\begin{equation}
\mu(f)<\infty.
\end{equation}

\subsubsection{\underline{Strong deformation}}

\begin{df} Let $f:M\to \C$ be a holomorphic function defined on a (non-compact)complete complex manifold and
$S\subset \C^m$ be a (open or closed) domain. A deformation of $f$
on $S$ is a holomorphic function $F(x,t)$,
$$
F:M\times S\to \C
$$
such that $F(x,0)=f$.
\end{df}

\begin{df} Let $F'(x,\lambda'):M\times S'\to \C$ be an another
deformation of $f:M\to \C$. $F'$ is said to be embedded in the
deformation $F(z,\lambda):M\times S\to \C$, if there is a
biholomorphic map $g(z,\lambda): M\to M$ for any $\lambda\in S$ and
an holomorphic submersion $\varphi:S\to S'$ such that
\begin{equation}
F'(x,\lambda')=F(g(z,\lambda),\varphi(\lambda)).
\end{equation}
If the submersion is a diffeomorphism, then we say that the two
deformations are equivalent. A deformation $F$ is called a maximum
deformation, if it can't be embedded into another different
deformation.
\end{df}

\begin{rem} In singularity theory, it is well known that there is a
versal deformation of a singularity (a germ of function with
isolated critical point). In our case, the deformation is global, we
don't know if there exists a versal deformation. It relates to the
automorphism group of $M$.
\end{rem}

\begin{prop} Let $F:M\times S\to \C$ be a deformation of $f$. Suppose that for any $t\in S$, $f_t$ is a tame function.
Then the function $\mu(f_t)$ is lower semicontinuous in $S$.
\end{prop}

\begin{proof} This is because the small continuous perturbation of $f$ will not
change the number of critical points contained in a compact set $K$
of the manifold $M$. However the perturbation may generates new
critical points of $f_t$ outside $K$. Therefore $\mu(f_t)$ is only
lower semicontinuous.
\end{proof}

\begin{df}\label{df:strong-deform} Let $F:M\times S\to \C$ be a deformation of $f:M\to \C$.
It is called a strong deformation of $f$ on $S\subset \C^m$, if the
following two conditions hold (denote by $f_t=F(\cdot,t)$):
\begin{enumerate}
\item $\sup_{t\in S}\mu(f_t)<\infty$.

\item For any $t\in S$, $f_t$ is strongly tame.

\item For any $t\in S$, $\Delta_t:=\Delta_{f_t}$ have common domains
in the space of $L^2$ forms.
\end{enumerate}
\end{df}

The third condition is a technique condition. Usually we consider
the deformation of the form:
\begin{equation}\label{df:usual-form}
f_t=f+\Sigma_{i=1}^m t_i g_i.
\end{equation}
The following lemma grantees the condition (3) in Definition
\ref{df:strong-deform}.

\begin{lm}\label{lm:deform-1} Let $f_t$ be a deformation of a strongly tame holomorphic function $f$
having the form (\ref{df:usual-form}), and satisfy: for any $C>0$,
$|\nabla f|^2-C|\nabla \pat g_i|\to \infty$ as $d(z,z_0)\to \infty$
and there exists $C_0$ such that $|\nabla g_i|\le C_0|\nabla f|$
near infinity. Then $f_t$ is a strong deformation for small
$t=(t_1,\cdots,t_m)$.
\end{lm}

\begin{proof} It suffices to prove the condition (3) in the
definition \ref{df:strong-deform}. We know that $\varphi\in
\Dom(\Delta_0)$ if and only if the graph norm
$$
||\varphi||^2_{g,0}:=\int_M |\bpat_f \varphi|^2+|\bpat_f^\dag
\varphi|^2+|\varphi|^2<\infty.
$$
By $L^1$ Stokes theorem, the above inequality is equivalent to the
following inequality
\begin{equation}
\int_M (\Delta_{\bpat}\varphi, \varphi)+(L_f(\varphi)+|\nabla
f|^2\varphi, \varphi)+(\varphi,\varphi)<\infty.
\end{equation}
Since $(M,g,f)$ is strongly tame, the above inequality shows the
equivalence of the graph norm $||\cdot||_{g,0}$ and the following
norm:
\begin{equation}\label{eq:doma-equiv}
\int_M (\Delta_{\bpat}\varphi, \varphi)+(|\nabla
f|^2+1)|\varphi|^2<\infty.
\end{equation}
Since for any $C>0$, $|\nabla f|^2-C|\nabla \pat g_i|\to \infty$ as
$d(z,z_0)\to \infty$ and there exists $C_0$ such that $|\nabla
g_i|\le C_0|\nabla f|$ near infinity, (\ref{eq:doma-equiv}) is
equivalent to
\begin{equation}
\int_M (\Delta_{\bpat}\varphi, \varphi)+(|\nabla
f_t|^2+1)|\varphi|^2<\infty.
\end{equation}
which is equivalent to the graph norm of $\Delta_t$:
$$
||\varphi||^2_{g,t}:=\int_M |\bpat_{f_t}
\varphi|^2+|\bpat_{f_t}^\dag \varphi|^2+|\varphi|^2.
$$
This proves the result.
\end{proof}

Fortunately, most interesting cases we considered in this paper
satisfy the hypothesis of the Lemma \ref{lm:deform-1}. Hence when we
talk about a strong deformation in the rest of this paper, we always
do the following assumption:\

\

\begin{enumerate}
\item[] \emph{\textbf{Assumption:} The deformation has the form (\ref{df:usual-form}) and satisfy the hypothesis of Lemma
\ref{lm:deform-1}}.\
\end{enumerate}

\

We can define the following subsets in $S$,
\begin{equation}
S_{< n}=\{t\in S| \mu(f_t)< n\}, S_n=S_{< n+1}-S_{< n}.
\end{equation}

\begin{prop} $S_{< n}$ is a closed set in $S_{< n+1}\subset S$.
\end{prop}

\begin{proof} Since $\mu(f_t)$ is lower semicontinuous, it is easy
to see that $S_n$ is an open set in $S_{<n+1}$, or equivalently,
$S_{< n}$ is a closed subset in $S_{<n+1}$.
\end{proof}

Suppose that $\mu^M=\sup_{t\in S}\mu(f_t), \mu^m=\inf_{t\in
S}\mu(f_t)$. Then the deformation parameter space can be stratified
as follows:
\begin{equation}
S=\bigsqcup_{n=\mu^m}^{\mu^M} S_n.
\end{equation}

\begin{df}
Define the subset $S_n^\circ$ of $S_n$ to be set of $t$ such that
all the critical points of $f_t$ are non-degenerate critical points,
i.e., all are Morse critical points.
\end{df}

In the following parts, we will consider two interesting
deformations: one is the deformation of the non-degenerate
quasihomogeneous polynomials and the other is the deformation of the
Laurent polynomials defined on the algebraic torus.

\subsubsection{\underline{Deformation of non-degenerate quasihomogeneous
polynomials}}\

\

Let $W=W(z_1,\cdots,z_n)$ be a non-degenerate quasihomogeneous
polynomial and $0$ is the only critical point on $\C^n$. Let
$\mu=\mu(f)$ be the corresponding Milnor number. Then we can
consider its miniversal deformation $F(z,t):\C^n\times \C^\mu\to
\C$. This miniversal deformation can be realized as follows, we can
take a basis $G_1,\cdots,G_\mu$ generating the local algebra
$Q_{W,0}=\O_{W,0}/J_{W,0}$ and meanwhile require that the degree of
the polynomial $G_i$ is less than $\mu$. The later fact is ensured
by Lemma 1 in PP. 122 of \cite{AGV}. Then the miniversal deformation
is given by
\begin{equation}
F(z,t)=W(z)+\sum^\mu_{j=1}t_j G_j(z).
\end{equation}
In particular, we take $G_1=1$, the constant map.

Following physics' notion(see for example, \cite{KTS}), we can
distinguish the role of each monomials.

\begin{df} Assume that $W$ has weight $1$. We can think of the
deformation of $F$ as a quasihomogeneous polynomial of $(z,t)$. Then
each deformation parameter $t_i$ is called:
\begin{enumerate}
\item relevant, if the weight of $t_i$ is positive;

\item marginal, if the weight of $t_i$ is zero.

\item irrelevant, if the weight of $t_i$ is negative.
\end{enumerate}
\end{df}

It is interesting to check the miniversal deformation of the
singularities on Arnold's classification table \cite{AGV}. Here we
consider the simple singularities and the unimodal singularities.

(1) Simple singularities includes the following $A,D,E$
singularities:
\begin{description}
\item[$A_n$] $W=x^{n+1}, \ n\geq 1;$
\item[$D_n$] $W=x^{n-1}+xy^2, \ n\geq 4;$
\item[$E_6$] $W=x^3+y^4;$
\item[$E_7$] $W=x^3+xy^3;$
\item[$E_8$] $W=x^3+y^5;$
\end{description}
The miniversal deformations of those singularities are relative,
therefore by Theorem \ref{thm:main-lowerdeform} such deformations
are strong deformation according to our definition.

(2) Unimodal singularities. Such type singularities includes the
three big groups:
\begin{itemize}
\item[(i)] The triply indexed series of 1-parameter families of hyperbolic singularities:
$$
T_{p,q,r}(a): W(x,y,z)=x^p+y^q+z^r+axyz,
\;\frac{1}{p}+\frac{1}{q}+\frac{1}{r}<1, a\neq 0.
$$
This is a strong deformation. Consider the miniversal deformation of
$T_{p,q,r}(a)$. In this case, the possible polynomial with the
highest degree has the form: $x^{p-2}y^{q-2}$, whose degree may
exceed $1$. Therefore the miniversal deformation of $T_{p,q,r}(a)$
contains irrelevant deformations.
\item[(ii)] Three one-parameter families of parabolic singularities:
\begin{description}
\item[$P8$] $W(x,y,z)=x^3+y^3+z^3+axyz, a^3+27\neq 0$;
\item[$X_9$]$W(x,y,z)=x^4+y^4+ax^2y^2, a^2\neq 4$;
\item[$J_{10}$] $W(x,y,z)=x^3+y^6+ax^2y^2, 4a^3+27\neq 0$.
\end{description}
If we take the parameter $a$ as the deformation parameter, then $a$
is a marginal variable. The other parameters appeared in the
miniversal deformation are relative variables. Therefore by Theorem
\ref{thm:main-lowerdeform}, the miniversal deformation is strong
deformation.

\item[(iii)] $14$ Exceptional singularities.
\begin{align*}
E_{12}: x^3+y^7+axy^5,\;&E_{13}: x^3+xy^5+ay^8;\\
E_{14}: x^3+y^8+axy^6,\;&Z_{11}: x^3y+y^5+axy^4;\\
Z_{12}: x^3y+xy^4+ax^2y^3,\;&Z_{13}: x^3y+y^6+axy^5;\\
W_{12}: x^4+y^5+ax^2y^3,\;&W_{13}: x^4+xy^4+ay^6;\\
Q_{10}: x^3+y^4+yz^2+axy^3,\;&Q_{11}: x^3+y^2z+xz^3+az^5;\\
Q_{12}: x^3+y^5+yz^2+axy^4,\;&S_{11}: x^4+y^2z+xz^2+ax^3z;\\
S_{12}: x^2y+y^2z+xz^3+az^5,\;&U_{12}: x^3+y^3+z^4+axyz^2.
\end{align*}
In this case, all the miniversal deformation contains the irrelevant
variable $a$.

So in this case the behavior of the functions $|\nabla
W(a)|^2-C|\nabla dW(a)|$ at the infinity is much influenced by the
irrelevant $a$ term and the method of Theorem
\ref{thm:main-hypersurface} can't be used. So in this case whether
$0$ is the discrete spectrum of the Schr\"odinger equations are not
clear yet.
\end{itemize}

Therefore we reach our conclusion:

\begin{thm} The miniversal deformation of the simple singularities $A_n, D_{n+1},E_6,E_7,E_8$ and the unimodal
singularities $P_8,X_9, J_{10}$, and the deformation of
$T_{p,q,r}(a)$ are all strong deformations.
\end{thm}

\subsubsection{\underline{Marginal deformation of nondegenerate quasihomogeneous
polynomials}}\

\

There is a particular interesting deformation of a non-degenerate
quasihomogeneous polynomial with given type. Actually $P_8, X_9,
J_{10}$ are such deformations. Such deformation has a global $\C^*$
action, hence the marginal deformation induces the complex
deformation of the hypersurface in projective spaces.

Let us simply recall some basic results in the deformation theory of
compact complex manifolds. Let $X_0$ be a compact complex manifold.
Then a smooth deformation of $X_0$ is a fibration $X\to (S,0)$,
where $S$ is the deformation space and the fiber at $0$ is just
$X_0$. Two deformation $X\to (S,0)$ and $(Y\to (T,0))$ is called
equivalent if there is a fiber-preserving biholomorphic map between
$X$ and $Y$. If $(T,0)\to (S,0)$ is a holomorphic map, then we can
have the induced deformation $Y\to (T,0)$ which is the pull-back of
the fiber $X\to (S,0)$.

\begin{df} A deformation of $X_0$ is called complete if any other
deformation of $X_0$ can be induced from it. If the inducing map is
unique, we call the deformation universal. If only the derivative at
the base point is unique it is called versal. Versal families are
unique up to isomorphism.
\end{df}

A local deformation of complex structure in a small neighborhood
$U_i$ is given by a holomorphic vector field $v_i$. The difference
$v_i-v_j$ in $U_i\cap U_j$ defines a \v{C}ech cocycle $\theta(v)$
with value in the sheaf $\Theta_{X_0}$ of germs of holomorphic
vector fields on the fiber $X_0$. Therefore each deformation $v$ of
$X_0$ gives a KS (Kodaira-Spencer) class by the KS map:
$$
\rho:T_{S,0}\to \theta(v)\in H^1(X,\Theta_X).
$$

We have the Kodaira's completeness theorem (see for example,
\cite{Ko} or \cite{CMP}).

\begin{thm}\

\begin{enumerate}
\item A smooth family of compact complex manifolds with a surjective
KS map is complete at the base point. In particular, if the KS map
is a bijection, the family if versal.
\item If $H^2(X,\Theta_X)=0$. Then a versal deformation for $X_0$
exists whose KS map is an isomorphism.
\item If moreover $H^0(X,\Theta_X)=0$, this versal deformation is
universal.
\end{enumerate}
\end{thm}

However, the existence of the versal deformation did not known until
Kuranish's theorem \cite{Kur} after generalizing the deformation
notion to analytic space:

\begin{thm} For any compact complex manifold $X_0$ there exists a
versal deformation with a bijective KS map. If
$H^0(X_0,\Theta_{X_0})=0$, such a deformation can be chosen
universal. The base $S$ of the deformation is the zero locus of the
Kuranishi map $K: H^1(X_0,\Theta_{X_0})\to H^2(X_0, \Theta_{X_0})$
which satisfies $dK(0)=0$.
\end{thm}

In particular, if $K\equiv 0$ (for example, $H^2(X_0,
\Theta_{X_0})=0$), then the Kuranishi space $S$ is smooth.

The deformation theory of Calabi-Yau manifolds are in particular
interesting because its importance in string theory. The following
theorem due to Tian \cite{Ti} says that there is no obstruction for
CY manifold:

\begin{thm} If $X_0$ is a compact K\"ahler CY manifold, then the
local universal deformation space of $X_0$ is isomorphic to an open
set in $H^1(X_0, \Theta_{X_0})$.
\end{thm}

Now if $X_0$ is a $3$-dimensional compact K\"ahler CY manifold, then
by serre duality, there is
$$
H^2(X_0, \Theta_{X_0})\cong H^1(X_0,\Omega^1(K_{X_0}))\cong
H^1(X_0,\Omega^1)\not\cong 0.
$$
This says even the obstruction space $H^2(X_0, \Theta_{X_0})$ is not
zero, the Kuranishi space is still smooth.

The deformation theory of a smooth hypersurface in projective space
can also be described clearly.

Let $\IP^{n-1}$ be the projective space with homogeneous coordinates
$[z_1,\cdots,z_n]$. Let $f$ be a nondegenerate homogeneous
polynomial with degree $d$, then $f$ defines a smooth hypersurface
$X=\{f=0\}$ in the projective space. Let $S^d$ be the parameter
space for the tautological family of degree $d$ hypersurface in
$\IP^{n-1}$. So in general the deformation has the form:
\begin{equation}
f_t(z)=t_1z_1^d+t_2z_1^{d-1}z_2+\cdots+t_\mu z_n^d.
\end{equation}
The number $\mu$ of the monomial $z_1^{k_1}\cdots z_n^{k_n}$ of
degree $k_1+\cdots+k_n=d$ is
\begin{equation}
\mu=\dbinom{n-1+d}{d}
\end{equation}

The tangent space at $X_0$ is
$$
T=T_{X_0}\cong S^d/\C f,
$$
On the other hand, the group $GL(n)$ acts on the homogeneous
polynomials and the tangent space of the orbit space is $J^d_f/\C
f$, where $J^d_f$ is the ideal generated by these polynomials
$z_j\frac{\pat f}{\pat z_i},i,j=1,\cdots,n$. In fact, we can choose
a one-parameter transformation $g_t$: $z_j\to z_j+tz_i$ and fix all
other variables. Then
$$
\frac{d}{dt}(g_t\circ f)|_{t=0}=z_i\frac{\pat f}{\pat z_j},
$$
which shows that the tangent space of the $GL(n)$ orbit at $f$ is
given by $J^d_f/\C f$. Therefore, we have the isomorphism:
\begin{equation}
T_0 S\cong S^d/J^d_f.
\end{equation}
Let $\nu_{X_{0}/\IP^{n-1}}$ be the normal bundle of $X_0$ in
$\IP^n$. Then the infinitesimal deformation of $X_0$ in $\IP^{n-1}$
is classified by $H^0(X_0, \nu_{X_{0}/\IP^{n-1}})$. For any $v\in
T_{S,0}$, we can define a characteristic map $\sigma: T_{S,0}\to
H^0(X_0, \nu_{X_{0}/\IP^{n-1}})$ such that
\begin{equation}
T_0 S\cong S^d/J_f\cong H^0(\nu_{X_{0}/\IP^{n-1}}).
\end{equation}

Let
$$
\delta: H^0(\nu_{X_{0}/\IP^{n-1}})\to H^1(\Theta_{X_0})
$$
be the coboundary map defined by the exact sequence of sheaves:
$$
0\to \Theta_{X_0}\to \Theta_{\IP^{n-1}}|_{X_0}\to
\nu_{X_{0}/\IP^{n-1}}\to 0.
$$
Then it is easy to see the KS map $K$ is the combination map
$\delta\circ \sigma$. The map $K$ sends the element $G\in S^d$ to
$f+tG$.

On the other hand, we have the exact sequence:
\begin{equation}
0\to \Theta_{\IP^{n-1}}(-d)\to \Theta_{\IP^{n-1}}\to
\Theta_{\IP^{n-1}}|_{X_0}\to 0.
\end{equation}
By Bott vanishing theorem, we have
$$
H^1(\Theta_{\IP^{n-1}})\cong H^{n-2}(\Omega^1_{\IP^{n-1}}(-n))=0,
$$
and
$$
H^2(\Theta_{\IP^{n-1}}(-d))=0,
$$
unless $n=4, d=4$. By exact sequence of the cohomology group, we
have
$$
H^1(\Theta_{\IP^{n-1}}|_{X_0})=0,
$$
unless $n=4, d=4$. This shows that the coboundary map $\delta$ is
surjective and so the KS map is surjective. Since the map
\begin{equation}
S^d/J^d_f\to H^1(X_0,\Theta_{X_0})
\end{equation}
is an isomorphism, by Kodaira completeness theorem this provides a
versal family. Since the automorphism group
$H^0(X_0,\Theta_{X_0})=0$, this is also a universal deformation. In
summary, we have

\begin{thm} Assume that $n\ge 4, d\ge 3$. Except the case $n=4, d=4$, the universal
deformation of a smooth hypersurface in
$\IP^{n-1}$ with degree $d$ is given by the isomorphism:
\begin{equation}
S^d/J^d_f\to H^1(X_0,\Theta_{X_0}).
\end{equation}
The dimension of the moduli space $S$ is
\begin{equation}
\mu=\dbinom{n-1+d}{d}-n^2.
\end{equation}
\end{thm}

\begin{ex} Consider the moduli space of elliptic curves, then the
dimension is
$$
\binom{2+3}{3}-3^2=1.
$$
\end{ex}

\begin{ex} Consider the moduli space of the quintic polynomials in
$\IP^4$, then the dimension is
$$
\binom{4+5}{5}-5^2=101.
$$
\end{ex}

Our analysis of the moduli of the deformation space implies the
following conclusion about the number of modules:

\begin{thm}[LG/CY correspondence between moduli numbers] Let $f$ be a smooth hypersurface of degree $d$ in the projective space
$\IP^{n-1}$. If $n\ge 4$ and $d\ge 3$ but except the case $n=4,
d=4$, then the dimension of its deformation space equals to the
number of marginal deformations in the universal unfolding of the
singularity $f$.
\end{thm}

\begin{proof} We know that the number of moduli of the hypersurface of
degree $d$ in the projective space $\IP^{n-1}$ is
\begin{equation}
\mu=\dbinom{n-1+d}{d}-n^2
\end{equation}
On the other hand, as singularity germ $f$ is holomorphic equivalent
to its normal form, the homogeneous Fermat polynomials
$f_0=z_1^d+\cdots+z_n^d$. Its Milnor algebra at $0$ is generated by
the polynomials $z_1^{k_1}\cdots z_n^{k_n}, 0\le k_i\le
d-2,i=1,\cdots,n$. So we need to compute the number of polynomials
$z_1^{k_1}\cdots z_n^{k_n}$  with $k_1+\cdots+k_n=d$ in the Milnor
algebra. Imagine that we have $d$ white balls and $n-1$ red balls in
one line, and then in total we have $n-1+d$ balls in one line and
order them from $1$ to $n-1+d$. If we ignore the colors of the
balls, then there are
\begin{equation}\label{ident: moduli-number}
\dbinom{n-1+d}{n-1}
\end{equation}
ways to choose $n-1$ balls. This is equivalent to insert $n-1$ red
balls into the $d$ white balls. However, we must rule out the
choices with respect to the decomposition $d=d+0$ and $d=(d-1)+1$.
For $d=d+0$, there are $n$ ways to insert the $d$ white balls
together into the two neighboring red balls. For $d=(d-1)+1$, there
are $n(n-1)$ ways to put $d-1$ white balls together between two
neighboring red balls and the rest one white ball to other gap. So
we have
$$
\dbinom{n-1+d}{n-1}-n-n(n-1)
$$
marginal deformations in the universal unfolding of the singularity
$z_1^d+\cdots+z_n^d$. This proves the conclusion.
\end{proof}

\begin{rem} For any $n,d$, the number of marginal deformations in
the universal unfolding of $f=z_1^d+\cdots+z_n^d$ is given by
(\ref{ident: moduli-number}). In case $n=4, d=4$ which corresponds
to $K3$ surface in $\IP^3$, the moduli of the complex structure is
$H^1(X,\Theta_X)=20$ (ref. \cite{Ko}, Page 247). The reason is that
one dimensional deformation is non-algebraic (is transcendental).
The number of the marginal deformations in the universal unfolding
of $f$ is $19$. The moduli of the quadratic surface has been given
in some formula of \cite{Ko} and the dimension of the Milnor number
of the singularity $f$ is $1$.
\end{rem}

\subsubsection{\underline{Deformation of nondegenerate and convenient Laurent
polynomials}}\

\

S. Barannikov \cite{Bar} has studied the Frobenius manifold
associated to a special Laurent polynomial and compare it with the
quantum cohomology of $\C P^n$. Later Douai and Sabbah
(\cite{Do},\cite{DS},\cite{DS2},\cite{Sa}) have studied the
deformation theory for general nondegenerate and convenient Laurent
polynomials defined on the algebraic torus $(\C^*)^n$. They can
construct the Frobenius manifold structure based on this
deformation. We will show the deformations they considered are
actually strong deformation in our notation.

Let $f$ be a convenient and nondegenerate Laurent polynomial. The
$\C$-vector space
$Q_f=\C[z_1,\cdots,z_n,z_1^{-1},\cdots,z_n^{-1}]/J_f$ is finite
dimensional and its dimension $\mu(f)$ is the sum of the Milnor
numbers of $f$ at each critical point. Let $\{g_i,i=1,\cdots,r\}$ be
the set of monomials such that their corresponding lattice points
are contained in the interior of the Newton polyhedron of $f$ and
injects to the Jacobi space $Q_f$. Then the deformation is
\begin{equation}\label{eq:deform-Laurent-1}
F(z,t)=f(z)+\sum_{j=1}^r t_j g_j.
\end{equation}
We call it as the subdiagram deformation. It was shown in PP. 23 of
\cite{Sa} that for any $t=(t_1,\cdots,t_r)\in \C^r$ the Laurent
polynomial $F(z,t)$ is convenient and nondegenerate. Therefore by
Proposition \ref{prop:Laurent-strong-1}, we have the following
result.

\begin{thm} Let $f$ be a convenient and nondegenerate Laurent
polynomial defined on the algebraic torus $(\C^*)^n$. Then the
deformation $F(z,t)$ with base $\C^r$ defined by identity
(\ref{eq:deform-Laurent-1}) is a strong deformation.
\end{thm}

\subsection{A priori estimate, existence and regularity}

In this section, we always keep the following basic assumptions:
\begin{itemize}
\item $(M,g,f)$ is a strongly tame section-bundle system.

\item $\lambda_0$ is not a spectrum point of
$\Delta_f$, and the section $\psi\in \Dom(\Delta_f)$ satisfies the
equation $(\lambda_0-\Delta_f)\psi=\varphi$.
\end{itemize}

Since $(M,g,f)$ is supposed to be strongly tame, there exists a
compact set $K_{\lambda_0}\subset M$ such that for any $z\not\in
K_{\lambda_0}$, the following inequality holds
\begin{equation}\label{ineq:aprio-ineq-1}
2\lambda_0+\frac{1}{2}|\nabla f|^2\le 2|L(\nabla \pat f)|+|\nabla
f|^2\le 2|\nabla f|^2.
\end{equation}
where $L(\nabla\pat f):=g^{\mub \nub}\nabla_\nu f_l
\iota_{\pat_\mub}dz^l\wedge$ and $L_f(\cdot)=L(\nabla\pat f)
\cdot+\overline{L(\nabla\pat f) \cdot}$.

We have the global energy estimate.

\begin{lm}\label{lm:ener-esti} Under the basic assumption in this
section, the following inequality holds
\begin{equation}
\int_M|\bpat \psi|^2+|\bpat^\dag \psi|^2+|\nabla f|^2|\psi|^2\le
C\int_M |\varphi|^2,
\end{equation}
where $C$ depends on $(M,g,f), \lambda_0$ and the distance of
$\lambda_0$ to the spectrum of $\Delta_f$.
\end{lm}

\begin{proof} We have
$$
(\Delta_f \psi,\psi )_{L^2}=(\lambda_0
\psi,\psi)_{L^2}-(\varphi,\psi)_{L^2}.
$$
Since $\Delta_f=\Delta_\bpat+L(\nabla \pat
f)\circ+\overline{L(\nabla \pat f)\circ}+|\nabla f|^2$, we have
$$
\int_M |\bpat \psi|^2+|\bpat^\dag \psi|^2+|\nabla
f|^2|\psi|^2+\int_M 2Re (L(\nabla \pat f)\circ\psi, \psi)\le
(\lambda_0+1)||\psi||_{L^2}+||\varphi||^2_{L^2}.
$$
Notice that (\ref{ineq:aprio-ineq-1}) holds for $z\in
M-K_{\lambda_0}$, so plus the interior integration over
$K_{\lambda_0}$ we have the estimate:
\begin{align*}
&\int_M|\bpat \psi|^2+|\bpat^\dag \psi|^2+\frac{1}{2}|\nabla
f|^2|\psi|^2\\
&\le
(\lambda_0+2)||R_{\lambda_0}(\Delta_f)\varphi||_{L^2}+||\varphi||_{L^2}
\le C||\varphi||^2_{L^2(M)}.
\end{align*}
So we proved the conclusion.
\end{proof}

The following lemma gives the local energy estimate near the
infinite far place:
\begin{lm}\label{lm:apri-esti-1} Suppose that the basic assumption in this section holds.
Then for any ball $B_R(z_0)\cap K_{\lambda_0}=\emptyset$ , we have
\begin{equation}
\int_{B_{\frac{R}{2}}(z_0)}|\nabla \psi|^2+|\nabla f|^2|\psi|^2\le
C\int_{B_R(z_0)}|\varphi|^2+|\psi|^2,
\end{equation}
where $C$ only depends on the geometry of $(M,g,f)$ and $\lambda_0$.
\end{lm}

\begin{proof} Let $\chi(z)$ be a smooth cut-off function with support
in $B_R(z_0)$ and equals $1$ on $B_{\frac{R}{2}}(z_0)$. By the
equation of $\psi$, we obtain
$$
\int_M \lambda_0 (\psi, \chi^2\psi)-(\Delta_f\psi,
\chi^2\psi)=\int_M (\varphi,\psi)\chi^2.
$$
Note that $\Delta_f=\Delta_{\bpat}+L(\nabla\pat
f)+\overline{L(\nabla\pat f)}+|\nabla f|^2$, replacing it into the
above inequality, there is
$$
\int_{B_R} \chi^2 |\nabla \psi|^2+|\nabla f|^2\chi^2|\psi|^2+2
Re(L(\nabla\pat f)\circ
\psi,\psi)\chi^2-\lambda_0\chi^2|\psi|^2=-\int_{B_R}
(\varphi,\psi)\chi^2-\int_{B_R} (\nabla \psi,2\chi\nabla \chi\psi).
$$
So we have
$$
\int \chi^2 |\nabla \psi|^2+(|\nabla f|^2-2|L(\nabla\pat f)|
-\lambda_0)|\psi|^2\chi^2\le \int
|\varphi||\psi|\chi^2+|\nabla\psi||\nabla \chi|\chi |\psi|.
$$
By Cauchy inequality, there is
$$
\int_{B_R} \chi^2 |\nabla \psi|^2+(|\nabla f|^2-2|L(\nabla\pat f)|
-\lambda_0-\epsilon)|\psi|^2\chi^2\le  C_\epsilon \int_{B_R} (
|\varphi|^2 \chi^2+|\nabla \chi|^2|\psi|^2).
$$
So we get the conclusion.
\end{proof}

The following lemma gives the weak maximum principle near the
infinite far place:
\begin{lm}\label{lm:apri-maxi-1} Suppose that the basic assumption in this section holds. If $B_R(z_0)\cap
K_{\lambda_0}=\emptyset$, then
$$
\sup_{B_{\frac{R}{2}}}|\psi(z)|\le
C(||\psi||_{L^2(B_R)}+||\varphi||_{L^2(B_R)} ),
$$
where $C$ only depends on $(M,g,f)$ and $\lambda_0$.
\end{lm}

\begin{proof} We have the identity
\begin{align*}
\Delta |\psi|^2=&(\Delta\psi,\psi)+(\psi,\Delta\psi)-2|\nabla
\psi|^2\\
=&(-|\nabla f|^2)|\psi|^2-2|\nabla \psi|^2+(-|\nabla f|^2 |\psi|^2-4
Re(L_f\circ(\psi),\psi)+2\lambda_0|\psi|^2)-2 Re(\psi,\varphi),
\end{align*}
Hence if the point $z\not\in K_{\lambda_0}$, we have
$$
\Delta |\psi|^2+2|\nabla\psi|^2\le 2|\psi||\varphi|,
$$
and then in weak sense that
\begin{equation}
\Delta |\psi|\le |\varphi|.
\end{equation}
By weak maximum principle (see \cite{HL,GT}), if $B_R(z_0)\cap
K_{\lambda_0}=\emptyset$, then
$$
\sup_{B_{\frac{R}{2}}}|\psi(z)|\le
C(||\psi||_{L^2(B_R)}+||\varphi||_{L^2(B_R)} ).
$$
\end{proof}

\begin{crl} Suppose that the basic assumption in this section holds.
Then as $d(z,z_0)\to \infty$.
$$
|\psi(z)|\to 0.
$$
and the following inequality holds:
\begin{equation}
\sup_{M-K_{\lambda_0}}|\psi(z)|\le C||\varphi||_{L^2(M)},
\end{equation}
where $C$ only depends on $(M,g,f)$ and $\lambda_0$.
\end{crl}

\begin{proof} Since $\psi,\varphi\in L^2(M)$, for any $\epsilon>0$, there exists a $R_0>0$ such that
$$
||\psi||_{L^2(M_{R_0}^c)}+||\varphi||_{L^2(M_{R_0}^c)}\le \epsilon,
$$
where $M_{R_0}^c:=\{z\in M|d(z,z_0)>R_0\}$. By Lemma
\ref{lm:apri-maxi-1}, there is
$$
\sup_{B_{\frac{R}{2}}}|\psi(z)|\le C\epsilon.
$$
This shows that $|\psi(z)|\to 0$ and the second conclusion holds
naturally.
\end{proof}

Now we want to estimate the higher order derivatives of $\psi$. We
define a subspace in $L^2(M)$ space, $W^{k,2}(M)$, whose norm is
given by
\begin{equation}
||\varphi||_{W^{k,2}_f(M)}=\left(\sum_{I,|I|\le k} \int_M |D^I
\varphi|^2 dv_M \right)^{1/2}.
\end{equation}
Here the first order operator $D$ represents $\bpat_f$ or
$\bpat^\dag_f$. We know that if $(M,g,f)$ is strongly tame and
$k=1$, then the norm $||\cdots||_{W^{1,2}_f(M)}$ is equivalent to
the following norm:
$$
\left( \int_M |\nabla\varphi|^2 + (|\nabla f|^2+1) dv_M
\right)^{1/2}.
$$

\begin{thm}\label{thm:apri-high} Suppose that the basic assumption in this section holds.
Furthermore assume that $\varphi\in W^{k,2}_f(B_R(z_0))$, where
$B_R(z_0)\cap K_{\lambda_0}=\emptyset$. Then there exists a constant
$C_k$ depending only on the geometry of $(M,g,f)$ and $\lambda_0,k$
such that
\begin{equation}
\sum_{I,j,|I|+j=k}\int_{B_{\frac{R}{2}}(z_0)}|\nabla D^I\psi|^2
|\nabla f|^{2j}\le
C_k\left(||\varphi||^2_{W^{k,2}_f(B_R(z_0))}+||\psi||^2_{L^2(B_R(z_0))}.
\right)
\end{equation}
\end{thm}

\begin{proof} We consider $k=1$ case in detail. Notice that
$$
(\lambda_0-\Delta_f)\bpat_f\psi=\bpat_f\varphi,\;(\lambda_0-\Delta_f)\bpat_f^\dag
\psi=\bpat_f^\dag \varphi.
$$
Replacing $\bpat_f\psi$ and $\bpat_f^\dag\varphi$ into the
inequality in Lemma \ref{lm:apri-esti-1} respectively, we get two
inequalities and then sum them to get
$$
\int_{B_{\frac{R}{2}}(z_0)} |\nabla \bpat_f \psi|^2+|\nabla
\bpat_f^\dag \psi|^2+|\nabla f|^2(|\bpat_f\psi|^2+|\bpat^\dag_f
\psi|^2)\le
C\int_{B_{\frac{3R}{4}}}\left(|\bpat_f\psi|^2+|\bpat_f^\dag
\psi|^2+|\bpat_f^\dag \varphi|^2+|\bpat_f \varphi|^2 \right)
$$
Applying the strongly tame condition to the left hand side and
applying Lemma \ref{lm:apri-esti-1} to the right hand side, we
obtain
\begin{equation}
\int_{B_{\frac{R}{2}}(z_0)} |\nabla \bpat_f \psi|^2+|\nabla
\bpat_f^\dag \psi|^2+|\nabla f|^2|\nabla \psi|^2+|\nabla
f|^4|\psi|^2\le
C\left(||\varphi||^2_{W^{1,2}_f(B_R(z_0))}+||\psi||^2_{L^2(B_R(z_0))}\right).
\end{equation}
$k\ge 2$ cases can be obtained by recursion method.
\end{proof}

Now we can get the uniform estimate for the derivatives of $\psi$:

\begin{thm}\label{thm:high-maxi} Suppose that the basic assumption in this section holds and
 assume that $\varphi\in W^{k,2}_f(M)$. If $B_R(z_0)\cap
K_{\lambda_0}=\emptyset$, then there exists a constant $C_k$ such
that for any $l,|l|\le k+1$ the following holds:
\begin{equation}
\sup_{B_{\frac{R}{2}}(z_0)}|D^l\psi|\le
C_k\left(||\varphi||_{W^{k,2}_f(B_R(z_0))}+||\psi||_{L^2(B_R(z_0))}.
\right)
\end{equation}
where $C$ only depends on the geometry of $(M,g,f)$ and
$\lambda_0,k$.
\end{thm}

\begin{proof} $D^l\psi$ satisfies
$$
(\lambda_0-\Delta_f)(D^l\psi)=D^l\varphi.
$$
By Lemma \ref{lm:apri-maxi-1}, we have
$$
\sup_{B_{\frac{R}{2}}(z_0)}|D^l\psi|\le
C_k\left(||D^l\varphi||_{L^2(B_{\frac{3R}{4}}(z_0))}+||D^l\psi||_{L^2(B_{\frac{3R}{4}}(z_0))}
\right)
$$
By Theorem \ref{thm:apri-high} and the tameness of $(M,g,f)$, we
have the estimate
\begin{align*}
||D^l\psi||^2_{L^2(B_{\frac{3R}{4}}(z_0))}\le& C\int
|\bpat_fD^{l-1}\psi|^2+|\bpat_f^\dag
D^{l-1}\psi|^2\\
\le&C\int_{B_{\frac{7R}{8}}}|\nabla
D^{l-1}\psi|^2+|D^{l-1}\psi|^2|\nabla f|^2\\
\le&
C_{k-1}\left(||\varphi||^2_{W^{k,2}_f(B_R(z_0))}+||\psi||^2_{L^2(B_R(z_0))}
\right).
\end{align*}
Therefore, we have the estimate
\begin{equation}
\sup_{B_{\frac{R}{2}}(z_0)}|D^l\psi|\le
C_k\left(||\varphi||_{W^{k,2}_f(B_R(z_0))}+||\psi||_{L^2(B_R(z_0))}.
\right).
\end{equation}
\end{proof}

\begin{crl} Suppose that the basic assumption in this section holds and
 assume that $\varphi\in W^{k,2}_f(M)$. Then for any $l,|l|\le k$,
$$
|D^l\psi(z)|\to 0,\;\text{as}\;|z|\to \infty.
$$
Furthermore, we have the estimate
\begin{equation}
\sup_{M-K_{\lambda_0}}|D^l\psi(z)|\le C_k||\varphi||_{W^{k,2}_f(M)},
\end{equation}
where $C_k$ only depends on the geometry of $(M,g,f)$ and
$\lambda_0,k$.
\end{crl}

\subsubsection{\underline{Equivalence of the norms}}

Here we will discuss the relations between $D^l$ and $\nabla^l$.

We have four basic operators $\pat,\bpat,\pat f\wedge, (\pat
f\wedge)^\dag$. They satisfy the following commutation relations
\begin{align*}
&[\bpat^\dag,\pat f\wedge]=g^{\mub \nu}\nabla_\nu
f_l\iota_{\pat_\mub}dz^l\wedge=f^\mub_\nu
\iota_{\pat_\mub}dz^\nu\wedge\\
&[\bpat,(\pat f\wedge)^\dag]=\overline{f^\mub_\nu} \iota_{\pat_\mu}
dz^\nub\wedge,\\
&[\bpat,\pat f\wedge]=0,\;[\bpat^\dag,(\pat f\wedge)^\dag]=0.
\end{align*}

The first order differential operators of $D^l$ have two which have
the formulas:
$$
\bpat_f=\bpat+\pat f\wedge=\bpat+f_a
dz^a\wedge,\;\bpat_f^\dag=\bpat^\dag+f^b\iota_{\pat_b}.
$$
The second order operators have two which have the formulas:
\begin{align*}
&\bpat^\dag_f\bpat_f=\bpat^\dag\bpat+f^a\iota_{\pat_a}\bpat-f_adz^a\wedge\cdot
\bpat^\dag+f^{\mub}_\nu\iota_{\pat_\mub}dz^\nu\wedge+f^a\iota_{\pat_a}\cdot(f_b
dz^b\wedge)\\
&\bpat_f\bpat^\dag_f=\bpat\bpat^\dag+f_a dz^a\wedge
\bpat^\dag-f^b\iota_{\pat_b}\bpat+\overline{f^\mub_\nu}\iota_{\pat_\mu}dz^\nub\wedge+f_af^bdz^a\wedge\cdot\iota_{\pat_b}.
\end{align*}
The higher order differential operators have the four types:
$$
\Delta_f^s, \Delta_f^s\cdot
\bpat^\dag_f\bpat_f,\Delta_f^s\bpat_f^\dag,\Delta_f^s\bpat_f.
$$
Here
$$
\Delta_f=\Delta_\bpat+f^{\mub}_\nu\iota_{\pat_\mub}dz^\nu\wedge+\overline{f^\mub_\nu}\iota_{\pat_\mu}dz^\nub\wedge+|\nabla
f|^2.
$$
Each operator then has the types $D_1D_2D_1\cdots D_1$ or
$D_1D_2D_1\cdots D_2$, where $D_1$ can represent $\bpat_f$ or
$\bpat_f^\dag$ and $D_2$ then represents the rest one. We need move
all the terms $f_adz^a\wedge $ or $f^a\iota_{\pat_a}$ from the
sequence to the left hand side. When commuting with $\bpat$ or
$\bpat^\dag$, a higher order derivatives of $f$ will generate. Hence
finally, all the terms in $D_1D_2D_1\cdots D_1$ have the form:
$$
c_{IJ}\nabla^{I_k} f\nabla^{J_{s-k}}, 0\le k\le s,
$$
where $s$ is the length of $D_1D_2D_1\cdots D_1$ and
$\nabla^{I_k}=\nabla^{i_1}\cdots\nabla^{i_k},
i_l\in\{1,\cdots,n,\bar{1},\cdots,\bar{n}\}$.

If the length is $2s$, then the highest order derivatives is
permutation of $\bpat,\bpat^\dag$'s and the $0$-th order term is
$$
|\nabla f|^{2s}-\sum_{I,|I|=2s}c_I \nabla^I f\cdot L,
$$
where $L$ is a multiplication operator preserving the degree of the
forms.

\begin{df} Let $(M,g,f)$ be a strongly tame section-bundle system. If for any
$s=1,2,\cdots,$ and for any $C>0$, the following relations hold:
\begin{equation}
|\nabla f|^{2s}-C\sum_{I,|I|=2s} |\nabla^I f|\to
\infty,\;\text{as}\;d(z,z_0)\to\infty,
\end{equation}
then $f$ is said to be strongly regular tame and $(M,g,f)$ is said
to be a strongly regular tame section-bundle system.
\end{df}

\begin{ex} Let $W$ be a nondegenerate quasi-homogeneous polynomial,
then $(\C^N, W)$ with the standard K\"ahler metric is a strongly
regular tame section-bundle system.
\end{ex}

\begin{ex} Let $f$ be a nondegenerate and convenient Laurent
polynomial defined on the algebraic torus $T$, then $(T,f)$ with the
standard K\"ahler metric is a strongly regular tame section-bundle
system.
\end{ex}

\begin{prop} Let $(M,g,f)$ be a strongly regular tame section-bundle system.
Then the norm $||\cdot||_{W^{k,2}_f}$ is equivalent to the following
norm
$$
\left(||\cdot||_{W^{k,2}}^2+||(|\nabla f|+1)\cdot||_{L^2}^2
\right)^{1/2}.
$$
\end{prop}

\begin{proof} Apply pointwise interpolation theorem to derivatives
with middle order.
\end{proof}

\subsubsection{\underline{Existence and regularity}}

We know that if $\lambda_0$ is not a spectrum point of $\Delta_f$,
then $(\lambda_0-\Delta_f)\psi=\varphi$ has a unique solution in
$W^{1,2}_f(M)$ if $\varphi\in L^2(M)$. This is a weak solution of
the Schr\"odinger equation. The existence of the weak solution is
equivalent to the application of the Lax-Milgram theorem for
quadratic form which is coercive and has positive lower bound. If
$\varphi\in W^{k,2}_f(M)$, then $\psi\in W^{k+1,2}_f(M)$. We can
define a function class $\mathscr{S}_f$ consisting of the smooth
function $u$ such that for any $I,J, |I|,|J|=0,1,\cdots,$ the
following holds:
$$
\sup_M |\nabla^I f||\nabla^J u|<\infty.
$$
If $f$ is a polynomial, then $\mathscr{S}\subset \mathscr{S}_f$, the
function space of rapid decrease. In particular, the function space
of compact support $C_0(M)\subset \mathscr{S}_f$ for any holomorphic
function $f$.

We have the existence and regularity theorem:

\begin{thm}\label{thm:exist-non-homo} Let $(M,g,f)$ be a strongly tame section-bundle system and assume
that $\lambda_0$ is not a spectrum point of $\Delta_f$. If
$\varphi\in \mathscr{S}_f$, then the equation
$$
(\lambda_0-\Delta_f)\psi=\varphi
$$
has a unique solution in $W^{k,2}_f(M)$ for any $k$.
\end{thm}

\subsection{Continuity of the spectrum}

Above all, we cite some conclusions about unbounded self-adjoint
operators in functional analysis. The reader can refer to Theorem
VIII 20, Theorem 23 and Theorem 25 of the book \cite{RS} to find the
proof of Proposition \ref{prop: func-conv}.

\begin{df} Let $A_n,n=1,\cdots,$ and $A$ be self-adjoint operators.
Then $A_n$ is said to converge to $A$ in the norm resolvent sense if
their resolvent $R_\lambda(A_n)\to R_\lambda(A)$ in norm for all
$\lambda$ with $\im\lambda\neq 0$. $A_n\to A$ is said to converge in
strong resolvent sense, if the resolvent $R_\lambda(A_n)\to
R_\lambda(A)$ strongly for all $\lambda$ with $\im\lambda\neq 0$.
\end{df}

Here we only use the norm resolvent convergence.

\begin{prop}\label{prop: func-conv} The following conclusions hold:
\begin{enumerate}
\item Let $\{A_n\}_{n=1}^\infty$ and $A$ be self-adjoint operators
with a common domain $D$ and norm $||\cdot||_D$ with
$||\varphi||_D=||A\varphi||+||\varphi||$. If
$$
\sup_{||\varphi||_D=1}||A_n\varphi-A\varphi||\to 0,
$$
then $A_n\to A$ in the norm resolvent sense.

\item If $A_n\to A$ in the norm resolvent sense and $f$ is a
continuous function on $\R$ vanishing at $\infty$, then
$||f(A_n)-f(A)||\to 0$.

\item Let the interval $I=[a,b]\subset \R$ has no intersection with
the spectrum $\sigma(A)$ of $A$,then for large $n$ the projection
$P_I(A_n)$ is well-defined and satisfies
$$
||P_I(A_n)-P_I(A)||\to 0.
$$
\end{enumerate}
\end{prop}

Consider the strong deformation with the form $f_t=f+\sum_i t_i
g_i$.

\begin{lm}\label{lm:deform-2} Let $\varphi\in \Dom(\Delta_0)$ and
$||\varphi||_{g,0}=1$, then
$$
||\Delta_t\varphi-\Delta_0\varphi||_{L^2}\to 0.
$$
\end{lm}

\begin{proof} Since $\Delta_t-\Delta_0$ is a symmetric operator, it
suffices to prove that
$$
|(\Delta_t-\Delta_0\varphi, \varphi)|\to 0,
$$
for any $\varphi$ such that $||\varphi||_{g,0}<\infty$. We have
\begin{align*}
|(\Delta_t-\Delta_0\varphi, \varphi)|=&\left|(L(\nabla\pat
g_i)(\varphi),\varphi)(t_i)+|t_i|^2|\nabla g_i|^2|\varphi|^2+2t_i
Re(\nabla f\varphi, \nabla g_i\varphi) \right|\\
\le & |t|||\nabla f \varphi||^2\le |t|||\varphi||_{g,0}.
\end{align*}
Thus
$$
||\Delta_t\varphi-\Delta_0\varphi||_{L^2}\to 0.
$$
\end{proof}

Therefore by \ref{lm:deform-2} and Proposition \ref{prop:
func-conv}, we have the corollary:

\begin{crl}Let the interval $I=[a,b]\subset \R$ has no intersection with
the spectrum $\sigma(\Delta_f)$ of $\Delta_f$,then there exists a
constant $\delta>0$, if $|t|<\delta$, then
$$
\dim \im P_I(\Delta_t)=\dim \im P_I(\Delta_0).
$$
\end{crl}

By Theorem \ref{thm:prel-disc-spectrum}, we can list all the
eigenvalues of $\Delta_t$ in the following order:
$$
0=\lambda_1(t)\le
\lambda_2(t)\le\cdots\lambda_k(t)\le\cdots\to\infty.
$$
We have the continuity theorem of eigenvalues:

\begin{thm}\label{thm:defor-conti-spectrum} $\lambda_k(t)$ is a continuous function for $t\in S$.
\end{thm}

\begin{proof} This theorem can be proved using the induction method
with respect to $k$. The only required fact is Lemma
\ref{lm:deform-2}. The reader can refer to Theorem 7.2 of \cite{Ko}
for a detail description of the proof.
\end{proof}

Let $P_{0,t}$ be the projection operator from the $L^2$ space to the
space of $\Delta_t$-harmonic forms. By using Lemma
\ref{lm:deform-2},we can easily prove the following result:

\begin{thm} $\dim P_{0,t}$ is uppersemicontinuous in $t\in S$.
\end{thm}

\subsection{Estimate of the eigenforms and the Green function}

We use the maximum principle of the scalar Laplace operator to build
the decay estimate of the eigenforms of $\Delta_f$. Consider the
fundamental solution of the following linear scalar equation:
\begin{equation}\label{eq:fund-solu}
(\Delta+a^2)E(z)=\delta(z),
\end{equation}
where $k$ is a given nonzero constant.

\begin{lm}Let $(M,g)$ be a $n$-dimensional complete non-compact Riemannian manifold with
bounded geometry. Then the fundamental solution $E(x)$ of the
operator exists and is unique. It has the approximating estimate as
$d(x,x_0)\to \infty$:
\begin{equation}
E_a(x)=c(d(x,x_0))^{(n-1)/2}e^{-a d(x,x_0)}(1+o(1)).
\end{equation}
Here $d(x,x_0)$ is the distance function from the point $x$ to
$x_0$.
\end{lm}

\begin{proof} If $M$ is the standard Euclidean space, then the
function $(d(x,x_0))^{(n-1)/2}e^{-a d(x,x_0)}$ is the fundamental
solution of the equation (\ref{eq:fund-solu}). If $M$ is the
complete Riemannian manifold with bounded geometry, the conclusion
is obtained by using comparison principle.
\end{proof}

\begin{thm}\label{thm:decay-eigen-1} Let $(M,g,f)$ be a strongly tame section-bundle system
and $\varphi$ is an eigenform of $\Delta_f$ corresponding to the
eigenvalue $\lambda$. Then there exists a constant $C$, for any
$a>0$ there is :
\begin{equation}
|\varphi|\le C e^{-a d(z,z_0)},
\end{equation}
where $z_0$ is an arbitrarily given base point on $M$.
\end{thm}

\begin{proof} Since $\varphi$ is an eigenform of the self-adjoint
operator $\Delta_f$, it has $W^{2,2}_{loc}$ smoothness. By $L^p$ and
Schauder theory of elliptic operators, $\varphi$ is $C^\infty$.

The Laplace operator on $M$ is given by
$$
\Delta=-g^{\nub\mu}\nabla_\mu\nabla_\nub.
$$
Let $\varphi$ be a eigenform of $\Delta_f$ with eigenvalue
$\lambda$, i.e.,
$$
\lambda\varphi=\Delta_f\varphi=\Delta
\varphi+L_f\circ(\varphi)+|\nabla f|^2\varphi.
$$
Then We have
\begin{align*}
\Delta
|\varphi|^2=&(\Delta\varphi,\varphi)+(\varphi,\Delta\varphi)-2|\nabla
\varphi|^2\\
=&(2\lambda-|\nabla f|^2)|\varphi|^2-2|\nabla \varphi|^2+(-|\nabla
f|^2 |\varphi|^2-2 Re(L_f\circ(\varphi),\varphi)),
\end{align*}
or for any $a>0$ there is
$$
(\Delta+a^2)|\varphi|^2=(a^2+2\lambda-|\nabla
f|^2)|\varphi|^2-2|\nabla \varphi|^2+(-|\nabla f|^2 |\varphi|^2-2
Re(L_f\circ(\varphi),\varphi)).
$$
There exists a constant $R$ depending only on $M,a$ and the geometry
of the section bundle system $(M,g,f)$ such that outside the ball
$B_R(z_0)$ the following inequality holds
$$
(\Delta+a^2)(|\varphi|^2-ME_a(z))\le 0.
$$
Now we can choose $M$ large enough such that on $\pat B_R(z_0)$,
there is
$$
|\varphi|^2(z)-ME_a(z)\le 0.
$$
Using the maximum principle in $M-B_R(z_0)$, we obtain the
conclusion.
\end{proof}

By the continuity theorem of the spectrum, Theorem
\ref{thm:defor-conti-spectrum}, we have:

\begin{crl}\label{crl:eign-decay} Suppose that $(M,g,f_t)$ is a strong deformation of the section-bundle system $(M,g,f)$ in
parameter space $S$. Then for any eigenform $\varphi_k(t)$ of
$\Delta_t$ with respect to the eigenvalue $\lambda_k(t)<\lambda_0$,
there exists a constant $C_0$ only depending on $(M,g,f)$ and
$\lambda_0$ such that
$$
|\varphi|\le C_0 e^{-a d(z,z_0)}.
$$
\end{crl}

\subsubsection{\underline{Estimate of the higher order
derivatives}}\

Let $(\varphi,\lambda)$ be a solution of the eigenvalue problem
$(\lambda-\Delta_f)\varphi=0$. Choose a point $\lambda_0\in\R$ which
is not a spectrum point of $\Delta_f$. Then we have
\begin{equation}\label{eq:decay-esti-1}
(\lambda_0-\Delta_f)\varphi=(\lambda_0-\lambda)\varphi.
\end{equation}
By bootstrap argument, we know that $\varphi\in W^{k,2}_f(M)$ for
any $k$. By Theorem \ref{thm:apri-high},\ref{thm:high-maxi}, we have
the following estimate.

\begin{prop}\label{prop: eigenform-decay} Let $K_{\lambda_0}$ be the domain in Lemma
\ref{lm:apri-esti-1} and $B_R(z)\cap K_{\lambda_0}=\emptyset$. Then
there exists a constant $C_k$ depending only on the geometry of
$(M,g,f)$ and $\lambda_0,\lambda,k$ and any $a>0$ such that
\begin{equation}
\sum_{I,j,|I|+j=k}\int_{B_{\frac{R}{2}}(z_0)}|\nabla D^I\varphi|^2
|\nabla f|^{2j}\le C_k e^{-a (d(z,z_0)-R)},
\end{equation}
and the pointwise estimate:
\begin{equation}
|D^l\varphi(z)|\le C_k e^{-a d(z,z_0)}.
\end{equation}
\end{prop}

\begin{proof} Since for any $I,|I|=0,1,\cdots,$ $D^I\varphi$ is a
solution of the equation (\ref{eq:decay-esti-1}), we can use the
local energy estimate, Lemma \ref{lm:apri-esti-1}, to get the
control of the local $W^{k,2}_f$ norm by the local $L^2$ norm of
$\varphi$. Replacing this estimate into Theorem
\ref{thm:apri-high},\ref{thm:high-maxi}, we get the control by the
local $L^2$ norm. Finally we use Theorem \ref{thm:decay-eigen-1}
estimate to get the decay estimate.
\end{proof}

Similarly, we have the uniform decay estimate for the strong
deformation.

\begin{crl} Suppose that $(M,g,f_t)$ is a strong deformation of the section-bundle system $(M,g,f)$ in
parameter space $S$. Then for any eigenform $\varphi(t)$ of
$\Delta_t$ with respect to the eigenvalue $\lambda(t)<\lambda_0$
($\lambda_0$ is assumed not to be a spectrum point of any
$\Delta_t$), there exists a constant $C_k$ only depending on
$(M,g,f),\lambda_0$ and any $a>0$ such that the following estimate
hold:
\begin{align*}
&\sum_{I,j,|I|+j=k}\int_{B_{\frac{R}{2}}(z_0)}|\nabla D^I\varphi|^2
|\nabla f|^{2j}\le C_k e^{-a (d(z,z_0)-R)}\\
&|D^l\varphi(t,z)|\le C_k e^{-a d(z,z_0)},\;\text{for any}I, |I|\le
k.
\end{align*}
\end{crl}

\subsubsection{\underline{Estimate of the Green function}}

Let $G(z,w)$ be the Green function of $\Delta_f$, i,e, it satisfies
$$
(\Delta_f)_z G(z,w)=\delta(z-w).
$$
In the domain $\{(z,w)\in M\times M|dist(z,w)\ge 1\}$, $G(z,w)$ is a
harmonic form and so has the following decay estimate:

\begin{prop}\label{prop: Green-func-decay} For any $(z,w)\in M\times M$ such that the distance
$dist(z,w)$ is large enough, there exists a constant $C_k$ only
depending on $(M,g,f),\lambda_0$ and any $a>0$ such that the
following estimate hold:
$$
|D_z^l G(z,w)|\le C_k e^{-a d(z,w)},\;\text{for any}I, |I|\le k.
$$
In particular, if $f$ is strongly regular tame, then
$$
|\nabla^I_z G(z,w)|+|\nabla f(z)|^{2|I|}|G(z,w)|\le C_k e^{-a
d(z,w)},\;\text{for any}I, |I|\le k.
$$
\end{prop}

The diagonal $\D\subset M\times M$ is the singular set of $G(z,w)$.
The asymptotic property of $G(z,w)$ as $d(z,w)\to 0$ is given by the
following result
\begin{prop}\label{prop:Green-func-appr} As $z\to w$, there holds
\begin{equation}
G(z,w)\to\begin{cases} c_n d(z,w)^{-(2n-2)},&\;n\ge 2 \\
c\log(d(z,w)),&\;n=1.
\end{cases}
\end{equation}
\end{prop}

\begin{proof} By the work of Malgrange \cite{Ma} and Li-Tam \cite{LT}, there exists a bounded symmetric
Green's function $E_0(x,y)$ on any
complete Riemannian manifold, i.e., $E_0(x,y)$ satisfies:
\begin{itemize}
\item $E_0(x,y)=E_0(y,x)$;
\item
$(\Delta_{0})_xE_0(x,y)=\delta(x-y),\;(\Delta_{0})_yE_0(y,x)=\delta(y-x)$.
\item $E_0(x,y)$ is bounded in any domain in $M\times M$ which has
positive distance to the diagonal $\D$.
\end{itemize}
Here $\Delta_0=-\frac{1}{\sqrt{g}}\frac{\pat}{\pat
x_i}(\sqrt{g}g^{ij}\frac{\pat}{\pat x_j})$ is the scalar Laplacian
operator.

Since our operators are all real, we will use real coordinate in the
following proof and denote the points $z,w$ by $x,y$. Notice that
the real dimension of $M$ is $2n$.

Now the twisted Laplacian of $k$-form has the expression:
$$
\Delta_f=\Delta_{\bpat}+L_f+|\nabla f|^2=\Delta_0 +(R+L_f+|\nabla
f|^2).
$$
Define
$$
G_0(x,y)=\sum_{I,J,|I|+|J|=k}E_0(x,y)dz^I\wedge dz^\Jb,
$$
Then the action of the scalar Laplacian is
$$
\Delta_0 G_0(x,y)=\delta(x-y),
$$
where $\delta(x-y)$ is viewed as a vector.

Define $F=(R+L_f+|\nabla f|^2)$. Then the Green function $G(x,y)$ of
$\Delta_f$ satisfies
$$
(\Delta_0+F)G(x,y)=\delta(x-y).
$$
Assume that $G(x,y)=G_0(x,y)+R_0(x,y)$, then we get the equation of
$R_0(x,y)$:
\begin{equation}
\Delta_f R_0(x,y)=-F\circ G_0(x,y).
\end{equation}
Since the injective radius has positive lower bound $\rho_0$, we can
take normal coordinates system around $y$. Locally we can take the
Taylor expansion of $F$:
$$
F(x,y)=\sum^{2n-2}_{j=1}\frac{F^{(j)}(y)}{j!}(x-y)^j+F_{2n-1}(x,y).
$$
We also write $R_0(x,y)$ as a linear combination of $2n-1$ unknown
vector-valued functions:
$$
R_0(x,y)=\sum_{j=1}^{2n-2}R_j(x,y)+R_{2n-1}(x,y).
$$
Consider the $2n-2$ Dirichelet boundary value problems:
\begin{equation}
\begin{cases}
\Delta_f R_j(x,y)=-\frac{F^{j}(y)}{j!}(x-y)^j G_0(x,y)\\
R_j(x,y)=0,\;\text{on}\;\pat B_{\rho_0}(y).
\end{cases}
\end{equation}
The key point here is that $(x-y)^j G_0(x,y)$ is at least
$L^{1+\epsilon}$ integrable for some $\epsilon>0$. In fact, we have
$$
(x-y)^j G_0(x,y)\in L^q,
$$
where
$$
\begin{cases}
1<q<\frac{2n}{2n-2-j},\;&n>1\\
1<q<\infty,\;&n=1.
\end{cases}
$$
By $L^p$ estimate, we know that $R_j(x,y)\in W^{2,
q},j=1,\cdots,2n-2$, whose singularity is weaker than $G_0(x,y)$.
The rest equation is
\begin{equation}
\Delta_f R_{2n-1}(x,y)=F_{2n-1}G_0(x,y),
\end{equation}
where the right hand side is a continuous function. So $R_(2n-1)$ is
at least $C^1$ differentiable. Hence we obtain the decomposition of
$G(x,y)$:
\begin{equation}
G(x,y)=G_0(x,y)+R_1(x,y)+\cdots+R_{2n-2}(x,y)+R_{2n-1}(x,y),
\end{equation}
where the highest singularity comes from $G_0(x,y)$. So we proved
our conclusion.
\end{proof}

\subsection{Stability}

Let $\Delta_t$ be a strong deformation of $\Delta_0$ for $t\in S$.
In this part, we want to construct the differentiability of the
Projection operator $P_I(\Delta_t)$ and the Green operator $G_t$.

\begin{df} Let $\LL_t$ be a family of linear operators acting
on $L^2(\Lambda^*(M))$. If for any section $\psi_t(z):=\psi(t,z)$
and any space derivatives $\pat^I,|I|=0,1,\cdots$, the section
$\pat^l\psi_t(z)$ is $C^r$ differential with respect to the
variables $(t,z)$, then the section $\psi_t(z)$ is said to be $C^r$
differentiable. If $\psi_t$ is $C^r$ differentiable section with
compact support in $z$ direction for each $t$ and the section $\LL_t
\psi_t$ is $C^r$ differentiable with respect to $t,z$, then $\LL_t$
is called $C^r$ differentiable.
\end{df}

Since $S$ is assumed to be a compact neighborhood around $t=0$, we
can choose $\lambda_0$ belonging to the regular point set of any
$\Delta_t$ such that
$R_{\lambda_0}(\Delta_t)=(\lambda_0-\Delta_t)^{-1}$ exists.
Therefore, the differentiability of $\Delta_t$ is equivalent to the
differentiability of $(\lambda_0-\Delta_t)$. We have the estimate:
\begin{equation}
||R_{\lambda_0}(\Delta_t)\varphi||_{L^2(M)}\le
C||\varphi||_{L^2(M)},\;\forall \varphi\in L^2.
\end{equation}

We need the following computation:

\begin{lm} Let $f_t=f+\sum_i t_i g_i$. Then we have
\begin{equation}\begin{cases}
\pat_i \Delta_t =L(\nabla\pat g_i)\circ+\nabla f \cdot
\overline{\nabla g_i}+\bar{t}_i|\nabla
g_i|^2\\
\pat_\jb\pat_i\Delta_t=\delta_{ij}|\nabla g_i|^2\\
\pat_j\pat_i\Delta_t=0\\
\pat_\ib\Delta_t=\overline{L(\nabla \pat g_i)}+\nabla g_i\cdot
\overline{\nabla f}+t_i|\nabla g_i|^2.
\end{cases}
\end{equation}
\end{lm}

These are all $0$ order multiplication operators. Since $f_t$ is a
strong deformation, we have pointwise estimate:
\begin{equation}\label{eq:stab-1}
|\pat_i \Delta_t|,|\pat_\jb\pat_i\Delta_t|,|\pat_\ib\Delta_t|\le
C(|\nabla f|^2+1),
\end{equation}
where $C$ depends only on $M,g,f,g_i$.

\begin{thm}\label{thm:stabi-resolv} Suppose that $\Delta_t$ is a strong deformation of
$\Delta_0$ in $S$. Then the resolvent $R_{\lambda_0}(\Delta_t)$ is
$C^\infty$ differentiable in $S$.
\end{thm}

\begin{proof} Let $(\lambda_0-\Delta_t)\psi_t=\varphi_t$, where $\varphi_t$
is $C^r$ differentiable in $(t,z)$ and for each $t$ $\varphi$ is a
compactly supported section. We want to prove that $\psi_t$ is $C^r$
differentiable in $(t,z)$. We prove by induction in $r$ and firstly
prove the continuity of $\psi_t$.

Let $\Omega'\subset\subset \Omega\subset M$ be a compact domain.
Then by the apriori estimate of the elliptic operators, we have
\begin{equation}
||\psi||_{W^{m+2,2}(\Omega')}\le
C(||(\lambda_0-\Delta_t)\psi||_{W^{m,2}(\Omega)}+||\psi||_{L^2(\Omega)}),
\end{equation}
where $C$ only depends on $(M,g,f), m$ and the distance between
$\pat\Omega$ and $\pat \Omega'$.

Since $\lambda_0-\Delta_t$ is an isomorphism, we have
\begin{equation}
||\psi||_{L^2(\Omega)}\le ||\psi||_{L^2(M)}\le
C||(\lambda_0-\Delta_t) \psi||_{L^2(M)}.
\end{equation}

Therefore, we have estimate
\begin{equation}\label{ineq:deform-conti-1}
||\psi||_{W^{m+2,2}(\Omega')}\le
C(||(\lambda_0-\Delta_t)\psi||_{W^{m,2}(\Omega)}+||(\lambda_0-\Delta_t)
\psi||_{L^2(M)}).
\end{equation}
By (\ref{ineq:deform-conti-1}), we have
\begin{align}\label{inequ:deform-smooth-2}
&||\psi_t-\psi_s||_{W^{m+2,2}(\Omega')}\le
C\left(||(\lambda_0-\Delta_t)(\psi_t-\psi_s)||_{W^{m,2}(\Omega)}+||(\lambda_0-\Delta_t)
(\psi_t-\psi_s)||_{L^2(M)}\right)\nonumber\\
\le&C\left(
||\varphi_t-\varphi_s||_{W^{m,2}(\Omega)}+||(\Delta_t-\Delta_s)\psi_s||_{W^{m,2}(\Omega)}+
||\varphi_t-\varphi_s||_{L^2(M)}+||(\Delta_t-\Delta_s)\psi_s||_{L^2(M)}.
\right)
\end{align}
Let $t\to s$. The first term vanishes, since for any $I, |I|\le m,$
$$
|\nabla^I\varphi_t-\nabla^I\varphi_s|\to 0,
$$
because of the $C^0$ continuity of $\varphi_t$ with respect to $t$.
The second term vanishes because the coefficients of the operator
$\Delta_t$ depends on $t$ continuously (even smoothly). The third
term vanish because of the dominant convergence theorem. Now the
last term
$$
||(\Delta_t-\Delta_s)\psi_s||_{L^2(M)}\le C|t-s||(|\nabla
f|^2+1)\psi_s||_{L^2(M)}\le C|t-s|||\varphi_s||_{W^{1,2}_f(M)}\to 0,
\;\text{as}\;t\to s.
$$
Here we used the estimate (\ref{eq:stab-1}), the global estimate
from Lemma \ref{thm:apri-high} and the fact that $\varphi$ has
compact support.

Now by sobolev embedding theorem on $\Omega'$, if $m$ satisfies
$m>r-2+\frac{n}{2}$, then for any $I, |I|\le r$ there is
$$
|\nabla^I\psi_t(z)-\nabla^I\psi_s(z)|=o(1).
$$
Therefore we proved the $C^0$ continuity of
$R_{\lambda_0}(\Delta_t)$.

Now consider the $C^1$ continuity. Take a mollifier
$\rho_\epsilon(t)$ such that its support lies in a very small
neighborhood of $t=0$ and it tends to $\delta$ function as
$\epsilon\to 0$. For any section $\psi_t(z)$, we obtain a smooth
function $\psi^\epsilon_t(z)$ with respect to $t$. $\psi^\epsilon$
is called the mollification of $\psi$ in $t$ direction. Then we have
the equation:
$$
(\lambda_0-\Delta_t)\psi^\epsilon_t=\varphi^\epsilon_t.
$$
Take the derivative $\pat_{t_i}$, we have
\begin{equation}
(\lambda_0-\Delta_t)\pat_{i}\psi^\epsilon_t=\pat_{i}\varphi^\epsilon_t+(\pat_i\Delta_t)\psi^\epsilon_t.
\end{equation}
By (\ref{ineq:deform-conti-1}), there is
\begin{align*}
&||\pat_{i}\psi^\epsilon_t||_{W^{m+2,2}(\Omega')}\le
C\left(||\pat_{i}\varphi^\epsilon_t+(\pat_i\Delta_t)\psi^\epsilon_t||_{W^{m,2}(\Omega)}
\right.+||\pat_{i}\varphi^\epsilon_t+(\pat_i\Delta_t)\psi^\epsilon_t||_{L^2(M)}\\
&\le
C\left(||\pat_{i}\varphi_t||_{W^{m,2}(\Omega)}+||\psi_t||_{W^{m,2}(\Omega)}
+||\pat_{i}\varphi_t||_{L^2(M)}+||(\pat_i\Delta_t)\psi^\epsilon_t||_{L^2(M)}\right).\\
\end{align*}
The last term is controlled by the inequality:
$$
||(\pat_i\Delta_t)\psi^\epsilon_t||_{L^2(M)}\le C||(|\nabla
f|^2+1)\psi||_{L^2(M)}\le C||\varphi||_{W^{1,2}_f(M)}.
$$

By Sobolev embedding theorem and the property of mollification, we
known that $\pat_i\psi_t$ exists.

To prove its continuity, we want to get a uniform control of the
term $|\nabla^I(\pat_i\psi_t)-\nabla^I(\pat_i \psi_s)|$ in
$\Omega'$.

This term is still controlled by the Sobolev norm on $\Omega$. by
(\ref{inequ:deform-smooth-2}) we have
\begin{align*}
&||\pat_i\psi_t-\pat_i\psi_s||_{W^{m+2,2}(\Omega')}\le
||\pat_i\varphi_t+(\pat_i \Delta_t)\psi_t-\pat_i\varphi_s-(\pat_i\Delta_s)\psi_s||_{W^{m,2}(\Omega)}\\
&+||(\Delta_t-\Delta_s)[\pat_i\varphi_s+(\pat_i\Delta_s)\psi_s]||_{W^{m,2}(\Omega)}+||\pat_i\varphi_t+(\pat_i
\Delta_t)\psi_t-\pat_i\varphi_s-(\pat_i\Delta_s)\psi_s||_{L^2(M)}\\
&+||(\Delta_t-\Delta_s)\pat_i\psi_s||_{L^2(M)}
\end{align*}

Let $t\to s$, then the first and the second term vanish obviously,
because of the continuity of $\pat_i\varphi_t,\psi_t$ and the
coefficients of $\Delta_t$. Since
$$
||(\pat_i \Delta_t)\psi_t||_{L^2(M)}\le C||(|\nabla
f|^2+1)\psi_t||_{L^2(M)}\le C ||\varphi_s||_{W^{1,2}_f(M)},
$$
we can use the dominant convergence theorem to deduce that the third
term tends to zero as $t\to s$.

For the last term we have
$$
||(\Delta_t-\Delta_s)\pat_i\psi_s||_{L^2(M)}\le C|t-s|||(|\nabla
f|^2+1)\pat_i\psi_s||\le C|t-s|
||\pat_i\varphi_s||_{W^{1,2}_f(M)}\to 0,
$$
as $t\to s$. So up to now, we have proved the $C^1$ smoothness.

After proving the $C^1$ smoothness, we can consider the equation
satisfied by the higher order derivatives. By the same way but
tedious computation, we can prove any $C^r$ smoothness of
$R_{\lambda_0}(\Delta_t)$.
\end{proof}

\begin{thm}\label{thm:stabi-proj} Let $I=[a,b]\subset \R$ and $a,b$ is not the spectrum
point of $\Delta_{t_0}$. Then there exists a $\delta>0$ such that
for any $t,|t-t_0|<\delta$, the projection $E_{t}(I)$ with respect
to $I$ is well-defined and $C^\infty$ differentiable.
\end{thm}

\begin{proof} Let $\Sigma$ be a simple closed curve in $\C$ containing $I$ and
intersects the real axis only at $a,b$ points. Then
$$
E_{t_0}(I)=\oint_\Sigma R_{\lambda}(\Delta_{t_0})d\lambda.
$$
Since the spectrum of $\Delta_t$ is continuous, there is a
$\delta>0$ such that for any $t,|t-t_0|<\delta,$ the interval $I$
has no intersection points with the spectrum of $\Delta_t$. Hence
for any $\lambda\in \Sigma$, the resolvent $R_\lambda(\Delta_t)$ is
well-defined and is $C^\infty$ differentiable with respect to $t$.
Therefore the projection $E_t(I)$ is $C^\infty$ differentiable.
\end{proof}

\begin{thm}\label{thm:stabi-Green} Let $G_t$ be the Green function of $\Delta_t$ with
parameter $t\in S$. If $(M,g,f_t)$ is a strong deformation of
$(M,g,f)$ on $S$, then $G_t$ is a $C^\infty$ differentiable
operator.
\end{thm}

\begin{proof} Since $S$ is compact, there exists a point
$\lambda_0\in \R$ such that $\lambda_0$ is less than the first
nonzero eigenvalues of all of $\Delta_t, t\in S$. Let $\Sigma$ be a
simple half closed curve in $\C$ containing $[\lambda_0,+\infty)$
and intersects the real axis only at $\lambda_0$. Then
$$
G_t=\int_\Sigma \frac{1}{\lambda}R_\lambda(\Delta_t)d\lambda.
$$
Note that the integration is absolutely uniform convergent, so the
differentiability of $R_\lambda(\Delta_t)$ for $t\in S$ implies the
$C^\infty$ differentiability of $G_t$.
\end{proof}

\section{$tt^*$ geometry}

\subsection{Hilbert bundle and Hodge bundle}\

\

Let $f_{\tau,t}(z):=f_\tau(z):=\tau f(t,z)$ be a family of
holomorphic functions defined on $\C^*\times S\times M$, where
$S\subset \C^m$ is a domain and having coordinates
$t=(t_1,\cdots,t_m)$ and $\C^*\equiv \C-\{0\}$. Then we have a
family of operators
$\Delta_{f_\tau},\bpat_{f_\tau},\pat_{f_\tau},\bpat_{f_\tau}^\dag,\pat_{f_\tau
}^\dag$ which depend on the parameter $(\tau,t)\in \C^*\times S$. So
the space of harmonic forms $\ch^*_{\tau,t}$ also depends on the
parameter $(\tau,t)$.

\emph{We still assume that $(M,g)$ is a K\"ahler manifold with
bounded geometry.}



We have the trivial complex Hilbert bundle
$\Lambda_\C^*:=L^2\Lambda^* \times \C^*\times S\to \C\times S$ which
is a graded bundle. The Hermitian metric on the Hilbert bundle
$\Lambda^n_\C$ is the usual one on the fiber:
$$
(\alpha,\beta)_{L^2}=\int_{M}\alpha\wedge *\bar{\beta},
$$
which is independent of $(\tau,t)$.

Fix $(\tau_0,t_0)\in S$ and let $(\tau,t)$ be a point in a small
neighborhood $U_{0}$ of $(\tau_0,t_0)$. For any $(\tau,t)\in U_{0}$,
we can choose uniformly bounded eigenforms
$\{\alpha_a(\tau,t)\}_{a=1}^\infty$ of $\Delta_{f_\tau}$ such that
they form a complete basis of $L^2\Lambda^k(M)$ and the
corresponding eigenvalues have the order:
$$
0=\lambda_0\le \lambda_1\le \cdots.
$$
Set
$$
g_{a\bb}(\tau,t)=g(\alpha_a,\alpha_b),
$$
which satisfies $g_{a\bb}=\overline{g_{b\ab}}$. We have the metric
tensor
$$
g=g_{a\bb}\alpha^a\otimes \alpha^\bb=g^{\bb a}\alpha_a\otimes
\alpha_\bb,
$$
where $\{\alpha^a\}$ is the dual basis in the dual Hilbert bundle
and $\{g^{\bb a}\}$ is the inverse operator defined by the dual
metric on the dual Hilbert bundle.

If $\alpha=\sum_a f_a\alpha_a,\beta=\sum_b h_b\beta_b$, then
$$
g(\alpha,\beta)=\sum_{ab}f_ag_{a\bb}\overline{h_b}=(g_{a\bb}
\overline{h_b}\alpha^a)(\alpha).
$$
Hence $g$ can be viewed as a bounded operator mapping the section of
the Hilbert bundle to its dual bundle.

There are some unbounded $\C$-linear operator acting on the Hilbert
bundle. The first is the differentiation. Since $\pat_i
\alpha_a:=\frac{\pat}{\pat t^i} \alpha_a(t)$ is an exponential decay
function by Theorem \ref{crl:eign-decay}, we can define the tensors
$\Gamma_i$ and $\Gamma_{\ib}$ as follows (infinite dimensional case)
\begin{align*}
&(\Gamma_i)_{a\bar{b}}=g(\pat_i\alpha_a,
\alpha_b)=g_{c\bar{b}}(\Gamma_i)_{a}^c\\
&(\Gamma_{\ib})_{a\bb}=g(\pat_{\ib}\alpha_a,
\alpha_b)=g_{c\bb}(\Gamma_\ib)_a^c.
\end{align*}
and the covariant derivative is defined as
$$
D_i:=\pat_i-\Gamma_i,\;D_{\bar{i}}:=\bpat_i-\Gamma_{\bar{i}}
$$
Now let $\{\alpha'_a\}$ be another basis on the same fiber, and
$$
S:=S_{a'b}
$$
is the transformation bounded operator, it is easy to see that
$$
\Gamma=S^{-1}\cdot \Gamma'\cdot S+S^{-1}\pat_i S.
$$
Hence $D_i,D_{\bar{i}}$ is really a covariant derivative. Define the
connection
\begin{equation}
D=\sum_i (dt^i D_i+d\bar{t}^i D_{\bar{i}}).
\end{equation}
$D$ is a Hermitian connection with respect to the bundle metric $g$,
since it is easy to prove that
$$
D_ig_{a\bar{b}}=D_{\bar{i}}g_{a\bar{b}}=0.
$$

There is another operator $\tau\pat_\tau$ acting on the Hilbert
bundle. We define
\begin{equation}
(\Gamma_{\tau})_{a\bb}=g(\tau\pat_\tau\alpha_a,\alpha_b).
\end{equation}
Similarly, we have $\bar{\tau}\pat_{\bar{\tau}}$.

There is a special element given by the isomorphism of Theorem
\ref{crl:stein} which we assume to be
\begin{equation}
\alpha_1=1+\bpat_f R_1,
\end{equation}
which $R_1$ is a smooth $n-1$ form has at most polynomial growth.
Partial Christoffel symbols involving $\alpha_1$ will vanish which
is given by the following proposition:

\begin{prop}\label{prop:little-monster} We have
\begin{equation}
(\Gamma_i)_{1\bb}=(\Gamma_i)_1^b=(\Gamma_\tau)_{1\bb}=(\Gamma_\tau)_1^b=0,\forall
i,b.
\end{equation}
\end{prop}

\begin{proof}We have
\begin{equation*}
(\Gamma_i)_{1\bb}=g(\pat_i \alpha_1,\alpha_b)=g(\bpat_f(\pat_i
R_1),\alpha_b)=0,\forall b,
\end{equation*}
and
$$
(\Gamma_i)_1^b=(\Gamma_i)_{1\ab} g^{\ab b}=0.
$$
The others can be obtained similarly.
\end{proof}

\begin{rem} If there is a projection operator defined on the Hilbert
bundle $\Lambda_\C^*$ with is compatible with all $\Delta_{\tau,t}$,
then the Hilbert bundle can be split into two parts and the
connection $D$ can be restricted to the connection on each
subbundles. In particular, $D$ can be restricted to the holomorphic
subbundle, the real bundle and the Hodge bundle which we will
defined subsequently.
\end{rem}

\subsubsection{\underline{Real structure and Maslov index}}

The canonical real structure in $L^2\Lambda^*$ is given by the
complex conjugate, which we denote by $\tau_R$. Let
$\{\alpha_a\}_{a=1}^\infty$ be a basis of $L^2\Lambda^*$ and then
$\alpha_{\ab}:=\overline{\alpha_a}=\tau_R \cdot \alpha_a$ is also a
basis of $L^2\Lambda^*$. $\tau_R\cdot \alpha_a$ can be represented
by the linear combination of $\{\alpha_a\}$:
\begin{equation}
\tau_R\cdot \alpha_a=M^b_{\ab}\alpha_b=M_\ab^\bb \alpha_\bb,
\end{equation}
Similarly,
$$
\tau_R\cdot \alpha_\ab=M^\bb_{a}\alpha_\bb=M_a^b \alpha_b.
$$
Obviously we have
$$
M^b_a=M^\bb_\ab=\delta^b_a,\;\overline{M^\bb_a}=M^b_\ab.
$$

The anti-linear map $\tau_R$ can be written as
$\tau_R=M^a_{\bb}\alpha_a\otimes \alpha^b$. It is easy to prove that
\begin{equation}
\eta_{cb}=g(\alpha_c,\tau_R\cdot \alpha_b)=g_{c\ab}M^\ab_b.
\end{equation}
If $\alpha_c$ and $\alpha_b$ are two forms, then we have
\begin{equation}
\eta_{cb}=\eta_{bc}.
\end{equation}
If $\{\alpha_a(t)\}$ is a real basis, then $M\equiv I$. This shows
that
$$
DM\equiv 0,
$$
i.e., the matrix $M$ is a parallel tensor.

Assume that there is a spectrum gap for Laplacians $\Delta_{f_\tau}$
on $S$, i.e, there is an open interval $(\mu,\nu)$ such that its
intersection with any spectrum of $\Delta_{f_\tau}$ is empty. In
this case, the bundle $\Lambda_\C$ is split into the direct sum of
two bundles. We denote the bundle with finite rank $\mu_F$ by
$\Lambda^F_\C$.

Fix the parameter $\tau$ and take a real orthonormal basis
$\{\alpha_a(t_0)\}_{a=1}^{\mu_F}$ of the fiber space
$\Lambda^F_\C(t_0)$ at the point $t_0\in S$. Denote by $\Lambda_\R$
the real space generated by this basis. Using the parallel
transportation defined by $D$ and a path $l_{tt_0}$ connecting $t$
and $t_0$, one can obtain the unitary operator $P_{l_{tt_0}}:
\Lambda_\C(t_0)\to \Lambda_\C(t)$. $P_{l_{tt_0}}(\Lambda_\R)$
defines a totally real Hilbert space in $\Lambda^F_\C(t)\cong
\C^{\mu^F}$.

Let $\R(\mu^F)$ be the set of totally real subspaces. Then it is a
homogeneous space
$$
\R(\mu^F)=GL(\mu^F,\C)/GL(\mu^F,\R).
$$
Define
$$
\widetilde{\R}(\mu^F)=\{A\in GL(\mu^F,\C)| A\bar{A}=I\}.
$$
Then we have the following lemma.
\begin{lm}[\cite{FOOO}] The map
$$
L: \R(\mu^F)\to \tilde{\R}(\mu^F);\;A\cdot \R\to A^{-1}\bar{A}
$$
is a diffeomorphism with respect to the obvious smooth structure on
$\R(\mu^F)$ and $\widetilde{\R}(\mu^F)$.
\end{lm}

\begin{df}[Generalized Maslov index] Let $l:S^1\to \R(\mu^F)$ be a
loop such that at point $t$, there is
$$
l(t)=P_{l_{tt_0}}(\Lambda_\R).
$$
Then the generalized Maslov index $m(l)$ is defined to be the
winding number of
$$
\det\circ L\circ l: S^1\to \C\setminus \{0\}.
$$
\end{df}

If we write the real structure matrix $M$ in terms of the parallel
orthonormal frame $\{\alpha_a(t)\}_{a=1}^{\mu^F}$, then
\begin{equation}
m(l)=\deg\circ \det M.
\end{equation}

The generalized Maslov index defined here satisfies the properties
that an usual Maslov index should satisfy.

Let $\ch\subset \Lambda_\C$ be the Hodge bundle over $\C^*\times S$,
and each fiber at $(\tau,t)\in \C^*\times S$ is the space of all
harmonic $n$ forms of the Laplacian $\Delta_{f_\tau}$.

Suppose that the primitive forms
$\{\alpha_1(\tau,t),\cdots,\alpha_\mu(\tau,t) \}$ give a local frame
of $\ch$ near a point $(\tau_0,t_0)\in \C^*\times S$ which is the
restriction of a frame in $\Lambda_\C$. Then we have the identities:
\begin{equation}
\bpat_{f_\tau}\alpha_a(\tau,t)=\bpat_{f_\tau}^\dag\alpha_a(\tau,t)=0,\;a=1,\cdots,\mu.
\end{equation}
Since $\Lambda\alpha_a(\tau,t)=0$, by Hodge identity
$\frac{1}{i}[\pat_f,\Lambda]=\bpat_f^\dag$ we know that
$\Lambda\pat_f\alpha_a(\tau,t)=0$. Since there is no primitive $n+1$
form, this shows that
\begin{equation}
\pat_{f_\tau}\alpha_a(\tau,t)=0.
\end{equation}
Similarly, we can prove that $\pat_{f_\tau}^\dag\alpha_a(\tau,t)=0.$

In summary, we have the important properties of the frame of
primitive harmonic $n$-forms in $\ch$.

\begin{prop} Suppose that $\{\alpha_1(\tau,t),\cdots,\alpha_\mu(\tau,t) \}$ is a local frame
of $\ch$ consisting of the primitive harmonic $n$-forms, then we
have the identities
\begin{equation}
\bpat_{f_\tau}\alpha_a(\tau,t)=\bpat_{f_\tau}^\dag\alpha_a(\tau,t)=\pat_{f_\tau}\alpha_a(\tau,t)
=\pat_{f_\tau}^\dag\alpha_a(\tau,t)=0,\;a=1,\cdots,\mu.
\end{equation}

\end{prop}

\subsection{Cecotti-Vafa's equation and "Fantastic" equation}\

\subsubsection{\underline{Cecotti-Vafa's equations}}\

Now we fix the parameter $\tau$ in the following discussion and
denote by $\pat_i f=\pat_i f_\tau$ the derivative of $f_{\tau}$ with
respect to $t_i$. Later we always use $i,j,k,\cdots$ to represent
the terms with respect to the parameter $t^i,t^j,t^k$ and etc. In
the following part, we will simply write $f_{\tau,t}=\tau f_t$ as
$f$ if there is no danger of confusion. It is easy to obtain the
following lemma:

\begin{lm}\label{lm:Derivative}
\begin{equation}
\begin{cases}
[\pat_i,\bpat_f]=\pat(\pat_i f)\wedge,\;&[\pat_i,\bpat_f^\dag]=0\nonumber\\
[\pat_\ib,\bpat_f]=0,\;&[\pat_\ib,\bpat_f^\dag]=\overline{\pat_if_a}(dz^a\wedge)^\dag
=\overline{\pat_if_a}g^{\ab b}\iota_{\pat_b}.
\end{cases}
\end{equation}
\end{lm}

\begin{lm} By our stability analysis, the Projection operator $P_0$ and the Green operator
$G=G_f$ are smooth with respect to $\pat_i,\pat_\ib$. We have the
commutation relations:
\begin{equation}
\begin{cases}
[\pat_i, \Delta_f]=[\pat(\pat_i f)\wedge,
\bpat_f^\dag],\;&[\pat_\ib, \Delta_f]=[\overline{\pat(\pat_i
f)\wedge}, \pat_f^\dag]\\
[\pat_i,G]=-G(\pat_i P_0+[\pat_i,\Delta_f]\cdot
G),\;&[\pat_\ib,G]=-G(\pat_\ib P_0+[\pat_\ib,\Delta_f]\cdot G).
\end{cases}
\end{equation}
When restricted to the space of primitive harmonic forms, we have
\begin{equation}
[\pat_i,P_0]=G\cdot\Delta_f\cdot \pat_i,\;[\pat_\ib,
P_0]=G\cdot\Delta_f\cdot \pat_\ib
\end{equation}
and
\begin{equation}
\pat_i G=-G(\pat_i+[\pat_i,\Delta_f]\cdot G).
\end{equation}
\end{lm}

\begin{proof}
It suffices to prove the identities related to the derivative
$\pat_i$.

We have
\begin{align*}
&[\pat_i, \Delta_f]=[\pat_i, [\bpat_f,\bpat_f^\dag]]\\
=&[[\pat_i,\bpat_f],\bpat_f^\dag]+[[\pat_i,
\bpat_f^\dag],\bpat_f]=[\pat(\pat_i f)\wedge,\bpat_f^\dag].
\end{align*}
This proves the first identity.

Differentiating the formula:
$$
P_0+\Delta_f G=I,
$$
we have
$$
\pat_i P_0+[\pat_i,\Delta_f]\cdot G+\Delta_f\cdot\pat_i G=0,
$$
which gives
$$
\pat_i G=[\pat_i,G]=-G\cdot(\pat_iP_0+[\pat_i,\Delta_f]G).
$$
Let $\alpha_a$ be a primitive harmonic form. Take the derivative to
the identity $P_0\alpha_a= \alpha_a$, we have
$$
[\pat_i, P_0]\alpha_a=(I-P_0)\pat_i\alpha_a=\Delta_f
G(\pat_i\alpha_a).
$$
So we are done.
\end{proof}

\begin{df} Define tensors $B_i, \Bb_i:=B_\ib$ with the components:
\begin{equation}
(B_i)_{a\bb}=g((\pat_i
f)\alpha_a,\alpha_b),\;(B_\ib)_{a\bb}=g((\overline{\pat_i
f})\alpha_a,\alpha_b),
\end{equation}
and the matrix-valued 1-forms:
\begin{equation}
B=\sum_{i=1}^\mu B_i dt^i,\;\Bb=\sum_{i=1}^\mu B_\ib dt^\ib.
\end{equation}
\end{df}

\begin{rem}
It is easy to see that the matrix $\Bb_i$ is the complex conjugate
transpose of $B_i$, i.e., in local coordinates,
$(\Bb_i)_{b\bar{a}}=\overline{(B_i)_{a\bar{b}}}$. So in some
references, people write $\Bb=B^\dag$, i.e.,
$(\cdot)^\dag=\overline{(\cdot)}^T$, the complex transpose of
matrix. Here the reader should be careful that if $\dag$ used for
matrix it has different meaning to $\dag$ used for differential
operators. Also the short line over the head of $\Bb$ does not mean
the complex conjugate of the matrix $B$, instead we take $\Bb$ as an
independent matrix from $B$. If we want to consider the relations
with the matrix $B$, one should use $B^\dag=\Bb$.
\end{rem}

We have the formula:

\begin{lm} For any $i=1,\cdot,m,\;a=1,\cdots,\mu,$ we have
\begin{equation}
\begin{cases}
(\pat_i f)\alpha_a=(B_i)^b_a\alpha_b+\bpat_f (\gamma_i)_a\\
\overline{(\pat_i f)}\alpha_a={(\Bb_i)^b_a}\alpha_b+\pat_f (\overline{\gamma_i})_a,\\
\end{cases}
\end{equation}
where
\begin{equation}
\begin{cases}
(\gamma_i)_a=\bpat_f^\dag\cdot G\cdot[(\pat_i f)\alpha_a]\\
(\overline{\gamma_i})_a=\pat_f^\dag\cdot G\cdot[(\overline{\pat_i
f)}\alpha_a]
\end{cases}
\end{equation}
Here $G$ is the Green function.
\end{lm}

\begin{proof} Using Hodge decomposition formula, we have
\begin{align*}
(\pat_i f)\alpha_a&=P_0[(\pat_i f)\alpha_a]+\Delta_f \cdot G\cdot
(\pat_i f)\cdot \alpha_a\\
&=(B_i)^b_a\alpha_b+\bpat_f\bpat_f^\dag G (\pat_i f)\alpha_a.
\end{align*}
Set
$$
(\gamma_i)_a=\bpat_f^\dag G (\pat_i f)\alpha_a,
$$
then we get the first identity. The second one can be proved in the
same way.
\end{proof}

\begin{lm}\label{lm:tt*-2}
For any $i=1,\cdot,m,\;a=1,\cdots,\mu,$ we have
\begin{equation}
\begin{cases}
D_i\alpha_a=\pat_f (\gamma_i)_a\\
D_\ib\alpha_a=\bpat_f (\overline{\gamma_i})_a,\\
\end{cases}
\end{equation}
\end{lm}

\begin{proof} $\alpha_a(t)$ satisfies the equation:
$$
\Delta_{f_t}\alpha_a(t)=0.
$$
Taking derivative $\pat_i$ to the above equality, we have
\begin{align*}
\Delta_{f_t}(\pat_i \alpha_a)+[\pat(\pat_i
f)\wedge\cdot\bpat_f^\dag+\bpat_f^\dag\cdot \pat(\pat_i
f)]\alpha_a=0.
\end{align*}
Hence
$$
\Delta_f(\pat_i\alpha_a)=-\bpat_f^\dag (\pat_f (\pat_i f
\alpha_a))=\pat_f\bpat_f^\dag (\pat_i f)\alpha_a,
$$
where we used the fact that $\pat_f\alpha_a=0$. Therefore,
$$
D_i\alpha_a=\pat_i\alpha_a-\Gamma_i\alpha_a=\pat_f G\bpat_f^\dag
(\pat_i f)\alpha_a.
$$
The other identity can be proved similarly.
\end{proof}

Some special data can be easily obtained in the following
proposition.

\begin{prop}\label{prop:B-value} We have
\begin{align}
&(B_1)_a^b=\tau \delta^b_a,\;(\bar{B}_1)_a^b=\taub \delta^b_a\\
&(\gamma_1)_a=(\bar{\gamma}_1)_a=0, \forall a=1,\cdots,\mu\\
&D_1\alpha_a=D_\eb\alpha_a=0,\;\forall a=1,\cdots,\mu
\end{align}
\end{prop}

\begin{rem}\label{rem:opera-matrix} We have defined the operator $B_i$ which acts on
$\alpha_a$ as $B_i\cdot \alpha_a=(B_i)_a^b \alpha_b$. On the other
hand, $(B_i)_{a\bb}=g(B_i\cdot \alpha_a, \alpha_b)=(B_i)_a^c
g_{c\bb}$. Hence, we have used the convention that the action of the
operator $B_i$ is a matrix multiplication by the matrix $(B_i)_a^b$,
where $a$ is the row index and $b$ is the column index. In this
convention, if $B_i,B_\jb$ are operators, then
\begin{align*}
&[B_i,B_\jb]\cdot \alpha_a=B_i((B_\jb)_a^b\alpha_b)-B_{\jb}((B_i)_a^b\alpha_b)=(B_\jb)_a^b(B_i)_b^c\alpha_c-(B_i)_a^b(B_\jb)_b^c\alpha_c\\
=&(B_\jb\cdot B_i-B_i\cdot
B_\jb)_a^c\alpha_c=[B_\jb,B_i]\alpha_a=-[B_i,B_\jb]\alpha_a.
\end{align*}
The front $[B_i,B_\jb]\cdot$ is understood as operator action and
the rear Lie algebra $-[B_i,B_\jb]$ is understood as the matrix
multiplication. We will use the later expression in our formulas.
\end{rem}

\begin{thm}\label{thm:CV-equa}  The connection $D$ and the operators $B_i$ satisfy the
following Cecotti-Vafa's equations on $\ch$, i.e, after $\mod
\ch^\perp$,
\begin{align*}
&D_{\bar{i}}B_j=D_iB_\jb=0,\;D_iB_j=D_jB_i,\;D_{\bar{i}}B_\jb=D_{\bar{j}}B_\ib\\
&[D_i,D_j]=[D_{\bar{i}},D_{\bar{j}}]=[B_i,B_j]=[B_\ib,B_\jb]=0\\
&[D_i,D_{\bar{j}}]=-[B_i,B_\jb].
\end{align*}
\end{thm}

\begin{proof}
Take the covariant derivatives to $B_j$:
\begin{align*}
D_i(B_j)_{a\bb}&=\pat_i (B_j)_{a
\bb}-(\Gamma_i)^c_{a}(B_j)_{c\bb}-(\Gamma_i)^\cb_{\bb}(B_j)_{a\cb}
\end{align*}
Since
\begin{align*}
\pat_i (B_j)_{a\bb}=&g(\pat_i\pat_j f \alpha_a, \alpha_b)+g(\pat_j
f\pat_i \alpha_a, \alpha_b)+g(\pat_jf
\alpha_a,\bpat_i\alpha_b)\\
&=g(\pat_i\pat_j f \alpha_a,
\alpha_b)+(\Gamma_i)^c_a(B_j)^{d}_c g_{d\bb}+(\Gamma_i)^\cb_\bb(B_j)^{d}_a g_{d\cb}\\
&+g(\pat_j f\pat_f(\gamma_i)_a,\alpha_b)+g(\pat_j
f\alpha_a,\bpat_f((\bar{\gamma}_i)_b))\\
=&g(\pat_i\pat_j f \alpha_a,
\alpha_b)+(\Gamma_i)^c_a(B_j)^{d}_c g_{d\bb}+(\Gamma_i)^\cb_\bb(B_j)^{d}_a g_{d\cb}\\
&+g(\pat_j f\pat_f(\gamma_i)_a,\alpha_b)+g(\pat_j f\cdot\pat_f \cdot
\bpat_f^\dag G(\pat_i f)\alpha_a,\alpha_b)+g(\pat_j
f\alpha_a,\bpat_f\cdot \pat^\dag_f\cdot G\overline{(\pat_i f)}\alpha_b)\\
=&g(\pat_i\pat_j f \alpha_a,
\alpha_b)+(\Gamma_i)^c_a(B_j)^{d}_c g_{d\bb}+(\Gamma_i)^\cb_\bb(B_j)^{d}_a g_{d\cb}\\
&-g(G(\pat(\pat_i f)\wedge \alpha_a),\overline{\pat(\pat_j f)}\wedge
\alpha_b)-g(G\pat(\pat_j f)\wedge \alpha_a,
\overline{\pat(\pat_i f)}\wedge {\alpha}_b)\\
\end{align*}
We know that
\begin{align*}
&D_i(B_j)_{a\bb}=g(\pat_i\pat_j f \alpha_a, \alpha_b)\\
-&g(G(\pat(\pat_i f)\wedge \alpha_a),\overline{\pat(\pat_j f)}\wedge
\alpha_b)-g(G\pat(\pat_j f)\wedge \alpha_a,
\overline{\pat(\pat_i f)}\wedge {\alpha}_b)\\
.
\end{align*}
By symmetry we have
$$
D_i B_j=D_j B_i.
$$
Take the complex conjugate transpose, we have $\bar{D}_i
\bar{B}_j=\bar{D}_j \bar{B}_i$ or $D_\ib B_\jb=D_\jb B_\ib$.

We have
\begin{align*}
&g((\pat_i f)(\pat_j f)\alpha_a,\alpha_b)=g((\pat_i f)( (B_j)_a^c
\alpha_c+\bpat_f (\gamma_j)_a),\alpha_b)\\
&=(B_j)_a^c (B_i)_{c\bb}+g(\bpat_f\cdot \pat_i f (\gamma_j)_a,
\alpha_b)\\
&=(B_j)_a^c (B_i)_{c\bb},
\end{align*}
which induces that
$$
B_iB_j=B_jB_i.
$$

We have the identity:
$$
D_\ib(B_j)_{a\bb}=\pat_\ib(B_j)_{a\bb}-(\Gamma_\ib)^c_a(B_j)_{c\bb}-(\Gamma_\ib)^\db_\bb(B_j)_{a\db}.
$$
Compute the first term:
\begin{align*}
\pat_\ib(B_j)_{a\bb}=&\pat_\ib g((\pat_j f)\alpha_a,\alpha_b)\\
=&g((\pat_j f)\pat_\ib\alpha_a,\alpha_b)+g((\pat_j f)\alpha_a,\pat_i
\alpha_b)\\
=&(\Gamma_\ib)^c_a(B_j)_{c\bb}+\overline{(\Gamma_i)^c_b}(B_j)_{a\cb}.
\end{align*}
Hence we proved that
$$
D_\ib(B_j)=0.
$$
Take the complex conjugate transpose, we obtain
$$
D_i B_\jb=0.
$$
By Lemma \ref{lm:tt*-2}, we have the following computation:
\begin{align*}
D_i(\bar{D}_j
\alpha_a)=&\pat_i(\bpat_f(\bar{\gamma}_j)_a)-(\Gamma_i)^b_a(\bpat_f
(\bar{\gamma}_j)_b)\\
=&\bpat_f(D_i(\bar{\gamma}_j)_a)+\pat(\pat_i f)\wedge
(\bar{\gamma}_j)_a\\
=&\bpat_f(D_i(\bar{\gamma}_j)_a)+\pat((\pat_i
f)(\bar{\gamma}_j)_a)-(\pat_i f)\pat(\bar{\gamma}_j)_a\\
=&\bpat_f(D_i(\bar{\gamma}_j)_a)+\pat((\pat_i
f)(\bar{\gamma}_j)_a)-(\pat_i f)[\overline{(\pat_j
f)}\alpha_a-(\bar{B}_j)^b_a\alpha_b-\overline{\pat f}\wedge
(\bar{\gamma}_j)_a]\\
=&\bpat_f(D_i(\bar{\gamma}_j)_a)+\pat_f((\pat_i
f)(\bar{\gamma}_j)_a)+(\pat_i f)(\bar{B}_j)^b_a\alpha_b-(\pat_i
f)\overline{\pat_j f}\alpha_a.
\end{align*}
Similarly, we have
$$
\bar{D}_j(D_i\alpha_a)=\pat_f(\bar{D}_j(\gamma_i)_a)+\bpat_f(\overline{(\pat_j
f)}(\gamma_i)_a)-\overline{(\pat_j f)}(\pat_i
f)\alpha_a+\overline{(\pat_j f)}(B_i)^b_a\alpha_b.
$$
Hence we have
\begin{align*}
[D_i,\bar{D}_j]\alpha_a=&\bpat_f(D_i(\bar{\gamma}_j)_a)-\pat_f(\bar{D}_j(\gamma_i)_a)
+\pat_f((\pat_i f)(\bar{\gamma}_j)_a)-\bpat_f(\overline{(\pat_j
f)}(\gamma_i)_a)+(\pat_i f)(\bar{B}_j)^b_a\alpha_b-\overline{(\pat_j
f)}(B_i)^b_a\alpha_b\\
=&\bpat_f[D_i(\bar{\gamma}_j)_a-\overline{(\pat_j
f)}(\gamma_i)_a]-\pat_f[\bar{D}_j(\gamma_i)_a-(\pat_i
f)(\bar{\gamma}_j)_a]+[(\pat_i
f)(\bar{B}_j)^b_a-\overline{(\pat_j f)}(B_i)^b_a]\alpha_b\\
=&\bpat_f[D_i(\bar{\gamma}_j)_a-\overline{(\pat_j
f)}(\gamma_i)_a+(\bar{B}_j)_a^b
(\gamma_i)_b]-\pat_f[\bar{D}_j(\gamma_i)_a-(\pat_i
f)(\bar{\gamma}_j)_a+(B_i)^b_a (\bar{\gamma_j})_b]\\
+&[(B_i
)_b^c(\bar{B}_j)^b_a-(B_\jb)^c_b (B_i)^b_a]\alpha_c\\
\end{align*}
Acting by the projection $P_0$, we obtain
$$
[D_i,\bar{D}_j]\alpha_a=-[B_i,\bar{B}_j]\alpha_a.
$$
The other identities $[D_i, D_j]=0$ and etc. can be easily seen to
be true when in view of the above proof.
\end{proof}

\begin{df} Let $t=(t_1,\cdots,t_m)$ be the local holomorphic coordinates of the
deformation parameter space $S$. Define the sections $D,\bar{D}, B,
\Bb \in \Omega^1(S)\otimes \Gamma(\End(\ch))$ as:
\begin{align*}
&D=\sum_{i=1}^m dt^i D_i,\;\bar{D}=\sum_{i=1}^m dt^\ib D_\ib\\
&B=\sum_{i=1}^m dt^i B_i,\;\Bb=\sum_{i=1}^m dt^\ib \Bb_i.
\end{align*}
\end{df}

Using the above definitions, the $tt^*$ equation can be written in
the following forms.

\begin{crl}
\begin{align*}
&(1)\;[D,B]=DB=0,\;[\bar{D},B]=\bar{D}B=[D,\Bb]=D\Bb=0\\
&(2)\;D^2=\frac{1}{2}[D,D]=0,\;(\bar{D})^2=\frac{1}{2}[\bar{D},\bar{D}]=0\\
&(3)\;B\wedge B=0,\;\Bb\wedge \Bb=0\\
&(4)\;[D,\bar{D}]=-[B,\Bb].
\end{align*}
\end{crl}

\subsubsection{\underline{"Fantastic" equations}}\

Now we allow the change of the parameter $\tau$ and consider the
action of the operator $\tau\pat_\tau$. We will continue to use the
simple notation $f=f_\tau$.

\begin{lm} We have the commutation relations:
\begin{align}
&[\tau\pat_\tau, \pat_f]=0,[\bar{\tau}\pat_{\bar{\tau}},\bpat_f]=0\\
&[\tau\pat_\tau, \bpat_f]=\pat
f\wedge,[\bar{\tau}\pat_{\bar{\tau}},\pat_f]={\bpat
\fb}\wedge\\
&[\tau\pat_\tau,
\bpat^\dag_f]=0,[\bar{\tau}\pat_{\bar{\tau}},\pat^\dag_f]=0\\
&[\bar{\tau}\pat_{\bar{\tau}},\bpat^\dag_f]=(\pat f\wedge)^\dag,
[\tau\pat_\tau,\pat^\dag_f]=(\bpat \fb\wedge)^\dag.
\end{align}
\end{lm}
\begin{proof} A straightforward computation.\end{proof}

\begin{lm}\label{lm:tt*2-1}The commutation relations hold:
\begin{align}
&[\tau\pat_\tau,\Delta_f]=[\pat f\wedge, \bpat_f^\dag]=\nabla^{\nub}(\pat f)\wedge \iota_{\pat_{\nub}}+|\pat f|^2\\
&[\bar{\tau}\pat_{\bar{\tau}},\Delta_f]=[\bpat \fb\wedge,
\pat_f^\dag]=\nabla^{\nu}(\bpat \fb)\wedge \iota_{\pat_{\nu}}+|\pat
f|^2.
\end{align}
\end{lm}
\begin{proof} The result is due to the equality we proved before:
\begin{equation}
[\pat f\wedge, \bpat^\dag]=g^{\nub\mu}\nabla_\mu(\pat_l f)dz^l\wedge
\iota_{\pat \nub}.
\end{equation}
\end{proof}

By Lemma \ref{lm:tt*2-1}, we have
\begin{equation}
\Delta_f(\tau\pat_\tau \alpha_a)=-[\pat f\wedge,
\bpat_f^\dag]\alpha_a=-\bpat^\dag_f(\pat f\wedge
\alpha_a)=\pat_f\bpat^\dag_f (f\alpha_a).
\end{equation}
By Hodge decomposition formula, there exists a matrix
$\Gamma_\tau=\Gamma_{\tau a}^b$ such that
\begin{equation}
(\tau\pat_\tau\alpha_a)=\Gamma_{\tau
a}^b\alpha_b+G\cdot\pat_f\bpat^\dag_f (f\alpha_a)=\Gamma_{\tau
a}^b\alpha_b+\pat_f \bpat^\dag_f G\cdot (f\alpha_a).
\end{equation}
Similarly, there exists a matrix $\bar{\Gamma}_\tau\equiv
\Gamma_{\bar{\tau}}$ such that
\begin{equation}
(\bar{\tau}\pat_{\bar{\tau}}\alpha_a)=\Gamma_{\bar{\tau}
a}^b\alpha_b+\bpat_f \pat^\dag_fG\cdot (\fb \alpha_a).
\end{equation}

\begin{df} Define two operators:
$$
D_\tau=\pat_\tau-\frac{1}{\tau}\Gamma_\tau,\;D_{\bar{\tau}}=\pat_{\bar{\tau}}-\frac{1}{\bar{\tau}}\Gamma_{\bar{\tau}}.
$$
\end{df}

So we have the formula by the above discussions:
\begin{lm}
\begin{align}
&\tau D_\tau\alpha_a=\pat_f (\gamma_\tau)_a\\
&\bar{\tau}D_{\bar{\tau}}\alpha_a=\bpat_f
(\overline{\gamma_\tau})_a,
\end{align}
where
\begin{align}
&(\gamma_\tau)_a=\bpat^\dag_f G\cdot(f\alpha_a) \\
&(\overline{\gamma_\tau})_a=\pat^\dag_f G\cdot(\fb\alpha_a).
\end{align}
\end{lm}

Since we have already defined the covariant derivatives $D_i,D_\ib$
along the deformation direction in addition to $\tau$-direction, now
we get the covariant derivatives $D$ along any direction and the
total differentials are defined as:
\begin{equation}
D=\sum_i dt^iD_i+d\tau D_\tau,\;\bar{D}=\sum_i dt^\ib
D_\ib+d\bar{\tau} D_{\bar{\tau}}.
\end{equation}

\begin{lm} The operator $D$ is a metric connection with
respect to $g$.
\end{lm}

\begin{proof} We have the covariant derivative computation:
\begin{align*}
&D_\tau g_{a\bb}\\
=&\pat_\tau
g_{a\bb}-\frac{\Gamma_{\tau a}^c}{\tau}g_{c\bb}-\frac{\Gamma_{\tau\bb}^\cb}{\bar{\tau}} g_{a\cb}\\
=&g(D_\tau\alpha_a+\frac{\Gamma_\tau}{\tau}\cdot
\alpha_a,\alpha_b)+g(\alpha_a,D_{\bar{\tau}}\alpha_b+\frac{\Gamma_\taub}{\bar{\tau}}\cdot\alpha_b)-
g(\frac{\Gamma_\tau}{\tau}\cdot\alpha_a,\alpha_b)-g(\alpha_a,\frac{\Gamma_{\bar{\tau}}}{\bar{\tau}}\cdot\alpha_b)\\
=&0.
\end{align*}
Similarly, we can prove $D_{\bar{\tau}}g_{a\bb}=0$.

Hence including the result from the first $tt^*$ equations, we know
that $D$ is a metric connection with respect to $g$.
\end{proof}

Analogous to the relation between $D_i$ and $B_i$, we have operators
${\U_\tau},{\Ub_\tau}$ corresponding to $D_\tau,D_{\bar{\tau}}$:
which are defined as the fibration multiplication by $f$ and $\fb$.

Define
\begin{equation}
{{\U_\tau}}_{a\bb}:=g(f\cdot\alpha_a,\alpha_b),\;{{\U_\tau}}={{\U_\tau}}_{a\bb}\alpha^a\otimes
\alpha^{\bb}={{\U_\tau}}_{a}^b \alpha^a\otimes
\alpha_b={{\U_\tau}}_{ab}\alpha^a\otimes \alpha^b.
\end{equation}
Here ${{\U_\tau}}_{ab}={{\U_\tau}}_{a\cb}M^{\cb}_{\bb}$, where $M$
is the matrix representation of the real structure $\tau_R$. It is
easy to see that
\begin{equation}
{{\U_\tau}}_{ab}=\eta(f\alpha_a,\alpha_b)=\eta(\alpha_a,
f\alpha_b)={{\U_\tau}}_{ba},
\end{equation}
i.e., ${{\U_\tau}}={{\U_\tau}}^*$.

Similarly, we can define $(\Ub_\tau)_{a\bb}=g(\fb \alpha_a,
\alpha_b)$ and the corresponding operator $\Ub$.

\begin{lm}\label{lm:tt*2-2} We have relations:
\begin{align}
&f\alpha_a=\U_\tau\cdot \alpha_a+\bpat_f(\gamma_\tau)_a\\
&\fb\alpha_a=\Ub_\tau\cdot\alpha_a+\pat_f(\bar{\gamma}_\tau)_a.
\end{align}
\end{lm}

\begin{proof} It is a straightforward computation: it suffices to
prove the first identity,
$$
\bpat_f(\gamma_\tau)_a=\bpat_f\cdot\bpat_f^\dag\cdot
G\cdot(f\alpha_a)=\Delta G(f\alpha_a)=(f-\U)\cdot\alpha_a,
$$
where we used the fact that $\bpat_f(f\alpha_a)=0$.
\end{proof}

\begin{lm} The following identity holds:
\begin{equation}
B=B^*,\;{\U_\tau}={\U_\tau}^*.
\end{equation}
\end{lm}

\begin{proof} Since $\eta$ is a real symmetric bilinear form, and the
multiplication operator $\pat_i f$ can be inserted equivalently in
the two variable positions, so $B=B^*$. Similarly, we have
${\U_\tau}={\U_\tau}^*$.
\end{proof}

\begin{thm}\label{thm:Fan-eq} The operators $D_i,D_\tau, B,{\U_\tau}$ satisfy the following
relations (We call it as "Fantastic" equations) on $\ch$, i.e.,
after $\mod \ch^\perp$,
\begin{align*}
&(1)\;[D_i,{\U_\tau}]+[B_i,\tau
D_\tau]=0,\;[D_\ib,\Ub_\tau]+[B_\ib,\bar{\tau}
D_{\bar{\tau}}]=0\\
&(2)\;[D_i,\Ub_\tau]=0,\;[D_\ib,{\U_\tau}]=0\\
&(3)\;[\taub D_{\bar{\tau}},B_i]
=[\tau D_{\tau},B_\ib]=0,\;\\
&(4)\;[B_i,{\U_\tau}]=0,\;[B_\ib,\Ub_\tau]=0\\
&(5)\;[\tau D_\tau,D_\ib]=-[{\U_\tau},B_\ib],\;[{\bar{\tau}}D_{\bar{\tau}},D_i]=-[\Ub_\tau,B_i]\\
&(6)\;[\tau D_\tau,D_i]=[\taub D_{\bar{\tau}},D_\ib]=0\\
&(7)\;[\tau D_\tau, \Ub_\tau]=[\bar{\tau}D_{\bar{\tau}},{\U_\tau}]=0\\
&(8)\;[\tau D_\tau,\bar{\tau}D_{\bar{\tau}}]=-[{\U_\tau},\Ub_\tau]\\
\end{align*}
\end{thm}

\begin{proof}
\begin{enumerate}
\item
To prove the identities in the first row, it suffices to prove the
first one.

At first we have the computation
\begin{align*}
&\eta(D_i\alpha_a,f\alpha_b)\\
=&\eta(f\pat_f (\gamma_i)_a, \alpha_b)=-\eta(f\cdot \bpat^\dag_f
\pat_f G (\pat_i f)\alpha_a,\alpha_b)\\
=&-\eta(G\cdot (\pat(\pat_i f)\wedge \alpha_a),\pat_f(f\alpha_b))
=-\eta(G\cdot (\pat(\pat_i f)\wedge \alpha_a),\pat f\wedge
\alpha_b),
\end{align*}
and
\begin{align*}
&\eta((\pat_i f)\alpha_b,\tau D_\tau \alpha_a)\\
=&\eta((\pat_i
f)\alpha_b,\pat_f((\gamma_\tau)_a))=\eta(\pat_f \bpat^\dag_f G f\alpha_a, (\pat_i f)\alpha_b)\\
=&-\eta(G(\pat f\wedge \alpha_a), \pat(\pat_i f)\wedge \alpha_b).
\end{align*}
Then we have
\begin{align*}
&\left\{[B_i,\tau D_\tau]+[D_i,{\U_\tau}]\right\}_{ab}=-(\tau D_\tau)(B_i)_{ab}+D_i{\U_\tau}_{ab}\\
=&-(B_i)_{ab}-\eta((\pat_i f)\alpha_b,(\tau
D_\tau\alpha_a))-\eta((\pat_i f)\alpha_a,(\tau D_\tau \alpha_b))\\
+&(B_i)_{ab}+\eta(D_i\alpha_a,f\alpha_b)+\eta(f\alpha_a,D_i\alpha_b)\\
=&\eta(G(\pat f\wedge \alpha_a),\pat(\pat_i f)\wedge
\alpha_b)+\eta(G(\pat f\wedge \alpha_b),\pat(\pat_i f)\wedge
\alpha_a)\\
-&\eta(G(\pat(\pat_i f)\wedge \alpha_a), \pat f\wedge
\alpha_b)-\eta(G(\pat(\pat_i f)\wedge \alpha_b),\pat f\wedge
\alpha_a)\\
=&0.
\end{align*}
\item To prove the second row, it suffices to prove that
$[D_i,\Ub_\tau]=0$.

We have
\begin{align*}
D_i({\bar{{\U_\tau}}}_{ab})&=\pat_i{\bar{{\U_\tau}}}_{ab}-(\Gamma_i)^c_a{\bar{{\U_\tau}}}_{cb}-(\Gamma_i)^c_b{\bar{{\U_\tau}}}_{ac}.\\
&=\eta(\fb\pat_i\alpha_a,\alpha_b)+\eta(\fb\alpha_i,\pat_i\alpha_b)-(\Gamma_i)^c_a{\bar{{\U_\tau}}}_{cb}-(\Gamma_i)^c_b{\bar{{\U_\tau}}}_{ac}=0.
\end{align*}

Here we used the fact that
$$
\eta(\fb\cdot
\pat_i\alpha_a,\alpha_b)=\eta(\fb(\Gamma_i)^c_a\alpha_c,\alpha_b)+\eta(\fb
\pat_f(\gamma_i)_a,\alpha_b)=(\Gamma_i)^c_a\bar{{\U_\tau}}_{cb}+0.
$$

\item To prove the second row, it suffices to prove that
$[D_\tau,\Bb_i]=0$.

We have
\begin{align*}
\tau D_\tau({B_\ib}_{ab})&=\tau \pat_\tau({B_\ib}_{ab})-(\Gamma_\tau)^c_a{B_\ib}_{cb}-(\Gamma_\tau)^c_b{B_\ib}_{ac}.\\
&=\eta(\overline{\pat_i f}\tau
\pat_\tau\alpha_a,\alpha_b)+\eta(\overline{\pat_i f}\alpha_a,\tau
\pat_\tau\alpha_b)-(\Gamma_\tau)^c_a \Bb_{icb}-(\Gamma_\tau)^c_b
\Bb_{iac}=0.
\end{align*}

Here we used the fact that
$$
\eta(\overline{\pat_i f}\tau
\pat_\tau\alpha_a,\alpha_b)=\eta(\overline{\pat_i
f}(\Gamma_\tau)^c_a\alpha_c,\alpha_b)+\eta(\overline{\pat_i
f}\pat_f\bpat^\dag_f G\cdot\pat_f
({f}\alpha_a),\alpha_b)=(\Gamma_\tau)^c_a \Bb_{icb}+0.
$$
\item Since $\bar{{\U_\tau}}=\bar{{\U_\tau}}^*$ and $\Bb=\Bb^*$, it suffices to prove
that
\begin{equation}
\Bb\bar{{\U_\tau}}=(\Bb\bar{{\U_\tau}})^*.
\end{equation}
Now we have
\begin{align*}
&\eta((\Bb_i\cdot
\Ub_\tau)\alpha_a,\alpha_b)=\eta(\Bb_i\cdot(\fb\alpha_a-\pat_f((\bar{\gamma}_\tau)_a)),\alpha_b)\\
=&\eta(\Bb_i\cdot \alpha_a,\fb\alpha_b)=\eta(\overline{\pat_i
f}\alpha_a-\pat_f((\bar{\gamma}_i)_a),\fb\alpha_b)\\
=&\eta(\overline{(\pat_i f)}\fb \alpha_a,
\alpha_b)=\eta(\overline{(\pat_i f)}\fb \alpha_b,
\alpha_a)=\eta(\alpha_a,(\Bb_i\cdot \Ub_\tau)\alpha_b).
\end{align*}
So we get the proof of the fourth row.

\item Prove the 5th row. For simplicity we can take $\alpha_a$ be the real basis. We have
\begin{align*}
&\eta([\tau D_\tau,D_\ib]\alpha_a,\alpha_b)\\
=&\eta(\tau D_\tau D_\ib
\alpha_a,\alpha_b)-\eta(D_\ib \tau D_\tau \alpha_a,\alpha_b)\\
=&-\eta(D_\ib \alpha_a,\tau D_\tau \alpha_b)+\eta(\tau D_\tau
\alpha_a, D_\ib \alpha_b)\\
=&-\eta(\bpat_f(\bar{\gamma}_i)_a,\pat_f(\gamma_\tau)_b)+\eta(\pat_f(\gamma_\tau)_a,\bpat_f
(\bar{\gamma}_i)_b)\\
=&-g(\bpat_f(\bar{\gamma}_i)_a,\bpat_f(\bar{\gamma}_\tau)_b)+g(\pat_f(\gamma_\tau)_a,\pat_f
(\gamma_i)_b)\\
=&-g(\pat_f\pat_f^\dag G \overline{(\pat_i f)}\alpha_a,\bpat_f^\dag
\bpat_f G (\fb \alpha_b))+g(\pat_f^\dag\pat_f G (f\alpha_a),\bpat_f
\bpat_f^\dag G (\pat_i f)\alpha_b)\\
=&-g((I-P_0)\overline{(\pat_i f)}\alpha_a, (I-P_0)(\fb
\alpha_b))+g((I-P_0)(f\alpha_a),(I-P_0)(\pat_i f)\alpha_b)\\
=&-\eta(\overline{(\pat_i f)}\alpha_a, (I-P_0)(f
\alpha_b))+\eta((f\alpha_a),(I-P_0)(\overline{\pat_i f})\alpha_b)\\
=&\eta(\overline{(\pat_i f)}\alpha_a,P_0 f\alpha_b)-\eta(f\alpha_a,
P_0 \overline{(\pat_i f)}\alpha_b)\\
=&\eta(B_\ib \alpha_a, {\U_\tau} \alpha_b)-\eta(\U_\tau \alpha_a, B_\ib\alpha_b)\\
=&\eta([\U^*_\tau \cdot B_\ib-B^*_{\ib}\cdot\U_\tau]\cdot\alpha_a,\alpha_b ) \\
=&\eta([\U_\tau, B_\ib]\cdot\alpha_a,\alpha_b)\\
=&\eta(-[\U_\tau, B_\ib]\alpha_a,\alpha_b).
\end{align*}
We have used the facts that $\U^*_\tau=\U_\tau, B^*_{\ib}=B_{\ib}$
for the last second equality. By Remark \ref{rem:opera-matrix} the
action of the operator $[\U_\tau, B_i]\cdot$ is given by the matrix
multiplication $-[\U_\tau, B_i]$.

\item To prove the 6th row, we have
\begin{align*}
&\eta([\tau D_\tau, D_i]\alpha_a,\alpha_b)\\
=&\eta(\tau D_\tau
D_i\alpha_a,\alpha_b)-\eta(D_i\tau D_\tau \alpha_a,\alpha_b)\\
=&-\eta(D_i\alpha_a, \tau D_\tau \alpha_b)+\eta(\tau D_\tau
\alpha_a,D_i\alpha_b)\\
=&-\eta(\pat_f(\gamma_i)_a,\pat_f(\gamma_\tau)_b)+\eta(\pat_f(\gamma_\tau)_a,\pat_f(\gamma_i)_b)\\
=&-g(\pat_f(\gamma_i)_a,\bpat_f(\bar{\gamma}_\tau)_b)+g(\pat_f(\gamma_\tau)_a,\bpat_f(\bar{\gamma}_i)_b)\\
=&-g(\bpat_f^\dag\pat_f\bpat_f^\dag G(\pat_i f)\alpha_a,
(\bar{\gamma}_\tau)_b)+g(\bpat_f^\dag\pat_f\bpat_f^\dag
G(f\alpha_a),(\bar{\gamma}_i)_b)=0.
\end{align*}
\item Prove the 7th row. We have
\begin{align*}
&([\tau D_\tau, \Ub_\tau])_{ab}=\eta(\fb \tau D_\tau
\alpha_a,\alpha_b)+\eta(\fb\alpha_a, \tau D_\tau \alpha_b)\\
=&\eta(\fb\pat_f(\gamma_\tau)_a,\alpha_b)+\eta(\fb\alpha_a,
\pat_f(\gamma_\tau)_a)\\
=&0.
\end{align*}

\item Take $\alpha_a$ to be real basis, we
have
\begin{align*}
&\eta([\tau D_\tau, \bar{\tau}D_{\bar{\tau}}]\alpha_a, \alpha_b)\\
=&g([\tau D_\tau, \bar{\tau}D_{\bar{\tau}}]\alpha_a, \alpha_b)\\
=&-g(\bar{\tau}D_{\bar{\tau}}\alpha_a,\bar{\tau}D_{\bar{\tau}}\alpha_b)+g(\tau
D_\tau\alpha_a,\tau
D_\tau\alpha_b)\\
=&-g(\bpat_f(\bar{\gamma}_\tau)_a,\bpat_f(\bar{\gamma}_\tau)_b)+g(\pat_f(\gamma_\tau)_a,\pat_f(\gamma_\tau)_b)\\
=&-g(\bpat_f^\dag\bpat_f \pat_f^\dag
G(\fb\alpha_a),(\bar{\gamma}_\tau)_b)+g(\pat_f^\dag\pat_f\bpat_f^\dag
G(f\alpha_a),(\gamma_\tau)_b)\\
=&-g(\pat_f^\dag(\fb\alpha_a),(\bar{\gamma}_\tau)_b)+g(\bpat_f^\dag
(f\alpha_a),(\gamma_\tau)_b)\\
=&-g(\fb\alpha_a,
(I-P_0)(\fb\alpha_b))+g(f\alpha_a,(I-P_0)(f\alpha_b))\\
=&-\eta(\fb\alpha_a,
(I-P_0)(f\alpha_b))+\eta(f\alpha_a,(I-P_0)(\fb\alpha_b))\\
=&\eta(\fb\alpha_a, P_0(f\alpha_b))-\eta(f\alpha_a,P_0(\fb\alpha_b))\\
=&\eta(-[\U_\tau,\Ub_\tau] \alpha_a,\alpha_b)
\end{align*}


\end{enumerate}
\end{proof}

\subsubsection{\underline{Connection $\D$ on the Hodge bundle
$\ch$}}\

\

Define the covariant derivatives
$\nabla_i=D_i-B_i,\bar{\nabla}_i=D_{\bar{i}}-B_\ib$ and the
corresponding connections:
$$
\nabla_t: \ch\to \ch\otimes \Lambda^1(T^*S),\;\nabla_\tb: \ch\to
\ch\otimes \Lambda^1(\overline{T^*S})
$$
by
$$
\nabla_t=\sum_i dt^i\wedge
\nabla_i,\;\bar{\nabla}_t=\nabla_\tb=\sum_i d t^\ib\wedge
\nabla_\ib.
$$
$\nabla_t,\nabla_\tb$ are the $(1,0)$ and $(0,1)$ parts of the
following connection for fixed $\tau$:
\begin{equation}
\D_t=\nabla_t+\nabla_\tb:\ch\to \ch\otimes T^*_\C S.
\end{equation}
Now the Cecotti-Vafa's equations is equivalent to the following
identities, after $\mod \ch^\perp$,
\begin{equation}\label{ident-curv-2}
\nabla_t^2=\nabla_t^2=[\nabla_t,\nabla_\tb]=0.
\end{equation}

Define the covariant derivatives along the $\tau$-direction:
\begin{equation}
\nabla_\tau=D_\tau-\frac{\U_\tau}{\tau},\;\nabla_\taub=D_\taub-\frac{\Ub_\tau}{\taub},
\end{equation}
and
\begin{equation}
\nabla_\tau=d\tau\wedge \nabla_\tau,\;\nabla_{\taub}=d\taub\wedge
\nabla_\taub.
\end{equation}
Here we use the same notation to denote the differential and its
derivatives along the $\tau$ direction. We also define the 1-forms
\begin{equation}
\U=\U_\tau d\tau,\;\Ub=\Ub_\tau d\bar{\tau}.
\end{equation}

The connection $\D_\tau$ is defined as
\begin{equation}
\D_\tau=\nabla_\tau+\nabla_\taub.
\end{equation}
The total connection $\D$ is defined on the bundle $\ch\to
\C^*\times S$ and given by
\begin{equation}
\nabla:=\nabla_t+\nabla_\tau,\;\bar{\nabla}:=\nabla_\tb+\nabla_\taub,\;\D=\nabla+\bar{\nabla}.
\end{equation}

\begin{thm}\label{thm:Tota-GM} $\D$ is a nearly flat connection of the Hodge bundle $\ch\to
\C^*\times S$,i.e., $\D^2=0,\mod \ch^\perp$ or equivalently
\begin{equation}
\nabla^2=\bar{\nabla}^2=[\nabla,\bar{\nabla}]=0, \mod \ch^\perp.
\end{equation}
\end{thm}

\begin{proof} When expanding the components of $\D$, $\D^2=0$ is the
conclusion of the identities appeared in the Cecotti-Vafa's
equations and the Fantastic equations. We can prove
$[\nabla,\bar{\nabla}]=0$ as an example. We have

\begin{align*}
&[\nabla,\bar{\nabla}]=[(D_t+D_\tau)-(B+\frac{\U}{\tau}),(D_{\tb}+D_{\taub})-(\Bb+\frac{\Ub}{\taub})]\\
=&-[D_t+D_\tau,
\Bb+\frac{\Ub}{\taub}]+\left\{[D_t+D_\tau,D_\tb+D_\taub]+[B+\frac{\U}{\tau},\Bb+\frac{\Ub}{\taub}]\right\}
-[B+\frac{\U}{\tau},D_\tb+D_\taub].
\end{align*}
The first term vanishes, since we have $[D_t,\Bb]=0$ by CV equation
(1); $[D_\tau, \Ub]=0$ by Fantastic equation (7); $[D_t,\Ub]=0$ by
Fantastic equation (2); $[D_\tau,\Bb]=0$ by Fantastic equation (3).
The third term vanishes due to the same reason. The second term is
\begin{align*}
&[D_t,D_\tb]+[D_t,D_\taub]+[D_\tau, D_\tb]+[D_\tau,D_\taub]
+[B,\Bb]+[B,\frac{\Ub}{\taub}]+[\frac{\U}{\tau},\Bb]+[\frac{\U}{\tau},\frac{\Ub}{\taub}].
\end{align*}
By CV equation (4), $[D_t,D_\tb]+[B,\Bb]=0$; By Fantastic equation
(5), we have
$$
[D_t,D_\taub]+[B,\frac{\Ub}{\taub}]=\left\{[D_i, D_\taub]+[B_i,
\frac{\Ub_\tau}{\taub}]\right\}dt^i\wedge d\taub=0,
$$
and
$$
[D_\tau,D_\tb]+[\frac{\U}{\tau},\Bb]=\left\{[D_\ib, D_\tau]+[B_\ib,
\frac{\U_\tau}{\tau}]\right\}dt^\ib\wedge d\tau=0;
$$
By Fantastic equation (8), we have
$$
[D_\tau,D_\taub]+[\frac{\U}{\tau},\frac{\Ub}{\taub}]=\left\{[D_\tau,D_\taub]+
[\frac{\U_\tau}{\tau},\frac{\Ub_\tau}{\taub}]\right\}d\tau\wedge
d\taub=0.
$$
Therefore, we proved $[\nabla,\bar{\nabla}]=0$.
\end{proof}

\subsubsection{\underline{Asymptotic estimates}}\

\

In this part, we will use the Witten-Helffer-Sj\"ostrand method to
estimate the growth order of the harmonic $n$ forms $\alpha_a$ and
the operators $B_i, \U_\tau$ with respect to $\tau$.

Assume that $f_\tau(t)$ is a Morse function with $\mu$
non-degenerate critical points $\{p_i\}$. Let $\alpha_a(\tau,t)$ be
a primitive harmonic $n$-form, then $\alpha_a(\tau,t)$ also
satisfies the equation:
$$
\Box_{f_\tau}\alpha=0,
$$
where $\Box_{f_\tau}=d_{2\re f_\tau}\circ d_{2\re
f_\tau}^\dag+d_{2\re f_\tau}^\dag\circ d_{2\re f_\tau}$, where
$d_{2\re f_\tau}=d+d(2\re f_\tau)\wedge$, is the Witten deformation
operator. Let $\{D_{2\epsilon}(p_a);z_1,\cdots,z_n\}$ be a good
coordinate disc having radius $2\epsilon$, centered at $p_a$ such
that it has flat metric and $f_\tau$ has the form
$$
f(\tau,t)=\tau(z_1^2+\cdots+z_n^2).
$$

Choose a smooth function $\gamma_\epsilon(z)$ satisfying:
$$\gamma_\epsilon(z)=
\begin{cases}
1& |z|\le \epsilon\\
0& |z|\ge 2\epsilon.
\end{cases}
$$

Denote by
\begin{equation}
\varphi^a_{\tau}(z)=C_n \tau^{n/2}e^{-\tau \sum_{i=1}^n |z_i|^2},
\end{equation}
where $C_n$ is chosen such that
$$
||\varphi^a_\tau||_{L^2(\C^n)}=1.
$$
Define
\begin{equation}
\tilde{\varphi}^a_\tau(z)=\frac{\gamma_\epsilon(z)\varphi^a_\tau(z)}{||\gamma_\epsilon(z)\varphi^a_\tau(z)||_{L^2(\C^n)}},
\end{equation}
then $\tilde{\varphi}^a_\tau\in \Omega^n_c(M)$ and satisfies
\begin{equation}
||\tilde{\varphi}^a_\tau||_{L^2(M)}=1.
\end{equation}

As first, we have spectrum gap theorem essentially due to Witten
(here we consider $L^2$ integrable form),

\begin{thm}[Witten] There exist constants $C_1,C_2,C_3$ and $T_0$
depending only on $(M,g,f)$ so that if $|\tau|>T_0$, then
$$
\text{Spec}(\Box_{2\re (f_\tau)})\cap (C_1
e^{-C_2|\tau|},C_3|\tau|)=\emptyset.
$$
\end{thm}

Hence as $|\tau|\to \infty$, we have the decomposition
\begin{equation}
(L^2(\Omega^n(M), d_{2\re(f_\tau)})=(L^2(\Omega^n(M))_{sm},
d_{2\re(f_\tau)})\oplus (L^2(\Omega^n(M))_{la}, d_{2\re(f_\tau)}).
\end{equation}

Note that the Morse function $2\re(f)$ is very special, since all
the $\mu$ critical points have Morse indices $n$.

As $|\tau|\to\infty$, the form $\alpha_a$ will concentrate near the
critical points. Without loss of generality, we can assume that
$\alpha_a$ concentrates at $p_a$ and satisfy
\begin{equation}
||\alpha_a||_{L^2(M)}=1.
\end{equation}
By Theorem 3.1, Proposition 3.3 of \cite{HS}, Theorem 8.30 of
\cite{BZ} and our asymptotic estimate near the infinity far place,
we can obtain the following asymptotic estimate:

\begin{thm}\label{thm:scal-esti} There exists $\epsilon>0$, $T_0$ and $C$ so that for
$|\tau|>T_0$ and any critical point $p_a$, there is
\begin{equation}
\sup_{z\in M\setminus D_{2\epsilon}(p_a)}|\alpha_a(\tau,z)|\le C
e^{-\epsilon |\tau|},
\end{equation}
and
\begin{equation}
|\alpha_a(\tau,z)-\tilde{\varphi}^a_\tau(z)|\le
C\frac{1}{|\tau|},\for z\in D_{2\epsilon}(p_a)\cap W^-(p_a).
\end{equation}
Here $W^-(p_a)$ is the unstable manifold of the critical point $p_a$
with respect to the gradient flow of $2\re(f)$.
\end{thm}

Applying the above theorem, it is easy to obtain the following
conclusion.

\begin{crl} Let $\{\alpha_a\}$ be a local unit frame of the
Hodge bundle $\ch$ at the base point $(\tau,t)\in \C^*\times S_m$.
Then if $|\tau|$ is large enough, we have the asymptotic formula:
\begin{align}
&\U_\tau=\tau\Uct,\;\B_i=\tau\Bci,\;A^\dag=\tau A^{\dag,\circ}\\
&\Ub_\tau=\bar{\tau}\U_{\bar{\tau}}^\circ,\;B_\ib=\bar{\tau}B_{\ib}^\circ,\;A=\bar{\tau}
A^{\circ}
\end{align}
where
$|\Uct|,|\Bci|,|\U_{\bar{\tau}}^\circ|,|B_{\ib}^\circ|,|A^{\dag,\circ}|,|A^\circ|\le
C$ and $C$ is a constant depending only on $M,g,f,t$ but not on
$\tau$.
\end{crl}

\subsection{Wall-crossing phenomenon and duality of Gauss-Manin
connections}\label{subsec:4.3}

\

Note that the action of the connections on the Hodge bundle is not
closed, and the Cecotti-Vafa and Fantastic equations only hold up to
a term in the perpendicular subspace $\ch^\perp$. In this section,
we want to construct the bundle $\ch_{\ominus, top}$ such that the
action of the connections is closed and the Cecotti-Vafa and
Fantastic equations really hold on $\ch_{\ominus, top}$. This is
done by considering the integration of the middle dimensional
harmonic forms over the Lefschetz thimble. So at first we study the
dual homology structure.

At the moment, we assume that the deformation space $S$ is a open
domain with constant global Milnor number on $M$.

Let $S_m(\tau)\subset S$ be the set of $t\in S$ such that all the
critical points of $f_{(\tau,t)}$ are Morse critical points.
Obviously $S_m(\tau)$ is independent of $\tau$. We can stratify the
space $S_m(\tau)$ in the following way. Let
$S_r(\tau):=S_m^0(\tau)\subset S_m(\tau)$ be the set of all $t$ in
$S_m(\tau)$ such that all the critical values have different
imaginary parts. We call the point $t$ in $S_r(\tau)$ as a regular
point. Let $S_m^k(\tau)\subset S,k\le \mu,$ be the set of all $t$ in
$S_m(\tau)$ such that there exist at most $k+1$ critical values of
$f_{(\tau,t)}$ have the same imaginary parts.

We have the following result and the proof is similar to Theorem
3.1.4 of \cite{FJR3}.

\begin{thm}\label{thm:wall-base}
\begin{enumerate}
\item $S_m(\tau)$ is a path connected dense open set in $S$.\\

\item $S_m(\tau)-S_r(\tau)$ is a union of real hypersurfaces
and separates the set $S_r(\tau)$ into a system of chambers.

\item $S_m^1(\tau)\subset S_m(\tau)$ is a path connected dense open set in $S$,
and $S_m^2(\tau)-S_m^1(\tau)$ is a real codimension $2$ real
analytic set in $S_m(\tau)$.

\item There is a decomposition $S=S_m^1(\tau)\cup (S-S_m^1(\tau))$, where $S-S_m^1(\tau)$ is a real codimension $2$ real
analytic set in $S$.

\end{enumerate}
\end{thm}

Define the following subsets in $\C^*\times S$:
$$
S_m=\cup_{\tau\in\C^*} S_m(\tau),\;S_m^k=\cup_{\tau\in
\C^*}S_m^k(\tau),k=0,1,\cdots,\mu.
$$

\begin{crl} Theorem \ref{thm:wall-base} holds for $S_m, S_m^k,S_r$
in $\C^*\times S$ instead of $S_m(\tau), S_m^k(\tau),S_r(\tau)$.
\end{crl}

\begin{ex} Let $f_{\tau,t}=\tau f_t=\tau(x^3+t_1 x+t_2)$, where $f_t$ is the miniversal deformation of
$f=x^3$. Then $f_t$ has two critical points $x_1(t),x_2(t)$ for any
$t$, which is given by $x^2=-\frac{t_1}{3}$. The walls are given by
the hypersurface:
$$
0=\im(f_t(x_1(t))-f_t(x_2(t)))=\im((x_1-x_2)(x_1^2+x_2^2+x_1x_2+t_1)),
$$
which is equivalent to
$$
\im (\sqrt{-1} t_1^{\frac{3}{2}})=0.
$$
Let $t_1=r_1e^{i\theta_1}$, then the equation of the walls in
$(t_1,t_2)$ space are given by
$$
\theta_1=\frac{2\pi k}{3},k=0,1,2.
$$
So if $\tau=1$, we have $3$ walls separating the $t=(t_1,t_2)$ space
and intersecting at $0$. Note that the equations of the walls are
independent of $t_2$.

For the total space $\C^*\times S$, we have the wall equation:
$$
\im \tau (\sqrt{-1} t_1^{\frac{3}{2}})=0,
$$
which is the rotation of the walls at $\tau=1$ along the $\tau$
direction.
\end{ex}

\begin{ex} Let $p\ge 3 $ and $f(x)=f(x_1, \cdots,  x_n)=x_1^p+\cdots+x_n^p$. Then the residue classes of
monomials $\{x=x_1^{\nu_1}\cdots x_n^{\nu_n},0\le \nu_i\le p_i-2,i=1\cdots n\}$ form the basis of the Milnor ring
$\C[x]/(\pat_x f)$ . The moduli space of marginal deformation and relative deformation forms our $S$, and the marginal deformation forms
part of $S-S_m(\tau)$ and is a subspace of codimension greater than $1$. A typical example is $f(x_1,x_2,x_3)=x_1^3+x_2^3+x_3^3$, and $S$ is given by the miniversal deformation $f(t)=x_1^3+x_2^3+x_3^3+t_8x_1x_2x_3+t_7x_1x_2+t_6x_2x_3+t_5x_3x_1+t_4 x_1+t_3 x_2+t_2 x_1+t_1$ and the marginal deformation is only given by $t_8$.

\end{ex}

Now we begin our first consideration: fix $\tau$. For simplicity, we
use $S_m, S_m^k,S_r$ to replace the notations $S_m(\tau),
S_m^k(\tau),S_r(\tau)$ in Section \ref{subsub:mono-1}-Section
\ref{subsub:mono-2}.

\subsubsection{\underline{Relative homology and Lefschetz
thimble}}\label{subsub:mono-1}\

\

Let $(M,g)$ be a stein manifold and $(\tau,t)\in S_m\subset S$,i.e.,
$f_{(\tau,t)}$ is a tame holomorphic Morse function. Since
$f_{(\tau,t)}=\tau f_t$, so if $f_t\in S_m$ then for any $\tau\in
\C^*$, $f_{(\tau,t)}\in S_m$. So $f_{(\tau,t)}$ has finitely many
critical points on $M$. We think $(M,g)$ as a real Riemannian
manifold and consider the negative gradient system generated by the
dual vector field of $2Re(f_{(\tau,t)})$. Let $\alpha>0$ and define
the sets
$$
f_{(\tau,t)}^{\ge \alpha}:=\{z\in M: |2Re(f_{(\tau,t)}(z))\ge
\alpha\},\;f_{(\tau,t)}^{\le - \alpha}:=\{z\in M|
2Re(f_{(\tau,t)}(z))\le -\alpha\}.
$$
By Morse theory, we know that if there is no critical values of
$2Re(f_{(\tau,t)})$ between $[\alpha,\beta]$, then the set
$f_{(\tau,t)}^{\le \alpha}$ is the deformation kernel of
$f_{(\tau,t)}^{\le \beta}$ by the flow generated by the vector field
$\frac{\nabla Re(f_{(\tau,t)})}{|\nabla Re(f_{(\tau,t)})|}$.

So if $\alpha$ is large enough, there is no critical points in
$f_{(\tau,t)}^{\le -\alpha}$ and each $f_{(\tau,t)}^{\le -\alpha}$
has the same homotopy type. We denote the equivalence class by
$f_{(\tau,t)}^{-\infty}$. Similarly, we can define
$f_{(\tau,t)}^{+\infty}$ for $f_{(\tau,t)}^{\ge \alpha}$ if $\alpha$
is large enough. We have the relative homology group $H_*(M,
f_{(\tau,t)}^{-\infty},\Z)$ and $H_*(M, f_{(\tau,t)}^{+\infty},\Z)$.
We can perturb $f_{(\tau,t)}$ a little bit such that $(\tau,t)\in
S_r$. Then we have the relative homology group $H_*(M,
f_{(\tau,t)}^{-\infty},\Z)$. Since $f_{(\tau,t)}$ is a holomorphic
Morse function, its real part $Re(f_{(\tau,t)})$ is a real Morse
function with Morse index $n$. We can define the unstable manifold
at a critical point $z_a$ of $f_{(\tau,t)}$ (or $Re(f_{(\tau,t)})$):
$$
\cC_a^-=: \{z\in M: z\cdot s\to z_a, \text{as}\;s \to -\infty\},
$$
where $z\cdot s$ represents the action of the flow at time $s$.
Similarly, we can define the stable manifold at $z_a$:
$$
\cC_a^+=: \{z\in M: z\cdot s\to z_a, \text{as}\;s \to +\infty\}.
$$

\begin{df} We call $\cC_a^\pm$ as the positive or negative
Lefschetz thimble of $f_{(\tau,t)}$ at the critical point $z_a$.
\end{df}

We have the following theorem, which should be known (but we can't
find the appropriate references for general $M$).

\begin{thm} Suppose that $f$ is a holomorphic Morse function defined
on the complex manifold $M$ which is fundamental tame, then only the
middle dimensional relative homology $H_n(M,f^{-\infty},\Z)$ is
nontrivial and it is a free abelian group and is generated by the
Lefschetz thimble $\{\cC_a^-,a=1,\cdots,\mu\}$, where $\mu$ is the
number of critical points of $f$ or Milnor number.
\end{thm}

\begin{proof} We can construct the relative Morse complex $(C_k,\pat)$ , where
each group $C_k=\oplus_a \Z \ka_a$ is a free abelian group and
generated by the critical points $\ka_a$ with index $k$. Now only
the middle dimensional homology group is nontrivial. Hence we get
the conclusion.
\end{proof}

\begin{rem} There is a natural map
$$
\pat: H_n(M,f^{-\alpha},\Z)\to H_{n-1}(f^{-1}(-\alpha),\Z)
$$
given by taking the boundary of $\cC_a^-$. In fact, we have
$$
\pat \cC_a^-= Re(f)^{-1}(-\alpha)\cap \cC_a^-.
$$
is a real $n-1$ cycle in $ f^{-1}(-\alpha)$. However, in general we
don't know if this map is surjective or injective because of the
topology of $M$. In the case that $M=\C^n$, this map is an
isomorphism and $\pat \cC_a^-$ are nontrivial cycles and are called
the vanishing cycles related to the critical points $z_a$.
\end{rem}

If the perturbation $t$ is small and preserve the tameness, then the
set $f_t^{\le -\alpha}$ is homotopic to $f^{\le -\alpha}$ for large
$\alpha$. Therefore we have the isomorphism even if $f$ is not a
Morse function:

\begin{equation}\label{eq:lef-isom}
H_*(M, f^{-\infty},\Z)\cong H_*(M,f_t^{-\infty},\Z).
\end{equation}

If $(\tau,t)\in S_m$, we can define an intersection pairing:
$$
\langle\cdot,\cdot\rangle :H_*(M,f_{(\tau,t)}^{-\infty},\Z)\times
H_*(M,f_{(\tau,t)}^{\infty},\Z)\to \Z
$$
by choosing two special basis in the two relative homology groups.

Let $\cC^-_a$ and $\cC^+_a$ be the unstable and stable (smooth)
manifold of the critical point $\kappa_a$, then
$\{\cC^-_a,a=1,\cdots,\mu\}$ and $\{\cC^+_b,b=1,\cdots,\mu\}$ form
the two basis and we define
$$
\langle\cC^-_a,\cC^+_b\rangle=I_{a^-b^+}=
\begin{cases}
1\;&a=b\\
0\;&a\neq b.
\end{cases}
$$

Similarly, one can define the intersection number
$$
I_{a^+b^-}=
\begin{cases}
(-1)^n\;&a=b\\
0\;&a\neq b
\end{cases}
$$
and the intersection pairing:
$$
\langle\cdot,\cdot\rangle :H_*(M,f_{(\tau,t)}^{\infty},\Z)\times
H_*(M,f_{(\tau,t)}^{-\infty},\Z)\to \Z.
$$

\begin{df} Since for any $t\in S$, $f_{(\tau,t)}$ is assumed to be
fundamental tame and has the same Milnor number $\mu$. We obtain a
$\Z$-coefficient local system $H^{\ominus}$ over $\C^*\times S$
whose fiber at $(\tau,t)\in \C^*\times S$ is
$H_{(\tau,t)}^{\ominus}=H_n(M,f_{(\tau,t)}^{-\infty},\Z)$.
Similarly, we can obtain the dual local system $H^{\oplus}$ over
$\C^*\times S$ with fiber
$H_{(\tau,t)}^{\oplus}=H_n(M,f_{(\tau,t)}^{+\infty},\Z)$. Define the
corresponding bundle as $H^\ominus=H^\ominus \otimes \C$ and
$H^\oplus=H^\oplus\otimes \C$. The total space is defined as
$H_{tot}=H^{\ominus}\oplus H^{\oplus}$. We call such bundles as the
relative homological Milnor fibrations.
\end{df}

Note that the fibers have the isomorphism:
$$
H_{(\tau,t)}^{\oplus}\cong H_{(-\tau,t)}^{\ominus}.
$$
and the intersection pairing $I$ can be thank as the following
pairing:
$$
I:H_{(\tau,t)}^{\ominus}\times H_{(-\tau,t)}^{\ominus}\to \Z.
$$

For $(\tau,t)\in S_r$, we can order each critical point $\ka_a$ in
the way that $a<b$ if and only if $\im(\varpi_a)<\im(\varpi_b)$,
where $\varpi_a$ is the critical value corresponding to the critical
point $\ka_a$. Therefore we get the order of the Lefschetz thimbles
$H^\pm_a(\tau,t)$ in the fiber $H_{(\tau,t)}^\ominus$ and
$H_{(\tau,t)}^\oplus$.

The $\Z$-lattices of the relative homological Milnor fibrations give
the flat connection $\D_{top}$ of the corresponding bundle which is
called the topological Gauss-Manin connection. Let $l(s),s\in
[-1,1]$ be a path in one chamber of $S_r$, then the parallel
transportation by $\D_{top}$ will not change the homotopy type of
the Lefschetz thimbles $\cC_a(\tau,t)^\pm$ and preserve the order.

Now we consider a path $l(s),s\in [-1,1]$ in $S_m$ go from the
chamber $cham^-$ to the chamber $cham^+$ and intersect with the wall
only at $s=0$. Let $\cC^-_a(-)$ ($\cC^-_a(+)$) be the basis of the
fibration at chamber $cham^-$ ($cham^+$). We assume that the
positions of the critical values $\varpi_i$ and $\varpi_{i+1}$ have
been changed and the other critical values keep fixed. Then the
wall-crossing formula changes the relative positions of only two
critical values among all critical values.

\begin{thm}[Wall-crossing formula]\label{thm: wall-crossing} Let $\{\ka_a(\pm 1)\}_{a=1}^\mu$
be the set of the ordered critical points at $s=\pm 1$. Assume
without loss of generality that $\ka_a(\pm 1)=\ka_a$ is fixed for
$a\neq i$, $\ka_a(s=\pm 1)=\ka_a(\pm 1)$ and $\im(\varpi_i(0))=\im
(\varpi_{i+1})$.

If the perturbation satisfies $\re (\varpi_i(s))<\re
(\varpi_{i+1})$, we have the left transformation
\begin{equation}
\begin{cases}
\cC^-_a(+)=\cC^-_a(-),&\:\forall a\neq i,i+1\\
\cC^-_i(+)=\cC^-_{i+1}(-)+I_{i^+,(i+1)^-}\cC^-_{i}(-),&\\
\cC^-_{i+1}(+)=\cC^-_{i}(-)&.
\end{cases}
\end{equation}
where $I_{i^+,(i+1)^-}=\#(\cC^+_i(s=0)\cap \cC^-_{i+1})$ is the
intersection number of the stable manifold of the critical point
$\ka_i(s=0)$ with the unstable manifold of the critical point
$\ka_{i+1}$.

If the perturbation satisfies $\re (\varpi_i(s))>\re
(\varpi_{i+1})$, we have the right transformation
\begin{equation}
\begin{cases}
\cC^-_a(+)=\cC^-_a(-),&\:\forall a\neq i,i+1\\
\cC^-_i(+)=\cC^-_{i+1}(-),&\\
\cC^-_{i+1}(+)=\cC^-_{i}(-)+I_{i^+,(i+1)^-}\cC^-_{i+1}(-)&.
\end{cases}
\end{equation}
\end{thm}

\begin{proof} The proof is similar to the case $M=\C^n$. The reader
can see \cite{E}.
\end{proof}

Similarly, one can get the transformation formula of the dual
Lefschetz thimbles in $H^\oplus$.

Since the connection $\D_{top}$ is flat, the parallel transportation
along a closed loop $\gamma\in \pi_1(S_m, (\tau_0,t_0))$ defines the
monodromy action $h^*_\gamma$ on the fiber
$H^\ominus_{(\tau_0,t_0)}$. The action is independent on the
representative element in its homotopy class $[\gamma]$ in
$\pi_1(S_m, (\tau_0,t_0))$. Thus we have the monodromy
representation:
$$
T:\pi_1(S_m, (\tau_0,t_0))\to \aut H^\ominus_{(\tau_0,t_0)}:
[\gamma]\to h^*_\gamma.
$$
Fix a basis in $H^\ominus_{(\tau_0,t_0)}$ and let it moves along a
loop in $S_m$. When the loop crosses the wall of each chamber, the
basis will be changed because the two vicinal critical points will
exchange instantons. The explicit representation is given by the
Wall-crossing formula, Theorem \ref{thm: wall-crossing}. Therefore,
it is easy to see the following conclusion holds:

\begin{prop} Let $l_{\tau_0}=\{(e^{i\theta}\tau_0, t_0)|\theta\in [0,2\pi]\}$ be a loop in $S_m$. Then the isomorphism $h_l^*\in
\aut(H^\ominus_{(\tau_0,t_0)})$  induced by the parallel
transportation by $\D_{top}$ along $l$ is the same as the monodromy
transformation $T$.
\end{prop}

\subsubsection{\underline{Witten index and intersection matrix}}\

\

Let $(\tau_0,t_0)\in S_r$ be the base point and $\{\cC^-_a(0)\}$ be
a ordered basis of $H^\ominus_{t_0}$. The special path
$l_{\tau_0}=\{(e^{i\theta}\tau_0, t_0)|\theta\in [0,\pi]\}\subset
S_m$ maps $f_{t_0}$ to $-f_{t_0}$ and induces the action
$$
I_W^-: H^\ominus_{(\tau_0,t_0)}\to H^\ominus_{(-\tau_0,t_0)}\cong
H^\oplus_{(\tau_0,t_0)}.
$$
$I_W^-$ is a reflection about the origin in the image plane $\C$.
The base transformation is given by
\begin{equation}
\cC^+_{\mu-a}(0)=(I_W^-)(\cC^-_a(0))=\cC^-_a(\pi)=(I_W^-)_{ba}\cC^-_b(0),
\end{equation}
where
$$
(I_W^-)_{ba}:=\left\{\begin{array}{ll} \#(\cC^-_b\cap
\cC^+_a),\;&\text{if}\;\im
f(p_a)<\im f(p_b)\\
1,\;&\text{if}\;\cC_a=\cC_b\\
0,\;&\text{if}\;\im f(p_a)>\im f(p_b)
\end{array}\right.
$$

Similarly, the loop induces another map:
$$
I_W^+: H^\oplus_{(\tau_0,t_0)}\to H^\oplus_{(-\tau_0,t_0)}\cong
H^\ominus_{(\tau_0,t_0)},
$$
and the base transformation is given by
\begin{equation}
\cC^-_{\mu-a}(0)=\cC^+_a(\pi)=(I_W^+)_{ba}\cC^+_b(0),
\end{equation}
where
$$
(I_W^+)_{ba}:=\left\{\begin{array}{ll} \#(\cC^+_b\cap
\cC^-_a),\;&\text{if}\;\im
f(p_a)>\im f(p_b)\\
(-1)^n,\;&\text{if}\;\cC_a=\cC_b\\
0,\;&\text{if}\;\im f(p_a)<\im f(p_b)
\end{array}\right.
$$

\begin{df} We call the transformation group $I_W^-\in \End(H^\ominus_{(\tau_0,t_0)},
H^\ominus_{(-\tau_0,t_0)})$ (and $I_W^+$) as the Witten map and the
corresponding matrices as the Witten indices.
\end{df}

\begin{prop} Given base point $(\tau_0,t_0)\in S_r$, the monodromy transformation $T$
and the Witten maps have the following relations:
\begin{equation}
T=(I_W^-)^2, \;T=(I_W^+)^2,\;I_W^+=((I_W^-)^{-1})^T.
\end{equation}
\end{prop}

\begin{proof} The front two identities are obvious since the path
defining the Witten maps is half the loop defining the monodromy
transformation. The third one is true since we have the pairing
$\#(\cC^-_a(t)\cap \cC^+_b(t))=\delta_{ab}$.
\end{proof}

\begin{lm}\label{lm:witten-1} Let $P$ be a
parallel transformation along a path $\gamma(\tau)$ from the fiber
$H^\ominus_{(\tau_0,t_0)}$ to $H^\ominus_{(\tau_0,t_1)}$ (or from
$H^\oplus_{(\tau_0,t_0)}$ to $H^\oplus_{(\tau_0,t_1)}$) across a
wall, then we have $P\circ I_W^\pm=I_W^\pm\circ P$.
\end{lm}

\begin{proof} The deformation is given by
$e^{i\theta}f_{\gamma(\tau)},\theta\in [0,\pi],\tau\in [-1,+1]$ and
the corresponding deformation parameters lie on $S_m^1$. It is
routine to construct a homotopy exchange $I_W^\pm\circ P$ and
$P\circ I_W^\pm$.
\end{proof}

\begin{df} We can also define a non-degenerate symmetric bilinear form in
$H^\ominus_{(\tau,t)}$ for any $(\tau,t)\in S_r$ as follows: fix a
basis $\{\cC^-_a\}_{a=1}^\mu$, and define
\begin{equation}
I_W(\cC^-_b,\cC^-_a):=\langle \cC^-_b,I_W^-(\cC^-_a)\rangle
\end{equation}
Similarly, we can define the non-degenerate symmetric bilinear form
$I_W$ in $H^\oplus_{(\tau,t)}$ via $I_W^+$.
\end{df}

Via Lemma \ref{lm:witten-1}, it is easy to get the following result.

\begin{prop} The non-degenerate bilinear form $I_W$ is parallel with respect to the
topological Gauss-Manin connection $\D_{top}$. Therefore, $I_W$ is
well-defined on $S_m^1$ which is independent of the choice of the
basis.
\end{prop}

\begin{rem} We can extend the intersection pairing
$I_W,\langle\cdot,\cdot\rangle$ to the whole space $\C^*\times S$ by
using isomorphism (\ref{eq:lef-isom}) to identify them to the nearby
perturbed quantities and then show they are independent of the
perturbation. Some times we will not distinguish the notation $I_W$
to represent the intersection pairing $I_W$ for the same homology
group and the pairing $\langle \cdot,\cdot\rangle$ between two
homology groups.
\end{rem}

\begin{rem} The intersection form $\langle \cdot,\cdot\rangle $
 gives a symplectic structure on the bundle $\ch^{tot}\to S_m^1$.
\end{rem}

\subsubsection{\underline{Poincar\'e duality}}\

\

The intersection number can be expressed as the integration of the
dual forms:

\begin{equation}
(I_W)_{a^-b^+}=\langle \cC^-_a,\cC^+_b\rangle=\#(\cC^-_a\cap
\cC^+_b)=\int_{\cC_a^-}PD(\cC_b^+),
\end{equation}
where $PD:H_n(M,f_{(\tau,t)}^{\infty}, \R)\to
H^n(M,f_{(\tau,t)}^{-\infty}, \R)$ is the Poincare dual operator.
Sometimes we also denote the integration as the intersection form:
\begin{equation}
\int_{\cC_a^-}PD(\cC_b^+)=:\langle \cC^-_a,PD(\cC^+_b)\rangle.
\end{equation}

For simplicity, we denote $f:=f_{(\tau,t)}$ in this part of
discussion.

\begin{df} Define $\ch^{(\tau,t)}_{\ominus,top}:=H^n(M,f_{(\tau,t)}^{-\infty}, \R)$ and $
\ch^{(\tau,t)}_{\oplus,top}:=H^n(M,f_{(\tau,t)}^{+\infty}, \R)$. Let
$\ch_{\ominus,top},\ch_{\oplus,top}$ be the bundles with fiber at
$(\tau,t)$ to be $\ch^{(\tau,t)}_{\ominus,top}$ and
$\ch^{(\tau,t)}_{\oplus,top}$ respectively. Obviously, we have
$$
\ch^{(-\tau,t)}_{\oplus,top}=\ch^{(\tau,t)}_{\ominus,top}.
$$
\end{df}

There is an intersection pairing be tween
$\ch^{(\tau,t)}_{\ominus,top}$ and $\ch^{(\tau,t)}_{\oplus,top}$
induced by the intersection pairing $I_W$ of the relative homology
groups.

We want to represent those cohomology classes of
$\ch^{(\tau,t)}_{\ominus,top}$ and the pairing in differential forms
and the corresponding integration.

Let $\{\alpha_c,c=1,\cdots,\mu\}$ be a local frame of $\ch_\ominus$
consisting of the $L^2$ harmonic $n$-forms. Then $\alpha_c$ are
primitive forms and satisfy
$$
\bpat_f \alpha_c=0,\;\pat_f \alpha_c=0.
$$
\begin{lm} Let $S_c^-=e^{(f+\bar{f})}\alpha_c$ and
$S_c^+=e^{-(f+\bar{f})}*\alpha_c$. Then $S_c^-$ and $S_c^+$ are
$d$-closed $n$-forms on $M$.
\end{lm}

\begin{proof} We have
$$
d(e^{(f+\bar{f})}\alpha_c)=e^{(f+\bar{f})}[df\wedge+d\bar{f}\wedge+d]\alpha_c=0.
$$
On the other hand, since $*\alpha_c$ are closed $\bpat_{-f}$ and
$\pat_{-f}$ forms, $S_c^+$ are also closed $d$-forms.
\end{proof}
We have
$$
\eta_{cd}=\int_{M}\alpha_c\wedge
*\alpha_d=\int_{M}S_c^-\wedge S_d^+
$$
Define
$$
\Pi_{ac}^-=\int_{\cC_a^-}S_c^-,\;\;\Pi^+_{bd}=(-1)^n
\int_{\cC_b^+}S_d^+.
$$
Assume that
$$
PD(\cC_a^+)=\sum_d c_{ad}^+ S_d^-,\;PD(\cC_b^-)=\sum_d c_{bd}^-
S_d^+.
$$
Then we have
$$
\Pi^+_{bd}=(-1)^n\int_{M} PD(\cC_b^+)\wedge S_d^+=(-1)^n\sum_k
\int_{M}c_{bk}^+ S_k^-\wedge S_d^+=(-1)^n\sum_k \eta_{kd}c_{bk}^+.
$$
and similarly
$$
\Pi_{ac}^-=\sum_k \eta_{kc}c_{ak}^-.
$$
So we obtain
\begin{align*}
c_{bk}^+=(-1)^n\sum_d \Pi^+_{bd}\eta^{dk},\;c_{ak}^-=\sum_d
\Pi_{ad}^-\eta^{dk}
\end{align*}

Thus we can prove that
\begin{align}
&(I_W)_{a^-b^+}=\#(\cC^-_a\cap
\cC^+_b)=\int_{\cC_a^-}PD(\cC_b^+)\nonumber\\
&= (-1)^n\sum_{d,e} \int_M c^-_{ad}S_d^+\wedge c^+_{be}S^-_e=
\sum_{l,d} \Pi^-_{al} \eta^{ld} \Pi^+_{bd}\nonumber
\end{align}

\begin{prop} The real structure $\tau_R=M$, and the matrices $I_W, \Pi^-,\Pi^+, \eta$ have the following
relations for $(\tau,t)\in \C^*\times S_r$,
\begin{equation}
I_W=\Pi^-(\tau,t) \cdot \eta^{-1}(\tau,t)\cdot
\Pi^+(\tau,t),M=\overline{\Pi}\Pi^{-1},\;g=\eta\cdot M.
\end{equation}
\end{prop}

\begin{crl} The matrices $\Pi^\pm(\tau,t),(\tau,t)\in \C^*\times S_r$ are nondegenerate
matrices, hence the spaces
$\ch_{\oplus,top}^{(\tau,t)},\ch_{\ominus,top}^{(\tau,t)}$ are
generated by the $d$-closed forms $S_c^-$'s and $S_c^+$'s.
\end{crl}

\begin{df} The matrix $\Pi^\pm(\tau,t),(\tau,t)\in \C^*\times S_r$ are called the periodic matrices.
\end{df}

\begin{df} Note that $\ch_{\ominus}$ is the Hodge bundle with fiber
at $(\tau,t)$ to be the space of $\Delta_{f_{(\tau,t)}}$-harmonic
$n$-forms. We define the bundle isomorphism
$\psi_{\ominus}:\ch_\ominus\to \ch_{\ominus,top}$ as
\begin{equation}
[\psi_{\ominus}(\alpha)](\tau,t)=e^{f_{(\tau,t)}+\fb_{(\tau,t)}}\alpha(\tau,t),\;(\tau,t)\in
\C^*\times S.
\end{equation}
Similarly, we have the isomorphism $\psi_{\oplus}:\ch_\oplus\to
\ch_{\oplus,top}$ as
\begin{equation}
[\psi_{\oplus}(\alpha)](\tau,t)=e^{-f_{(\tau,t)}-\fb_{(\tau,t)}}\alpha(\tau,t),\;(\tau,t)\in
\C^*\times S.
\end{equation}

Define the operator $\hstar:\ch_{\ominus,top}\to \ch_{\oplus,top}$
as
$$
\hstar(S^-)=\hstar(e^{f+\fb}\alpha)=e^{-(f+\fb)}*\alpha.
$$
\end{df}
Note that $*\alpha(\tau,t)$ is a $\Delta_{-f}$-harmonic form and
lies in $\ch_\oplus^{(\tau,t)}\equiv \ch_\ominus^{(-\tau,t)}$.

\begin{prop} The following relations hold:

\begin{enumerate}
\item
\begin{equation}
\hstar^2=*^2=(-1)^n,\;
\end{equation}

\item the diagram below commutes.
$$
\begin{CD}
\ch_\ominus@>{{\psi_{\ominus}}}>>\ch_{\ominus,top}\\
@VV{{*}}V@VV{{\hstar}}V\\
\ch_\oplus @>{\psi_{\oplus}}>>\ch_{\oplus,top}
\end{CD}
$$
\end{enumerate}
\end{prop}

\begin{df} Define the pairing $\heta$ in $\ch_{\ominus,top}$ as
\begin{equation}
\heta(S^-_1,S^-_2):=\langle S^-_1,\hstar S^-_2\rangle.
\end{equation}
Plus the real structure $\tau_R$, we can define a Hermitian metric
on $\ch_{\ominus,top}$ as
\begin{equation}
\hg(S^-_1,S^-_2)=\heta(S^-_1,\tau_R\cdot S^-_2)
\end{equation}
\end{df}

\begin{thm}\label{thm:bund-isom} The bundle map $\psi_\ominus$ provides an isomorphism
between two real Hermitian holomorphic bundles:
$(\ch_\ominus,g,\tau_R)$ and $(\ch_{\ominus,top}, \hg,\tau_R)$ and
the same for $\psi_\oplus$.
\end{thm}

\begin{proof} Since the function $e^{f+\fb}$ is a real function, the isomorphism $\psi_\oplus$
preserves the real structure. The other facts are easy to see. Note
that if $\bpat$ defines the holomorphic structure on
$\ch_{\ominus,top}$, the pull-back operator
$\psi_{\ominus,\tau}^*(\bpat)=e^{-f-\fb}\cdot\bpat\cdot
e^{f+\fb}=\bpat+\bpat\fb$ defines the holomorphic structure on
$\ch_\ominus$.
\end{proof}

\begin{rem} The operator $\hstar$ does not correspond to the Witten
map $I_W^-$, since there is $(I_W^-)^2=T$, the monodromy
transformation.
\end{rem}

\subsubsection{\underline{Flat connection and horizontal sections of bundle $\ch_{\ominus,top}$}}\

\

Fix $(\tau_0,t_0)\in S_r$, and let $(\tau,t)\in S_r,
f:=f_{(\tau,t)}=\tau f_t$ and $\cC^-_c(\tau,t)\in
H_n(M,f_{(\tau,t)}^{-\infty},\Z)$.

We know that the $n$ form
$S_a^-=S_a^-(\tau,t)=e^{f_{(\tau,t)}+\overline{f_{(\tau,t)}}}\alpha_a$
is $d$-closed. Now using the formulas
$$
D_i\alpha_a=\pat_f(\gamma_i)_a,\;(\gamma_i)_a=\bpat_f^\dag\cdot
G\cdot [(\pat_i f\alpha_a-(B_i)^b_a\alpha_b)],
$$
we have
\begin{align*}
D_i S_a^-&=\pat_i(e^{f+\bar{f}}\alpha_a)-\Gamma_i(e^{f+\bar{f}}\alpha_a)\\
&=e^{f+\bar{f}}(\pat_i f\alpha_a+D_i\alpha_a)\\
&=(B_i)_a^b
e^{f+\bar{f}}\alpha_b+e^{f+\bar{f}}(\pat_f(\gamma_i)_a+\bpat_f(\gamma_i)_a)\\
&= B^b_{ia}S_b^-+d(e^{f+\bar{f}}(\gamma_i)_a).
\end{align*}

This is equivalent to
\begin{equation}
\nabla_i S_a^-=d(e^{f+\bar{f}}(\gamma_i)_a).
\end{equation}

Similarly, we have
\begin{equation}
\nabla_\ib S_a^-=d(e^{f+\bar{f}}(\gamma_\ib)_a).
\end{equation}
On the other hand, we can consider the covariant derivative along
$\tau$ direction. We knew that
$$
D_\tau\alpha_a=\pat_f(\frac{(\gamma_\tau)_a}{\tau}),\;(\gamma_\tau)_a=\bpat_f^\dag\cdot
G\cdot(f\alpha_a).
$$
So
\begin{align*}
D_\tau S_a^-=&e^{f+\fb}(\frac{f}{\tau}\alpha_a+D_\tau \alpha_a)\\
=&e^{f+\fb}\frac{1}{\tau}(\U\cdot\alpha_a+\bpat_f(\gamma_\tau)_a+\pat_f(\gamma_\tau)_a),\\
\end{align*}
which is equivalent to
\begin{equation}
\nabla_\tau S_a^-=\frac{1}{\tau}d(e^{f+\fb}(\gamma_\tau)_a).
\end{equation}
Similarly, we have
\begin{equation}
\nabla_\taub S_a^-=\frac{1}{\taub}d(e^{f+\fb}(\bar{\gamma}_\tau)_a).
\end{equation}
In summary, we have
\begin{lm}
\begin{align}
&\nabla_i S_a^-=d(e^{f+\bar{f}}(\gamma_i)_a),\;\nabla_\ib
S_a^-=d(e^{f+\bar{f}}(\gamma_\ib)_a)\\
&\nabla_\tau S_a^-=\frac{1}{\tau}d(e^{f+\fb}(\gamma_\tau)_a),\;
\nabla_\taub S_a^-=\frac{1}{\taub}d(e^{f+\fb}(\bar{\gamma}_\tau)_a).
\end{align}
\end{lm}

Define $\gamma_a=\sum_{i} (\gamma_i)_a dt^i+\sum_{\ib}(\gamma_\ib
)_a
dt^\ib+(\gamma_\tau)\frac{d\tau}{\tau}+(\gamma_\taub)\frac{d\taub}{\taub}\in
\Omega^1(\C^*\times S)\otimes L^2\Lambda^{n-1}(M)$. We have the
following result by the above analysis.

\begin{thm}\label{thm:S-clos-form} $S^-_a$ is a closed $n$-form and integrable over Lefschetz thimble $\cC^-_a(\tau,t)$.
There exists a form $\gamma_a\in \Omega^1(\C^*\times S)\otimes
L^2\Lambda^{n-1}(M)$ such that
\begin{equation}
\D S^-_a=d_z(e^{f+\fb}\gamma_a), a=1,\cdots,\mu.
\end{equation}
Similar results hold for $S^+_a$.
\end{thm}

By the above theorem, we obtain the covariant derivatives with
respect to the position $a$,
\begin{equation}
\D \Pi^-_{ca}(\tau,t)=0, \;c=1,\cdots,\mu,\;a=1,\cdots,\mu,
\end{equation}
where the matrix $\Pi^-=(\Pi^-_{ca}(\tau,t))_{\mu\times \mu}$
well-defined on $S_r$ is the period matrix which was defined before.

Now for fixed $(\tau_0,t_0)\in S_r$ and $\cC^-_c(\tau_0,t_0)$, the
$\mu$ vectors
$\Pi^-_{a}(\tau,t)=(\Pi^-_{1a}(\tau,t),\cdots,\Pi^-_{\mu
a}(\tau,t))^T,a=1,\cdots,\mu$ are horizontal sections of the bundle
$\ch_{\ominus,top}$ and generate a basis of the $\C$-coefficient
local system in $\ch_{\ominus,top}$. Furthermore, we can choose
$S^-_a$ such that $\Pi^-_{a},a=1,\cdots,\mu$ are integral vectors in
$\C^\mu$, i.e., forms an integral lattice in $\ch_{\ominus,top}$.

The period matrix $\Pi^-(\tau,t)$ has the expression:
\begin{equation}
\Pi^-=(\Pi^-_1,\cdots,\Pi^-_\mu)
\end{equation}
and is the fundamental solution matrix of the equation:
\begin{equation}\label{eq:flat-1-2}
\D y(t)=0.
\end{equation}

Such identity also implies the following important conclusion:

\begin{thm}\label{thm:true-flat} The Cecotti-Vafa equations and the Fantastic equations
hold on $\ch_{\ominus,top}\to \C^*\times S_m$ and the connection
$\D$ is flat.
\end{thm}

\begin{proof} Applies the operators in the left hand side of those
formulas of Theorem \ref{thm:CV-equa} and Theorem \ref{thm:Fan-eq}
to the frame consisting of the sections $\Pi^-_{c}(\tau,t)$, then
one can get the conclusion. $\D^2=0$ is a natural conclusion.
\end{proof}

\subsubsection{\underline{Identification of the connection $\D_{top}$ and the connection $\D$ }}\

\begin{lm}\label{lm:GM=topGM} Assume that $(\tau,t)\in S_r$. Let $S^-_a$
be the same as in Theorem \ref{thm:S-clos-form} and
$\cC^-_f,f=1,\cdots,\mu$ be the dual cycle of $S^-_a$. For any
$a=1,\cdots,\mu$, we have
\begin{equation}
d\langle \cC^-_f(\tau,t),S_a^-(\tau,t)\rangle=0=\langle
\cC^-_f(\tau,t), \D (S_a^-(\tau,t))\rangle.
\end{equation}
\end{lm}

\begin{proof} Since $\cC^-_f(\tau,t)$ be the dual cycle, we have,
$$
\langle \cC^-_f(\tau,t),S_a^-(\tau,t)\rangle=\delta_{fa},
$$
and so
$$
d\langle \cC^-_c(\tau,t),S_a^-(\tau,t)\rangle=0.
$$
On the other hand, we have
$$
\langle \cC^-_f(\tau,t), \D (S_a^-(\tau,t))\rangle=0,
$$
and the conclusion is proved.
\end{proof}

\begin{thm} The connection $\D$ can be identified with the topological Gauss-Manin connection $\D_{top}$
over $S_r$ of the dual bundle $\ch_{\ominus,top}$ of $H^\ominus$.
\end{thm}

\begin{proof} Since any local section $S(t)$ of $\ch_{\ominus,top}$ can
be written as the linear combination of the product terms
$g(\tau,t)S_a^-(\tau,t)$, by Lemma \ref{lm:GM=topGM} it is easy to
get
\begin{equation}\label{eq:GM-topo-1}
d\langle \cC^-_f(\tau,t),S(\tau,t)\rangle=\langle \cC^-_f(\tau,t),\D
S(\tau,t)\rangle.
\end{equation}
Now if $\cC^-_f(\tau,t)$ is a horizontal section w.r.t.$\D_{top}$ of
the bundle $H^\ominus$ and $S^-_a$ are the dual forms, then we can
prove in the same way as in Lemma \ref{lm:GM=topGM} that
\begin{equation}\label{eq:GM-topo-2}
d\langle \cC^-_f(\tau,t),S^-_a(\tau,t)\rangle=\langle
\D_{top}\cC^-_f(\tau,t), S^-_a(\tau,t)\rangle,
\end{equation}
and furthermore, the above equality holds for any section
$\cC^-_f(\tau,t)$ of $H^\ominus$. Combining the two equalities
(\ref{eq:GM-topo-1}) and (\ref{eq:GM-topo-2}), we actually obtain
$$
d\langle \cC^-(\tau,t),S^-(\tau,t)\rangle=\langle
\D_{top}\cC^-(\tau,t),S^-(\tau,t)\rangle +\langle \cC^-(\tau,t),\D
S^-(\tau,t)\rangle.
$$
for any local sections $\cC^-(\tau,t)$ of $H^\ominus$ and
$S^-(\tau,t)$ of $\ch_{\ominus,top}$. This proves the conclusion.
\end{proof}

\subsubsection{\underline{Holomorphic structure on
$\ch_{\ominus,top}$}}\

\

Since the $(0,1)$ part of $\D$ is integrable, by the celebrated
Koszul-Margrange integrability theorem (see \cite{DK}, Chapter 2),
we know that $\ch$ has a holomorphic structure and the holomorphic
structure is determined by the $(0,1)$ part $\D^{0,1}$ of $\D$.

The identity $[D_i,D_{\bar{j}}]=-[B_i,B_\jb]$ shows that the
curvature $D_t:=D$ along the deformation direction is
\begin{equation}\label{ident-curv-1}
F^D_{i\bar{j}}=-[B_i,B_\jb].
\end{equation}
Hence usually $D$ is not a flat connection on the Hodge bundle
$\ch$. Similarly the equation (8) and (5) in Fantastic equations
show that $D_\tau$ is not flat and the $t$-direction and the $\tau$
direction has nontrivial curvature.

\subsubsection{\underline{Correspondence}}\

We can think of the family of the closure of each chamber of $S_r$
in $S_m^1$ as an atlas covering $S_m^1$. If we fix a basis
$\{\cC^-_c\}$ of $H^\ominus$ in this chamber, the horizontal
sections $\Pi^-_c$ whose component $\Pi^-{ca}$ are the integration
of $S^-_a$ over the fixed cycle $\cC^-_c$ are well defined over the
closure of each chamber and form a trivialization of the bundle
$\ch_{\ominus,top}$. The transformation between two chambers is
given by the Picard-Lefschetz transformation formula.

Now the bundle $\ch_{\ominus,top}$ equipped with the data
$(S^-_a,D+\bar{D},\heta,\tau_R,\hstar)$ can be identified with the
bundle $\ch_{\ominus,top}$ with the data
$(\Pi^-_a,\D,\heta,\tau_R,\hstar)$. The correspondence is given by
\begin{align*}
S^-_a&\Leftrightarrow \Pi^-_a\\
D+\bar{D}&\Leftrightarrow \D=\pat-(B+\frac{\U}{\tau})+\bpat-(\Bb+\frac{\Ub}{\taub})\\
\tau_R& \Leftrightarrow \tau_R\\
\hstar S^-_a=S^+_a&\Leftrightarrow \hstar \Pi^-_a:=\Pi^+_a\\
\heta(S^-_a,S^-_b)&\Leftrightarrow \heta(\Pi^-_a,\Pi^-_b):=\langle
\Pi^-_a, \hstar \Pi^-_b\rangle_I:=(\Pi^-_a)^T\cdot(I_W^T)^{-1}\cdot
(\Pi^+_b).
\end{align*}
If $\alpha_a,\alpha_b$ correspond to $S^-_a,S^-_b$ and to
$\Pi^-_a,\Pi^-_b$, then it is easy to check that
\begin{equation}
\heta(\Pi^-_a,\Pi^-_b)=\heta(S^-_a,S^-_b)=\eta(\alpha_a,\alpha_b).
\end{equation}
We also have the formula for the pairings:
\begin{equation}
\langle S^-_a, \hstar S^-_b\rangle=\langle \Pi^-_a,\hstar
\Pi^-_b\rangle_I.
\end{equation}
There is analogous result for $\ch_{\oplus,top}$.

Therefore, we have the following conclusion:

\begin{thm}\label{thm:correspondence} The correspondence between
$(\ch_{\ominus,top},S^-_a,D+\bar{D},\heta,\tau_R,\hstar)$ and
$(\ch_{\ominus,top},\Pi^-_a,\D,\heta,\tau_R,\hstar)$ is an
isomorphism between two real Hermitian bundles with the associated
Hermitian connections.
\end{thm}

\begin{rem} The data
$(\ch_{\ominus,top},\Pi^-_a,\D,\heta,\tau_R,\hstar)$ are
well-defined on $S_m^1\subset S$. However, since the fiber
$\ch_{\ominus,top}$ at $(\tau,t)$ is the dual space of $H^\ominus$,
the bundle $\ch_{\ominus,top}$ can be naturally extended over $S$ by
taking the dual space of $H^\ominus_{(\tau,t)}$ for $(\tau,t)\in
S-S_m^1$.
\end{rem}

\subsection{Deformation of flat connections}\label{subsec:4.4}

At first we discuss the deformation of flat connections of the
bundle $\ch_\ominus$ and then the deformation can be transferred to
the bundle $\ch_{\ominus,top}$ with base connection $\D$. All the
conclusions hold if we replace the connection $D'$ by $\pat$.

\subsubsection{\underline{Holomorphic frame of $D$}}\

By Remark \ref{rem:holo-decomp-form}, if $(\tau,t)\in S_m$, then we
can use the isomorphism $i_{0h}:\ch_{\ominus}\simeq
\Omega^n/df\wedge \Omega^{n-1}$ to construct a basis
$\alpha_a(\tau,t)$ of $\ch^{(\tau,t)}_\ominus$ such that
\begin{equation}
\alpha_a=i_{0h}(\alpha_a)+\bpat_f R_a,
\end{equation}
where $i_{0h}(\alpha_a)$ is holomorphic and uniquely determined by
$\alpha_a$. We call this basis as a holomorphic frame of the
connection $D$. Obviously, we have
\begin{equation}
\pat_\ib \alpha_a=\bpat_f (\pat_\ib R_a),\;\pat_\taub
\alpha_a=\bpat_f (\pat_\taub R_a).
\end{equation}
This shows that for holomorphic frame of $D$, the following
Christoffel symbols vanish:
\begin{equation}
(\Gamma_i)_{\ab \bb}=(\Gamma_\ib)_{ab}=(\Gamma_\taub)_{ab}=0.
\end{equation}
So the connection $D$ has the form $D=D'+\bpat$ in a holomorphic
frame $\{\alpha_a\}$ of the Hodge bundle $\ch$. The Cauchy-Riemann
operator $\bpat$ can be decomposed into the sum along $\tau$ and $t$
directions: $\bpat=\bpat_t+\bpat_\tau$. In this frame, we have
$$
\D=D'+\bpat-B-\frac{\U}{\tau}-\Bb-\frac{\Ub}{\taub}.
$$
Correspondingly we have the expressions on the topological bundle
$\ch_{\ominus,top}$. The Cecotti-Vafa's equations and Fantastic
equations also hold in this holomorphic frame.\

\subsubsection{\underline{Gauge Transformation}}\

In holomorphic frame, Equation (1) in Cecotti-Vafa's equation and
Equation (1) in the Fantastic equations have the following form
\begin{equation}\label{eq:gauge-0}
\begin{cases}
\bpat_t \Bb=0\\
\bpat_\tau \Bb=-\frac{1}{\taub}\bpat_t \Ub.
\end{cases}
\end{equation}

\begin{prop}\label{prop:gauge-exchan}
 The following matrix equation:
\begin{equation}\label{eq:gauge-1}
\begin{cases}
\bpat A=(\Bb+\frac{\Ub}{\taub})\\
D'A=0
\end{cases}
\end{equation}
has a local solution $A\in \End(\ch)$ in $\C^*\times S$ satisfying
the initial condition at a point $p_0=(\tau_0,t_0)$:
\begin{equation}
A(p_0)=Id.
\end{equation}
and furthermore, $A$ can be chosen to be symmetric w.r.t.the metric
$\eta$, i.e., $A=A^*$.
\end{prop}

\begin{rem} for any number $k_0\in \C$, $A+k_0 Id$ is also a
solution of (\ref{eq:gauge-1}).
\end{rem}

\begin{proof} Since the R.H.S. of the first equation satisfies the
integrable condition:
$$
\bpat(\Bb+\frac{\Ub}{\taub})=\bpat_t\Bb+\bpat_\tau\Bb+\frac{1}{\taub}\bpat_t
\Ub=0,
$$
where we used the equation (\ref{eq:gauge-0}). Plus the second gauge
equation and by Cauchy-Kovalevski theorem, we have the conclusion.

Due to fact that the connections $D',\bpat$ are metric connection
w.r.t. $\eta$ or $g$ and the operator $\Bb+\frac{\Ub}{\taub}$ is
symmetric w.r.t $\eta$, we can choose the solution $A$ to be
symmetric.
\end{proof}

\begin{lm}\label{eq:comm-rela-1} The solution $A$ in Proposition
\ref{prop:gauge-exchan} satisfies the commutation relations:
\begin{equation}
[\bpat A,A]=[\Bb+\frac{\Ub}{\taub},A]=0.
\end{equation}
\end{lm}

\begin{proof} By the CV and Fantastic equations, we can
easily check that
$$
\bpat([\Bb+\frac{\Ub}{\taub},A])=D'([\Bb+\frac{\Ub}{\taub},A])=0.
$$
Hence $[\Bb+\frac{\Ub}{\taub},A]$ is a constant, but by
normalization condition $A(p_0)=Id$, we know that it must be zero.
\end{proof}

\begin{crl}\label{crl:gauge-1} Let $A^\dag$ be the conjugate operator of $A$ with
respect to the metric $g$. Then we have
\begin{equation}\label{eq:gauge-crl-1}
\begin{cases}
D'A^\dag=B+\frac{\U}{\tau}\\
\bpat A^\dag=0,
\end{cases}
\end{equation}
and
\begin{equation}
(A^\dag)^*=A^*,\;[A^\dag, B+\frac{\U}{\tau}]=0.
\end{equation}
\end{crl}

\begin{proof} Taking the conjugate of Equation (\ref{eq:gauge-1}),
we have
$$
(\bpat A)^\dag=(\Bb+\frac{\Ub}{\taub})^\dag=B+\frac{\U}{\tau}.
$$
Compute the matrix $(\bpat A)^\dag$. We have
\begin{align*}
&g((\bpat
A)^\dag\alpha_a,\alpha_b)=g(\alpha_a, (\bpat A)\alpha_b)\\
=&g(\alpha_a, \bpat (A\alpha_b))-g(A^\dag\alpha_a, \bpat \alpha_b)\\
=&\pat g(A^\dag\alpha_a,\alpha_b)-g(D'\alpha_a,A\alpha_b)-\left(\pat
g(A^\dag\alpha_a,\alpha_b)-g(D' (A^\dag \alpha_a),\alpha_b)
\right)\\
=&g(D'(A^\dag \alpha_a)-A^\dag (D'\alpha_a),\alpha_b)\\
=&g([D', A^\dag]\alpha_a,\alpha_b).
\end{align*}
Therefore, we have
$$
D'A^\dag=B+\frac{\U}{\tau}.
$$
Similarly, we can prove the other equalities.
\end{proof}

\begin{prop}Define the gauge transformation
$S=e^A$. Then
\begin{equation}
S^{-1}\D S=\nabla^G+\bpat,
\end{equation}
where \begin{equation}
\nabla^G=D'-e^{-A}(B+\frac{\U}{\tau})e^A=e^{-A}
(D'-B-\frac{\U}{\tau})e^A.
\end{equation}
is a $(1,0)$-form satisfying:
\begin{equation}
(\nabla^G)^2=0,[\bpat,\nabla^G]=0.
\end{equation}
\end{prop}

\begin{proof} The direct computation shows that
\begin{align*}
S^{-1}\D S=&D'+S^{-1}D'S-S^{-1}(B+\frac{\U}{\tau})S+\bpat+S^{-1}\{\bpat S-(\Bb+\frac{\Ub}{\taub}) S\}\\
=&\bpat+D'+e^{-A}D'e^A-e^{-A}(B+\frac{\U}{\tau})e^A+e^{-A}(\bpat
e^A-(\Bb+\frac{\Ub}{\taub})e^A)\\
=&\bpat+D'-e^{-A}(B+\frac{\U}{\tau})e^A.
\end{align*}
In the above computation, we have implicitly used the commutation
relation $[\bpat A,A]=0$ which was proved in Lemma
\ref{eq:comm-rela-1}. This commutation relation will be used
frequently in the later computation without saying.

The other conclusions are the corollary of the fact $\D^2=0$.
\end{proof}

\subsubsection{\underline{One parameter transformation group of flat
connections}}\

\

Now we consider a family of connections: for $s\in \R$,
\begin{equation}
\nabla^{G,s}=s D'-e^{-s A}(B+\frac{\U}{\tau})e^{s A}.
\end{equation}

\begin{lm} $\nabla^{G,s}$ satisfies:
\begin{equation}
(\nabla^{G,s})^2=0,\;[\bpat, \nabla^{G,s}]=0,\;[D',\nabla^{G,s}]=0.
\end{equation}
\end{lm}

\begin{proof} The first equality is obvious. For the second one, we
have the computation:
\begin{align*}
[\bpat, \nabla^{G,s}]=&[\bpat,e^{-sA}(s D'-B-\frac{\U}{\tau})e^{s A}]\\
=&[\bpat, e^{-sA}](s D'-B-\frac{\U}{\tau})e^{s A}+e^{-sA}[\bpat,(s
D'-B-\frac{\U}{\tau})e^{sA}]\\
=&[\bpat, e^{-sA}](s D'-B-\frac{\U}{\tau})e^{s A}+ e^{-sA}[\bpat, (s
D'-B-\frac{\U}{\tau})]e^{s A}
-e^{-sA}(s D'-B-\frac{\U}{\tau})[\bpat,e^{sA}]\\
=&e^{-sA}\left\{ -s^2 [\Bb+\frac{\Ub}{\taub}, D']+s\left(
[\bpat,D']+[\Bb+\frac{\Ub}{\taub},B+\frac{\U}{\tau}]\right)-[\bpat, B+\frac{\U}{\tau}] \right\}e^{s A} \\
=&0
\end{align*}
We have used CV equations and Fantastic equations in the last
equality.

For the third one, we use the relation $[D',B+\frac{\U}{\tau}]=0$
and the fact that $D'A=0$, we have
$$
[D',\nabla^{G,s}]=e^{sA}[D',s D'-B-\frac{\U}{\tau}]e^{-s A}=0.
$$
So we are done.
\end{proof}

Take the derivative with respect to $s$, one has:
\begin{equation}
\nabla^{\eta,s}:=\frac{\pat}{\pat s}\nabla^{G,s}=e^{-s A}\{D'+[A,
B+\frac{\U}{\tau}]\}e^{s A}.
\end{equation}

Note that
\begin{equation}
\nabla^\eta:=\nabla^{\eta,0}=D'+[A, B+\frac{\U}{\tau}].
\end{equation}
is a $(1,0)$-type connection.

Using this connection, the other connections at time $s$ can be
represented by means of the adjoint action and the corresponding
integrals.

\begin{prop}\label{prop:conn-eta} We have the formulas:
\begin{align}
&\nabla^{\eta,s}=\text{Ad}(e^{-s A})(\nabla^\eta)=e^{-s \text{ad}(A)}(\nabla^\eta).\\
&\nabla^{G,s}=-(B+\frac{\U}{\tau})+\int^s_0
(e^{-s\adj(A)})(\nabla^\eta) ds.
\end{align}
and
\begin{align}
&\nabla^\eta=D'-D'([A^\dag,A])\\
&\nabla^{G,s}=sD'-D'(e^{-s \adj(A)}(A^\dag)).
\end{align}
\end{prop}

\begin{thm}\label{thm:Saito-stru} The connection $\nabla^\eta$ satisfies the following identities:
\begin{align}
&[\bpat, \nabla^\eta]=0,\;\nabla^\eta \eta=0,\;\bpat
\eta=0,\;\nabla^\eta (B+\frac{\U}{\tau})=0,\;(\nabla^\eta)^2=0\\
&\nabla^\eta
A^\dag-[B+\frac{\U}{\tau},[A^\dag,A]]-(B+\frac{\U}{\tau})=0\label{eq:conn-eta-2}\\
\end{align}
\end{thm}

\begin{proof} Since $g$ and the real structure is parallel to the
connection $D$, we have $(D'+\bpat)\eta=0$, which is equivalent to
$D'\eta=0,\bpat\eta=0$.

For each $s$, we have
$$
(\nabla^{G,s})^2=0,\;[\bpat, \nabla^{G,s}]=0.
$$
Take the derivative to the above identities at $s=0$, there is
$$
\nabla^\eta (B+\frac{\U}{\tau})=0, \;[\bpat, \nabla^\eta]=0.
$$

Since
$$
[A,B+\frac{\U}{\tau}]^*=[B^*+\frac{\U^*}{\tau},A^*]=[B+\frac{\U}{\tau},A]=-[A,B+\frac{\U}{\tau}],
$$
this together with the fact $D'\eta=0$ shows that
$\nabla^\eta\eta=0$.

Finally, we want to prove $(\nabla^\eta)^2=0$. We take the
derivatives twice to $(\nabla^{G,s})^2=0$, then
$$
[\frac{d}{ds}\nabla^{G,s},\frac{d}{ds}\nabla^{G,s}]=-[\nabla^{G,s},
\frac{d^2}{ds^2}\nabla^{G,s}].
$$
Take $s=0$, then we have
\begin{equation}\label{eq:gauge-2}
(\nabla^\eta)^2=-[B+\frac{\U}{\tau},[A,\nabla^\eta]].
\end{equation}
Compute the right hand side,
\begin{align*}
-[B+\frac{\U}{\tau},[A,\nabla^\eta]]&=-[[B+\frac{\U}{\tau},A],\nabla^\eta]+[[B+\frac{\U}{\tau},\nabla^\eta],A]\\
&=[\nabla^\eta,\nabla^\eta]=2(\nabla^\eta)^2,
\end{align*}
where we used the facts that $[D',\nabla^\eta]=0=\nabla^\eta
(B+\frac{\U}{\tau})$. Replacing the above equality into
(\ref{eq:gauge-2}), we obtain
$$
(\nabla^\eta)^2=0.
$$
To prove Equation (\ref{eq:conn-eta-2}), we use Corollary
\ref{crl:gauge-1} and Proposition \ref{prop:conn-eta} as follows.
Since
$$
\D^{1,0}=\nabla^\eta-[A, B+\frac{\U}{\tau}]-(B+\frac{\U}{\tau}),
$$
we have
$$
\D^{1,0}A^\dag=D'A^\dag-[B+\frac{\U}{\tau},A^\dag]=B+\frac{\U}{\tau}-[B+\frac{\U}{\tau},A^\dag].
$$
Hence
$$
[\nabla^\eta-[A,B+\frac{\U}{\tau}]-(B+\frac{\U}{\tau}),A^\dag]=B+\frac{\U}{\tau}-[B+\frac{\U}{\tau},A^\dag].
$$
This gives Equation (\ref{eq:conn-eta-2}).
\end{proof}

Define
\begin{equation}
\eta^s=(e^{-s A})^*\eta=\eta(e^{s A}\cdot,e^{s
A}\cdot),\;B^s=\text{Ad}(e^{-s A})(B),\;\U^s=\text{Ad}(e^{-s
A})({\U}).
\end{equation}

It is easy to see the following result hold:
\begin{prop} For any $s\in \R$, there is
\begin{equation}
[\bpat, \nabla^{\eta,s}]=0,\;\nabla^{\eta,s} \eta^{s}=0,\;(\bpat-s
(\bar{B}+\frac{\Ub}{\taub})) \eta^{s}=0,\;\nabla^{\eta,s}
(B^s+\frac{\U^s}{\tau})=0,\;(\nabla^{\eta,s})^2=0.
\end{equation}
\end{prop}

\begin{proof}The second, fourth and fifth equalities are the result
of composite computation. The third one is easy to see. The first
one comes from the identity:
$$
[D',\bpat]=-[B+\frac{\U}{\tau},\bar{B}+\frac{\Ub}{\taub}].
$$
\end{proof}

\subsubsection{\underline{Holomorphic structure of $(\ch_{\ominus,top},
\D)$}}\

\

Now we can do gauge transformation to the connection $\D$ of
$\ch_{\ominus,top}$ by replacing the frame $\Pi^-_a$ by
$\Pi^{-,h}_a:=e^{-A}\cdot\Pi^-_a$. In this frame, the connection has
the decomposition:
$$
\D=\nabla+\bar{\nabla}=\nabla^G+\bpat.
$$
and $e^{-A}\cdot\Pi^-_a, a=1,\cdots,\mu$ are the solutions of the
following equations:
\begin{equation}
\nabla^G y(t)=0,\; \bpat y(t)=0.
\end{equation}
The fundamental solution matrix is $\Pi^{-,h}:=e^{-A}\cdot
\Pi^-=(e^{-A}\cdot \Pi^-_1,\cdots,e^{-A}\cdot \Pi^-_\mu)$.

In summary, we have the following conclusion:

\begin{thm}\label{thm:hori-base} Given a holomorphic frame $\{\alpha_a,a=1,\cdots,\mu\}$
of the Hodge bundle $\ch$ with respect to the connection
$D+\bar{D}$, there is the correspondent frame
$\{S^-_a=e^{f+\fb}\alpha_a,a=1,\cdots,\mu\}$ of $\ch_{\ominus,top}$.
The integration of $S^-_a$ over the Lefschetz thimbles induces the
horizontal frame $\{\Pi^-_a,a=1,\cdots,\mu\}$ of the bundle
$(\ch_{\ominus,top},\D)$. The frame $\{\Pi^{-,h}_a,a=1,\cdots,\mu\}$
is the holomorphic and horizontal frame of the bundle
$(\ch_{\ominus,top}, e^{-A}\cdot \D\cdot e^A=\nabla^G+\bpat)$. In
particular, the fundamental solution matrix $\Pi^{-,h}(\tau,t)$
gives a holomorphic mapping
\begin{equation}
\Pi^{-,h}: S_r\to GL(\mu,\C).
\end{equation}
\end{thm}

\subsection{Monodromy, Picard-Fuchs equation and mixed Hodge
structure}\label{subsec:4.5}

\

This section will discuss the global behavior of the bundle
$\ch_{\ominus,top}$ with flat connection $\D$ over the deformation
space $\C^*\times S$. There are two ways to look at this bundle: The
first way is that we can fix the parameter $\tau\in \C^*$ and
discuss the restricted bundle $\ch_{\ominus,top}\to S$ with the
corresponding structure; the second way is that we can fix the
deformation parameter $t\in S$ and consider the bundle
$\ch_{\ominus,top}\to \C^*$. The second way is naturally related to
the meromorphic bundle over $\IP^1$ and the Riemann-Hilbert-Birkhoff
problem. All we want to understand is the monodromy operator, which
will give us the grading of the bundle, i.e., the mixed Hodge
structure of $\ch_{\ominus,top} \to \C^*\times S$.

Above all we take the first consideration and fix the parameter
$\tau\in \C^*$ and consider the restricted bundle
$\ch_{\ominus,top}\to S$.

Let $i_1:\C^*\hookrightarrow\C^*\times S,\;i_2:S\hookrightarrow
 \C^*\times S$ be the inclusion map. We consider firstly the
 pull-back structure of the bundle $\ch_{\ominus,top}\to \C^*\times S$ to $S$ by
 $i_2$. So we fix $\tau=\tau_0$ and denote by $f_t:=f_{(\tau_0,t)}$.

Now $\ch_{\ominus,top}\to S$ is a flat bundle and the integrable
connection $\D$ is the natural extension of the topological
Gauss-Manin connection $\D_{top}$ defined on $\ch_{\ominus,top}\to
S_m^1$. The singular complimentary set $S-S_m^1$ is of complex
codimension $1$. We want to study the limit behavior of the
geometrical structure around the singular set.

Note that the set $S-S_m^1$ is the disjoint union of two sets
$S-S_m$ and $S_m-S_m^1$. The front set contains more singular points
$t$ that $f_t$ has degenerate critical points. The rear set contains
the Morse critical points which are the intersection points of "two
walls".

Because the monodromy transformation $T$ is the same when restricted
to the complex $1$-dimensional space $l$ corresponding to one
deformation direction except the direction of constant term in $f_t$
(See Theorem 4.43 of \cite{{Zo}}). We can only consider the
deformation along this direction.

Since we are discussing the local deformation, we can assume without
loss of generality that $0\in S-S_m^1$ and let $f_t=f_0+tz_1$
(otherwise, we can do a linear coordinate transforation). Choose a
regular point $t_0$ as the base point on $l$. Since $f_t$ is
fundamentally tame on $M$, we can choose a loop $\gamma\subset
S_m^1\cap l$ based on $t_0$ such that the interior of
$f_t^{-1}(\gamma)$ contains the whole critical points of $f_t$ for
small $t\in l$.

Since $t_0$ is regular, we can choose integral basis in
$\ch_{\ominus,top,t_0}$, i.e., the dual pairing of the elements in
this basis with the Lefschetz thimbles are integers. In the integral
basis, the monodromy matrix $T=(h^*_\gamma)^{-1}$ has integer
entries.

We can also choose the integral basis to be holomorphic. So in
holomorphic frame, the period matrix is holomorphic and satisfies
the following equation:
\begin{equation}
\nabla^G \Pi^{-,h}=(\frac{d}{dt}+ \Gamma(t))\Pi^{-,h}=0.
\end{equation}
Consider the equation satisfied by the parallel section $\phi(t)$,
\begin{equation}
\frac{d \phi}{dt}+ \Gamma(t)\phi=0.
\end{equation}
Because of the instanton exchange phenomenon, the solutions are
multiple valued sections. One of the fundamental matrix of solutions
of this system is the periodic matrix, since the period matrix
$\Pi^{-,h}$ is non degenerate.

The matrix representation of the monodromy operator $T$ is uniquely
determined by a basis in $\ch_{\ominus,top,t_0}$. Now we fix a basis
of $\ch_{\ominus,top,t_0}$ and write
$$
T=e^{2\pi iR}.
$$
Notice that the matrix function $t^R$ has the same monodromy of the
fundamental matrix $\Pi^{-,h}$ of solutions. So
$$
\Pi^{-,h}=Z(t)t^R,
$$
such that $Z(t)$ is a single valued function (with possible
singularities at $t$) on $S_m$. Let $C$ be a constant matrix, then
we have
$$
\Pi^{-,h} C=Z(t)C t^{C^{-1}RC}.
$$
Therefore, if we choose the suitable initial basis $C$ at $t_0$, the
matrix $R$ can be reduced to the Jordan normal form. At first, this
Jordan norm form can be divided into a big block diagonal matrix and
the number of the big blocks is exactly the number of isolated
critical points, i.e., the global monodromy matrix can be reduced to
the local monodromy group $T_i$ around the critical point $z_i$.
Secondly the local monodromy group corresponding to a critical point
is still a block matrix and each block is a Jordan standard block
(still denoted by) $R_i$ with all the diagonal entries to be
$\alpha_i$, the entries in the lower diagonal are $1$ and other
entries vanish. Let $\alpha_{ii_j}$ be the eigenvalue of one Jordan
block $R_{ii_j}$ of the local monodromy group $T_i$. Then
$$
R_{ii_j}=\alpha_{ii_j} I+N_{ii_j},
$$
where $N_{ii_j}$ is the corresponding nilpotent matrix. In summary,
for a local monodromy $T_i$, there is the matrix decomposition:
\begin{equation}
T_i=e^{2\pi R_i}=\lambda_i e^{2\pi iN_i}=\lambda_i T^u_i,
\end{equation}
where $\lambda_i=\diag({e^{2\pi i\alpha_{ii_1}},\cdots,e^{2\pi
i\alpha_{ii_l }}})$ is the eigenvalue of $T_i$ and $T_i^u$ is the
unipotent part of the monodromy:
\begin{equation}
N_i=\frac{1}{2\pi i}\log T_i^u.
\end{equation}

\begin{thm}[Regularity theorem] The periodic matrix is uniformly
bounded near the critical points $t_i \in l$. In particular, if
$0\in S-S_m^1$, then the periodic matrix is uniformly bounded near
$0\in l$.
\end{thm}

\begin{proof} We can assume that the isolated critical point $t_i$ is the origin.
Due to the decay property of the harmonic forms and the tameness of
$f_0$, it is easy to see that the periodic matrix $\Pi^{-,h}_{a\ib}
$ is uniformly bounded for small $t$.
\end{proof}

Once we have the regularity theorem, the connection matrix of
$\nabla^G$ has the decomposition under a (possible meromorphic)gauge
transformation at $t_i$:
$$
\Gamma_i(t)=\Gamma_{i0}/t+\Gamma_{i1}(t),
$$
where $\Gamma_{i0}$ is a constant matrix. This fact is based on the
singularity study of the coefficient matrix $\Gamma_i(t)$ and the
fundamental matrix of solutions of the Picard-Fuchs type equation.
The detail proof can be found in section 7.3-7.7 of \cite{Ku}.

Therefore we can define the residue of the connection $\nabla^G$ at
$t_i$ as
$$
\Gamma_{0i}=:\res_{t_{i}}\nabla^G.
$$

The proof of the following two theorems are the same to that in
\cite{Ku}.

\begin{thm} The local monodromy group $T_i$ around $t_i$ can be extended
to the point $t_i$ such that $\lim_{t\to t_i} T_i(t)=T_{i,\inf}$,
where
\begin{equation}
T_{i,\inf}=e^{-2\pi i\res_{t_i}\nabla^G}.
\end{equation}
\end{thm}

Here $T_{i,\inf}$ is called the infinitesimal monodromy
transformation at $t_i$.

For local monodromy transformation $T_i$ around a point $t_i$ such
that $f_{t_i}$ has only isolated critical point, we have the
following monodromy theorem.

\begin{thm}\
\begin{enumerate}
\item $T_i$ can be reduced to the Jordan normal form, and the
eigenvalue of each Jordan block is the root of unity, i.e., there is
a integer $N$ such that $T^N_i$ is unipotent.

\item The size of the Jordan block does not exceed $n+1$, i.e., $(T_i^N-I)^{n+1}=0$.
\end{enumerate}
\end{thm}

\begin{proof} The proof is the same to the monodromy theorems
appeared in the deformation theories, for instance, the deformation
theory of projective varieties or the deformation theory of
singularities (see the comment in \cite{Ku}). The key point is that
the monodromy matrix is unimodular integral matrix.
\end{proof}

\begin{rem} The local monodromy transformation $T_i$ and the
infinitesimal monodromy transformation $T_{i,\inf}$ have the same
eigenvalues, since the characteristic polynomial of $T_i$ has
integer coefficients and does not change as the base point $t_0$
tending to $t_i$.
\end{rem}

Now we study the action of the local or infinitesimal monodromy
groups.

\begin{lm} Let $C_h$ be a chamber in $S$ and $t_0\in C_h$. Then
$$
T_{t_0,inf}=Id.
$$
\end{lm}

\begin{proof} Since $C_h$ is open, there is open set $U_0\ni t_0$.
It suffices to prove that for any loop $\gamma\subset U_0$, the
monodromy transformation along $\gamma$ is the identity.

Fix a basis $\{\pi^-_1(0),\cdots,\pi^-_{\mu}(0)\}$ of
$\ch_{\ominus,top,t_0}$ and assume that
$\{\pi^-_1(\tau),\cdots,\pi^-_{\mu}(\tau)\}$ is the parallel
transportation by connection $\nabla^G$ along the loop $\gamma$,
where $\tau\in [0,1]$. Then the monodromy group is given by the
matrix $(\heta(\pi^-_i(0),\pi^-_j(1)))$. If $\pi^-_i(0)$ is chosen
to be the dual cocycle of the Lefschetz thimble $\{\cC_a(0)\},
a=1,\cdots,\mu$, then
$$
\heta(\pi^-_i(0),\pi^-_j(1))=\heta(s^-_i(0),s^-_j(1))=\int_{\cC_i(0)}s^-_j(1).
$$
Since the loop $\gamma$ is in the chamber $C_h$, the moving of the
Lefschetz thimble will not produce instantons, i.e., the Lefschetz
thimble keep their relative homological classes. Therefore, we have
$$
\int_{\cC_i(0)}s^-_j(1)=\int_{\cC_i(1)}s^-_j(1)=\delta_{ij}.
$$
This proves the result.
\end{proof}

The proof of the above lemma also shows that

\begin{crl} If $t_0$ is a point on the wall of chambers, then
the infinitesimal monodromy group $T_{t_0,\inf}$ is nontrivial.
\end{crl}

\subsubsection{\underline{Mixed Hodge structure}}\

\

Now fix a point $t_0\in S_m^1$ and let $T\in \aut(\ch_{t_0})$ be the
monodromy transformation. By monodromy theorem, we can choose a
special horizontal basis of $\ch_{\ominus,top,t_0}$ such that $T$
can be expressed as a Jordan normal matrix:
$$
T=\begin{pmatrix} J_{\mu_1}& & 0\\
 &\ddots&\\
 0&&J_{\mu_s}
\end{pmatrix},\; \text{where}\;
J_{\mu_l}=\begin{pmatrix} \beta_{\mu_l}& &\cdots &0\\
1& \ddots& &\vdots\\
\vdots&\ddots&\ddots&\\
0&\cdots&1&\beta_{\mu_l},
\end{pmatrix}
$$
and $\mu_1+\cdots+\mu_l=n$.

So we have the decomposition $T=e^{2\pi i R}=e^{2\pi i\beta}e^{2\pi
i N}=T_s T_u$, where
$\beta=\diag(\beta_{\mu_1},\cdots,\beta_{\mu_l}), $
$\lambda_i=e^{2\pi i\beta_i}$ is the eigenvalue of $T$, and
$$
N=\frac{1}{2\pi i}\log T_u
$$
is nilpotent matrix. Here the matrix $N$ is uniquely determined by
$T$ and the matrix $\beta$ is uniquely determined modulo a scalar
matrix.

Now the parallel transportation of the connection $\nabla^G$ gives a
splitting of the bundle over $S_m^1$:
$$
\ch_{\ominus,top}=\oplus_{i=1}^l V_i,
$$
where the subbundles $V_i$ correspond to the Jordan block $J_i$ and
is generated by the set of sections:
\begin{equation}
V_i=\{\pi^-| (J_{\mu_i}-\beta_{\mu_i})^{\mu_i}\pi^-=0\}.
\end{equation}

Therefore for any horizontal section $\pi^-$ of $\ch_{\ominus,top}$,
we have the decomposition:
$$
\pi^-=\sum_i \pi^-_i,\:\pi^-_i\in V_i.
$$
If $\pi^-$ is a section of any subbundle $V_i$, then $\pi^-$ is said
to be homogeneous and is denoted by $\pi^-_h$. Its spectrum with
respect to $\beta$ is defined as $\beta(\pi^-_h)$ equals to some
$\beta_{\mu_i}$. We can shift the matrix $\beta$ by a scalar matrix
such that its minimum spectrum point lies in $[0,1)$ and denote by
$\chat$ the maximum spectrum point. Let $\Sp(\beta)$ be the spectrum
set of the matrix $\beta$.

On the other hand, we can consider the monodromy transformation of
the bundle $\ch_{\oplus,top}\to S_m^1$, which is the inverse of $T$.
The corresponding matrices are $-\beta, -N$. Now we consider the
total space $\ch_{tot,top}=\ch_{\ominus,top}\oplus \ch_{\oplus,top}$
and it has a natural Hodge filtration defined as follows. At first,
we define the degree of each homogeneous element $\pi=\pi_-(\pi_+)$
in $\ch_{\ominus,top}$ or $\ch_{\oplus,top}$ by shifting its
spectrum by $\chat$,i.e.,
\begin{equation}
\deg(\pi)=\beta(\pi)+\chat.
\end{equation}
Hence the degree of any homogeneous element lies in the interval
$[0,2\chat]$ and we denote by $\Lambda_{\deg}$ the set of degrees of
all homogeneous element.

\begin{df} The Hodge filtration $F^p \ch_{tot,top},p\in \Lambda_{\deg}$ is a subbundle
of $\ch_{tot,top}$ which is generated by all the homogeneous
horizontal sections $\pi^h$ with $\deg(\pi^h)\le 2\chat-p$.
Therefore, we obtain a decreasing filtration of $\ch_{tot,top}$:
\begin{equation}
0\subset F^{2\chat}\subset\cdots\subset F^0=\ch.
\end{equation}
With the real operator $\tau_R$, the lattice $H_\Z$ generated by the
integral horizontal sections of $\ch_{tot,top}$, the triple $(H_\Z,
\tau_R, F^\bullet \ch_{tot,top})$ is called forming a pure Hodge
structure with weight $\Lambda_{\deg}$.
\end{df}

\begin{rem} In Griffiths' work constructing the variation of
polarized Hodge structure of smooth projective varieties, the Weil
operator has the crucial role which gives an automorphism of the
Hodge bundle with the fixed type. In our case the corresponding
operator is $\hstar$ operator, which does not preserve the bundle
$\ch_{\ominus,top}$ but exchange $\ch_{\ominus,top}$ and
$\ch_{\oplus,top}$. This is why we take account the total space
$\ch_{tot,top}$ (or correspondingly $\ch_{tot}$ of Hodge bundle).
\end{rem}

Since our construction used horizontal sections of $\ch_{tot,top}$,
the filtration naturally satisfies the Griffiths' transversality
condition: $\nabla^G F^p\subset F^{p-1}\otimes \Omega^1(S_m^1)$. The
pairing $\langle \phi,\psi\rangle_I$ is $(-1)^n$ symmetric,
non-degenerate and satisfies the positivity:
$$
\langle \phi,\hstar\tau_R(\phi)\rangle_I\ge 0.
$$
By the duality of the Gauss-Manin connection $\D$ with the
topological Gauss-Manin connection $\D_{top}$, we know that $\D$ is
the metric connection with respect to the nondegenerate form
$\langle\cdot,\cdot\rangle$. We actually proved the following
conclusion:

\begin{thm} The triple $(H_\Z, \D, \langle\cdot,\cdot\rangle,
F^\bullet\ch_{tot})$ forms a polarized variation of Hodge structure
on $S_m^1$ with weight $\Lambda_{\deg}$.
\end{thm}

Since $\nabla^G$ preserves the degree of a homogeneous section, it
commutes with $T_s$ and then with $T_u$. So the commutation of the
nilpotent operator $N$ and $\nabla^G$ allows us define a weight
filtration $W_\bullet$ of the bundle $\ch_{tot,top}$ over $S_m^1$ by
the action of the monodromy group such that
\begin{align*}
&(i)\;N(W_i)\subset W_{i-2},\\
&(ii)\;N^r: Gr_r \to Gr_{-r} \;\text{is an isomorphism for }\;r.
\end{align*}

The following conclusion is obvious:

\begin{thm} $(H_\Z, W_\bullet, F^\bullet)$ forms a mixed Hodge
structure of $\ch_{tot,top}$ over $S_m^1$.
\end{thm}

Sometimes we are only interested in a local deformation space of
$S$, for example, in the chamber $C_h$. when considering the local
cases, the local monodromy group will become smaller. In particular,
if we consider the deformation in the chamber $C_h$, the restricted
Hodge bundle is semi-simple and can be decomposed as the direct sum
of $1$-dimensional line bundles. In this case, the weight filtration
is trivial.

\begin{rem} Once we obtained a polarized VHS, one can define the
period domain in the flag manifold and study the behavior of the
period mapping. For instance, since the transversality holds the
period mapping is a holomorphic map from $S_r$ to the period domain.
We can also study the limit Mixed Hodge structure and etc.

Therefore the concepts and the structures of VHS,period mapping and
period domain has been generalized to much broad area beyond that
has already been seen in complex geometry.

After that a natural problem is to compare the HS we built here with
those already well-known. We will consider those problems in the
forthcoming papers.
\end{rem}

\subsubsection{\underline{Quasi-homogeneous case}}\label{subsub:mono-2}\

Since the monodromy transformation is given by the Picard-Lefschetz
transformation formula for the Lefschetz thimbles, which is
determined by the topology of the manifold $M$, there is no general
formula for the monodromy transformation. However, if $M=\C^n$, then
the Picard-Lefschetz transformation formula is the same to that for
the vanishing cycles (see \cite{E}). We have the following
conclusion:

\begin{prop}\label{prop:VHS-1} Let $f$ be a holomorphic function on $\C^n$ with isolated critical point and
$(\C^n,f_t)$ be a strong deformation of the holomorphic function $f$
based on $S$ such that the global Milnor number of each $f_t$ equals
to that of $f$.  Then the global monodromy group of the
corresponding Hodge bundle is the monodromy group of the singularity
$f:(\C^n,0)\to (\C,0)$.
\end{prop}

\begin{rem} The miniversal deformation of the simple singularities
$A_n, D_n, E6,E_7,E_8$, the hyperbolic singularities $T_{p,q,r}$ and
the parabolic singularities $P_8,X_9,J_{10}$ in Arnold's list
satisfy the requirement in Proposition \ref{prop:VHS-1}.
\end{rem}

If the deformation of a non-degenerate quasi-homogeneous polynomial
$W$ satisfying the requirements in Proposition \ref{prop:VHS-1}, we
have the following properties (see PP. 29 of \cite{Ku}):

\begin{prop} The monodromy group $T$ on the deformation space $S$ is semi-simple and its
eigenvalues are
$$
\lambda_I=e^{2\pi i\alpha_I}, \;I=(i_1,\cdots,i_n)\in
A,\alpha_I=\sum_{j=1}^n q_j(i_j+1).
$$
Here $A$ is the index set of the local algebra $Q_{W,0}$, and
$\alpha_I$ is actually the degree of the element $z^I dz^1\wedge
\cdots\wedge dz^n=z^{i_1}\cdots z^{i_n}dz^1\wedge \cdots \wedge
dz^n$.
\end{prop}

In this case the monodromy operator is semi-simple and the weight
filtration is trivial. We can construct the filtration $F^\bullet$
for the bundle $\ch_{tot,top}$, which is related to the
Schr\"odinger operator $\Delta_W$ as before. The weight is the set
$\Lambda_{\deg}$ with the maximum degree $2\chat=2\sum_i (1-2q_i)$.
If $\chat$ is an integer, then there is an nonempty subset
$\Lambda'_{\deg}\subset \Lambda_{\deg}$ which contains all integral
degrees. The information of the filtration $F^p, p\in
\Lambda'_{\deg}$ was conjectured can be identified with the Hodge
structure on the hypersurface $W=0$ in the projective space
${\mathbb{P}}^{n-1}$ (see \cite{Ce1}). This is just the B side of
Gepner's idea \cite{Ge} about LG/CY correspondence. We will also
check this point in our future study.

\subsubsection{\underline{A family of flat bundles over
$\C^*$}}\label{subsec:4.5.3}

Now we consider the inclusion $i_t:=i_1:\C^*\hookrightarrow
\C^*\times S_m^1$ and the pull-back structure $(i_t^*\ch_{tot,top}
\to \C^*,i_t^*(\D))$. $i_t^*(\D)$ is the induced connection and has
the splitting:
$$
i^*_t(\D)=i_t^*(\nabla^G)+\bpat_\tau.
$$
Since $\nabla^G=\pat-e^{-\adj A}(B+\frac{\U_\tau}{\tau}d\tau)$, we
have
\begin{equation}
i_t^*(\nabla^G)=\pat_\tau-\frac{e^{-\adj
A(\tau,t)}(\U_\tau(\tau,t))}{\tau}d\tau.
\end{equation}

Now the horizontal section $y(\tau,t)$ still satisfies the above
complex o.d.e:
\begin{equation}\label{eq:C*-1}
\frac{dy(\tau,t)}{d\tau}=\left(\frac{e^{-\adj
A(\tau,t)}(\U_\tau(\tau,t))}{\tau}\right)y.
\end{equation}

Since we are interested in the case of $\tau\to\infty$, we change
the coordinate: $\tau=1/s$. So $\pat_\tau=-s^2\pat_s$ and the
equation (\ref{eq:C*-1}) becomes
\begin{equation}\label{eq:hori-ode}
\frac{dy(\tau,s)}{ds}=-\left(\frac{e^{-\adj
A(\tau,t)}(\U_\tau(\tau,t))}{s}\right)y.
\end{equation}

Combining with Theorem \ref{thm:hori-base}, we have the following
interesting conclusion.

\begin{thm} The $\mu$ vectors $\{\Pi^{-,h}_a(1/s,t),a=1,\cdots,\mu\}$
satisfy simultaneously the following two equations:
\begin{equation}\label{eq:hori-two-equa}
\begin{cases}
\nabla^G_t y=0\\
\frac{dy(\tau,s)}{ds}=-\left(\frac{e^{-\adj
A(1/s,t)}(\U_\tau(1/s,t))}{s}\right)y.
\end{cases}
\end{equation}
\end{thm}

Now for fixed $t\in S_m^1(\tau)$, $\tau f_t$ has $\mu$ nondegenerate
critical points and it is possible that for some $\tau$ there are
$3$ critical values having the same imaginary parts. The solutions $
\Pi^{-,h}_a(\tau,t)$ are not well-defined at those points and there
is monodromy around these points. Therefore the equation
(\ref{eq:hori-ode}) has more singular points except $0,\infty$.
Equation (\ref{eq:hori-ode}) gives a monodromy deformation of o.d.e.
The study of such o.d.e. is related to the Riemann-Hilbert problem,
Painlev\'e equations and other integrable systems.

\subsection{Period matrix, Primitive vector and flat coordinate
systems}\label{subsec:4.6}

\

\subsubsection{\underline{Computation of the period matrix}}\

\

At the first glance, it seems that the computation of the period
matrix needs to solve the Schr\"odinger equation, which is almost an
impossible task. However, Under some mild assumption to the section
bundle system $(M,g,f)$ and its deformation, we can give the
explicit formula of the period matrix, equivalently this gives the
formula of the horizontal sections, which is the solution of the
Gauss-Manin connection.

Let $(M,g,f)$ be a strongly tame section-bundle system and $M$ be
stein. Let $(\tau,t)\in S_r$ and consider the deformation
$f_{(\tau,t)}=\tau f_t$.

Remember that we have by Theorem \ref{crl:stein} an isomorphism
$$
i_{0h}:\ch\to \ch_{hol}=\Omega^n/d f_{(\tau,t)}\wedge \Omega^{n-1},
$$
given by
\begin{equation}\label{eq:isom-holo-harm}
\alpha_a=i_{0h}(\alpha_a)+\bpat_{f_{(\tau,t)}} R_a.
\end{equation}
where $R_a$ is a smooth $(n-1)$-form.

At first we prove an important lemma.

\begin{lm}\label{lm:peri-vanish} Let $\{\alpha_a,a=1,\cdots,\mu\}$ be a frame of the Hodge
bundle $\ch$ over an open neighborhood $U$. Let $R$ be a smooth
$n-1$ form in $U$. If at any point $(\tau,t)\in U$, there is
\begin{equation}
\sum_a \int_M |R|\cdot|\alpha_a|<\infty,
\end{equation}
then
\begin{equation}
\int_{\cC^-_a} e^{f+\fb}\bpat_f R=\int_{\cC^-_a}e^{f+\fb}\pat_f R=0.
\end{equation}
In particular, this is true if $R$ has only polynomial growth.
\end{lm}

\begin{proof}It suffices to prove the first identity. We have
\begin{align*}
&\int_{\cC^-_a}e^{f+\bar{f}}\bpat_f R=\int_M PD(\cC^-_a)\wedge e^{f+\bar{f}}\bpat_f R\\
=&\int_M c_{ac} S^+_c\wedge e^{f+\bar{f}}\bpat_f R=\int_M c_{ac} e^{-f-\bar{f}}*\bar{\alpha_c}\wedge e^{f+\bar{f}}\bpat_f R\\
=&\pm\int_M c_{ac}(\bpat_f R,\alpha_c)=0.
\end{align*}
Here we have used the $L^1$ Stokes theorem and the fact that
$\bpat^\dag_f\alpha_c=0$.
\end{proof}

\begin{prop}\label{prop:peri-expre} If at any $(\tau,t)$, there is
$$
\sum_{a,b}\left(|(i_{0h}(\alpha_a),\alpha_b)_{L^2}|+\int_M
|R_a|\cdot|\alpha_a|\right)<\infty,
$$ then we have
\begin{equation}
\Pi^{-,h}_b=(\Pi^{-,h}_{1b},\cdots,\Pi^{-,h}_{\mu
b})^T,\;\Pi^{+,h}_b=(\Pi^{+,h}_{1b},\cdots,\Pi^{+,h}_{\mu b})^T,
\end{equation}
and
\begin{equation}
\Pi^{-,h}_{ab}=e^{-A}\cdot\int_{\cC^-_a}e^{f+\bar{f}}i_{0h}(\alpha_b),\;\Pi^{+,h}_{ab}=(-1)^n
e^A\cdot\int_{\cC^+_a}e^{-f-\bar{f}}* \overline{i_{0h}(\alpha_b)}.
\end{equation}
In particular, this is true if $f$ is a polynomial on $M$.
\end{prop}

\begin{proof} By Equation (\ref{eq:isom-holo-harm}), we have
\begin{align*}
&\int_{\cC^-_a}e^{f+\bar{f}}\alpha_b=\int_{\cC^-_a}e^{f+\bar{f}}i_{0h}(\alpha_b)+
\int_{\cC^-_a}e^{f+\bar{f}}\bpat_f R_b\\
=&\int_{\cC^-_a}e^{f+\bar{f}}i_{0h}(\alpha_b).
\end{align*}
Similarly, we can prove the formula for $\Pi^{+,h}_{ab}$.

If $f$ is polynomial, then $f_{(\tau,t)}$ is polynomial and the
generators in $\Omega^{n}(M)/df_{(\tau,t)}\wedge \Omega^{n-1}(M)$
has only polynomial growth and the function $R_a$ can be chosen has
only polynomial growth.
\end{proof}

Choose $\{\alpha_n\}$ to be a unit frame of the Hodge bundle $\ch$.
We can apply the above formula to the strong tame deformation of a
nondegenerate quasi-homogeneous polynomial.

\begin{crl} Let $(\C^n,W)$ be a strongly tame section-bundle system with $W$
a nondegenerate quasihomogeneous polynomial. Let $W_t,
t=(t_1,\cdots,t_m)$ be a strong deformation of $W$ such that $\pat_i
W,i=1,\cdots,m$ form a $m$-dimensional $\C$-linear vector space in
$\Omega^n/dW\wedge \Omega^{n-1}$, then the periodic matrix have the
following expression:
\begin{equation}
\Pi^{-,h}_i=(\Pi^{-,h}_{1i},\cdots,\Pi^{-,h}_{\mu
i})^T,\;\Pi^{+,h}_i=(\Pi^{+,h}_{1i},\cdots,\Pi^{+,h}_{\mu i})^T,
\end{equation}
and
\begin{align}
&\Pi^{-,h}_{ai}=\tau^{n/2} e^{-A}\int_{\cC^-_a}e^{\tau
W+\overline{\tau W}}\pat_i W dz_1\wedge\cdots\wedge dz_n,\\
&\Pi^{+,h}_{ai}=(-1)^n\taub^{n/2}e^A  \int_{\cC^+_a}e^{-\tau
W-\overline{\tau W}}*\overline{\pat_i W}d\bar{z}_1\wedge\cdots\wedge
d\bar{z}_n.
\end{align}
for $1\le a\le \mu,1\le i\le m$.
\end{crl}

\begin{proof} By Theorem \ref{thm:scal-esti}, the growth order of $\alpha_i$ near the critical point is $n/2$
as $|\tau|\to \infty$, hence we have  $i_{0h}(\alpha_i)=\tau^{n/2}
\pat_i W$, which has only polynomial growth, and since
$\alpha_i(\tau,t)$ is exponentially decaying, we have
$$
\sum_{i,b}|(\tau^{n/2} \pat_i W, \alpha_a)_{L^2}|<\infty.
$$
So we are done.
\end{proof}

Similarly, we have the following conclusion:

\begin{crl} Let $((\C^*)^n, \frac{dz_1}{z_1}\wedge \cdots\wedge \frac{dz_n}{z_n}, \tau f_{t})$
be a strong deformation of nondegenerate and convenient Laurent
polynomials. Let $\pat_i f,i=1,\cdots,m$ form a $m$-dimensional
$\C$-linear vector space in $\Omega^n/df\wedge \Omega^{n-1}$, then
the periodic matrix have the following expression:
\begin{equation}
\Pi^{-,h}_i=(\Pi^{-,h}_{1i},\cdots,\Pi^{-,h}_{\mu
i})^T,\;\Pi^{+,h}_b=(\Pi^{+,h}_{1i},\cdots,\Pi^{+,h}_{\mu i})^T,
\end{equation}
and
\begin{align}
&\Pi^{-,h}_{ai}=\tau^{n/2} e^{-A}\cdot \int_{\cC^-_a}e^{\tau
f+\overline{\tau f}}\pat_i f \frac{dz_1}{z_1}\wedge \cdots\wedge
\frac{dz_n}{z_n}\\
&\Pi^{+,h}_{ai}=(-1)^n \taub^{n/2} e^A\cdot \int_{\cC^+_a}e^{-\tau
f-\overline{\tau f}}*\overline{\pat_i f}
\overline{\frac{dz_1}{z_1}}\wedge \cdots\wedge
\overline{\frac{dz_n}{z_n}},
\end{align}
for $ 1\le a\le \mu,1\le i\le m$.
\end{crl}

\subsubsection{\underline{Primitive vector}}\

\

Assume that  the set of $\{\pat_j f_{(\tau,t)},j=1,\cdots,m\}$,
where $m=\dim S$, forms a partial basis of the algebra in
$\Omega^{n}(M)/df_{(\tau,t)}\Omega^{n-1}(M)$ and corresponds to the
set $\{\alpha_j,j=1,\cdots, m\}$ of holomorphic frame of harmonic
$n$ forms. Then the $m$ vector $\Pi^{-,h}_j,j=1,\cdots,\mu$ can be
expressed by the vector $\Pi^{-,\circ}_1$ as below:

\begin{prop}
\begin{align}
&\Pi^{-,h}_j=\pat_j \Pi^{-,\circ}_1\\
&\tau \pat_\tau \Pi^{-,\circ}_1=e^{-ad(A)}(\U_\tau)\cdot
\Pi^{-,\circ}_1
\end{align}
where
\begin{align}
&\Pi^{-,\circ}_1=(\Pi^{-,\circ}_{11},\cdots,\Pi^{-,\circ}_{\mu 1})^T\nonumber\\
&\Pi^{-,\circ}_{a1}=\tau^{n/2-1} e^{-A}\cdot\int_{\cC_a}e^{\tau f+\taub\fb}dz^1\wedge\cdots\wedge
dz^n.
\end{align}
\end{prop}

\begin{proof}By Proposition \ref{prop:peri-expre} and the fact that $\pat A=0$, we have
\begin{align*}
&\Pi^{-,h}_{cj}=e^{-A}\cdot\int_{\cC^-_c}e^{\tau f+\taub\fb}\pat_j fdz_1\wedge
\cdots \wedge dz_n\\
&=e^{-A}\cdot \pat_j\cdot \int_{\cC^-_c}e^{\tau f+\taub\fb}=\pat_j e^{-A}\cdot
\frac{1}{\tau}\Pi^{-}_{c1}=\pat_j \Pi^{-,\circ}_{c1}.
\end{align*}
Similarly, using the fact that $\pat_\tau A=0$ we can prove the
second identity.
\end{proof}

Analogously, the $m$ vector $\Pi^+_j,j=1,\cdots,\mu$ can be
expressed as the derivatives of another vector $\Pi^{+,\circ}_1$:
\begin{equation}
\Pi^+_j=\pat_j \Pi^{+,\circ}_1,
\end{equation}
where
\begin{align}
&\Pi^{+,\circ}_1=(\Pi^{+,\circ}_{11},\cdots,\Pi^{+,\circ}_{\mu 1})^T\nonumber\\
&\Pi^{+,\circ}_{a1}=(-1)^n\taub^{n/2-1}\int_{\cC_a}e^{-\tau f-\taub\fb}*
dz^\eb\wedge\cdots\wedge dz^\nb.
\end{align}
In particular, we have the identity:
\begin{equation}
\Pi^{-,\circ}_1=\frac{1}{\tau}\Pi^{-,h}_1,\Pi^{+,\circ}_1=\frac{1}{\taub}\Pi^{+}_1
\end{equation}
Furthermore, if $f_{(\tau,t)}$ is a universal deformation of $\tau
f$, then all the horizontal sections $\Pi^{-,h}_a
(\Pi^+_a),a=1,\cdots,\mu$ are generated by the section
$\Pi^{-,\circ}_1(\Pi^{+,\circ}_1)$.

\begin{df} Assume that $f_{(\tau,t)}$ is a universal deformation of $\tau
f$. The vector $\Pi^{-,\circ}_1 (\Pi^{+,\circ}_1)$ is called the
primitive vector (imitating Saito's notation \cite{Sai1}), since the
period matrices can be obtained by taking the Jacobi matrix of the
map $\Pi^{-,\circ}_1$ or $\Pi^{+,\circ}_1$ :
\begin{equation}
\Pi^{-,h}=\frac{\pat(\Pi^{-,\circ}_{11},\cdots,\Pi^{-,\circ}_{\mu
1})}{\pat(t_1,\cdots,t_\mu)},\;\Pi^{+,h}=\frac{\pat(\Pi^{+,\circ}_{11},\cdots,\Pi^{+,\circ}_{\mu
1})}{\pat(t_1,\cdots,t_\mu)}.
\end{equation}
\end{df}

The primitive vector $\Pi^{-,\circ}_1$ gives a $C^\infty$
multiple-value mapping:
\begin{equation}
\Pi^{-,\circ}_1: S_m^1\to \C^\mu.
\end{equation}

Since $\Pi^{+,\circ}_1$ is not independent, which can be obtained by
$\Pi^{-,\circ}_1$, it suffices to consider $\Pi^{-,\circ}_1$. For
convenience, we denote by $\Pi_1:=\Pi^{-,\circ}_1$.

If we identify the holomorphic frame of $\ch_{\ominus,top}$ with the
frame of the holomorphic tangent bundle of $S$ and use the pull-back
metric of the metric $\hat{\eta}$ on $\ch_{\ominus,top}$, then
locally the primitive vector gives a totally geodesic embedding into
the Euclidean space $\C^\mu$, i.e.,
\begin{equation}
\D d \Pi_1=0.
\end{equation}
Globally, we can obtain a map:
$$
\Pi_1: S_m^1\to \C^\mu/T,
$$
where $T$ is the monodromy transformation with respect to our
background frame on $\ch_{\ominus,top}$.

\begin{rem} An example given by $z_1^5+z_2^5+z_1^2z_2^2$ (singularity of type
$T_{2,5,5}$), Page 29 of \cite{Ku}) shows that the monodromy $T$ can
have infinite order.
\end{rem}

\subsubsection{\underline{Flat coordinate system}}\

\

At first we give a lemma on the harmonic coordinate system of a
smooth manifold $M$ with connection $\nabla$.

\begin{df} Let $(x^1,\cdots,x^n)$ be a local coordinate system of
a chart $U\subset M$. If there exist $n$ smooth functions
$u^i=u^i(x^1,\cdots,x^n)$ such that the Jacobi matrix
$J=\frac{\pat(u^1,\cdots,u^n)}{\pat(x^1,\cdots,x^n)}$ does not
vanish on $U$ and satisfy
\begin{equation}
\nabla d u^i=0,i=1,\cdots,n,
\end{equation}
then $(u^1,\cdots,u^n)$ is called a harmonic coordinate system.
\end{df}

The specialty of the harmonic coordinate system is due to the
following known result in geometrical analysis (ref.\cite{Jo}):

\begin{prop}\label{Prop:harm-coord} In harmonic coordinate system, the Christoffel symbols
$\Gamma^i_{jk}$ vanishes and the covariant derivative
$\nabla_i=\pat_i$.
\end{prop}

\begin{proof} A simple proof. $\nabla d
u^a=0$ is equivalent to the identities:
\begin{equation}
\frac{\pat}{\pat t^j}(\frac{\pat u^a}{\pat
t^k})-\Gamma^i_{jk}(\frac{\pat u^a}{\pat t^i})=0,\forall a,j,k.
\end{equation}
The Christoffel symbols under the coordinate system $\{u^a\}$ is
\begin{align*}
&\nabla_{\frac{\pat}{\pat u^k}}\frac{\pat }{\pat u^l}=\frac{\pat
t^i}{\pat u^k}\left(\frac{\pat}{\pat t^i}(\frac{\pat t^p}{\pat
u^l})+(\frac{\pat t^j}{\pat u^l})\Gamma^p_{ij}
\right)\frac{\pat}{\pat t^p}=0.
\end{align*}
\end{proof}

\begin{thm}\label{thm:prim-vect} The primitive form is a biholomorphic map from each chamber of $S_m^1$ to its
image in $\C^\mu$.
\end{thm}

\begin{proof} The complex Jacobian of $\Pi_1$:
$$
\frac{\pat(\Pi_{11},\cdots,\Pi_{\mu1})}{\pat(t^1,\cdots,t^\mu)}
$$
is just the period matrix $(\Pi^{-,h}_{ab})$, which is nondegenerate
by the fact that
$$
I_W=\Pi^{-}\cdot \eta^{-1}\cdot \Pi^{+}.
$$
Hence the primitive vector is a biholomorphic map from each chamber
of $S_m^1$ to its image in $\C^\mu$.
\end{proof}

By the above lemma, we can define the new coordinates in each
chamber of $S_m^1$: $u=(u^1,\cdots,u^\mu)=\Pi_1$.

Now the complex Jacobian $J$ of the coordinate transformation is the
period matrix $\Pi^{-,h}_{ab}$ which is holomorphic in $t$ and the
real Jacobian has the block diagonal matrix:
$$
\begin{pmatrix}
\Pi^{-,h}_{ab}&0\\
0&\overline{\Pi^{-,h}_{ab}}.
\end{pmatrix}
$$
Hence in the $u$-coordinate system, the connection matrix of
$\nabla^G$ vanishes and $\nabla^G=\pat$.

\begin{df} The local holomorphic coordinate systems $u=(u^1,\cdots,u^\mu)$ is called the flat coordinate system
with respect to the connection $\nabla^G$. In fact, we have
constructed a family of flat coordinate systems depending on
$\tau\in \C^*$.
\end{df}

\subsection{Harmonic Higgs bundle and Frobenius
manifold}\label{subsec:4.7}

\subsubsection{\underline{Harmonic Higgs bundle}}\

We will show that the bundle $(\ch_{\ominus,top},\D)$ satisfying the
Cecotti-Vafa's equations and "fantastic" equations will induce many
interesting structures,for example, the structure of harmonic Higgs
bundle

The definition of the harmonic Higgs bundle can be found in Appendix
\ref{app:1}.

Now we know that the data $(\ch_{\ominus,top},\hg,\tau_R)$ forms a
real Hermitian holomorphic bundle. On the other hand, we have

\begin{thm}\label{thm:harm-Higgs} The triple $(\ch_{\ominus,top},\hg, B+\frac{\U}{\tau})$ forms a harmonic Higgs
bundle over $\C^*\times S$. Moreover, if we think $\tau$ as the
parameter, then $(\ch_{\ominus,top},\hg,B)$ forms a family of
harmonic Higgs bundle over $S$.
\end{thm}

\begin{proof} At first we know that $\ch_{\ominus,top}$ has a holomorphic structure
with the Hermitian metric $\hg$ and $\D=\D^{1,0}+\D^{0,1}$ is a
Chern connection with respect to $\hg$. Here $\D^{0,1}$ gives the
holomorphic structure on $\ch_{\ominus,top}$ and the holomorphic
parallel frame is given by $\Pi^-_a,a=1,\cdots,\mu$. Secondly,
$B+\frac{\U}{\tau}: \ch_{\ominus,top}\to \Omega^1(\C^*\times
S)\otimes \ch_{\ominus,top}$ is holomorphic (i.e.,$\bpat
(B+\frac{\U}{\tau})=0$) and satisfies the equation
$(B+\frac{\U}{\tau})\wedge (B+\frac{\U}{\tau})=0$. So
$B+\frac{\U}{\tau}$ is a Higgs field on $\ch_{\ominus,top}$. Finally
the connection $\D$ is a flat connection, which is the conclusion of
the Cecotti-Vafa's equations and the "fantastic equations".
Therefore $(\ch_{\ominus,top}, \hg,B+\frac{\U}{\tau})$ is a harmonic
Higgs bundle.
\end{proof}

\subsubsection{\underline{Constructing the Frobenius manifold
structure via primitive vector}}\

Now we assume that the deformation $f_{(\tau,t)}$ satisfies the
"SUBBUNLE" condition:
\begin{enumerate}
\item The elements $i_{h0}(\pat_j f_{(\tau,t)}), j=1,\cdots,m$ generates a
subbundle $\ch_m$ of $\ch$.

\item $i_{h0}(\pat_1 f_{(\tau,t)})=i_{h0}(\tau dz^1\wedge \cdots \wedge dz^n)=\tau^{-n/2}\alpha_1$.

\item There is a perpendicular decomposition: $\ch=\ch_m\oplus
\ch_m^\perp$ with respect to $g$ such that $\ch_m$ is the invariant
subspace of $\ch$ under the action of all $\pat_j f_{(\tau,t)}$.
\end{enumerate}

If the strong deformation satisfies the above splitting condition,
then we can restrict our discussion on $\ch_m$. In this case, the
number of the deformation parameters equals the dimension of the
space $\ch_m$ of the $n$-harmonic forms. So without loss of
generality, we assume that $m=\mu$ and $\ch_m=\ch$.

In this part, we set $\tau=1$. For simplicity, we ignore the
discussion of the Euler field and leave the related topic in future
paper.

We want to construct the Frobenius manifold structure on the
holomorphic tangent space $TS_m^1$ with flat coordinates
$(u_1,\cdots,u_\mu)$ via the primitive vector. Let us recall the
definition of Frobenius manifold which was introduced by Dubrovin
\cite{Du}. We follow the description of Manin \cite{Ma}(with some
change of the notations) .

\begin{df} Let $S$ be a complex manifold. Consider a triple $(TS,
g,A)$ consisting of an affine flat structure, a metric and a $(0,3)$
type symmetric tensor $A$ such that the following structures hold:
\begin{enumerate}
\item Define an $\O_S$-linear symmetric multiplication
$\circ=\circ_{A,g}$ on $TS$ by:
\begin{equation}\label{eq:Frob-1}
A(X,Y,Z)=g(X\circ Y,Z)=g(X,Y\circ Z).
\end{equation}
If (\ref{eq:Frob-1}) holds, then we say that the metric is invariant
with respect to the multiplication. $(TS, g,A,\circ)$ is called a
pre-Frobenius manifold.

\item If in addition, for any point on $S$ there is a potential function
$F_0$ defined locally such that for any flat local tangent fields
$X,Y,Z$
\begin{equation}
A(X,Y,Z)=(XYZ)F_0,
\end{equation}
then the pre-Frobenius manifold is called potential.

\item A pre-Frobenius manifold is called associative, if the
multiplication $\circ$ is associative.

\item A pre-Frobenius manifold is called Frobenius, if it is
simultaneously potential and associative.

\end{enumerate}
\end{df}

We still assume that the strong deformation satisfies the
requirement in last section.

We have the bundle $\ch_{\ominus,top}$ defined over $\C^*\times
S_m^1$. By Theorem \ref{thm:hori-base} and Theorem
\ref{thm:harm-Higgs}, the bundle $\ch_{\ominus,top}$ has the
following structures:
\begin{enumerate}
\item nondegenerate bilinear symmetric form $\heta$, real structure
$ \tau_R$ and a corresponding Hermitian metric
$\hg=\heta(\cdot,\tau_R\cdot)$;
\item Gauss-Manin connection $\D=\nabla^G+\bpat$ which is the Chern
connection of $(\ch_{\ominus, top},\hg)$.
\item Flat structure given by the horizontal (multivalued) sections
$\{S^-_a,a=1,\cdots,\mu\}$ or $\{\Pi^{-,h}_{\cdot
a},a=1,\cdots,\mu\}$.
\item Higgs field $B=B_i dt^i=(B_i)_{ab}\alpha^a\otimes \alpha^b\otimes
dt^i$ which is holomorphic.
\end{enumerate}

In particular, the pairing $\heta$ and the metric $\hg$ is
holomorphic, since we have the relation $\heta_{ab}=\heta(s^-_a,
s^-_b)=\eta(\alpha_a,\alpha_b)=\eta_{ab}$, and the later is
holomorphic. Since the real structure commutes with $\bpat$, so
$\hg$ is holomorphic also.

Note that $B$ has the following representation:
\begin{equation}
(B_i)_{ab}=\heta(\pat_i f s^-_a, s^-_b)=\eta((\pat_i
f)\alpha_a,\alpha_b)=\int_M (\pat_i f)\alpha_a\wedge *\alpha_b.
\end{equation}

By Theorem \ref{thm:prim-vect}, the primitive vector $\Pi_1$ is a
non-singular holomorphic section of $\ch_{\ominus,top}$:
$$
\Pi_1:S_m^1\to \ch.
$$
We can use it to shift all the structures on $\ch_{\ominus,top}$ to
the holomorphic tangent bundle $TS_m^1$. This is done by defining
the map:
$$
\hat{\Pi_1}: TS_m^1\to \ch_{\ominus,top},
$$
given by
$$
X\to \iota_X(B)\cdot \Pi_1,
$$
where $\iota_X$ means the contraction with $X$.

\begin{prop} If we choose $\Pi_{j}=\pat_j\Pi_1=\Pi^{-,h}_{\cdot,j}$ as the frame
on the bundle $\ch_{\ominus,top}$, then we have
\begin{equation}
\hat{\Pi}_1(\pat_j)=\Pi_j,j=1,\cdots,\mu,
\end{equation}
\end{prop}

\begin{proof} We have
\begin{align*}
&\hat{\Pi}_1(\pat_j)=(B_j)_{1}^a\Pi_{a}=\heta^{al}\heta((\pat_j f)s^-_l,s^-_1)\Pi_{a}\\
=&\heta^{al}
\heta(s^-_l,s^-_j)\Pi_{a}=\heta^{al}\heta_{lj}\Pi_a=\Pi_j.
\end{align*}
\end{proof}

So we can define the structures on the tangent bundle $TS^1_m$:
\begin{equation}
\nabla^{tan}=\hat{\Pi}_1^*(\nabla^G),\;\eta^{tan}=\hat{\Pi}_1^*(\heta),\;e:=(\hat{\Pi}_1)^{-1}(\Pi_1)=\pat_1.
\end{equation}

\begin{lm}
The following conclusions hold:
\begin{enumerate}
\item $\nabla^{\tan}$ is a torsion-free, and flat connection and
satisfies: $\nabla^{tan}\eta^{tan}=0.$
\item $\nabla^{tan}e=0.$
\end{enumerate}
\end{lm}

\begin{proof} We have
$$
(\nabla^{tan})^2=\hat{\Pi}_1^*(\nabla^G\circ\nabla^G)=0,\;\nabla^{tan}\eta^{tan}=\hat{\Pi}_1^*(\nabla^G\heta)=0,
\;\nabla^{tan} e=\hat{\Pi}_1^*(\nabla^G \Pi_1)=0.
$$
On the other hand, we have
\begin{align*}
&\nabla^{tan}_{\pat_j}\pat_k-\nabla^{tan}_{\pat_k}\pat_j=\nabla^G_{\pat_j}(\hat{\Pi}_1(\pat_k))
-\nabla^G_{\pat_k}(\hat{\Pi}_1(\pat_j))\\
=&\nabla^G_{\pat_j}(B_k\cdot \Pi_1)-\nabla^G_{\pat_k}(B_j\cdot
\Pi_1)=\iota_{\pat_j\wedge \pat_k}(\nabla^G B)\cdot \Pi_1+0=0.
\end{align*}
This proves that $\nabla^{tan}$ is torsion-free.
\end{proof}

Now we define the multiplication in the tangent bundle $TS_m^1$:
\begin{equation}
\hat{\Pi}_1(X\circ Y):=(\iota_X B)\hat{\Pi}_1(Y)=(\iota_X B)(\iota_Y
B)\Pi_1.
\end{equation}

\begin{lm} The following conclusions hold:
\begin{enumerate}
\item The multiplication $\circ$ is commutative: $X\circ Y=Y\circ X$.

\item $e$ is the unit field such that $e\circ X=X$.

\item The multiplication is compatible with the metric $\eta^{tan}$:
\begin{itemize}
\item $\eta^{tan}(X\circ Y,Z)=\eta^{tan}(X,Y\circ Z).$
\item $\eta^{tan}(X\circ Y,e)=\eta^{tan}(X,Y)$.
\end{itemize}
\end{enumerate}
\end{lm}

\begin{proof} The commutativity of $\circ$ is due to the fact that
$B\wedge B=0$. We have
$\iota_{\pat_1}B=((B_1)_{ab})=(\delta_{ab})=Id$, which shows that
$e\circ X=X$. The conclusion (3) is easily obtained by using the
fact that $B=B^*$ with respect to $\eta$ and hence $\heta$.
\end{proof}

By the above analysis, we know that $e=\pat_1$ is an unit field
which is parallel with respect to $\nabla^{tan}$ and is the unit of
the multiplication $\circ$.

\begin{lm} $(TS_m^1,\eta_{tan},B,\circ)$ has pre-Frobenius manifold structure.
\end{lm}

Furthermore, we have

\begin{thm} $(TS_m^1,\eta_{tan},B,\circ)$ is potential and associative. Hence
it has Frobenius manifold structure.
\end{thm}

\begin{proof} To prove the associativity, we need to prove that
$$
(\pat_i\circ\pat_j)\circ \pat_k=\pat_i\circ(\pat_j\circ\pat_k),
$$
which is equivalent to the identity:
\begin{equation}\label{eq:Frob-assoc-1}
B^e_{ij}B^c_{ek}=B^e_{jk}B^c_{ei}.
\end{equation}
This is just the commutation relation $[B_i,B_k]=0$. Hence we proved
the associativity.

We can lower down the upper indices in (\ref{eq:Frob-assoc-1}) to
obtain the following equivalent formula:
\begin{equation}\label{eq:Frob-poten-1}
B_{ijc}\eta^{ce}B_{elk}=B_{ljc}\eta^{ce}B_{eik}.
\end{equation}

Now we want to show that $B$ is actually an integrable tensor.

The identity $\nabla B=0$ in flat coordinates
$\{u_1(\tau,t),\cdots,u_\mu(\tau,t)\}$ has the form
$$
\pat_b B_{acd}=\pat_a B_{bcd}.
$$
This is equivalent to
$$
\pat(B_{acd} du^a)=0.
$$
By Poincare lemma, there is a matrix $\hat{A}_{cd}$ such that
$$
B_{acd} du^a=\pat_a(\hat{A}_{cd}).
$$
Since $B_{acd}$ is symmetric in $c$ and $d$, we can define the
symmetric tensor
$$
A_{cd}=\frac{\hat{A}_{cd}+\hat{A}_{dc}}{2}.
$$
$A_{cd}$ also satisfies
$$
B_{acd} du^a=\pat_a(A_{cd}).
$$
Since $B_{acd}=B_{cad}$, we have
$$
\pat_a(A_{cd})=\pat_c(A_{ad}).
$$
This is again equivalent to
$$
\pat(A_{ad}du^a)=0.
$$
By Poincare lemma again, there exist $\Phi_d$ such that
$$
A_{ad}du^a=\pat(\Phi_d)\;\text{or}\;A_{ad}=\pat_a\Phi_d.
$$
Since $A_{ad}=A_{da}$, we have
$$
\pat_a\Phi_d=\pat_d\Phi_a.
$$
This shows that $\Phi_a du^a$ is exact, and by Poincare lemma again
there exists a function $F_0$ such that
$$
\pat_a F_0=\Phi_a.
$$
By our induction, we know that
\begin{equation}
B_{acd}=\pat_a\pat_c\pat_d F_0.
\end{equation}
Hence the pre-Frobenius manifold is potential.

So we have proved that $(TS,\eta, B,\circ)$ is a Frobenius manifold.
\end{proof}

Now the identity (\ref{eq:Frob-poten-1}) is just the WDVV equation.

\begin{equation}
\pat_a\pat_b\pat_c F_0 \eta^{ce}\pat_e\pat_d\pat_h
F_0=\pat_d\pat_b\pat_c F_0\eta^{ce}\pat_e\pat_a\pat_h F_0,
\end{equation}

\appendix

\section{Harmonic Higgs bundle}\label{app:1}

\begin{df}[Harmonic Higgs bundle] Let $E$ be a holomorphic bundle on a complex manifold
$M$, which is equipped with a Hermitian non-degenerate sesquilinear
form $\eta$ (not necessary positive) . We say that $(E,h)$ is a
Hermitian holomorphic bundle. $B$ is said to be a Higgs field if $B$
is a holomorphic $\O_M$-linear morphism $E\to
\Omega_M^1\otimes_{\O_M}E$ satisfying $B\wedge B=0$. $(E,B)$ is said
to be a Higgs bundle.

Let $(E,h,B)$ be a Hermitian holomorphic bundle with Higgs field
$B$. Suppose that $D$ is the Chern connection with respect to $h$
such that $D=D'+d''$. If the connection $\D=D+B+B^\dag$ is flat,
then we say that $(E,h,B)$ is a harmonic Higgs bundle. The
integrability is equivalent to the following $tt^*$-equations:
\begin{align}
&(d'')^2=0,\;d''(B)=0,\;B\wedge B=0;\\
&(D')^2=0,\;D'(B^\dag)=0,\;B^\dag\wedge B^\dag=0;\\
&D'(B)=0,\;d''(B^\dag)=0,\;D'd''+d''D'=-(B\wedge B^\dag+B^\dag\wedge
B).
\end{align}
Here $B^\dag$ means the $h$-adjoint of $B$.
\end{df}

\begin{df}[real structure] Let $(E,h)$ be a Hermitian holomorphic
bundle, an anti-linear operator $\tau_\R:E\to E$ is called a real
structure if the following conditions hold:
\begin{enumerate}
\item $\tau_\R^2=I$;
\item $h(\tau_\R\cdot,\tau_\R\cdot)=\overline{h(\cdot,\cdot)}$;
\item $D(\tau_\R)=0$.
\end{enumerate}
Obviously $h$ is real on $E_\R=\ker(\tau_\R-I)$.
\end{df}

One can easily show the following relations:
\begin{align*}
D(\tau_\R)=0\Longleftrightarrow D'\tau_\R=\tau_\R
d''\Longleftrightarrow \tau_R D'=d''\tau_\R.
\end{align*}

Using the real structure $\tau_R$, one can define a symmetric
quadratic form $g$ on the bundle $E$:
\begin{equation}
g(u,v):=h(u,\tau_\R v).
\end{equation}
Since $D$ is compatible with $h$ and $\tau_\R$, it is easy to show
that $g$ is holomorphic and $D'$ is compatible with $\tau_\R$, i.e,
\begin{equation}
D(g)=0.
\end{equation}
Let $P$ be a bundle operator, then we can define $P^*=\tau_\R P^\dag
\tau_\R$, which is just the $g$-adjoint of $P$.

\begin{df} A harmonic Higgs bundle $(E,h,B)$ with a real structure $\tau_\R$ is called a real harmonic
Higgs bundle, and denoted by $(E,h,B,\tau_\R)$. The tuple
$(E,h,\tau_\R)$ is called a real Hermitian holomorphic bundle.
\end{df}

\section{Abstract variation Hodge structure and Period
domain}\label{App:2}

The context of this part is somehow standard and can be found in
many books, for example, \cite{CMP},\cite{Ku} and etc.

\subsubsection{Pure Hodge structure}\

\begin{df} A pure Hodge structure (HS) of weight $k$ is a pair $\{H_\Z,
F^\bullet,I\}$, where $H_Z$ is a lattice with finite rank, $I$ is an
involution acting on $H=H_\Z\otimes \C$ and $F^\bullet$ is a
decreasing filtration of $H$:
$$
0=F^k\subset\cdots\subset F^0=H_\C,
$$
such that
$$
H=F^p\oplus \tau(F^{k-p+1})
$$
for all $p\in \Z$.
\end{df}

The HS of weight $k$ defines a Hodge decomposition by:
$$
H^{p,q}=F^p\cap \tau(F^q),
$$
and
$$
F^p=\oplus_{r\ge p}H^{r,k-r}=H^{p,n-p}\oplus
H^{p+1,n-p-1}\oplus\cdots \oplus H^{k,0},
$$
i.e.,
$$
F^k=H^{k,0},\;F^{k-1}=H^{k,0}\oplus H^{k-1,1},\cdots,F^0=H.
$$

\begin{df} $h^{p,q}=\dim_\C H^{p,q}$ are called the Hodge numbers
with respect to the pure HS.
\end{df}

\subsubsection{Polarized HSs}\

\begin{df} A polarized HS of weight $k$ is an HS $\{H_Z, F^\bullet,
I\}$ together with a non-degenerate integral bilinear form,
$$
\eta: H_Z\times H_Z\to \Z,
$$
symmetric, if $k$ is even and antisymmetric, if $k$ is odd, and a
$\R$-linear transformation $C:H\to H$ (called the generalized Weil
operator) such that the Riemannian bilinear relations hold:
\begin{align*}
&\eta(F^p,F^{k-p+1})=0,\\
&\eta(C\psi,\overline{\psi})>0.
\end{align*}
The bilinear form $\eta$ is called the polarization of the HS
$F^\bullet$.
\end{df}

\begin{ex} The deformation of the compact complex manifolds provide
the pure Hodge structure.

\end{ex}

\subsubsection{Mixed Hodge structure}

\begin{df} A mixed Hodge structure (MHS) is a triple $(H_Z, W_\bullet,
F^\bullet)$, where $H_\Z$ is a lattice, $W^\Q_\bullet$ is a finite
increasing filtration on $H_\Q=H_\Z\otimes \Q\otimes \C$ and
$F^\bullet$ is a finite decreasing filtration on $H=H_\C$. Denote by
$W_\bullet=W^\Q_\bullet$. Then the two filtrations should induce a
HS of weight $k$ on the quotient $Gr_kH:=Gr^W_k H:=W_k /W_{k-1}$.
Here $F^lGr^W_k H$ is the image of $F^l\cap W_k $ in $Gr^W_k H$,
i.e.,
$$
F^lGr^W_k H=(W_k \cap F^l+W_{k-1})/W_{k-1} .
$$
$F^\bullet $ is called the Hodge filtration and $W_\bullet $ the
weight filtration.
\end{df}

The compatible condition of $W_\bullet $ and $F^\bullet$ is
equivalent to the decomposition for any $k,l$:
$$
Gr_k H=F^lGr_k H\oplus I(F^{k-l-1}Gr_k H).
$$

\begin{df} Let $H^{p,q}=F^pGr_{p+q}H$, then
$Gr_k=\oplus_{p+q}H^{p,q}$. The numbers $h^{p,q}$ are called the
Hodge numbers of the MHS ($W_\bullet, F^\bullet$).
\end{df}

If there are two MHSs on spaces $H$ and $H'$, then the space of
linear mappings $\hom(H,H')$ is endowed with a natural MHS,
\begin{align*}
&W_a\hom(H,H')_Q=\{\varphi:H_\Q\to H'_\Q|\varphi(W_k)\subset
W'_{k+a},\forall k\}\\
&F^b\hom(H,H')=\{\varphi: H\to H'|\varphi(F^l)\subset
(F')^{l+b},\forall l\}.
\end{align*}
It is easy to see that
$$
\hom(H,H')^{a,b}=\oplus \hom(H^{p,q}, (H')^{p+a,q+b}).
$$

A morphism $\varphi, H\to H'$ is called a MHS morphism of type
$(a,a)$ if
$$
\varphi(W_k)\subset W'_{k+2a},\;\varphi(F^l)\subset (F')^{l+a}.
$$
$\varphi$ is called an MHS morphism if $a=0$. An important property
of a MHS morphism is their compatibility with the two filtrations,
i.e.,
$$
\varphi(W_k)=W'_{k+2a}\cap \im \varphi, \varphi(F^l)=(F')^{l+a}\cap
\im \varphi(H).
$$
Therefore the morphism of MHS can be applied to the morphism of
complexes of MHS.

\subsubsection{Weight filtration of a nilpotent operator}\

Let $N:H\to H$ be a nilpotent operator on a vector space $H$,
$N^{k+1}=0$. Then there exists a unique increasing filtration
$W_\bullet$ on $H$:
$$
\cdots\subset W_{-1}\subset W_0\subset W_1\subset\cdots,
$$
such that
\begin{align*}
&(i)\;N(W_i)\subset W_{i-2},\\
&(ii)\;N^r: Gr_r \to Gr_{-r} \;\text{is an isomorphism for }\;r.
\end{align*}
This filtration is called the weight filtration of a nilpotent
operator $N$ (with center $0$).

The filtration can be constructed as follows. Take a basis $u_i$
such that $N$ has the Jordan normal form and $u_i$ satisfies
$$
\begin{pmatrix}
0& &\cdots&0\\
1&\ddots & &\vdots\\
\vdots&\ddots& & \\
0&\cdots&1 &0 \end{pmatrix},Nu_1=u_2,\cdots,Bu_{q-1}=u_q,Nu_q=0.
$$
Number basis for each block such that the numbers are symmetric
relative to $0$ and $N$ decreases the index by $2$, i.e.,
$Nu_i=u_{i-2}$,
\begin{itemize}
\item[(i)] if $q=2k+1$ is odd, then
$$
u_{2k},\cdots,u_2,u_0,u_{-2},\cdots,u_{-2k}.
$$
\item[(ii)] if $q=2k$ is even, then
$$
u_{2k-1},\cdots,u_{-1},u_1,\cdots,u_{-(2k-1)}.
$$
\end{itemize}

Now $W_k$ is defined as the subspace generated by vectors of weight
$\le k$.

We can shift the index of the weight filtration $W_\bullet$ such
that it is symmetric about the index $n$: define
$$
W[-n]_k:=W_{k-n}.
$$
So we get a new weight filtration $W[-n]$ satisfying:
$$
N^r: Gr_{n+r}\to Gr_{n-r}\;\text{is an isomorphism for any}\;r.
$$

Let $\underline{H}=H/\ker(N)$ and $\underline{N}$ is the restriction
of $N$ to $\underline{H}$. Then we have an induced weight filtration
by $\underline{N}$: $\overline{W}_\bullet$. On the other hand, we
can directly restrict the filtration $W_\bullet[-n] $ to
$\underline{H}$ to get $\overline{W_\bullet[-n]}$. It is easy to see
that the following relation holds:
\begin{equation}
\overline{W_\bullet[-n]}=\overline{W}[-n-1]_\bullet.
\end{equation}

\subsubsection{Variation of HS}\

\begin{df}A (real) variation of Hodge structure consists of a triple
$(H, \nabla, F^\bullet) $, where $H$ is a real vector space on a
complex base manifold $S$, $\nabla$ is a flat connection on $H$, and
$F^\bullet$ is a filtration of $H_\C:=H\otimes \C$ by holomorphic
subbundle which gives the Hodge structure in each fiber and which
satisfies the following \underline{Griffiths's transversality
condition}: the $(1,0)$ part $\nabla'$ of the connection $\nabla$
satisfies
\begin{equation}
\nabla'(F^l)\subset F^{l-1}\otimes \Omega^{1,0}(S).
\end{equation}
\end{df}

Here the HS on each fiber need not be obtained from an integral
lattice. However, in most geometrical cases the HS on each fiber
comes from a locally constant system of free $\Z$-modules.

\begin{df} A polarization of the VHS $H$ consists of an underlying
rational structure $E=E_\Q\otimes\R$ and a nondegenerate bilinear
form
$$
\eta: E_\Q\times E_\Q\to \underline{\Q},
$$
where $\underline{\Q}$ is a trivial $\Q$-valued local system on $S$.

This form $\eta$ must satisfies the following properties:
\begin{enumerate}
\item $\eta$ induces a polarization on each fiber;

\item $\eta$ is flat with respect to the flat connection
$\nabla$,i.e.,
$$
d\eta(s,s')=\eta(\nabla s,s')+\eta(s,\nabla s').
$$
\end{enumerate}
\end{df}

\subsubsection{Period mapping and period domain}

Now a polarized VHS defines a monodromy representation:
$$
\rho:\pi_1(S,o)\to G_\Q=\aut((E_0)_\Q,\eta),
$$
where $o$ is a base point. Its image
$$
\Gamma=\rho(\pi_1(S,o))\subset G_\Q
$$
is called the monodromy group.

\begin{df}
Giving an integer $k$, a lattice $H_\Z$ and a bilinear form $\eta$
(on the lattice) and positive integers $h^{p,q}$ (which is the
dimension of $H^{k,l}$). The space $D$ consists of all polarized
Hodge structures $H_\Z, H^{p,q},\eta$ with dimension $\dim
H^{p,q}=h^{p,q}$ is called the period domain.
\end{df}

The period domain can be constructed explicitly. The filtration
$F^\bullet $ with fixed dimension $f^j=\dim F^j$ is a flag variety
and forms a subvariety of a product of Grassmannian varieties with
the natural complex structure. If the filtration satisfies the first
Riemann relation:
$$
\eta(F^p,F^{k-p+1})=0,
$$
then only the half filtration is free: $F^k\subset \cdots \subset
F^v, v=[k+1/2]$. Therefore all such filtrations satisfying the first
Riemann relation forms a variety $\check{D}$ in the Grassmannian
manifold associated to the complexified space $H_\C$ and fixed
dimension of subspaces.

The second bilinear relation is a positive condition which shows
that all such filtrations satisfying the first and the second
Riemann relations form an open subvariety $D$ in $\check{D}$.

The variety $\check{D}$ and $D$ are actually homogeneous spaces. Let
$G_\C=\aut(H,\eta)$ be the group of linear automorphisms of the
space $H$ preserving the form $\eta$. $G_\C$ acts on $\check{D}$
transitively. If $B$ is an isotropic subgroup at a fixed point (a
fixed filtration), then
$$
\check{D}=G_\C/B.
$$
In a suitable basis, $B$ consists of the block-triangular matrices.
such a group is called a parabolic subgroup. It can be shown that
$\check{D}$ is smooth.

Take the subgroup $G_\R=\aut(H_\R, \eta)\subset G_\C$ and its
subgroup $K=B\cap G_\R$ ($K$ is compact, since the restriction to
the quotient spaces $F^p_*/F^{p+1}_*$ are orthogonal groups). We
have
$$
D=G_\R/K.
$$

The following proposition can be found in PP. 143 of \cite{CMP}.

\begin{prop} Let $D=G_\R/K$ be a period domain for weight $k$ Hodge
structures with Hodge numbers $h^{p,q},p=0,\cdots,k$.
\begin{enumerate}
\item If $k=2m+1$, the form $\eta$ is antisymmetric and then
$G=Sp(n,\R)$, where $n=\dim_C H_\C$, and $K=\prod_{p\le m}
U(h^{p,q})$.
\item If $k=2m$, the form $\eta$ is symmetric and $G=SO(s,t)$, where
$s$ is the sum of the Hodge numbers $h^{p,q}$ with $p$ even, while
$t$ is the sum of the remaining Hodge numbers. Here we have
$K=\prod_{p<k}U(h^{p,q}\times SO(h^{m,m}))$.
\end{enumerate}
\end{prop}

\begin{df} Let $(H, \nabla, F^\bullet)$ be a VHS over a complex
manifold $S$ (or analytic space). Then each fiber of the bundle $H$
gives a point in the period domain $D=G/K$. Hence this gives a
mapping
$$
\tilde{\PP}: S\to D.
$$
which is a multivalued mapping because of the action of the
monodromy group $\Gamma$. Since $\Gamma $ is discrete, we can push
down $\tilde{\PP}$ to get a singular valued mapping:
$$
\PP: S\to D/\Gamma.
$$
$\PP$ is called the period mapping. Here the space $D/\Gamma$
becomes a complex orbifold.
\end{df}

The following result is well-known.
\begin{thm}[Griffiths Transversally Theorem] Let $(H, \nabla,
F^\bullet)$ be a VHS. Then the subbundle $F^p$ and the period
mapping $\PP$ are holomorphic.
\end{thm}

\end{document}